%% file: cap.tex
\documentclass[journal]{IEEEtran}
\usepackage{enumerate}
\usepackage{amssymb,stmaryrd,amsmath,amsfonts,rotating}
\usepackage[noadjust]{cite} \usepackage{color} 
\usepackage[vflt]{floatflt} \usepackage{epic}
\usepackage{color}
\usepackage{verse}

\newcommand{\x}{\underline{x}}
\newcommand{\xBPdens}{\Ldens{x}^{\BPsmall}}

\newcommand{\y}{\underline{y}}
\newcommand{\z}{\underline{z}}
\newcommand{\vc}{\underline{v}}
 
\newcommand{\xunstable}{x_{\text{u}}(\epsilon)}
\newcommand{\xunstab}{x_{\text{u}}}
 
\newcommand{\xstable}{x_{\text{s}}(\epsilon)}
\newcommand{\xstab}{x_{\text{s}}}

\newcommand{\bavg}{\overline{\mathfrak{B}}}
\newcommand{\xavg}{\bar{x}}

\newcommand{\cwdldrLdelta}{b(\dl, \dr, \delta, w, \Lc)} 
\newcommand{\cwdldrLdeltanew}{b(\dl, \dr,\frac{2(\dl-1)(\dr-1)}{w}+\delta, w, \Lc)} 
\newcommand{\ent}{\ensuremath{{\tt{h}}}}
\newcommand{\entLE}{{\tilde{\ent}}}
\newcommand{\xLE}{\ensuremath{{\tilde{x}}}}

\newcommand{\Msat}{K}
\RequirePackage{bbm}
\definecolor{darkgreen}{rgb}{0.14,0.5,0.14}

\input{./definition}
\newdimen\arrayruleHwidth
\makeatletter
\def\Hline{\noalign{\ifnum0=`}\fi\hrule \@height \arrayruleHwidth
   \futurelet \@tempa\@xhline}
	\makeatother
\allowdisplaybreaks
\begin{document}
\title{Spatially Coupled Ensembles Universally Achieve Capacity under Belief Propagation} 
\author{\IEEEauthorblockN{Shrinivas
Kudekar\IEEEauthorrefmark{1}, Tom Richardson\IEEEauthorrefmark{1} and R{\"u}diger
Urbanke\IEEEauthorrefmark{2} \\ }
\IEEEauthorblockA{\IEEEauthorrefmark{1}Qualcomm, USA\\ Email: \{skudekar,tjr\}@qualcomm.com} \\
\IEEEauthorblockA{\IEEEauthorrefmark{2}School of Computer and Communication Sciences\\ EPFL, Lausanne, Switzerland\\ Email: ruediger.urbanke@epfl.ch}\\
 }


\maketitle
\begin{abstract}
We investigate spatially coupled code ensembles.  For transmission
over the binary erasure channel, it was recently shown that spatial
coupling increases the {\em belief propagation} threshold of the
ensemble to essentially the {\em maximum a-priori} threshold of the
underlying component ensemble.  This explains why convolutional
LDPC ensembles, originally introduced by Felstr{\"{o}}m and Zigangirov,
perform so well over this channel.

We show that the equivalent result holds true for transmission over
general binary-input memoryless output-symmetric channels.
More precisely, given a desired error probability and a gap to capacity,
we can construct a spatially coupled ensemble which fulfills these constraints
{\em universally} on this class of channels under belief propagation decoding.
In fact, most {\em codes} in that ensemble
have that property.  The quantifier {\em universal} refers to the
{\em single} ensemble/code which is good for all channels but we
assume that the channel is known at the receiver.

The key technical result is a proof that under belief propagation decoding spatially
coupled ensembles achieve essentially the {\em area threshold} of
the underlying uncoupled ensemble.

We conclude by discussing some interesting open problems.
\end{abstract}

\section{Introduction}
\subsection{Historical Perspective} Ever since the publication of
Shannon's seminal paper \cite{Sha48} and the introduction of the
first coding schemes by Hamming \cite{Ham50} and Golay \cite{Gol49},
coding theory has been concerned with finding low-delay and
low-complexity capacity-achieving schemes.  The interested reader
can find an excellent historical review in \cite{FoC07}. Let us
just briefly mention some of the highlights before focusing on those
parts that are the most relevant for our purpose.

In the first 50 years, coding theory focused on the construction
of {\em algebraic} coding schemes and algorithms that were capable
of exploiting the algebraic structure.  Two early highlights of
this line of research were the introduction of  Bose-Chaudhuri-Hocquenghem (BCH)
codes \cite{Hoc59,BoC60} as well as Reed-Solomon (RS) codes
\cite{ReS60}.  Berlekamp devised an efficient decoding algorithm
\cite{Ber84} and this algorithm was then interpreted by Massey as
an algorithm for finding the shortest feedback-shift register that
generates a given sequence \cite{Mas69}.  More recently, Sudan
introduced a list decoding algorithm for RS codes that decodes
beyond the guaranteed error-correcting radius \cite{Sud97allerton}.
Guruswami and Sudan improved upon this algorithm \cite{GuS99} and
Koetter and Vardy showed how to handle soft information \cite{KoV03}.

Another important branch started with the introduction of {\em
convolutional codes} \cite{Eli55} by Elias and the introduction of
the {\em sequential decoding} algorithm by Wozencraft \cite{Woz57}.
Viterbi introduced the {\em Viterbi algorithm} \cite{Vit67}. It was
shown to be optimal by Forney \cite{For67} and Omura \cite{Omu69}
and to be eminently {\em practical} by Heller \cite{Hel68,Hel69}.

An important development in transmission over the continuous input,
band-limited, additive white Gaussian noise channel was the invention
of the {\em lattice codes}. It was shown in \cite{Bu75, Bu89, Lo97,
UrR98, EZ04} that lattice codes achieve the Shannon capacity.  A
breakthrough in bandwidth-limited communications came about when
Ungerboeck \cite{Ung82,Ung87a,Ung87b} invented a technique to combine
coding and modulation. Ungerboeck's technique ushered in a new era
of fast modems.  The technique, called {\em trellis-coded modulation}
(TCM), offered significant coding gains without compromising bandwidth
efficiency by mapping binary code symbols, generated by a convolutional
encoder, to a larger (non-binary) signal constellation. In \cite{For88a,
For88b} Forney showed that lattice codes as well as TCM schemes may
be generated by the same basic elements and the generalized technique
was termed {\em coset-coding}.

Coming back to binary linear codes, in 1993, Berrou, Glavieux and
Thitimajshima \cite{BGT93} proposed {\em turbo} codes. These codes
attain near-Shannon limit performance under low-complexity iterative
decoding.  Their remarkable performance lead to a flurry of
research on the ``turbo'' principle.  Around the same time, Spielman
in his thesis \cite{Spi96}, \cite{SiS96} and MacKay and Neal in
\cite{mncEL1, mncEL2, mncN, MacKay_Neal_Codes:95}, independently
rediscovered low-density parity-check (LDPC) codes and iterative
decoding, both introduced in Gallager's remarkable thesis \cite{Gal63}.
Wiberg showed \cite{Wib96} that both turbo codes and LDPC codes
fall under the umbrella of {\em codes based on sparse graphs} and
that their iterative decoding algorithms are special cases of the
{\em sum-product} algorithm. This line of research was formalized
by Kschischang, Frey, and Loeliger who introduced the notion of
{\em factor graphs} \cite{KFL01}.

The next breakthrough in the design of codes (based on sparse graphs)
came with the idea of using {\em irregular} LDPC codes by Luby,
Mitzenmacher, Shokrollahi and Spielman \cite{LMSS01b}, \cite{LMSS98}.
With this added ingredient it became possible to construct irregular
LDPC codes that achieved performance within $0.0045$dB of the Shannon
limit when transmitting over the binary-input additive white Gaussian
noise channel, see Chung, Forney, Richardson and Urbanke \cite{CRU01}.
The development of these codes went hand in hand with the development
of a systematic framework for their analysis by Luby, Mitzenmacher,
Shokrollahi and Spielman \cite{LMSS01, LMSSS97} and Richardson and
Urbanke \cite{RiU01}.

A central research topic for codes on graphs is the interaction of
the graphical structure of a code and its performance.  Turbo codes
themselves are a prime example how the ``right'' structure is
important to achieve good performance \cite{BGT93}.
Further important parameters and structures are, the degree
distribution (dd) and in particular the fraction of degree-two variable
nodes, multi-edge ensembles \cite{RiU04b}, degree-two nodes in a
chain \cite{DHM98}, and protographs \cite{TAD04isit,DJDT05}.

Currently sparse graph codes and their associated iterative
decoding algorithms are the best ``practical'' codes in terms of
their trade-off between performance and complexity and they are
part of essentially all new communication standards.

Polar codes represent the most recent development in coding theory
\cite{Ari09}.  They are provably capacity achieving on binary-input
memoryless output-symmetric (BMS) channels (and many others) and
they have low decoding complexity.  They also have no error floor
due to a minimum distance which increases like the square root of
the blocklength.  The simplicity, elegance, and wide applicability
of polar codes have made them a popular choice in the recent
literature. There are perhaps only two areas in which polar codes
could be further improved.  First, for polar codes the convergence
of their performance to the asymptotic limit is slow.  Currently
no rigorous statements regarding this convergence for the general
case are known. But ``calculations'' suggest that, for a fixed
desired error probability, the required blocklength scales like
$1/\delta^{\mu}$, where $\delta$ is the additive gap to capacity
and where $\mu$ depends on the channel and has a value around
$4$, \cite{HAU10,KMTU10}. Note that random block codes under MAP
decoding have a similar scaling behavior but with $\mu=2$.  This
implies a considerably faster convergence to the asymptotic behavior.
The value $2$ is a lower bound for $\mu$ for any system since the
variations of the channel itself imply that $\mu \geq 2$.  The
second aspect is {\em universality}: the code design of polar codes
depends on the specific channel being used and one and the same
design cannot simultaneously achieve capacity over a non-trivial
class of channels (under successive cancellation decoding).

Let us now connect the content of this paper to the previous
discussion.  Our main aim is to explain the role of a further
structural element in the realm of sparse graph codes (besides the
previously discussed such examples), namely that of ``spatial
coupling.''  We will show that this coupling of graphs leads to a
remarkable change in their performance.  Ensembles designed in this
way combine some of the nice elements of polar codes (namely the
fact that they are provably capacity achieving under low complexity
decoding) with the practical advantages of sparse graph codes (the
codes are competitive already for moderate lengths). Perhaps most
importantly, it is possible to construct {\em universal} such codes
for the whole class of BMS channels.  Here, universality refers to
the fact that one and the same ensemble is good for a whole class
of channels, assuming that at the receiver we have knowledge of the
channel.

\subsection{Prior Work on Spatially Coupled Codes}\label{sec:priorwork}
The potential of spatially coupled codes has long been recognized.
Our contribution lies therefore not in the introduction of a new
coding scheme, but in clarifying the mechanism that make these
codes perform so well.

The term {\em spatially coupled codes} was coined in \cite{KRU10}.
Convolutional LDPC codes (more precisely, terminated convolutional
LDPC codes), which were introduced by Felstr{\"{o}}m and Zigangirov in \cite{FeZ99},
and their many variants belong to this class.  Why do
we introduce a new term?  The three perhaps most important reasons
are: (i) the term ``convolutional'' conjures up a fairly specific node
interconnection structure whereas experiments have shown that the
particular nature of the connection is not important and that
the threshold saturation effect occurs as soon as the connection is sufficiently strong;
(ii) a well known result for convolutional codes says that the
boundary conditions are ``forgotten'' exponentially fast; but for
spatially coupled codes it is exactly the boundary condition which
causes the effect and there is no decay of this effect in the spatial
dimension of the code; (iii) the same effect has (empirically) been
shown to hold in many other graphical models, most of them outside
the realm of coding; the term ``spatial coupling'' is perhaps then
somewhat more generally applicable.

There is a considerable literature on convolutional-like LDPC
ensembles.  Variations on the constructions as well as some analysis
can be found in Engdahl and Zigangirov \cite{EnZ99}, Engdahl,
Lentmaier, and Zigangirov \cite{ELZ99}, Lentmaier, Truhachev, and
Zigangirov \cite{LTZ01}, as well as Tanner, D. Sridhara, A. Sridharan,
Fuja, and Costello \cite{TSSFC04}.  

In \cite{SLCZ04, LSZC10}, Sridharan, Lentmaier, Costello and
Zigangirov consider density evolution (DE) analysis for convolutional LDPC
ensembles and determine thresholds for the BEC. The equivalent
results for general channels were reported by Lentmaier, Sridharan,
Zigangirov and Costello in \cite{LSZC05, LSZC10}.  This DE analysis
is in many ways the starting point for our investigation.  By
comparing the thresholds to the thresholds of the underlying
ensembles under MAP decoding (see e.g.  \cite{RiU08}), it quickly
becomes apparent that an interesting effect must be at work.  Indeed,
in a recent paper \cite{LeF10}, Lentmaier and Fettweis followed
this route and independently formulated the equality of the belief
propagation (BP) threshold of convolutional LDPC ensembles and the
MAP threshold of the underlying ensemble as a conjecture.

A representation of convolutional LDPC ensembles in terms of a
protograph was introduced by Mitchell, Pusane, Zigangirov and
Costello \cite{MPZC08}. The corresponding representation for
terminated convolutional LDPC ensembles was introduced by Lentmaier,
Fettweis, Zigangirov and Costello \cite{LFZC09}. A variety of
constructions of LDPC convolutional codes from the graph-cover
perspective is shown by Pusane, Smarandache, Vontobel, and Costello
\cite{PSVC11}.

A pseudo-codeword analysis of convolutional LDPC codes was performed
by Smarandache, Pusane, Vontobel, and Costello in
\cite{SPVC06,SPVC09,PSVC11}. Such an analysis is important if we
want to understand the error-floor behavior of spatially coupled
ensembles.  

In \cite{PISWC10}, Papaleo, Iyengar, Siegel, Wolf, and Corazza study
the performance of windowed decoding of convolutional LDPC codes
on the BEC.  Such a decoder has a decoding complexity which is
independent of the chain length, an important practical advantage.
Luckily, it turns out that the performance under windowed decoding,
when measured in terms of the threshold, approaches the ``regular""
(without windowed decoding) threshold exponentially fast in the
window size, see \cite{IPSWVC10,ISUW11}. The threshold saturation
phenomenon therefore does not require an infinite window size.

The scaling behavior of spatially coupled ensembles, i.e., the relationship
between the chain length, the number of variables per section, and
the error probability is discussed by Olmos and Urbanke in \cite{OlU11}.

\subsection{Prior Results for the Binary Erasure Channel} It was
recently shown in \cite{KRU10} that for transmission over the BEC
spatially coupled ensembles have a BP threshold which is essentially
equal to the MAP threshold of the underlying uncoupled ensemble.
Further, this threshold is also essentially equal to the MAP threshold
of the coupled ensemble.  This phenomena was called {\em threshold
saturation} in \cite{KRU10} since the BP threshold takes on its
largest possible value (the MAP threshold).  This significant
improvement in the performance is due to the spatial coupling of
the underlying code. Those ``sections'' of the code that have already
succeeded in decoding can help their neighboring less fortunate
sections in the decoding process. In this manner, the information
propagates from the ``boundaries'', where the bits are known perfectly
towards the ``middle''.  In a recent paper \cite{LeF10}, Lentmaier
and Fettweis independently formulated the same statement as a
conjecture and provided numerical evidence for its validity. They
attribute the observation of the equality of the two thresholds to
G. Liva.

It was shown in \cite{MPZC08,DDJ06,SPVC09,SPVC06} that if we couple
component codes whose Hamming distance grows linearly in the
blocklength then also the resulting coupled ensembles have this
property (assuming that the number of ``sections'' or copies of the
underlying code is kept fixed). The equivalent result is true for
stopping sets.  This implies that for the transmission over the BEC
the block BP threshold is equal to the bit BP threshold and that
such ensembles do not exhibit error floors under BP decoding.

\subsection{Prior Results for General Binary-Input Memoryless Output-Symmetric
Channels} As pointed out in a preceding section, BP thresholds for
transmission over general BMS channels were computed by means of a
numerical procedure by Lentmaier, Sridharan, Zigangirov and Costello
in \cite{LSZC05}. Further, in \cite{MMRU09} (conjectured) MAP
thresholds for some LDPC ensembles were computed according to the
Maxwell construction. Comparing these two values, one can check
empirically that also for transmission over general BMS channels
the BP threshold of the coupled ensembles is essentially equal to
the (conjectured) MAP threshold of the underlying ensemble.  Indeed,
recently both \cite{KMRU10} as well as \cite{LMFC10} provided further
numerical evidence that the threshold saturation phenomenon also
applies to general BMS channels.  

For typical sparse graph ensembles the MAP threshold is not equal
to the Shannon threshold but the Shannon threshold can only be
reached by taking a sequence of such ensembles (e.g., a sequence
of increasing degrees). There are some notable exceptions, like MN
ensembles or HA ensembles.  Kasai and Sakaniwa take this as a
starting point to investigate in \cite{KaS11} whether by spatially
coupling such ensembles it is possible to create ensembles which
are universally capacity achieving under BP decoding.

\subsection{Spatial Coupling for General Communication Scenarios,
Signal Processing, Computer Science, and Statistical Physics} The
principle which underlies the good performance of spatially coupled
ensembles is broad. It has been shown to apply to a variety of
problems in communications, computer science, signal processing,
and physics.  To mention some concrete examples, the threshold
saturation effect (dynamical/algorithmic threshold of the system
being equal to the static or condensation threshold) of coupled
graphical models has been observed for rate-less codes by Aref and
Urbanke \cite{ArU11}, for channels with memory and multiple access
channels with erasure by Kudekar and Kasai \cite{KuKa11a,KuKa11b},
for CDMA channels by Takeuchi, Tanaka, and Kawabata \cite{TTK11},
for relay channels with erasure by Uchikawa, Kasai, and Sakaniwa
\cite{UKS11}, for the noisy Slepian-Wolf problem by Yedla, Pfister,
and Narayanan \cite{YPN11}, and for the BEC wiretap channel by
Rathi, Urbanke, Andersson, and Skoglund \cite{RUAS10}.  Uchikawa,
Kurkoski, Kasai, and Sakaniwa recently showed an improvement of the
BP threshold has also for transmission over the unconstrained AWGN
channel using low-density lattice codes \cite{UKKS11}. Further,
Yedla, Nguyen, Pfister and Narayanan, demonstrated the universality
of spatially-coupled codes in the 2-user binary input Gaussian
multiple-access channel and finite state ISI channels like the
dicode-erasure channel and the dicode channel with AWGN \cite{YNPN11,
NYPN11}. In \cite{YNPN11} they show in addition that for a fixed
rate pair, spatially-coupled ensembles universally saturate the
achievable region (i.e., the set of channel gain parameters that
are achievable for the fixed rate pair) under BP decoding.  Similarly,
in \cite{NYPN11} they provide numerical evidence that spatially
coupled ensembles achieve the symmetric information rate for the
dicode erasure channel and the dicode channel with AWGN.

In signal processing and computer science spatial coupling has found
success in the field of compressed sensing \cite{KP10,SKS11,KMSSZ11,DJM11}.
In \cite{KP10}, Kudekar and Pfister use sparse measurement matrices
with sub-optimal verification decoding and show that spatial coupling
boosts thresholds of sparse recovery. In \cite{KMSSZ11, DJM11},
Krzakala, M{\'e}zard, Sausset, Sun, and Zdeborova as well as Donoho,
Javanmard, and Montanari show that by carefully designing dense
measurement matrices using spatial coupling one can achieve the
best possible recovery threshold, i.e., the one achieved by the
optimal $\ell_0$ decoder.  Thus, the phenomena of threshold saturation
is also demonstrated in this case. This development is quite remarkable.

Statistical physics is another very natural area in which the
threshold saturation phenomenon is of interest.  For the so-called
random $K$-SAT problem, random graph coloring, and the Curie-Weiss
model, spatially coupled ensembles were investigated by Hassani,
Macris, and Urbanke, \cite{HMU10,HMU11a,HMU11b}.  In all these
cases, the threshold saturation phenomenon was observed.  This
suggests that it might be possible to study difficult theoretical
problems in this area, like the existence of the static threshold,
by studying the dynamical threshold of a chain of coupled models,
perhaps an easier problem.  Further spatially-coupled models were
considered by Takeuchi and Tanaka \cite{TTK11b}.

\subsection{Main Results and Consequences} In this paper we show that for
transmission over general BMS channels coupled ensembles exhibit the threshold
saturation phenomenon.  By choosing e.g. regular component ensembles of fixed
rate and increasing degree, this implies that coupled ensembles can achieve
capacity over this class of channels.  More precisely, for each $\delta>0$
there exists a coupled ensemble which achieves at least a fraction $1-\delta$
of capacity {\em universally}, under belief propagation decoding, over the whole class
of BMS channels.  The qualifier "universal" is important here.

Coupled ensembles inherit to a large degree the error floor behavior
of the underlying ensemble.  Further, such an ensemble can be chosen
so that it has a non-zero error correcting radius, and hence does
not exhibit error floors. To achieve this, it suffices to take the
variable-node degree to be at least five. This guarantees that a
randomly chosen graph from such an ensemble is an expander with
expansion exceeding three-quarters with high probability. This
expansion guarantees an error correcting radius under the
so-called flipping decoder \cite{Spi95} as well as under the BP
decoder, assuming that we suitably clip both the received as well
as the internal messages \cite{BuM00}.

Although one can empirically observe the threshold saturation
phenomenon for a wide array of component codes, we state and prove
the main result only for regular LDPC ensembles. This keeps the
exposition manageable.

\subsection{Outline} In Section~\ref{sec:review} we briefly review
regular LDPC ensembles and their asymptotic (in the blocklength)
analysis.  Much of this material is standard and we only include
it here to set the notation and to make the paper largely self-contained.
The two most important exceptions are our in-depth discussion of
the Wasserstein distance and the the so-called area threshold, in
particular the (Negativity) Lemma~\ref{lem:asymptoticnegativity}.

In Section~\ref{sec:coupled} we review some basic properties of
coupled ensembles.  Using simple extremes of information combining
techniques, we will see in Section~\ref{sec:firstresult} that
coupling indeed increases the BP threshold significantly, even
though these simple arguments are not sufficient to characterize
the BP threshold under coupling exactly.

We state our main result, namely that the BP threshold of coupled
ensembles is essentially equal to the area threshold of the underlying
component ensemble, in Section~\ref{sec:main}. We also discuss how
one can easily strengthen this result to apply to individual codes
rather than ensembles and how this gives rise to codes which are
universally close to capacity under BP decoding for the whole class
of BMS channels.

We end in Section~\ref{sec:conclusion} with a discussion of what
challenges still lie ahead.  In particular, spatial coupling has
been shown empirically to lead to the threshold saturation phenomenon
in a wide class of graphical models. Rather than proving each such
scenario in isolation, we want a common framework to
analyze all such systems.

Many of the proofs are relegated to the appendices.  This makes it
possible to read the material on two levels -- a casual level,
skipping all the proofs and following only the flow of the argument,
and a more detailed level, consulting the material in the appendices.

\section{Uncoupled Systems}\label{sec:review}
\subsection{Regular Ensembles} 
\begin{definition}[$(\dl, \dr)$-Regular Ensemble]
Fix $3 \leq \dl \leq \dr$, $\dl, \dr \in \naturals$, and $n$ so that
$n \dl/\dr \in \naturals$.  The $(\dl, \dr)$-regular LDPC ensemble
of blocklength $n$ is defined as follows. There are $n$ {\em variable}
nodes and $n \frac{\dl}{\dr}$ {\em check} nodes. Each variable node
has degree $\dl$ and each check node has degree $\dr$. Accordingly,
each variable node has $\dl$ {\em sockets}, i.e., $\dl$ places to
connect an edge to, and each check node has $\dr$ sockets.  Therefore,
there are in total $\dl n$ variable-node sockets and the same number of
check-node sockets. Number both kinds from $1$ to $n \dl$.  Consider the
set of permutations $\Pi$ on $\{1, \dots, n \dl\}$. Endow this set with a
uniform probability distribution.  To sample from the $(\dl,
\dr)$-regular ensemble, sample from $\Pi$ and connect the variable
to the check node sockets according to the chosen permutation. This
is the {\em configuration model} of LDPC ensembles. It is inspired by
the configuration model of random graphs \cite[Section 2.4]{Bol01}.
\qed
\end{definition}

\subsection{Binary-Input Memoryless Output-Symmetric
Channels}\label{sec:bmsc} Throughout we will assume that transmission
is taking place over a BMS channel. Let $X$ denote the input and let $Y$
be the output. Further, let $p(Y=y \mid X=x)$ denote the {\em transition
probability} describing the channel. An alternative characterization
of the channel is by means of its so-called $L$-distribution, denote
it by $\Ldens{c}$.  More precisely, $\Ldens{c}$ is the distribution of
\begin{align*} \ln \frac{p(Y \mid X=1)}{p(Y \mid X=-1)} \end{align*}
conditioned that $X=1$.

Given $\Ldens{c}$, we write $\Ddens{c}$, $\absDdens{c}$, and
$\absDdist{c}$ to denote the corresponding $D$ distribution, the
$|D|$ distribution and the cdf in the $|D|$-domain, respectively, 
see \cite[Section~4.1.4]{RiU08}.

Typically we do not consider a single channel in isolation but a
whole {\em family} of channels. We write $\{ \BMS(\sigma)\}$ to
denote the family parameterized by the scalar $\sigma$. Often it
will be more convenient to denote this family by $\{\Ldens{c}_\sigma\}$,
i.e., to use the family of $L$-densities which characterize the
channel family.  If it is important to make the range of the parameter
$\sigma$ explicit, we will write
$\{\Ldens{c}_\sigma\}_{\underline{\sigma}}^{\overline{\sigma}}$.

Sometimes it is convenient to use the {\em natural} parameter of
the family.  For example, for the three fundamental channels, the
BEC, the binary symmetric channel (BSC) and the binary additive
white-Gaussian noise channel (BAWGNC), the corresponding channel
families are given by $\{\text{BEC}(\epsilon)\}_{0}^{1}$,
$\{\text{BSC}(p)\}_{0}^{\frac12}$, and
$\{\text{BAWGNC}(\sigma)\}_{0}^{\infty}$.  Other times, it is more
convenient to use a common parameterization.  E.g., we will write
$\{\BMS(\ent)\}$ to denote a channel family where $\BMS(\ent)$
denotes the element in the family of {\em entropy} $\ent$.

Assume that we are given a channel family
$\{\BMS(\sigma)\}_{\underline{\sigma}}^{\overline{\sigma}}$.  We
say that the family is {\em complete} if
$\entropy(\BMS(\underline{\sigma}))=0$,
$\entropy(\BMS(\overline{\sigma}))=1$, and for each $\ent \in [0,
1]$ there exists a parameter $\sigma$ so that
$\entropy(\BMS(\sigma))=\ent$.  Here $\entropy(\cdot)$ is the entropy
functional defined in Section~\ref{sec:BPandDE}.

Let $p_{Z \mid X}(z\mid x)$ denote the transition probability
associated to a BMS channel $\Ldens{c}'$ and let $p_{Y \mid X}(y
\mid x)$ denote the transition probability of another BMS channel
$\Ldens{c}$.  We then say that $\Ldens{c}'$ is {\em degraded} with
respect to $\Ldens{c}$ if there exists a channel $p_{Z \mid Y}(z\mid
y)$ so that \begin{align*} p_{Z \mid X}(z \mid x) = \sum_{y} p_{Y \mid X}(y\mid
x) p_{Z \mid Y}(z\mid y).  \end{align*} We will use the notation
$\Ldens{c} \prec \Ldens{c}'$ to denote that $\Ldens{c}'$ is degraded
wrt $\Ldens{c}$ (as a mnemonic think of $\Ldens{c}$ as the erasure
probability of a BEC and replace $\prec$ with $<$).

A useful characterization of degradation, see \cite[Theorem
4.74]{RiU08}, is that $\Ldens{c} \prec \Ldens{c}'$ is equivalent
to 
\begin{align}\label{equ:degradation} \int_0^1 f(x) \absDdens{c}(x)
\,\dee x \leq  \int_0^1 f(x) \absDdens{c'}(x) \,\dee x 
\end{align}
for all $f(x)$ that are non-increasing and concave on $[0,1]$.
Here, $\absDdens{c}(x)$ is the so called $|D|$-density associated
to the $L$-density $\Ldens{c}$, see \cite[p. 179]{RiU08}.  In
particular, this characterization implies that $F(\Ldens{a})\leq
F(\Ldens{b})$ for $\Ldens{a} \prec \Ldens{b}$ if $F(\cdot)$ is
either the Battacharyya or the entropy functional.  This is true
since both are linear functionals of the distributions and their
respective kernels in the $|D|$-domain are decreasing and concave.
An alternative characterization in terms of the cumulative distribution
functions $\absDdist{c}(x)$ and $\absDdist{c'}(x)$ is that for all
$z \in [0, 1]$, 
\begin{align}\label{equ:degradationcdfs} \int_z^1
\absDdist{c}(x) \dee x \leq  \int_z^1 \absDdist{c'}(x) \,\dee x.
\end{align}

A BMS channel family $\{\BMS(\sigma)\}_{\underline{\sigma}}^{\overline{\sigma}}$ is said to be {\em ordered} by
degradation if $\sigma_1 \leq \sigma_2$ implies $\Ldens{c}_{\sigma_1}
\prec \Ldens{c}_{\sigma_2}$.
(The reverse order, $\sigma_1 \geq \sigma_2,$ is also allowed but we
generally stick to the stated convention.)

We say that an $L$-density $\Ldens{c}$ is {\em symmetric} if
$\Ldens{a}(-y) = \Ldens{a}(y) e^{-y}$. We recall that all densities
which stem from BMS channels are symmetric, see \cite[Sections
4.1.4, 4.1.8 and 4.1.9]{RiU08}. All densities which we consider are symmetric.  We
will therefore not mention symmetry explicitly in the sequel.

A BMS channel family $\{\Ldens{c}_\sigma\}$ is said to be {\em
smooth} if for all continuously differentiable functions $f(y)$ so
that $e^{y/2} f(y)$ is bounded, the integral $\int f(y)
\Ldens{c}_{\sigma}(y) \,\dee y $ exists and is a continuously
differentiable function with respect to $\sigma$, see \cite[Definition~4.32]{RiU08}.

The three fundamental channel families 
$\{\text{BEC}(\epsilon)\}_{0}^{1}$,
$\{\text{BSC}(p)\}_{0}^{\frac12}$, and
$\{\text{BAWGNC}(\sigma)\}_{0}^{\infty}$
are all complete, ordered, smooth, and symmetric.

\subsection{MAP Decoder and MAP Threshold}
The bit {\em maximum a posteriori} (bit-MAP) decoder for bit $i$ finds
the value of $x_i$ which maximizes $p(x_i \mid y_1^n)$. It minimizes
the bit error probability and is optimal in this sense.  The block {\em
maximum a posteriori} (block-MAP) decoder finds the codeword $x_1^n$
which maximizes $p(x_1^n \mid y_1^n)$. It minimizes the block error
probability and is optimal in this sense.

\begin{definition}[MAP Threshold]\label{def:mapthreshold}
Consider an ordered and complete channel family $\{\Ldens{c}_\ent\}$.
The {\em MAP threshold} of the $(\dl, \dr)$-regular ensemble for this channel
family is denoted by $\ent^{\MAPsmall}(\dl,\!\dr\!)$ and defined by
\begin{align*}
&  \inf\{\ent\in [0,1]: \liminf_{n\to \infty}
\mathbb{E}[\entropy(X_1^n\mid Y_1^n(\ent))/n] \!>\! 0\},
\end{align*}
where $\entropy(X_1^n\mid Y_1^n(\ent))$ is the conditional entropy of the
transmitted codeword $X_1^n$, chosen uniformly at random from the code, given the received message $Y_1^n(\ent)$
and where the expectation $\mathbb{E}[\cdot]$ is wrt the $(\dl, \dr)$-regular 
ensemble.  \qed
\end{definition}
{\em Discussion}:
Define $\text{P}_{e, i}=\text{Pr}\{X_i \neq \hat{X}_i(Y_1^n)\}$, where $\hat{X}_i(Y_1^n)$ is the MAP estimate of bit $i$ 
based on the observation $Y_1^n$. Note that by
the Fano inequality we have $\entropy(X_i \mid Y_1^n) \leq h_2(P_{e, i})$.
Assume that we are transmitting above $\ent^{\MAPsmall}(\dl,\!\dr\!)$ so that
$\mathbb{E}[\entropy(X_1^n\mid Y_1^n)/n] \geq \delta >0$.\footnote{We have 
$\mathbb{E}[\entropy(X_1^n\mid Y_1^n)/n]  \geq \frac12 \liminf_{n\to \infty}
\frac1n\mathbb{E}[\entropy(X_1^n\mid Y_1^n(\ent))] >0$
for all $n>n_0$, lets say. 
Further, for $1 \leq n \leq n_0$, $\mathbb{E}[\entropy(X_1^n\mid Y_1^n)/n]$ is strictly
positive unless the channel is trivial. The claim follows by taking the minimum of all of the bounds
for $1 \leq n \leq n_0$ as well as the bound for $n > n_0$.}
Then
\begin{align*}
h_2(\mathbb{E}[ \frac{1}{n} \sum_{i=1}^{n} \text{P}_{e, i}]) 
& \!\geq \!\mathbb{E}[ \frac{1}{n} \sum_{i=1}^{n} h_2(\text{P}_{e, i})] 
\geq \mathbb{E}[\sum_{i=1}^n \entropy(X_i \mid Y_1^n)/n] \\
& \geq \mathbb{E}[\entropy(X_1^n\mid Y_1^n)/n] \geq \delta > 0. 
\end{align*}
In words, if we are transmitting {\em above} the MAP threshold,
then the ensemble average bit-error probability is lower bounded
by $h_2^{-1}(\delta)$, a strictly positive constant. This ensemble
is therefore not suitable for reliable transmission above this
threshold.

In general we cannot conclude from $\mathbb{E}[\entropy(X_1^n \mid
Y_1^n)/n] \leq \delta$ that the average error probability is
small.\footnote{This is possible if we have the slightly stronger
condition $\mathbb{E}[\sum_{i=1}^n \entropy(X_i \mid Y_1^n)/n] \leq
\delta$.  In this case $\delta \geq \frac1{n}\mathbb{E}[\sum_{i=1}^n
\entropy(X_i \mid Y_1^n)]  = \frac1{n}\mathbb{E}[\sum_{i=1}^n
\mathbb{E}_{Y_1^n}[h_2(\min_x p(x \mid Y_1^n))]]  \geq
\frac1{n}\mathbb{E}[\sum_{i=1}^n \mathbb{E}_{Y_1^n}[2 \min_x
p(x \mid Y_1^n)]]  = \frac1{n}\mathbb{E}[\sum_{i=1}^n 2 
\text{P}_{e, i}]$, so that $\frac1{n}\mathbb{E}[\sum_{i=1}^n
\text{P}_{e, i}] \leq \frac12 \delta$. The last step in the previous chain of
inequalities follows since under MAP decoding the error probability conditioned
that we observed $y_1^n$ is equal to $\min_x p(x \mid y_1^n)$. An alternative way to
prove this is to realize that $\entropy(X_i \mid Y_1^n)$ represents a BMS channel 
with a particular entropy and to use extremes of information combining to find the
worst error probability such a channel can have. The extremal channel in this 
case is the BEC.}

\subsection{Belief Propagation, Density Evolution, and Some Important Functionals}\label{sec:BPandDE}
In principle one can investigate the behavior of coupled ensembles
under any message-passing algorithm. We limit our investigation
to the analysis of the BP decoder, the most powerful local
message-passing algorithm.  We are interested in the asymptotic
performance of the BP decoder,  i.e., the performance when the
blocklength $n$ tends to infinity.  This asymptotic performance is
characterized by the so-called density evolution (DE) equation
\cite{RiU01}.

\begin{definition}[Density Evolution]\label{def:de}
For $\ell \geq 1$, the DE equation for a
$(\dl, \dr)$-regular ensemble is given by
$$
\Ldens{x}_{\ell} = \Ldens{c} \vconv (\Ldens{x}^{\cconv \dr-1}_{\ell-1})^{\vconv \dl-1}.
$$
Here, $\Ldens{c}$ is the $L$-density of the BMS channel over which
transmission takes place and $\Ldens{x}_{\ell}$ is the density emitted by
variable nodes in the $\ell$-th round of density evolution. Initially we
have $\Ldens{x}_{0}=\Delta_0$, the delta function at $0$. The operators
$\vconv$ and $\cconv$ correspond to the convolution of densities at
variable and check nodes, respectively, see \cite[Section 4.1.4]{RiU08}.
\qed
\end{definition}
As mentioned, all distributions associated to BMS channels are
symmetric and symmetry is preserved under DE, see \cite[Chapter
4]{RiU08} for details.
There are a number of functionals of densities are of interest to us.
The
most important functionals are the Battacharyya, the entropy, and
the error probability functional.
For a density $\Ldens{a}$ these are denoted by $\batta(\Ldens{a})$,
$\entropy(\Ldens{a})$, and $\perr(\Ldens{a})$, respectively.
Assuming $\Ldens{a}$ is an $L$-density,
they are given by
\begin{align*}
\batta(\Ldens{a}) & = \int \Ldens{a}(y) e^{-y/2} \,\dee y, \;\;
\entropy(\Ldens{a})  = \int \Ldens{a}(y) \log_2(1\!+\!e^{-y}) \,\dee y, \\
\perr(\Ldens{a}) & = \frac12\int \Ldens{a}(y) e^{-(y/2+\vert y/2\vert)} \,\dee y.
\end{align*}

We end this section with the following useful fact. The proof can be found in 
Appendix~\ref{sec:someusefulfacts}.
\begin{lemma}[Entropy versus Battacharyya]\label{lem:entropyvsbatta} For
any $L$-density $\Ldens{a}$, 
$\batta^2(\Ldens{a}) \leq \entropy(\Ldens{a}) \leq \batta(\Ldens{a})$.  
\end{lemma}

\subsection{Extremes of Information Combining and
the Duality Rule}\label{sec:extremesofinfocombining} When we are operating
on BMS channels, the quantities appearing in the DE equations are
distributions. These are hard to track analytically in general,
unless we are transmitting over the BEC. Often we only need bounds.
In these cases {\em extremes of information combining} ideas are
handy, see \cite{HuH02,HuH03,LHHH03,SSZ03,SSZ05}, \cite[p. 242]{RiU08}.

\blemma[Extremes of Information Combining]\label{lem:extremes}
Let $F(\cdot)$ denote either $\entropy(\cdot)$ or $\batta(\cdot)$ and let $\alpha \in [0, 1]$.
Let $\Ldens{a}_\BECsmall$ and $\Ldens{a}_\BSC$ denote $L$-densities from
the families $\{\BEC(\epsilon)\}$ and $\{\BSC(p)\}$, respectively,
so that $F(\Ldens{a}_\BECsmall)= F(\Ldens{a}_\BSC)=\alpha$.
Then for any $\Ldens{b}$,
\begin{enumerate}[(i)]
\item $\min_{\Ldens{a}: F(\Ldens{a})=\alpha} F(\Ldens{a} \vconv \Ldens{b})= F(\Ldens{a}_{\BECsmall} \vconv \Ldens{b})$ \label{lem:extremesminvconv}
\item $\max_{\Ldens{a}: F(\Ldens{a})=\alpha} F(\Ldens{a} \vconv \Ldens{b})= F(\Ldens{a}_{\BSCsmall} \vconv \Ldens{b})$ \label{lem:extremesmaxvconv}
\item $\min_{\Ldens{a}: F(\Ldens{a})=\alpha} F(\Ldens{a} \cconv \Ldens{b})= F(\Ldens{a}_{\BSCsmall} \cconv \Ldens{b})$ \label{lem:extremesmincconv}
\item $\max_{\Ldens{a}: F(\Ldens{a})=\alpha} F(\Ldens{a} \cconv \Ldens{b})= F(\Ldens{a}_{\BECsmall} \cconv \Ldens{b})$ \label{lem:extremesmaxcconv}
\end{enumerate}
\elemma
{\em Discussion:} Although the extremes of information combining
bounds are only stated for pairs of distributions, they naturally
extend to more than two distributions.  E.g., we claim that
$\min_{\Ldens{a}: F(\Ldens{a})=\alpha} F(\Ldens{a}^{\vconv
d})= F(\Ldens{a}_{\BECsmall})^d=\alpha^d$. To see this, let
$\{\Ldens{a}_i\}_{i=1}^{d}$ be any set of distributions so that
$F(\Ldens{a}_i)=\alpha$. Then we can use
Lemma~\ref{lem:extremes} repeatedly to conclude that
\begin{align*}
F(\Ldens{a}_1 \vconv (\vconv_{i=2}^{d} \Ldens{a}_i) ) 
& \geq F(\Ldens{a}_{\BECsmall} \vconv (\vconv_{i =2}^{d} \Ldens{a}_i) ) \\ 
& = F(\Ldens{a}_2 \vconv (\Ldens{a}_{\BECsmall} \vconv (\vconv_{i =3}^{d} \Ldens{a}_i) ) \\ 
& \geq F(\Ldens{a}_\BECsmall \vconv (\Ldens{a}_{\BECsmall} \vconv (\vconv_{i =3}^{d} \Ldens{a}_i) ) \\ 
& = \cdots \\
& \geq F(\Ldens{a}_d \vconv (\Ldens{a}_{\BECsmall}^{\vconv d-1})) \\ 
& \geq F(\Ldens{a}_\BECsmall \vconv (\Ldens{a}_{\BECsmall}^{\vconv d-1} )) = \alpha ^d.
\end{align*}
The same remark and the same proof technique applies to the other cases.

\blemma[Duality Rule -- \protect{\cite[p. 196]{RiU08}}]\label{lem:dualityrule}
For any $\Ldens{a}$ and $\Ldens{b}$ 
$\entropy(\Ldens{a} \vconv \Ldens{b})+
\entropy(\Ldens{a} \cconv \Ldens{b})  =
\entropy(\Ldens{a})+
\entropy(\Ldens{b})$.
\elemma
{\em Note:} We give a simple proof of this identity at the end of the proof of 
Lemma~\ref{lem:entropyofcheck}.

\subsection{Fixed Points, Convergence, and BP Threshold}\label{sec:fps}
We say that the density $\Ldens{x}$ is a {\em fixed point} (FP) of DE for 
the $(\dl, \dr)$-regular ensemble and the channel
$\Ldens{c}$ if
\begin{align}\label{eq:regDE}
\Ldens{x} = \Ldens{c} \vconv (\Ldens{x}^{\cconv \dr-1})^{\vconv \dl-1}.
\end{align}
More succinctly, when the underlying ensemble is understood from the
context, we say that $(\Ldens{c}, \Ldens{x})$ is a FP.

One way to generate a FP is to initialize $\Ldens{x}_{0}$ with $\Delta_0$
and to run DE, as stated in Definition~\ref{def:de}.  We call such a
FP a FP of {\em forward} DE. The resulting FPs are the ``natural'' FPs
since they have a natural operational meaning -- if we pick sufficiently
long ensembles, these are the FPs which we can observe in simulations
when we run the BP decoder. 

\begin{definition}[Weak Convergence]
We say that a sequence of distributions $\{\Ldens{a}_i\}$ converges weakly to a limit
distribution $\Ldens{a}$ if for the corresponding cumulative distributions in
the $|D|$-domain, call them $\{\Ddist{A}_i\}$, 
for all bounded and continuous functions $f(x)$ on $[0, 1]$ we have
\begin{align*} 
\lim_{i \rightarrow \infty} \int_{0}^{1} f(x) \dee \absDdist{A}_i(x) = 
\int_{0}^{1} f(x) \dee \absDdist{A}(x). 
\end{align*}
An equivalent definition is that $\absDdist{A}_i(x)$ converges to 
$\absDdist{A}(x)$ at points of continuity of $x.$
\qed
\end{definition} 

A simple proof of the following lemma can be found at the end of
Section~\ref{sec:wassersteinanddegradation}.
\blemma[Convergence of Forward DE --
\protect{\cite[Lemma~4.75]{RiU08}}]\label{lem:convergence} The sequence
$\{\Ldens{x}_{\ell}\}$ of distributions of forward DE converges weakly to
a symmetric distribution.  
\elemma

\blemma[BP Threshold]
Consider an ordered and complete channel family $\{\Ldens{c}_\sigma\}$.
Let $\Ldens{x}_{\ell}(\sigma)$ denote the distribution in the $\ell$-th round
of DE when the channel is $\Ldens{c}_\sigma$.
Then the {\em BP threshold} of the $(\dl, \dr)$-regular ensemble 
is defined as
\begin{align*}
\sigma^{\BPsmall}(\dl, \dr) & = \sup\{\sigma: \Ldens{x}_{\ell}(\sigma) \stackrel{\ell \to \infty}{\rightarrow} \Delta_{+\infty}\}.
\end{align*}
\elemma
In other words, the BP threshold is characterized by the largest channel
parameter so that the forward DE FP is trivial.

We have just seen that the FPs of forward DE are important since they
characterize the BP threshold.  But there exist FPs that cannot be
achieved this way. Let us review a general method of constructing FPs.
Assume that, given a channel family $\{ \Ldens{c}_\sigma\}$, we need a FP
$\Ldens{x}$ which has a given {\em error probability} $\perr(\Ldens{x})$,
{\em entropy} $\entropy(\Ldens{x})$, or {\em Battacharyya parameter}
$\batta(\Ldens{x})$.  Such FPs can often be constructed, or at least
their existence can be guaranteed, by a procedure introduced in
\cite{MMRU09}. Let us recall this procedure for the case of fixed entropy.

Consider a smooth, complete, and ordered family $\{\Ldens{c}_\ent\}$
and the $(\dl, \dr)$-regular ensemble.
Let us denote by $\Tc_{\ih}$ the ordinary density evolution
operator at fixed channel $\Ldens{c}_\ih$. Formally,
\begin{eqnarray}
\Tc_{\ih}(\Ldens{a})  = \Ldens{c}_\ih \vconv (\Ldens{a}^{\cconv \dr-1})^{\vconv \dl-1}.
\end{eqnarray}
For any $\xl\in[0,1]$, we define the density evolution operator at fixed
entropy $\xl$, call it $\Td_{\xl}$, as
\begin{eqnarray}
\Td_{\xl}(\Ldens{a}) = \Tc_{\ih(\Ldens{a}, \xl)}(\Ldens{a}),
\end{eqnarray}
where $\ih(\Ldens{a}, \xl)$ is the solution of $\entropy(\Tc_{\ih}(\Ldens{a})) =\xl$.
Whenever no such value of $\ih$ exists, $\Td_{\xl}(\Ldens{a})$
is left undefined.
Since, for a given $\Ldens{a}$, the family
$\Tc_{\ih}(\Ldens{a})$ is ordered by degradation,
$\entropy(\Tc_{\ih}(\Ldens{a}))$ is a non-decreasing function of $\ih$.
As a consequence the equation  $\entropy(\Tc_{\ih}(\Ldens{a})) =\xl$
cannot have more than a single solution.
Furthermore, by the smoothness of the channel family
$\Ldens{c}_\ih$,  $\entropy(\Tc_{\ih}(\Ldens{a}))$ is continuous as 
a function of $\ih$.  Notice  that
$\entropy(\Tc_{0}(\Ldens{a})) = 0$: if the channel is noiseless the
output density at a variable nodes is noiseless as well.
Therefore, a necessary
and sufficient condition for a solution $\ih(\Ldens{a}, \xl)$
to exist (when the family $\{\Ldens{c}_\ent\}$ is complete) is that
$\entropy(\Tc_{1}(\Ldens{a}))= \entropy( (\Ldens{a}^{\cconv \dr-1})^{\vconv \dl-1}) \ge\xl$ (see Theorem 6 in \cite{MMRU09}).

\begin{definition}[DE at Fixed Entropy $\xl$]
Set $\Ldens{a}_0 = \Ldens{c}_\xl$.
For $\ell\ge 0$ compute  $\Ldens{a}_{\ell+1} =
\Td_{\xl}(\Ldens{a}_{\ell})$.  
\qed
\end{definition}
{\em Discussion:}
It can be shown that if the above procedure gives rise to an infinite
sequence, i.e., if $\Td_{\xl}(\cdot)$ is well-defined at each step,
then this sequence has a converging subsequence. In fact, in practice
one observes that the sequence itself converges.
The computation of the convolutions is typically done numerically
either by sampling
or via Fourier transforms as in ordinary density evolution.  Due
to the monotonicity of $\entropy(T_{\ih}(\Ldens{a}_{\ell}))$ in
$\ih$, the value of $\ih(\Ldens{a}_{\ell}, \xl)$ can be efficiently
found by a bisection method.  The procedure is halted when some convergence criterion
is met -- e.g., one can require that (a properly defined) distance
between $\Ldens{a}_\ell$ and $\Ldens{a}_{\ell+1}$ becomes smaller
than a threshold.

Any FP of the above transformation $\Td_{\xl}$, i.e., any $\Ldens{a}$
such that $\Ldens{a} = \Td_{\xl}(\Ldens{a})$, is also a FP of
ordinary density evolution for the channel $\Ldens{c}_\ih$ with $\ih =
\ih(\Ldens{a}, \xl)$.  Furthermore, if a sequence of densities such that
$\Ldens{a}_{\ell+1} = \Td_{\xl}(\Ldens{a}_{\ell})$ converges (weakly)
to  a density $\Ldens{a}$, then $\Ldens{a}$ is a FP of $\Td_{\xl}$,
with entropy $\xl$.

\subsection{BP Threshold for Large Degrees}
What happens to the BP threshold when we fix the design rate
$r=1-\dl/\dr$ and increase the degrees? The proof of the following
lemma, which uses basic extremes of information combining arguments,
can be found in Appendix~\ref{sec:proofofbpboundsuncoupled}.

\begin{lemma}[Upper Bound on BP Threshold]\label{lem:bpboundsuncoupled}
Consider transmission over an ordered and complete family
$\{\Ldens{c}_\ent\}$ of BMS channels using an $(\dl, \dr)$-regular
dd and BP decoding.  Let $r=1-\frac{\dl}{\dr}$ be
the design rate and let $\ent^{\BPsmall}(\dl, \dr)$ denote the BP threshold.
Then,
\begin{align*}
\ent^{\BPsmall}(\dl, \dr) \leq 
\frac{h_2(\frac{1}{2 \sqrt{\dr-1}})}{1 \!-\! ((1-r)\dr)e^{-2\sqrt{\dr\!-\!1}}}.
\end{align*}
In particular, by increasing $\dr$ while keeping the rate $r$ fixed,
the BP threshold converges to $0$.  
\end{lemma}

\subsection{The Wasserstein Metric: Definition and Basic Properties}
\label{sec:wassersteinmetric}
In the sequel we will often need to measure how close various
distributions are. Sometimes it is convenient to compare their entropy
or their Battacharyya constant. But sometimes a more general measure
is required. The Wasserstein metric is our measure of choice.
\begin{definition}[Wasserstein Metric -- \protect{\cite[Chapter 6]{Villani09}}]\label{def:wasserstein}
Let $\absDdens{a}$ and $\absDdens{b}$ denote two $|D|$-distributions.
The Wasserstein  metric,
denoted by $d(\absDdens{a}, \absDdens{b})$, is defined as
\begin{align}\label{eq:blmetric}
d(\absDdens{a}, \absDdens{b})=\!\!\!\!\!\!\sup_{f(x) \in \Lip(1)[0, 1]} \!\Big\vert \int_{0}^{1} \!\!f(x)(\absDdens{a}(x)\!-\! \absDdens{b}(x)) \,\dee x \Big\vert,
\end{align}
where $\Lip(1)[0, 1]$ denotes the class of Lipschitz continuous functions on $[0, 1]$
with Lipschitz constant $1$.
\qed
\end{definition}
{\em Discussion}: 
In the sequel we will say that a function $f(x)$ is $\Lip(c)$ as a shorthand
to mean that it is Lipschitz continuous with constant $c$. If we want
to emphasize the domain, then we write e.g., $\Lip(c)[0, 1]$. Why have
we defined the metric in the $|D|$-domain?  As the next lemma shows,
convergence in this
metric implies weak convergence. Since all the distributions of
interest are symmetric, it suffices to look at the $|D|$-domain
instead of the $D$-domain.  To ease our notation, however, we will formally
write expressions like $d(\Ldens{a}, \Ldens{b})$, i.e., we will allow
the arguments to be e.g. $L$-distributions.
It is then implied that the metric is determined
using the equivalent $|D|$-domain representations
as defined above.

\begin{lemma}[Basic Properties of the Wasserstein Metric]\label{lem:blmetric}
In the following, $\Ldens{a}$, $\Ldens{b}$, $\Ldens{c}$, and $\Ldens{d}$
denote $L$-distributions.  

In the $|D|$ domain we have the following expressions for 
$\batta(\Ldens{a})$ and $\entropy(\Ldens{a})$ (compare this to the
expressions in the $L$ domain given in Section~\ref{sec:BPandDE}),
\begin{align*}
\batta(\absDdens{a}) & = \int_0^1 \sqrt{1-x^2}\absDdens{a}(x) \dee x, \\
\entropy(\absDdens{a}) & = \int_0^1 h_2\Big(\frac{1-x}2\Big)\absDdens{a}(x) \dee x,
\end{align*}
where $h_2(x) = -x\log_2 x - (1-x)\log_2(1-x)$ is the binary entropy function.
See \cite{Zol97,Villani09} for more details on metrics for probability measures. 
\renewcommand\theenumi{\roman{enumi}} 
\renewcommand{\labelenumi}{(\roman{enumi})}
\begin{enumerate} \item {\em Alternative Definitions}: \label{lem:alternative}
\begin{align*}
d(\Ldens{a}, \Ldens{b}) & = 
\inf_{p(x, y): p(x) \sim \absDdens{a}; p(y) \sim \absDdens{b}} E[|X-Y|], \\
d(\Ldens{a}, \Ldens{b}) & = 
\int_{0}^{1} |\absDdist{a}(x)-\absDdist{b}(x)| \dee x.
\end{align*}
 
\item {\em Boundedness}: 
\label{lem:blmetricboundedness} 
$d(\Ldens{a}, \Ldens{b}) \leq 1$. 

\item {\em Metrizable and Weak
Convergence}:\label{lem:blmetricmetrizable} The Wasserstein metric
induces the weak topology on the space of probability measures on $[0,
1]$. In other words, the space of probability measures under the weak
topology is metrizable and convergence in the Wasserstein metric is
equivalent to weak convergence (see \cite[Theorem 6.9]{Villani09}).

\item {\em Polish Space}:
\label{lem:blmetricpolish}
The space of probability distributions on $[0, 1]$ metrized by the
Wasserstein distance is a complete separable metric space, i.e.,
a Polish space, and any measure can be approximated by a sequence of
probability measures with finite support, i.e., distributions of the
form $\sum_{i=1}^{n} c_i \delta(x-x_i)$, where $\sum_{i=1}^{n} c_i=1$,
$c_i \geq 0$, and $x_i \in [0, 1]$.  Further, the space is compact.
(See \cite[Theorem 6.18]{Villani09}.)

\item {\em Convexity}: 
\label{lem:blmetricconvexity} 
Let $\alpha\in [0,1]$.   Then
\begin{align*}
d(\alpha \Ldens{a} + \bar{\alpha}\Ldens{b}, \alpha \Ldens{c} + \bar{\alpha}\Ldens{d})\leq \alpha d(\Ldens{a} ,  \Ldens{c}) + \bar{\alpha}d(\Ldens{b} ,  \Ldens{d}).
\end{align*}
In general, if $\sum_i \alpha_i = 1$, then 
\begin{align*}
d(\sum\alpha_i \Ldens{a}_i, \sum \alpha_i\Ldens{b}_i)\leq \sum \alpha_i d(\Ldens{a}_i ,  \Ldens{b}_i).
\end{align*}

\item {\em Regularity wrt $\vconv$}: 
\label{lem:blmetricregularvconv}
The Wasserstein metric satisfies the regularity property
$d(\Ldens{a}\vconv \Ldens{c}, \Ldens{b}\vconv \Ldens{c}) \leq 2 d(\Ldens{a}, \Ldens{b})$,
so that
\begin{align*}
d(\Ldens{a}\vconv \Ldens{c}, \Ldens{b}\vconv \Ldens{d}) & \leq 
d(\Ldens{a}\vconv \Ldens{c}, \Ldens{b}\vconv \Ldens{c}) +
d(\Ldens{b}\vconv \Ldens{c}, \Ldens{b}\vconv \Ldens{d}) \\
& \leq 2 d(\Ldens{a}, \Ldens{b}) + 2 d(\Ldens{c}, \Ldens{d}),
\end{align*}
and for $i \geq 2$ and any distribution $\Ldens{c},$
$d(\Ldens{a}^{\vconv i}\vconv\Ldens{c}, \Ldens{b}^{\vconv i}\vconv\Ldens{c}) \leq 2 i d(\Ldens{a}, \Ldens{b})$.

\item {\em Regularity wrt $\cconv$}: 
\label{lem:blmetricregularcconv}
The Wasserstein metric satisfies the regularity property
$d(\Ldens{a}\cconv \Ldens{c}, \Ldens{b}\cconv \Ldens{c})\leq d(\Ldens{a}, \Ldens{b})\sqrt{1-\batta^2(\Ldens{c})} \leq d(\Ldens{a}, \Ldens{b})$,
so that
\begin{align*}
 d(\Ldens{a}\cconv \Ldens{c}, \Ldens{b}\cconv \Ldens{d}) & \leq 
d(\Ldens{a}\cconv \Ldens{c}, \Ldens{b}\cconv \Ldens{c}) +
d(\Ldens{b}\cconv \Ldens{c}, \Ldens{b}\cconv \Ldens{d}) \\
& \leq d(\Ldens{a}, \Ldens{b}) + d(\Ldens{c}, \Ldens{d}).
\end{align*}
Further,
\begin{align*}
d(\Ldens{a}^{\cconv i}, \Ldens{b}^{\cconv i})  \!\leq\! d(\Ldens{a}, \Ldens{b})\sum_{j=1}^i (1\!-\!\batta^2(\Ldens{a}))^{\frac{i\!-\!j}2}(1\!-\!\batta^2(\Ldens{b}))^{\frac{j\!-\!1}2}.
\end{align*}

\item {\em Regularity wrt DE}:
\label{lem:blmetricregularDE} Let $\Tc_{\Ldens{c}}(\cdot)$ denote the
DE operator for the dd $(\dl, \dr)$ and the channel $\Ldens{c}$.  Then
$d(\Tc_{\Ldens{c}}(\Ldens{a}), \Tc_{\Ldens{c}}(\Ldens{b}))  \leq \alpha d(\Ldens{a}, \Ldens{b})$, with
$$
\alpha=2 (\dl-1) \sum_{j\!=\!1}^{\dr\!-\!1} (1\!-\!\batta^2(\Ldens{a}))^{\frac{\dr\!-\!1\!-\!j}2}(1\!-\!\batta^2(\Ldens{b}))^{\frac{j\!-\!1}2}.
$$

\item  {\em Wasserstein Bounds Battacharyya and Entropy:}
\label{lem:blmetricwasserboundsbatta} 
\begin{align*}
 |\batta(\Ldens{a})-\batta(\Ldens{b})| & \leq 
\sqrt{d(\Ldens{a}, \Ldens{b})}
\sqrt{2-d(\Ldens{a}, \Ldens{b})}\,, \\ 
|\entropy\bigl(\Ldens{a})-\entropy(\Ldens{b}\bigr)| & \leq 
h_2\Bigl( \frac{d(\Ldens{a}, \Ldens{b})}{2} \Bigr) \\
& \leq
\frac{1}{\ln 2}
\sqrt{d(\Ldens{a}, \Ldens{b})}
\sqrt{2-d(\Ldens{a}, \Ldens{b})}\,.
\end{align*}

\item {\em Battacharyya Sometimes Bounds Wasserstein}:
\label{lem:blmetricbattaboundswasser} 
\begin{align*}
& d(\Delta_0, \Ldens{a}) 
\leq \sqrt{1-\batta(\Ldens{a})^2}
\leq \sqrt{2(1-\batta(\Ldens{a}))}, \\
& d(\Delta_{+\infty}, \Ldens{a}) 
\leq \batta(\Ldens{a}).
\end{align*}
\end{enumerate}
\elemma
{\em Discussion:} Perhaps the most useful property of the Wasserstein
metric is that it interacts nicely with the operations of variable-
and check-node convolution.  This is the essence of properties
(\ref{lem:blmetricregularvconv}), (\ref{lem:blmetricregularcconv}),
and (\ref{lem:blmetricregularDE}). For example, it is easy to see why property
(\ref{lem:blmetricregularDE}) might be useful: Given that two
distributions $\Ldens{a}$ and $\Ldens{b}$ are close, it asserts
that after one iteration of DE these two distributions are again
close.  Indeed, as we will see shortly, depending on the Battacharyya
parameter of the starting distributions the distance might in fact
become smaller, i.e., we might have a {\em contraction}.

\subsection{Wasserstein Metric and Degradation}
\label{sec:wassersteinanddegradation}
When densities ordered by degradation, some
the Wasserstein metric inherits some additional properties.
\begin{lemma}[Wasserstein Metric and Degradation]
\label{lem:degradationandwasserstein}
In the following $\Ldens{a}$ and $\Ldens{b}$ denote $L$-distributions.
\renewcommand\theenumi{\roman{enumi}} 
\renewcommand{\labelenumi}{(\roman{enumi})}
\begin{enumerate}
\item {\em Wasserstein versus Degradation}: 
\label{lem:blmetricdegradation} Let $\Ldens{a} \prec \Ldens{b}$.
Let $\absDdist{A}$ and $\absDdist{B}$ denote the corresponding
$|D|$-domain cdfs.  Define $D(\Ldens{a}, \Ldens{b}) =\int_0^1 x
(\absDdist{B}(x)-\absDdist{A}(x)) \dee x$.  Note that $D(\Ldens{a},
\Ldens{b})$ can be seen as a measure of how much $\Ldens{b}$ is
degraded wrt $\Ldens{a}$ since it is the average of the non-negative
integrals $\int_z^1 (\absDdist{B}(x)-\absDdist{A}(x)) \dee x$ (cf. \eqref{equ:degradationcdfs}).  Then
$$D(\Ldens{a}, \Ldens{b}) \geq d^2(\Ldens{a}, \Ldens{b})/4.$$
Furthermore, $D(\Ldens{a}, \Ldens{b}) \leq 1$ and for any symmetric
densities such that $\Ldens{a} \prec \Ldens{b} \prec \Ldens{c}$, $D(\Ldens{a},
\Ldens{c}) = D(\Ldens{a}, \Ldens{b}) + D(\Ldens{b}, \Ldens{c})$.

\item {\em Entropy and Battacharyya Bound Wasserstein Distance}:
\label{lem:blmetricentropy} 
Let $\Ldens{a} \prec \Ldens{b}$. 
Then
\[
d(\Ldens{a}, \Ldens{b})  \leq  
2\sqrt{(\ln 2)(\entropy(\Ldens{b}) - \entropy(\Ldens{a}))}
\leq
2\sqrt{\batta(\Ldens{b}) - \batta(\Ldens{a})}
\]
and
$\batta(\Ldens{b}) - \batta(\Ldens{a})
\leq
\sqrt{2(\entropy(\Ldens{b}) - \entropy(\Ldens{a}))}\,.
$
\item{\em Continuity for Ordered Families}: 
\label{lem:blmetriccontinuity} 
Consider a smooth family of $L$-distributions
$\{\Ldens{c}_\sigma\}_{\underline{\sigma}}^{\overline{\sigma}}$ ordered
by degradation so that $\batta(\cdot)$ is continuous wrt
$\sigma \in [\underline{\sigma}, \overline{\sigma}]$.
Then the Wasserstein metric is also continuous in $\sigma$.
\end{enumerate}
\end{lemma}
{\em Discussion:} Property (\ref{lem:blmetricdegradation}) is
particularly useful. Imagine a sequence of distributions
$\{\Ldens{a}_i\}_{i=0}^{n}$ ordered by degradation, i.e., $\Ldens{a}_0
\prec \Ldens{a}_1 \prec \dots \prec \Ldens{a}_n$. Then $\Ldens{a}_0
\prec \Ldens{a}_n$  and we know from \cite{RiU08} that $D(\Ldens{a}_0,
\Ldens{a}_n) =\int_0^1 z (\absDdist{A}_n-\absDdist{A}_0) \dee z$
is non-negative since it is the ``average'' of the non-negative
integrals $\int_y^1 (\absDdist{A}_n-\absDdist{A}_0) \dee z$. Now
note that $D(\cdot, \cdot)$ is additive and that $D(\Ldens{a}_0,
\Ldens{a}_n) \leq 1$. From these two facts we can conclude
that there must exist an index $i$, $0 \leq i \leq n-1$, so that
$D(\Ldens{a}_i, \Ldens{a}_{i+1}) \leq \frac{1}{n}$. More generally,
we can conclude for any $1 \leq k \leq n$ that there must exist an
index $i$, $0 \leq i \leq n-k$, so that $D(\Ldens{a}_i, \Ldens{a}_{i+k})
\leq \min\{\frac{k}{n-k+1}, 1\} \leq \frac{2 k}{n}$. This follows by upper bounding the average
of all these $n-k+1$ such distances. By property (\ref{lem:blmetricdegradation}) this
implies ``closeness"  also in the Wasserstein sense. In words, in
a sequence of distributions ordered by degradation we are always
able to find a subsequence of distributions which are  ``close'' in
the Wasserstein sense.

As an exercise in using the basic properties of the Wasserstein
distance, let us give a proof of Lemma~\ref{lem:convergence}.
\begin{IEEEproof}
Since we are considering a sequence of distributions obtained by forward DE, we have $\Ldens{x}_{\ell} \succ \Ldens{x}_{\ell+1}$ for $\ell \geq
0$.  Therefore, the quantities $D(\Ldens{x}_{\ell}, \Ldens{x}_{\ell+1})$
are non-negative and they are additive in the sense that $D(\Ldens{x}_0,
\Ldens{x}_n)=\sum_{\ell=0}^{n-1} D(\Ldens{x}_\ell, \Ldens{x}_{\ell+1})$.
Further, $D(\cdot, \cdot)$ is upper bounded by $1$. It follows that
$\{\Ldens{x}_\ell\}$ forms a Cauchy sequence wrt to $D(\cdot, \cdot)$
and hence also wrt $d(\cdot, \cdot)$. This in turn implies that
$\{\Ldens{x}_\ell\}$ converges wrt $d(\cdot, \cdot)$ and this
convergence is equivalent to weak convergence. Finally, symmetry
can be tested in terms of bounded continuous functionals and weak
convergence preserves such functionals.  
\end{IEEEproof}

\subsection{GEXIT Curve}\label{sec:gexitcurves}
As we have discussed in the preceding section, FPs of DE play a crucial
role in the asymptotic analysis. E.g., the BP threshold is characterized
by the existence/non-existence of a non-trivial FP of forward DE for a
particular channel.

An even more powerful picture arises if instead of looking at
a single FP at a time we visualize a whole {\em collection} of FPs.
In order to visualize many FPs at the same time it is
convenient to project them.  E.g., given the FP pair $(\Ldens{c},
\Ldens{x})$ we might decide to plot the point $(\entropy(\Ldens{c}),
\entropy(\Ldens{x}))$ in the two-dimensional unit box $[0, 1] \times
[0, 1]$. 
\begin{example}[BP EXIT Curve for BEC]\label{exa:exit} Note that for
the BEC, erasure probability is equal to Battacharyya parameter, and
also equal to entropy. Even though all these parameters are equal in
this case, our language will reflect that we are plotting entropy.

Rather than plotting $x$ itself it is convenient to plot the {\em EXIT}
value $(1-(1-x)^{\dr-1})^{\dl}$. This is the locally best estimate of
a bit based on the internal messages only, {\em excluding}
the direct observation.  For this choice the resulting
curve is usually called the {\em BP EXIT} curve, see \cite{tBr00,AKtB04}
and \cite[Sections 3.14 and 4.10]{RiU08}. It is the {\em BP} EXIT curve
since the estimate is a BP estimate. And it is the BP {\em EXIT} (where
the E stands for ``extrinsic'') curve since the estimate excludes the
received value associated to this bit.

The FP equation is $x=\epsilon (1-(1-x^{\dr-1}))^{\dl-1}$,
which we can solve for $\epsilon$ to get 
\begin{align}\label{equ:epsofx}
\epsilon(x) & = \frac{x}{(1-(1-x^{\dr-1}))^{\dl-1}}.
\end{align}
Using (\ref{equ:epsofx}) we can write down the parametric
characterization of the BP EXIT curve
\begin{align*}
\Bigl(\frac{x}{(1-(1-x^{\dr-1}))^{\dl-1}}, (1-(1-x)^{\dr-1})^{\dl}\Bigr).
\end{align*}
This curve is shown in the left-hand side in Figure~\ref{fig:ebpexit36bec} for
the $(3, 6)$-regular ensemble and has a typical $C$ shape.  In fact, one can
show that, in this case, for $\epsilon<\epsilon^{\BPsmall}(\dl, \dr)$ (the BP threshold) there is
only one FP at $x=0$ corresponding to perfect
decoding; for $\epsilon=\epsilon^{\BPsmall}(\dl, \dr)$ there are 2
FPs, one is at $x=0$ and the other is the FP corresponding to forward
DE; and for $\epsilon>\epsilon^{\BPsmall}(\dl, \dr)$  there
are exactly 3 FPs of DE, one of the FPs is at $x=0$  and  the remaining two FPs are strictly positive, one of which is {\em
stable}, denoted by $\xstab(\epsilon)$, whereas the other is {\em unstable},
denoted by $\xunstab(\epsilon)$.  The stable FP is the FP which is reached by
forward DE.  For details see Lemma~\ref{lem:propertyofh(x)}.

A quantity which will appear throughout this paper is the value of
the unstable FP when transmitting over BEC$(\epsilon=1)$. We denote this FP
by $\xunstab(1)$. More precisely,
 $\xunstab(1)$ is the smaller non-zero solution of $x=(1 -
 (1-x)^{\dr-1})^{\dl-1}$.  Note that $\xunstab(1)$ depends on the
 degrees, but we drop it from the notation for ease
of exposition.

\end{example} 
\begin{figure}[htp]
\centering
\input{ps/ebpexit36bec_arxiv}
\caption{\label{fig:ebpexit36bec} Left: The BP EXIT curve of the $(\dl=3,
\dr=6)$-regular ensemble when transmitting over the BEC.  The curve has a characteristic ``C'' shape.
Right: The construction
of the MAP threshold from the BP EXIT curve. The dark gray area is equal
to the design rate of the code.}
\end{figure}
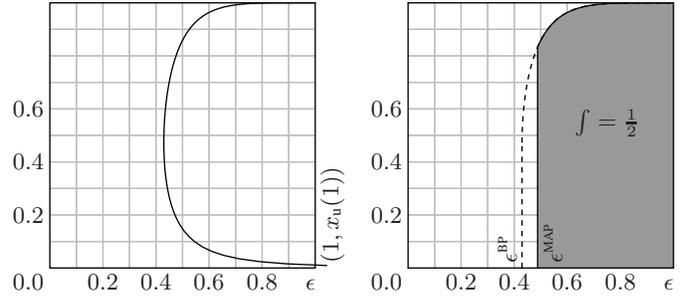
{\em Discussion}: 
The above example raises the following two questions.  (1) We
have a large degree of freedom in selecting the projection operator.
Which one is ``best''? (2) From the above example we see that the
set of FPs forms a smooth curve. Indeed, for the BEC it is not hard to
see that the only FPs are the ones on the curve together with all
the FPs of the form $(\Ldens{c}_\epsilon, \Delta_{+\infty})$, where
$\Ldens{c}_\epsilon$ is any element of the family of BEC channels and
$\Delta_{+\infty}$ corresponds to erasure value of 0.
Is this picture still valid for general channel families?

In the remainder of this section we address the first question, i.e.,
we will discuss a particularly effective choice of the projection operator.
In the next section we will address the question
of the existence and nature of this curve for the general case,
presenting some partial results.

A good choice for the projection operator for general channels is
the {\em GEXIT functional} \cite{MMRU09}.  For the BEC this coincides
with the EXIT functional that we saw in Example~\ref{exa:exit}.
For the general case take a FP
$(\Ldens{c}_{\sigma}, \Ldens{x}_\sigma)$ 
and define $\Ldens{y}=\Ldens{x}^{\cconv \dr-1}_\sigma$.
Then
\begin{align*}
G(\Ldens{c}_{\sigma}, \Ldens{y}^{\vconv \dl}) & = 
\frac{\frac{\dee}{\dee \sigma} \entropy(\Ldens{c}_{\sigma} \vconv \Ldens{y}^{\vconv \dl})}{
\frac{\dee}{\dee \sigma} \entropy(\Ldens{c}_{\sigma} )},
\end{align*} 
where we think of $\Ldens{y}$ as fixed with respect to $\sigma$.
In words, $G(\Ldens{c}_{\sigma}, \cdot )$ measures the ratio of the
change in entropy of $\Ldens{c}_\sigma \vconv \Ldens{y}^{\vconv
\dl}$ (the entropy of the decision of any variable node under BP
decoding) versus the change of entropy of the channel $\Ldens{c}_\sigma$
as a function of $\sigma$.

{\em Discussion:} Note that if the parameterization in $\sigma$ is
Lipschitz, i.e., if for some positive constant $\alpha$,
$|\entropy(\Ldens{c}_{\sigma_2}) -\entropy(\Ldens{c}_{\sigma_1})|
\leq \alpha |\sigma_2-\sigma_1|$, then the derivative $\frac{\dee}{\dee
\sigma} \entropy(\Ldens{c}_{\sigma} )$ exists almost everywhere.  Further,
in this case also $\entropy(\Ldens{c}_{\sigma} \vconv \Ldens{y}^{\vconv
\dl})$ is Lipschitz and hence differentiable almost everywhere. This follows since
by (the Duality Rule in) Lemma~\ref{lem:dualityrule}, for $\sigma_2 \geq \sigma_1$,
\begin{align*}
 [\entropy(\Ldens{c}_{\sigma_2} \vconv \Ldens{y}^{\vconv \dl}) & -
\entropy(\Ldens{c}_{\sigma_1} \vconv \Ldens{y}^{\vconv \dl})] 
\\ 
& + [\entropy(\Ldens{c}_{\sigma_2} \cconv \Ldens{y}^{\vconv \dl})-
\entropy(\Ldens{c}_{\sigma_1} \cconv \Ldens{y}^{\vconv \dl})] \\
&= [\entropy(\Ldens{c}_{\sigma_2})-\entropy(\Ldens{c}_{\sigma_1})] \leq 
\alpha |\sigma_2 - \sigma_1|,
\end{align*}
where the last step on the right-hand side assumes that the
parameterization is such that $\entropy(\Ldens{c}_{\sigma})$ increases
in $\sigma$.  The claim follows since both terms on the left
are non-negative (due to degradation), so that in particular the first term is upper
bounded by $\alpha |\sigma_2 - \sigma_1|$, i.e., it is Lipschitz.
This formulation also shows that the numerator is no larger than the
denominator (so that the ratio exists) and that the GEXIT value is
upper bounded by $1$ (and is non-negative).

We get the GEXIT {\em curve} by plotting $(\entropy(\Ldens{c}_{\sigma}),
G(\Ldens{c}_{\sigma}, \Ldens{y}^{\vconv \dl}))$ for a family of
FPs $\{\Ldens{c}_{\sigma}, \Ldens{x}_\sigma\}$.  This is shown in
Figure~\ref{fig:ebpexit36} for the $(3, 6)$-regular ensemble assuming
that transmission takes place over the BAWGNC.
In the last section we have already explained how we can construct
in the general case FPs by a numerical procedure. To plot
Figure~\ref{fig:ebpexit36} we have used this procedure to get a
complete family of FPs for all entropies from $0$ to $1$. In each
of the two pictures of Figure~\ref{fig:ebpexit36} there is a small
black dot. This dot marks a particular FP and the two small inlets
show the corresponding distribution of the channel $\Ldens{c}_\sigma$
as well as the message distribution emitted at the variable nodes,
call it $\Ldens{x}_\sigma$. 
For a detailed
discussion we refer the reader to \cite{MMRU09,RiU08}.
\begin{figure}[htp] \centering
\input{ps/ebpexit36} 
\caption{\label{fig:ebpexit36} 
The BP GEXIT curve for the $(\dl=3,
\dr=6)$-regular ensemble and transmission over the BAWGNC.  Each point on
the curve corresponds to a FP $(\Ldens{c}_{\sigma}, \Ldens{x}_\sigma)$ of DE.  
The two figures show the 
FP density $\Ldens{x}$ 
as well as the input density $\Ldens{c}_{\sigma}$
for two points on the curves (see inlets).} 
\end{figure}
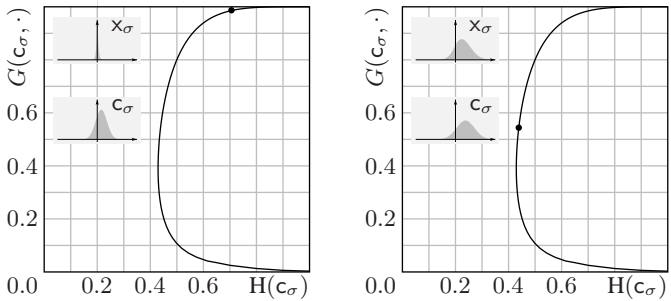

Why do we use this particular representation? As we will discuss
in detail in Section~\ref{sec:areatheorem}, assuming this curve
indeed exists and is ``smooth'', the area which is enclosed by it
is equal to $r=1-\dl/\dr$, the design rate of the ensemble.

This is easy to see for the BEC. To simplify notation,
denote the GEXIT value in this case by $G(\epsilon, y^{\dl})$, 
where $\epsilon$ is the erasure probability, $x$ is the FP 
for this channel parameter, and $y=1-(1-x)^{\dr-1}$.
We then have
$G(\epsilon, y^{\dl})=(1-(1-x)^{\dr-1})^{\dl}$.
Let us integrate the area which is enclosed by this curve.
We call the corresponding integral the GEXIT integral. For our particular
case it is given by 
\begin{align*}
& \int (1-(1-x)^{\dr-1})^{\dl} \,\dee \epsilon  = 
\int_{0}^{1} (1\!-\!(1\!-\!x)^{\dr-1})^{\dl} \epsilon'(x) \,\dee x  \\
= & \epsilon(x) (1-(1-x)^{\dr-1})^{\dl} \mid_{0}^{1} + \\ 
& - \dl (\dr-1) \int_{0}^{1} \epsilon(x) (1-x)^{\dr-2} (1-(1-x)^{\dr-1})^{\dl-1} 
\,\dee x  \\
= &  1 - \dl (\dr-1) \int_{0}^{1} x (1-x)^{\dr-2} \,\dee x \\
= &  1 + \dl x (1-x)^{\dr-1} {\Large \mid}_{0}^{1} - \dl \int_{0}^{1} (1-x)^{\dr-1} \,\dee x
=   1-\frac{\dl}{\dr}. 
\end{align*}
Perhaps surprisingly, the result stays valid for general channels as we
will discuss in Section~\ref{sec:areatheorem}. This property is one of
the main ingredients in our proof.

Note that given $\Ldens{c}_{\ent}$ and $\Ldens{z}_{\ent}$, the GEXIT
functional $G(\Ldens{c}_{\ent}, \Ldens{z}_{\ent})$ can be expressed
in the form $\int \Ldens{z}_{\ent}(w) f(\ent, w) \dee w$, where
$f(\ent, w)$ is called as the GEXIT kernel.
In the $|D|$-domain this kernel is given by
\begin{align}\label{eq:gexitkernel|D|domain} \int_0^1
\!\!\frac{\dee \Ldens{c}_{\ent}(z)}{\dee \ent}\underbrace{\!\!\!\Big(\!\!\!\sum_{i,j=\pm
1}\!\!\!\!\! \frac{(1\!+\!iz)(1\!+\!jw)}4
 \log_2\!\Big(\!1\!+\!
 \frac{(1\!-\!iz)(1\!-\!jw)}{(1\!+\!iz)(1\!+\!jw)}\Big)\!\Big)}_{=k(z, w)}\! \dee
 z.
\end{align}
For a proof of the following see Lemma 4.77, \cite{RiU08}.
\begin{lemma}[GEXIT for Smooth and Ordered Channels]\label{lem:gexitsmoothfamily}
For a smooth, ordered, channel family $\{\Ldens{c}_{\ent}\}_{\ent}$, $f(\ent, w)$, as a
function of $w$, exists, is continuous, non-negative, non-increasing and concave
on its entire domain. Further $f(\ent, 0)=1$ and $f(\ent, 1)=0$.
\end{lemma}
We remark that the above lemma also holds when $\{\Ldens{c}_{\ent}\}$ is piece-wise linear.

\subsection{Existence of GEXIT Curve}\label{sec:gexitexistence}
As we briefly discussed above, for the BEC it is trivial to see
that the BP GEXIT curve indeed exists. But for general BMS channels
this is not immediate.  The aim of this section is to show the
existence of the BP GEXIT curve for at least a subset of parameters.

Let us first recall the following lemma which was stated and proved in a
slightly weaker form in \cite{MMU08}. For the convenience of the reader
we reproduce the proof in Appendix~\ref{app:bounds}.

\begin{lemma}[Sufficient Condition for Continuity]\label{lem:condforcontinuity}
Assume that communication takes place over an ordered and complete family
$\{\Ldens{c}_{\ih}\}_{\ih}$, where $\ent=\entropy(\Ldens{c}_\ent)$,
using the \ddp $(\dl,\dr)$.

Then, for any $\ih \in [0, 1]$, there exists at most one 
density $\Ldens{x}_{\ih}$ so that $(\Ldens{c}_{\ih}, \Ldens{x}_{\ih})$ forms
a FP which fulfills 
\begin{eqnarray}
\batta(\Ldens{c}_{\ih}) (\dl-1)(\dr-1)(1-\batta(\Ldens{x}_{\ih})^2)^{\dr-2}<1
\, .\label{eq:UniquenessCond}
\end{eqnarray}
Furthermore, if such a density $\Ldens{x}_{\ih}$ exists, then it coincides with
the  forward DE FP. Finally,
$\batta(\Ldens{x}_{\ih})$ is Lipschitz continuous with respect to $\batta(\Ldens{c}_{\ih})$.
More precisely, if two FPs
$(\Ldens{c}_{\ih_1}, \Ldens{x}_{\ih_1})$ and $(\Ldens{c}_{\ih_2}, \Ldens{x}_{\ih_2})$ satisfy the condition
$\batta(\Ldens{c}_{\ih_i}) (\dl-1)(\dr-1)(1-\batta(\Ldens{x}_{\ih_i})^2)^{\dr-2} \le 1-\delta$
for some $\delta>0$, then  
\begin{eqnarray}
|\batta(\Ldens{x}_{\ih_1})-\batta(\Ldens{x}_{\ih_2})|\le \frac{1}{\delta} \,|\batta(\Ldens{c}_{\ih_1}) - \batta(\Ldens{c}_{\ih_2})|\, .
\label{eq:Lipschitz}
\end{eqnarray}
\end{lemma}

The following lemma states that, at least for sufficiently large entropies,
the BP GEXIT curve indeed exists and is well behaved.
\begin{lemma}[Continuity For Large Entropies]\label{lem:continuityforlargeentropy}
Assume that communication takes place over an ordered and complete family
$\{\Ldens{c}_{\ih}\}_{\ih}$, where $\ent=\entropy(\Ldens{c}_\ent)$,
using the \ddp $(\dl,\dr)$.
Consider the set of FP pairs
$\{(\Ldens{c}_{\ih}, \Ldens{x}_{\ih})\}$ obtained by applying
forward DE to each channel $\Ldens{c}_{\ih}$. Let 
\begin{align*}
a(x) & = (1-(1-x)^{\dr-1})^{\dl-1}, \\
b(x) & = (\dl-1)^2(\dr-1)^2 x (1-x)^{2(\dr-2)}, \\
c(x) & = \sqrt{x/a(x)}.
\end{align*}
Let $\xLE$ be the unique solution in $(0, 1]$ of the equation
\begin{align}\label{equ:existence}
a(x) - b(x)=0.
\end{align}
Then the family $\{(\Ldens{c}_{\ih}, \Ldens{x}_{\ih})\}_{\ih=\entLE(\dl,
\dr, \{\Ldens{c}_{\ent}\})}^{1}$, with $\entLE(\dl, \dr,
\{\Ldens{c}_{\ent}\}) =\ent_{\BMS}(c(\xLE))$, satisfies \eqref{eq:UniquenessCond}, is Lipschitz continuous
wrt to the Battacharyya parameter of the channel, where
$\ent_{\BMS}(\cdot)$ is the function which maps the Battacharyya
constant of an element of the family to the corresponding entropy.
Further, $\batta(\Ldens{x}_{\ih}) \geq \xunstab(1) >0$ for all $\ih
\geq \entLE(\dl, \dr, \{\Ldens{c}_{\ent}\})$\footnote{Note that we
have made the dependence on the channel family, $\{\Ldens{c}_{\ent}\}$,
explicit in the notation of $\entLE(\dl, \dr, \{\Ldens{c}_{\ent}\})$.}.
\end{lemma}

\begin{table}
\centering
\begin{tabular}{ccccccc}
dd & $\xLE$ & $\frac{\dl}{\dr}$ & $\batta=\entLE_{\BEC}$ & $\entLE_{\BAWGNC}$ & $\entLE_{\BSC}$ & $\overline{\ih}$ \\ \Hline 
(3, 4) & 0.5479 & 0.75 & 0.8156 & 0.7544 & 0.7428 & 0.8254 \\ 
(6, 8) & 0.4107 & 0.75 & 0.6822 & 0.5971 & 0.5694 & 0.6958 \\
(9, 12) & 0.3277 & 0.75 & 0.6024 & 0.5097 & 0.4719 & 0.6185 \\ 
(12, 16) & 0.2752 & 0.75 & 0.5483 & 0.4530 & 0.4087 & 0.5658 \\ \hline
(3, 6) & 0.3805 & 0.5 & 0.6787 & 0.5931 & 0.5651 & 0.7010 \\ 
(4, 8) & 0.3512 & 0.5 & 0.6384 & 0.5485 & 0.5152 & 0.6590\\ 
(5, 10) & 0.3192 & 0.5 & 0.6022 & 0.5094 & 0.4717 & 0.6229 \\ 
(6, 12) & 0.2916 & 0.5 & 0.5717 & 0.4773 & 0.4357 & 0.5924 \\ \hline
(3, 12) & 0.2127 & 0.25 & 0.4970 & 0.4012 & 0.3513 & 0.5335 \\ 
(4, 16) & 0.1957 & 0.25 & 0.4690 & 0.3736 & 0.3210 & 0.5005 \\
(5, 20) & 0.1774 & 0.25 & 0.4426 & 0.3481 & 0.2933 & 0.4721 \\ 
(6, 24) & 0.1616 & 0.25 & 0.4200 & 0.3267 & 0.2702 & 0.4483 \\ 
(7, 28) & 0.1483 & 0.25 & 0.4006 & 0.3086 & 0.2509 & 0.4281 
\end{tabular}
\caption{ \label{tab:bounds}
Top branches of GEXIT curves are Lipschitz continuous from indicated
channel entropy until $1$. The numbers $\xLE$, $\batta=\entLE_{\BEC}$,
$\entLE_{\BAWGNC}$, and $\entLE_{\BSC}$ are computed according to
Lemma~\ref{lem:continuityforlargeentropy}.  The final number
$\overline{\ih}$ is a universal upper bound, valid for all BMS
channels and it was computed according to
Lemma~\ref{lem:FPforlargeentropy}.
 }
\end{table} Table~\ref{tab:bounds} shows the resulting bounds for
various regular dds and various channels. These
bounds were computed as follows. For a fixed dd 
pair $(\dl, \dr)$ we first computed $\xLE$ numerically. This is
easy to do since we know that there is a unique solution of the
equation $a(x)-b(x)=0$ in $(0, 1]$. Further, $a(0)-b(0)=0$,
$a'(0)-b'(0)=-(\dl-1)^2(\dr-1)^2 < 0$, and $a(1)-b(1)=1$.
We can therefore find this unique solution efficiently via bisection.
Once $\xLE$ is found, we find the corresponding Battacharyya parameter
of the channel by computing $c(\xLE)$. Finally, we can convert this
into an entropy value via the appropriate function $\ent_{\BMS}(\cdot)$.
E.g. for the family of BSC channels we have
$\ent_{\BSCsmall}(x)=h_2(\frac12 (1-\sqrt{1-x^2}))$.

Although it is easy and stable to compute the above lower bound on
the entropy numerically, it will be convenient to have a universal
and analytic such lower bound. This is accomplished in the following
lemma, whose proof can be found in Appendix~\ref{app:bounds}.

\begin{lemma}[Universal Bound on Continuity
Region]\label{lem:FPforlargeentropy} Assume that communication takes
place over an ordered and complete family $\{\Ldens{c}_{\ih}\}_{\ih}$,
where $\ent=\entropy(\Ldens{c}_\ent)$, using the \ddp $(\dl,\dr)$
with $\dr \geq 4$ and $\dl\geq 3$. Let $a(x)$ be defined as in
Lemma~\ref{lem:continuityforlargeentropy}.  Consider the set of FP
pairs $\{(\Ldens{c}_{\ih}, \Ldens{x}_{\ih})\}$ which is derived by
applying forward DE to each channel $\Ldens{c}_{\ih}$.  Then the
GEXIT curve associated to $\{(\Ldens{c}_{\ih}, \Ldens{x}_{\ih})\}_{
\ih > \overline{\ih}}$, where
\begin{align*}
\overline{x} & = 1 - ((\dl-1)(\dr-1))^{-\frac1{\dr-2}}, \;\; \overline{\ih}  = \sqrt{\overline{x}/a(\overline{x})},
\end{align*}
is Lipschitz continuous wrt the Battacharyya parameter of the
channel. 
Also, $\entLE(\dl, \dr, \{\Ldens{c}_{\ent}\}) \leq \overline{\ih}$,
where $\entLE(\dl, \dr, \{\Ldens{c}_{\ent}\})$ is the quantity introduced in Lemma~\ref{lem:continuityforlargeentropy},
and 
$\overline{\ih} \leq \frac{e^{\frac14} \sqrt{2}}{(\dr-2)^{\frac{1}{4}}}$,
so that $\overline{\ih}$ tends to zero when $\dr$ tends to infinity.
\end{lemma} 
Table~\ref{tab:bounds} lists these universal upper bounds
$\overline{\ih}$ for all the dds.

The following corollary follows immediately from
Lemma~\ref{lem:condforcontinuity}, property (\ref{lem:blmetricentropy})
of Lemma~\ref{lem:degradationandwasserstein}, and property
(\ref{lem:blmetricwasserboundsbatta}) of Lemma~\ref{lem:blmetric}.
\begin{corollary}[Continuity of Entropy]\label{cor:continuitywrtwasserstein}
Let $\{\Ldens{c}_{\ent}\}$ be a smooth BMS channel family and let
$(\Ldens{c}_{\ih}, \Ldens{x}_{\ih})$ denote a forward DE FP pair
with channel entropy $\ih > \entLE(\dl, \dr, \{\Ldens{c}_{\ent}\})$, where $\entLE(\dl,
\dr, \{\Ldens{c}_{\ent}\})$ is the value defined in Lemma~\ref{lem:continuityforlargeentropy}.
Then for $\ih_1, \ih_2 > \entLE(\dl, \dr, \{\Ldens{c}_{\ent}\})$ we have
\begin{align*}
& \frac{(\ln(2))^2}{2} |\entropy(\Ldens{x}_{\ent_1})-\entropy(\Ldens{x}_{\ent_2})|^2 \leq 
d(\Ldens{x}_{\ent_1}, \Ldens{x}_{\ent_2}) \leq \\
& \leq \frac{32^{\frac14}}{\sqrt{\delta}}d(\Ldens{c}_{\ent_1}, \Ldens{c}_{\ent_2})^{\frac14} 
\leq \frac{2 \sqrt{2}}{\sqrt{\delta}} 
(\ln(2) |\entropy(\Ldens{c}_{\ent_1})-\entropy(\Ldens{c}_{\ent_2})|)^{\frac{1}{8}}.
\end{align*}
\end{corollary}

The proof of the following lemma can be found in Appendix~\ref{app:magic}.
\blemma[Entropy Product Inequality]\label{lem:magic}
Given $\Ldens{a}$ and $\Ldens{b}$,
\begin{align*}
\entropy(\Ldens{a}\vconv\Ldens{b})
& = 
\int_{0}^{1} \int_{0}^{1}
\absDdens{a}(x) \absDdens{b}(y) k(x,y) \dee x \dee y \\
& = \int_{0}^{1} \int_{0}^{1}
\tilde{\absDdist{A}}(x) \tilde{\absDdist{B}}(y) k_{xxyy}(x,y) \dee x \dee y,
\end{align*}
where
\begin{align*}
k_{xxyy}(x,y)
& =
\frac{2}{\ln(2)} \frac{1+3x^2y^2}{(1-x^2y^2)^3},
\end{align*}
and where the cumulative distributions
$\absDdist{A}(x) = \int_0^x \absDdens{a}(z)\text{d}z,$
$\absDdist{B}(x) = \int_0^x \absDdens{b}(z)\text{d}z$
are used to define 
$\tilde{\absDdist{A}}(x) = \int_x^1 \absDdist{A}(z)\text{d}z$
and
$\tilde{\absDdist{B}}(x) = \int_x^1 \absDdist{B}(z)\text{d}z$ and
the kernel $k(x,y)$ is as given in \eqref{eq:gexitkernel|D|domain}.
We claim that
\renewcommand\theenumi{\roman{enumi}} 
\renewcommand{\labelenumi}{(\roman{enumi})}
\begin{enumerate}
\item {\em Bound on Kernel:}\label{equ:boundonk}
\[
k_{xxyy}(x,y) \le  \frac{8}{\ln(2)} (1-x^2)^{-\frac{3}{2}}(1-y^2)^{-\frac{3}{2}}.
\]
\item {\em Bound for Partially Degraded Case:} \label{equ:partiallydegradedcase}
Let $\Ldens{a'}$ be degraded with respect to the channel density $\Ldens{a}$ and let
$\Ldens{b'}$ be such that $d(\Ldens{b}', \Ldens{b}) \leq \delta$.
Then
\[
\entropy((\Ldens{a'}-\Ldens{a})\vconv(\Ldens{b'}-\Ldens{b}))
\le
\frac{8}{\ln(2)} \batta(\Ldens{a'}-\Ldens{a}) \sqrt{2 \delta}.
\]
\item {\em Bound for Fully Degraded Case:} \label{equ:degradedcase}
Let $\Ldens{a'}$ be degraded with respect to the channel density $\Ldens{a}$ and let
$\Ldens{b'}$ be degraded with respect to the channel density $\Ldens{b}.$
Then
\[
\entropy((\Ldens{a'}-\Ldens{a})\vconv(\Ldens{b'}-\Ldens{b}))
\le
\frac{8}{\ln(2)} \batta(\Ldens{a'}-\Ldens{a}) \batta(\Ldens{b'}-\Ldens{b})\,.
\]
\end{enumerate}
\end{lemma}

\begin{corollary}[Continuity of the BP GEXIT Curve]\label{cor:continuityofBPGEXIT}
Let $\{\Ldens{c}_{\ent}\}$ be a smooth BMS channel family and let
$(\Ldens{c}_{\ih}, \Ldens{x}_{\ih})$ denote a forward DE FP pair
with channel entropy $\ih > \entLE(\dl, \dr, \{\Ldens{c}_{\ent}\})$, where $\entLE(\dl, \dr, \{\Ldens{c}_{\ent}\})$ is the value defined in Lemma~\ref{lem:continuityforlargeentropy}.
Then, $G(\Ldens{c}_{\ih},(\Ldens{x}_{\ih}^{\cconv \dr-1})^{\vconv \dl}) $ is continuous wrt to $\ih$.
\end{corollary}
\begin{proof}
The GEXIT functional is defined as 
\[
G_\ent=\frac{\partial}{\partial \ent'}\entropy(\Ldens{c}_{\ent'} \vconv \Ldens{z}_\ent)\biggr|_{\ent'=\ent}\,.
\]
We will find it more convenient to parameterize the densities using
$b=b(\ent)=\batta(\Ldens{c}_\ent).$ Let us define
\[
D(b',b)=\frac{\partial}{\partial b'} \entropy(\Ldens{c}_{b'} \vconv \Ldens{z}_{b})\,.
\]
We claim that $D(b',b)$ is continuous in both its arguments.  Note
that $G_\ent = D(b(\ent),b(\ent)) \frac{\dee b(\ent)}{\dee \ent}$
and, correspondingly, we define $G_b = D(b, b).$
To show continuity of $D$ in the first component
note that $(D(b'', b)-D(b', b)) \rightarrow 0$ by the smooth channel
family assumption.  To show continuity of $D$ in the second component consider
$\entropy((\Ldens{c}_{b'''}-\Ldens{c}_{b''})\vconv
(\Ldens{z}_{b'}-\Ldens{z}_b)).$ By (the Entropy Product Inequality) Lemma~\ref{lem:magic}, property (\ref{equ:degradedcase}), we have
\begin{align*}
|\entropy((\Ldens{c}_{b'''}\!-\!\Ldens{c}_{b''})\!\vconv\! (\Ldens{z}_{b'}\!-\!\Ldens{z}_b))|
&\!\le\!
\frac{8}{\ln 2}\!\left|\batta(\Ldens{c}_{b'''}\!-\!\Ldens{c}_{b''})\right|\left|\batta(\Ldens{z}_{b'}\!-\!\Ldens{z}_b)\right|
\\ 
&=
\frac{8}{\ln 2}|b'''-b''| \left|\batta(\Ldens{z}_{b'}-\Ldens{z}_b)\right|,
\end{align*}
from which we obtain
\[
\left|(D(b'',b')-D(b'',b))\right| \le \frac{8}{\ln 2} \left|\batta(\Ldens{z}_{b'}-\Ldens{z}_b)\right|,
\]
showing that $D$ is actually Lipschitz in its second argument.  It
follows, in particular, that $G_b$ is continuous in $b$. Since the Battacharyya
parameter is a bounded functional and the channel family is smooth, 
 we have $\frac{\dee b(\ent)}{\dee \ent}$ is continuous in $\ent$. Consequently,  $G_\ent$ is continuous in
$\ent$.  
\end{proof}

\subsection{Area Theorem}\label{sec:areatheorem} In
Section~\ref{sec:gexitcurves} we introduced the GEXIT curve associated
to a regular ensemble, see e.g. Figure~\ref{fig:ebpexit36}.  In
Section~\ref{sec:gexitexistence} we then derived conditions which
guarantee that this curve indeed exists and is continuous in a given region. We will
now discuss the GEXIT integral, the area under the GEXIT curve.
In order to derive some properties of this integral, we will first
introduce GEXIT integrals in a slightly more general form before
we apply them to ensembles.

\begin{definition}[Basic GEXIT Integral]\label{def:gexitintegralbasic}
Given two families $\{\Ldens{c}_{\sigma}\}_{\underline{\sigma}}^{\overline{\sigma}}$
and $\{\Ldens{z}_{\sigma}\}_{\underline{\sigma}}^{\overline{\sigma}}$,
the {\em GEXIT integral}  
$\{\Ldens{c}_{\sigma}, \Ldens{z}_{\sigma}\}_{\underline{\sigma}}^{\overline{\sigma}}$
is defined as 
\begin{align*}
G(\{\Ldens{c}_{\sigma}, \Ldens{z}_{\sigma}\}_{\underline{\sigma}}^{\overline{\sigma}}) & =
\int_{\underline{\sigma}}^{\overline{\sigma}} 
\entropy (\frac{\dee \Ldens{c}_{\sigma}}{\dee \sigma} \vconv \Ldens{z}_{\sigma}) \,\dee \sigma.
\end{align*}
\qed
\end{definition}
{\em Discussion:} In the above definition, and some definitions
below, we need regularity conditions to ensure that the integrals
exist. Rather than stating some general conditions here, we will
discuss and verify them in the specific cases. E.g., one case we
will discuss is if the channel family $\Ldens{c}_\sigma$ is smooth
and $\Ldens{z}_\sigma$ is a polynomial in $\sigma$ with ``coefficients''
which are fixed densities.

\begin{definition}[GEXIT Integral of Code]\label{def:gexitintegral}
Consider a binary linear code of length $n$ whose graphical
representation is a tree.  Assume that we are given an ordered
family of channels
$\{\Ldens{c}_{\sigma}\}_{\underline{\sigma}}^{\overline{\sigma}}$.  Assume
that when all variable nodes ``see'' the channel $\Ldens{c}_{\sigma}$
the distribution of the resulting extrinsic BP message density at
the $i$-th variable node is $\Ldens{z}_{\sigma, i}$.  Then the {\em
GEXIT integral} associated of the $i$-th variable node is
$G(\{\Ldens{c}_{\sigma}, \Ldens{z}_{\sigma,
i}\}_{\underline{\sigma}}^{\overline{\sigma}})$.  \qed
\end{definition}
{\em Discussion}: Note that the distribution $\Ldens{z}_{\sigma,
i}$ is the best guess we can make about bit $i$ given the code
constraints and all observations except the direct observation on
bit $i$. This is why we have called the distribution the {\em
extrinsic} message density. Note further that we have assumed that
the graphical structure of the code is a tree. Therefore, BP equals
MAP, the optimal such estimator.

The GEXIT integral applied to an ensemble is just the integral under the GEXIT curve
of this ensemble.
\begin{definition}[GEXIT Integral of Ensemble]\label{def:gexitforensemble}
Consider the $(\dl, \dr)$-regular ensemble and assume
that $\{\Ldens{c}_{\sigma}, \Ldens{x}_{\sigma}\}_{\underline{\sigma}}^{\overline{\sigma}}$ is
a family of FPs of DE.
Define $\Ldens{y}_{\sigma}=\Ldens{x}_{\sigma}^{\cconv \dr-1}$. Then 
\begin{align*}
G(\dl, \dr, \{\Ldens{c}_{\sigma}, \Ldens{x}_{\sigma} \}_{\underline{\sigma}}^{\overline{\sigma}}) & =   
\int_{\underline{\sigma}}^{\overline{\sigma}} \entropy \bigr(\frac{\dee \Ldens{c}_{\sigma}}{\dee \sigma} \vconv 
\Ldens{y}_{\sigma}^{\vconv \dl}\bigr) \,\dee \sigma.
\end{align*} 
\qed
\end{definition}
In the sequel it will be handy to explicitly evaluate the
integral.  The proof of the following lemma is contained in
Appendix~\ref{app:areaunderBPGEXIT}.  
\begin{lemma}[Evaluation of GEXIT Integral]\label{lem:areaunderBPGEXIT} Assume that communication
takes place over an ordered, complete and piece-wise smooth family
$\{\Ldens{c}_{\ih}\}_{\ih}$, using the degree-distribution pair
$(\dl,\dr)$.  Let $\{\Ldens{c}_{\ih}, \Ldens{x}_{\ih}\}_{\ih}$ be
the FP family of forward DE.  Set $\Ldens{x}=\Ldens{x}_{\ent^*}$,
$\ent^* \geq \entLE(\dl, \dr, \{\Ldens{c}_\ent\})$, where $\entLE(\dl, \dr, \{\Ldens{c}_\ent\})$
is the quantity introduced in Lemma~\ref{lem:continuityforlargeentropy}.
Then,
\begin{align*}
G(\dl, \dr, \{\Ldens{c}_{\ih}, \Ldens{x}_{\ih} \}_{\ih^*}^{1}) & =   
1 - \frac{\dl}{\dr} - A,
\end{align*}
where 
\begin{align*}
A & =  
\entropy(\Ldens{x}) +  (\dl-1-\frac{\dl}{\dr}) 
\entropy(\Ldens{x}^{\cconv \dr}) - (\dl-1) \entropy(\Ldens{x}^{\cconv \dr-1}).
\end{align*}
\end{lemma}
{\em Discussion:} Note that this GEXIT integral has a simple graphical
interpretation; it is the area under the GEXIT curve as e.g. shown
in the right-hand picture of Figure~\ref{fig:ebpexit36bec}.  The
condition $\ent^* \geq \ent(\dl, \dr,  \{\Ldens{c}_\ent\})$ ensures that this
curve is well defined and integrable.

We have seen in the last section that the value of a GEXIT integral
of an ensemble is determined by the expression $A$. We will soon
see that it is crucial to describe the region where $A$ is negative.
The following lemma, whose proof can be found in
Appendix~\ref{app:asymptoticnegativity}, gives a characterization
of this property.

\begin{lemma}[Negativity]\label{lem:asymptoticnegativity} Let
$(\Ldens{c}, \Ldens{x})$ be an {\em approximate} FP of the $(\dl,
\dr)$-regular ensemble of design rate $r=1-\dl/\dr$. Assume that $\dr \geq 1 + 5 (\frac1{1-r})^{\frac43}$ and
for some fixed $0\leq \delta \leq (\frac{\ln(2) \dl}{16 \sqrt{2}\dr})^2$,
$d(\Ldens{x}, \Ldens{c}\vconv(\Ldens{x}^{\cconv
\dr-1})^{\vconv \dl-1}) \leq \delta$.  Let
\begin{align*}
A & =\entropy(\Ldens{x}) + (\dl-1-\frac{\dl}{\dr}) \entropy(\Ldens{x}^{\cconv \dr})  - 
(\dl-1) \entropy(\Ldens{x}^{\cconv \dr-1}).
\end{align*}
For $0 \leq \kappa \leq \frac1{4 e} \frac{\dl}{\dr}$,
if $\entropy(\Ldens{x}) \in [(\frac34)^{\frac{\dl-1}{2}}+\frac1{(\dr-1)^3},  \frac{\dl}{\dr} - 
\dl e^{-4 (\dr-1) (\frac{2 \dl}{11 e \dr})^{\frac{4}{3}}}-\kappa]$, then
$A \leq -\kappa$.
\end{lemma}
{\em Discussion:} In words, for sufficiently high degrees,
$A(\Ldens{x})$ is strictly negative for all $\Ldens{x}$ with entropies
in the range $(0, \dl/\dr)$. Note that $\dl/\dr$ corresponds to the
Shannon threshold for a code of rate $1-\dl/\dr$. In the preceding
lemma we introduced the notion of an {\em approximate} FP of DE:
we say that $(\Ldens{c}, \Ldens{x})$ is a $\delta$-approximate FP
if for some $\delta>0$ we have $d(\Ldens{x},
\Ldens{c}\vconv(\Ldens{x}^{\cconv \dr-1})^{\vconv \dl-1}) \leq
\delta$.

\subsection{Area Threshold}
The most important goal of this paper is to show that
suitable coupled ensembles achieve the capacity. The preceding
(Negativity) Lemma~\ref{lem:asymptoticnegativity} is an important
tool for this purpose. But we will in fact prove a refined statement,
namely we will determine the threshold for fixed dds.
This threshold is the so-called area threshold and it was first introduced
in \cite{MMRU09} in the context of the Maxwell construction. 

\begin{definition}[Area Threshold]\label{def:areathresholddef}
Consider the $(\dl, \dr)$-regular ensemble and transmission over a
complete and ordered channel family
$\{\Ldens{c}_\ih\}_{\ent=0}^{1}$.
For each $\ih \in [0, 1]$, 
let $\Ldens{x}_\ih$ be the forward DE FP associated to channel $\Ldens{c}_\ih$.
The {\em area threshold}, denote it by $\ent^A(\dl, \dr, \{\Ldens{c}_\ent\})$, is defined as
\begin{align*}
\ih^A(\dl, \dr, \{\Ldens{c}_\ent\}) & = \sup\{\ih \in [0, 1]: 
A(\Ldens{x}_\ih, \dl, \dr) \leq 0\},
\end{align*}
where $A(\Ldens{x}_{\ih}, \dl, \dr)$ is equal to $A$, which is given in Lemma~\ref{lem:areaunderBPGEXIT}, 
evaluated at the FP $\Ldens{x}_{\ih}$, when transmitting with the $(\dl, \dr)$-regular ensemble.
\qed
\end{definition}
Note that $A(\Delta_{+\infty}, \dl, \dr)=0$ and that $\Ldens{x}_\ih =
\Delta_{+\infty}$ for all $\ih < \ih^{\BP}(\dl, \dr, \{\Ldens{c}_{\ent}\})$.  Therefore the
set over which we take the supremum is non-empty and $\ih^{\BP}(\dl,
\dr, \{\Ldens{c}_{\ent}\}) \leq \ih^A(\dl, \dr, \{\Ldens{c}_\ent\})$.  Also note that we
have made the dependence of the area threshold on the
channel family and the dd explicit.\footnote{We keep the explicit notation
of $\entLE(\dl, \dr, \{\Ldens{c}_\ent\})$ and $\ent^A(\dl, \dr,
\{\Ldens{c}_\ent\})$ in the statements of the lemmas and theorems
but drop it in the proof for ease of exposition.}

Table~\ref{tab:areathreshold} gives some values for $\ih^{A}(\dl, \dr, \{\Ldens{c}_\ent\})$ for
various dds and channels.  

\begin{table} \centering \begin{tabular}{clllll} 
d.d. & rate & $\ih^{\text{Sha}}$ & $\ih_{\BEC}^{A}$ & $\ih_{\BSC}^{A}$ & $\ih_{\BAWGNC}^{A}$
\\ \Hline 
$(5, 6)$ & $0.1667$ & $0.8333$ & $0.8333$ & $0.8332$  & $0.8333$  \\ 
$(4, 5)$ & $0.2$ & $0.8$ & $0.7997$ & $0.7992$  & $0.7994$  \\ 
$(3, 4)$ & $0.25$ & $0.75$ & $0.7460$ & $0.7407$  & $0.7428$  \\ 
$(4, 6)$ & $0.3333$ & $0.6667$ & $0.6657$ & $0.6633$ & $0.6645$  \\ 
$(3, 5)$ & $0.4$ & $0.6$ & $0.5910$ & $0.5772$  & $0.5841$  \\ 
$(3, 6)$ & $0.5$ & $0.5$ & $0.4881$ & $0.4681$  & $0.4794$  \\ 
$(3, 7)$ & $0.5714$ & $0.4286$ & $0.4154$ & $0.3912$  & $0.4057$  \\ 
$(3, 8)$ & $0.6250$ & $0.3750$ & $0.3613$ & $0.3345$  & $0.3514$  \\ 
$(3, 9)$ & $0.6667$ & $0.3333$ & $0.3196$ & $0.2912$  & $0.3099$  \\ 
\end{tabular} 
\caption{ \label{tab:areathreshold} Numerically computed area thresholds for some
dds and channels. }
\end{table}
Recall that the GEXIT integral has a simple graphical interpretation
-- it is the area under the GEXIT curve, assuming of course
that both the curve and the integral exist.  The area threshold is
therefore that channel parameter $\ih^A(\dl, \dr, \{\Ldens{c}_\ent\})$ such that the GEXIT integral
from $\ih^A(\dl, \dr, \{\Ldens{c}_\ent\})$ to $1$ is equal to $1-\frac{\dl}{\dr}$, the design rate.

Consider e.g. the case of the $(10, 20)$-regular dd
depicted in Figure~\ref{fig:maxwell1020}.
\begin{figure}[htp]
\centering
\input{ps/maxwellbsc}
\caption{\label{fig:maxwell1020} 
The area threshold for the $(10,20)$-regular ensemble and
transmission over the BSC.  
We have $\ih^A \approx 0.49985$. For comparison, the BP
threshold is at a channel entropy of roughly $0.2528$.} 
\end{figure}
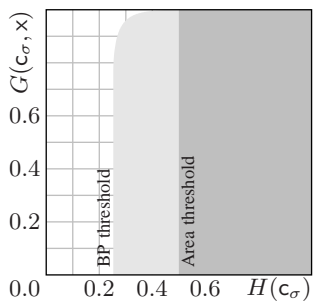
From Lemma~\ref{lem:FPforlargeentropy} we know that the GEXIT curve
is Lipschitz continuous at least in the range $\ih \in [0.341, 1]$.
An explicit check shows that $A(\Ldens{x}_{\ent=0.341}) < 0$, so
that $\ent^A \geq 0.341$.  We know that for $\ent \in [0.341. 1]$
the expression $1-\frac{\dl}{\dl}-A(\Ldens{x}_\ent)$ corresponds
to the area under this GEXIT curve between $\ent$ and $1$. This
expression is therefore a decreasing function in $\ent$, or
equivalently, $A(\Ldens{x}_\ent)$ is an increasing function in
$\ent$.  Using bisection, we can therefore efficiently find the
area threshold and we get $\ih^A \approx 0.49985$.  Note that for
this case the area threshold has the interpretation as that unique
channel parameter $\ent^A$ so that the enclosed area under the GEXIT
curve between $\ent^A$ and $1$ is equal to $1-\frac{\dl}{\dr}$.
This is obviously the reason for calling $\ent^A$ the area threshold.

The same interpretation applies to any dd ($\dl, \dr)$ and any BMS
channel where the area threshold $\ih^A(\dl, \dr, \{\Ldens{c}_\ent\})$
is such that the GEXIT curve from $\ih^A(\dl, \dr, \{\Ldens{c}_\ent\})$
up till $1$ exists and is integrable.  Empirically this is true for
{\em all} regular dds and all BMS channels. Consider e.g. the case
of the $(3, 6)$ ensemble and transmission over the BAWGNC, see Figure~\ref{fig:maxwell}. From
Table~\ref{tab:bounds} we are assured that this curve exists and
is smooth at least in the range $\ih \in [0.5931, 1]$. This region
is unfortunately too small.  But it is easy to compute the curve
numerically over the whole range. Since the resulting curve is
smooth everywhere, it is easy to compute the area threshold numerically
in this way. We get $\ih^A \approx 0.4792$.
\begin{figure}[htp]
\centering
\input{ps/maxwell}
\caption{\label{fig:maxwell} 
The area threshold for the $(3,6)$-regular ensemble and
transmission over the BAWGNC.  This upper bound is given by the entropy
value where the dark gray vertical line hits the $x$-axis.  
Numerically the upper bound
is at a channel entropy of roughly $0.4792$. For comparison, the BP
threshold is at a channel entropy of roughly $0.4291$.} 
\end{figure}
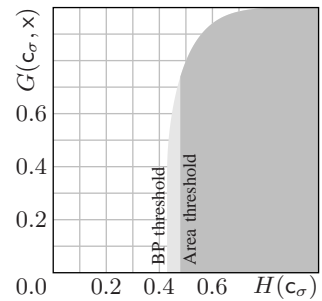

Fortunately, if we fix the rate then for all dd 
of sufficiently high degree this interpretation applies.
\blemma[Area Threshold Approaches Shannon]
\label{lem:areathresholdapproachesshannon}
Consider a sequence of $(\dl, \dr)$-regular ensembles of fixed
design rate $r=1-\dl/\dr$ and with $\dl, \dr$ tending to infinity. 

Assume that $\dr \geq 1 + 5 (\frac1{1-r})^{\frac43}$ and
that $\entLE(\dl, \dr, \{\Ldens{c}_{\ent}\}) < \frac{\dl}{\dr}-\dl e^{-4 (\dr-1)
(\frac{2 (1-r)}{11 e })^{\frac{4}{3}}} $, where $\entLE(\dl, \dr, \{\Ldens{c}_\ent\})$
is defined in Lemma~\ref{lem:continuityforlargeentropy}.
Then for any BMS channel family $\{\Ldens{c}_\ent\}$
\begin{align*} 
\frac{\dl}{\dr}-\dl e^{-4
(\dr-1) (\frac{2 (1-r)}{11 e })^{\frac{4}{3}}} \leq \ih^A(\dl, \dr, \{\Ldens{c}_\ent\}) \leq
\frac{\dl}{\dr}.
\end{align*} 
Furthermore, $A(\Ldens{x}_{\ih^A}, \dl, \dr) = 0$ and, for fixed
rate and increasing degrees, the sequence of the area thresholds
$\ih^{\text{A}}(\dl, \dr, \{\Ldens{c}_\ent\})$ converges to the
Shannon threshold $\ih^{\text{Shannon}}(\dl, \dr)=\frac{\dl}{\dr}
= 1 - r$ {\em universally} over the whole class of BMS channel
families.
\elemma 
\begin{IEEEproof} 
Note that $\entLE \leq \bar{\ih} \leq
\frac{e^{\frac14} \sqrt{2}}{(\dr-2)^{\frac14}} \downarrow_{\dr \rightarrow
\infty} 0$, where $\bar{\ih}$  is the universal upper
bound on $\entLE$ in Lemma~\ref{lem:FPforlargeentropy}.
Thus, $\entLE < \frac{\dl}{\dr}-\dl e^{-4 (\dr-1)
(\frac{2 (1-r)}{11 e })^{\frac{4}{3}}}$ is fulfilled for sufficiently
large degrees.

Let us begin with the lower bound on $\ih^A$. Consider any $\entLE < \ih < \frac{\dl}{\dr}-\dl e^{-4 (\dr-1) (\frac{2 (1-r)}{11 e
})^{\frac{4}{3}}}$.  Let $\Ldens{x}_{\ih}$ be the corresponding BP
FP. Clearly, $\entropy(\Ldens{x}_{\ih}) < \frac{\dl}{\dr}-\dl e^{-4
(\dr-1) (\frac{2 (1-r)}{11 e })^{\frac{4}{3}}}$. Suppose that
$\entropy(\Ldens{x})\in [(\frac34)^{\frac{\dl-1}{2}}+\frac1{(\dr-1)^3},
\frac{\dl}{\dr} - \dl e^{-4 (\dr-1) (\frac{2 \dl}{11 e
\dr})^{\frac{4}{3}}}-\kappa]$. Then from the (Negativity)
Lemma~\ref{lem:asymptoticnegativity} it follows that $A(\Ldens{x}_{\ih})
< 0$ and hence $\ih^A \geq \frac{\dl}{\dr}-\dl e^{-4
(\dr-1) (\frac{2 (1-r)}{11 e })^{\frac{4}{3}}}$. Now suppose that
$\entropy(\Ldens{x}_{\ih})  < (\frac34)^{\frac{\dl-1}{2}}+\frac1{(\dr-1)^3}$
(the left boundary in the Negativity lemma). Since $\ih > \entLE$, we know from Corollary~\ref{cor:continuitywrtwasserstein}
that $\entropy(\Ldens{x}_{\ih})$ is a continuous function wrt $\ih$
with $\entropy(\Ldens{x}_{\ih=1}) = 1$.  Thus, from the mean value
theorem, there must exists a channel entropy $\ih^*$ such that
$\entropy(\Ldens{x}_{\ih^*}) $ lies within the interval
prescribed by the Negativity lemma. Therefore, also in this case
$\ih^A \geq \frac{\dl}{\dr}-\dl e^{-4 (\dr-1) (\frac{2
(1-r)}{11 e })^{\frac{4}{3}}}$.

Let us now consider the upper bound. From above arguments, since $\entLE < \ih^A$,
the BP GEXIT integral from $\ih^A$ to $1$ is given by
Lemma~\ref{lem:areaunderBPGEXIT}. If we combine this with the
definition of the area threshold, i.e.,  the expression $A$
in Lemma~\ref{lem:areaunderBPGEXIT} is non-positive at $\ih^A$,
we get that the BP GEXIT integral at the area threshold is at least equal to $1 -
\frac{\dl}{\dr}$. Now, note that the BP GEXIT curve is always upper
bounded by $1$ and so the integral from $\ih^A$ to $1$ can be at
most equal to $1-\ih^A$. Putting things together we have that $\ih^A
\leq \ih^{\text{Shannon}}=\frac{\dl}{\dr}$.

Let us prove the last claim of the lemma. We want to show that
at the area threshold $A(\Ldens{x}_{\ih^A}, \dl, \dr) = 0$. Recall
that the area threshold was defined as the supremum over all $\ih$
so that $A(\Ldens{x}_{\ih}, \dl, \dr)$ is less than or equal to zero. Therefore,
 all we need to show is that $A(\Ldens{x}_{\ih}, \dl, \dr)$ is
continuous as a function of $\ih$ around $\ih^A$.

Note that $\ih^A$ is strictly larger than $\tilde{\ih}$.
Thus, from Corollary~\ref{cor:continuitywrtwasserstein} we conclude that the
Wasserstein distance $d(\Ldens{x}_{\ih}, \Ldens{x}_{\ih^A})$ is continuous wrt
$\ih$. It is not hard to verify that $A(\Ldens{x}_{\ih}, \dl, \dr)$ is also continuous wrt
the Wasserstein distance. Combining, we get that $A(\Ldens{x}_{\ih}, \dl, \dr)$ is
continuous wrt $\ih$ around $\ih^A$.  
\end{IEEEproof}

\section{Coupled Systems}\label{sec:coupled}
\subsection{Spatially Coupled Ensemble} 
Our goal is to show that coupled ensembles can achieve capacity on
general BMS channels.  Let us recall the definition of an ensemble
which is particularly suited for the purpose of analysis. We call
it the $(\dl, \dr, \Lc, w)$ ensemble.  This is the ensemble we use
throughout the paper. For a quick historical review on some of the
many variants see Section~\ref{sec:priorwork}.

The variable nodes of the ensemble are at positions $[-\Lc, \Lc]$, $\Lc
\in \naturals$. At each position there are $M$ variable nodes, $M \in
\naturals$. Conceptually we think of the check nodes to be located at
all integer positions from $[- \infty, \infty]$.  Only some of these
positions actually interact with the variable nodes.  At each position
there are $\frac{\dl}{\dr} M$ check nodes. It remains to describe how the
connections are chosen.  We assume that each of the $\dl$ connections of
a variable node at position $i$ is uniformly and independently chosen from
the range $[i, \dots, i+w-1]$, where $w$ is a ``smoothing'' parameter. In
the same way, we assume that each of the $\dr$ connections of a check
node at position $i$ is independently chosen from the range $[i-w+1,
\dots, i]$. A detailed construction of this ensemble can be found in
\cite{KRU10}.

For the whole paper we will always be interested in the limit when $M$
tends to infinity while $\Lc$ as well as $\dl$ and $\dr$ stay fixed. In
this limit we can analyze the system via density evolution, simplifying
our task.

Not surprisingly, spatially coupled ensembles inherit many of their
properties from the underlying ensemble. Perhaps most importantly,
the local connectivity is the same. Further, the design rate of the
coupled ensemble is close to that of the original one.  A proof of
the following lemma can be found in \cite{KRU10}.
\begin{lemma}[Design Rate]\label{lem:designrate}
The design rate of the ensemble $(\dl, \dr, \Lc, w)$, with $w \leq \Lc$,
is given by
\begin{align*}
R(\dl, \dr, \Lc, w) & = 
(1-\frac{\dl}{\dr}) - \frac{\dl}{\dr} \frac{w+1-2\sum_{i=0}^{w} 
\bigl(\frac{i}{w}\bigr)^{\dr}}{2 \Lc+1}.
\end{align*}
\end{lemma}

There is an entirely equivalent way of describing a spatially coupled
ensemble in terms of a circular construction. This construction has
the advantage that it is completely symmetric. This simplifies some
of the ensuing proofs.

\begin{definition}[Circular Ensemble] \label{def:circularensemble}
Given an $(\dl, \dr, \Lc, w)$ ensemble we can associate to it a {\em
circular} ensemble. This circular ensemble has $w-1$ extra sections, all
of whose variable nodes are set to zero.  To be concrete, we assume that
the sections are numbered from $[-\Lc, \Lc+w-1]$, where the sections in $[-\Lc,
\Lc]$ are the sections of the original ensemble and the sections in $[\Lc+1,
\Lc+w-1]$ are the extra sections. In this new circular ensemble all index
calculations (for the connections) are done modulo $2 \Lc+w$ and indices
are mapped to the range $[-\Lc, \Lc+w-1]$. For all positions in the range
$i \in [\Lc+1, \Lc+w-1]$ the channel is $\Ldens{c}_i=\Delta_{+\infty}$, and
consequently, $\Ldens{x}_i=\Delta_{+\infty}$. For all 
``regular'' positions $i \in [-\Lc, \Lc]$ the associated channel is
the standard channel $\Ldens{c}$.  This circular
ensemble has design rate equal to $1-\dl/\dr$. 
\qed
\end{definition}

As we will see, it is the global structure which helps all
the individual codes to perform so well -- individually they can
only achieve their BP threshold, but together they reach their MAP
performance.

\subsection{Density Evolution for Coupled Ensemble}\label{sec:decoupled}
Let us describe the DE equations for the $(\dl, \dr, \Lc, w)$ ensemble.
In the sequel, densities are $L$-densities.  Let $\Ldens{c}$ denote the
channel and let $\Ldens{x}_i$ denote the density which is emitted by
variable nodes at position $i$. Throughout the paper, $\Delta_{+\infty}$
denotes an $L$-density with all its mass at $+\infty$ and represents
the perfect decoding density. Also, $\Delta_{0}$ denotes an $L$-density
with all its mass at $0$ and represents a density with no information.

\begin{definition}[DE of the $(\dl, \dr, \Lc, w)$ Ensemble]\label{def:decoupled}
Let $\Ldens{x}_i$, $i\in \integers$, denote the average $L$-density which
is emitted by variable nodes at position $i$. For $i \not \in [-\Lc, \Lc]$
we set $\Ldens{x}_i=\Delta_{+\infty}$.  In words, the boundary variable nodes have
perfect information. For $i \in [-\Lc, \Lc]$, the FP condition implied by
DE is
\begin{align}\label{eq:densevolxi}
\Ldens{x}_i 
& = \Ldens{c} \vconv 
\Bigl(\frac{1}{w} \sum_{j=0}^{w-1} \bigl(\frac{1}{w} \sum_{k=0}^{w-1} 
\Ldens{x}_{i+j-k} \bigr)^{\cconv \dr-1} \Bigr)^{\vconv \dl-1}.
\end{align}
Define
\begin{align*}
g(\Ldens{x}_{i\!-\!w\!+\!1}, \!\dots\!,  \Ldens{x}_{i\!+\!w\!-\!1}) =  
\Bigl(\frac{1}{w} \sum_{j=0}^{w-1} \bigl(\frac{1}{w} \sum_{k=0}^{w-1} 
\Ldens{x}_{i\!+\!j\!-\!k} \bigr)^{\cconv \dr\!-\!1} \Bigr)^{\vconv \dl\!-\!1}.
\end{align*}
Note that 
$g(\Ldens{x}, \dots, \Ldens{x}) = (\Ldens{x}^{\cconv \dr-1})^{\vconv \dl-1}$,
where the right-hand side represents DE (without the effect of the
channel) for the underlying $(\dl, \dr)$-regular ensemble. Also define
\begin{align*}
\hat{g}(\Ldens{x}_{i\!-\!w\!+\!1}, \!\dots\!,  \Ldens{x}_{i\!+\!w\!-\!1}) = 
 \Bigl(\frac{1}{w} \sum_{j=0}^{w-1} \bigl(\frac{1}{w} \sum_{k=0}^{w-1} 
\Ldens{x}_{i\!+\!j\!-\!k} \bigr)^{\cconv \dr\!-\!1} \Bigr)^{\vconv \dl}.
\end{align*}
As before we see that $\hat{g}(\Ldens{x}, \!\dots\!,  \Ldens{x})$ denotes the
EXIT value of DE for the underlying $(\dl, \dr)$-regular ensemble. 
It is not hard to see \cite{RiU08} that both $g(\Ldens{x}_{i-w+1},
\dots, \Ldens{x}_{i+w-1})$ as well as $\hat{g}(\Ldens{x}_{i-w+1}, \dots,
\Ldens{x}_{i+w-1})$ are monotone wrt degradation in all their arguments
$\Ldens{x}_j$, $j=i-w+1, \dots, i+w-1$.  More precisely, if we degrade any
of the densities $\Ldens{x}_j$, $j=i-w+1, \dots, i+1-1$, then $g(\cdot)$
(respectively $\hat{g}(\cdot)$) is degraded.  We say that $g(\cdot)$ (respectively $\hat{g}(\cdot)$)
is {\em monotone} in its arguments.  \qed \end{definition}

\blemma[Sensitivity of DE]\label{lem:sensitivity}
Fix the parameters $(\dl, \dr)$ and $w$ and assume that
$d(\Ldens{a}_i, \Ldens{b}_i) \leq \kappa$, $i=-w+1, \dots, w-1$. Then
\begin{align*}
d(
\Ldens{c}\vconv g(\Ldens{a}_{-w+1}, \dots, \Ldens{a}_{w-1})& , 
\Ldens{c}\vconv g(\Ldens{b}_{-w+1}, \dots, \Ldens{b}_{w-1})) \\
& \leq 2 (\dl-1)(\dr-1) \kappa.
\end{align*}
\elemma
\begin{proof}
For $i \in [0, w-1]$, define $\tilde{\Ldens{a}}_i = \frac1{w} \sum_{k=0}^{w-1} \Ldens{a}_{i-k}$ and
$\tilde{\Ldens{b}}_i = \frac1{w} \sum_{k=0}^{w-1} \Ldens{b}_{i-k}$. 
Set $\Ldens{c}_i = \tilde{\Ldens{a}}_i^{\cconv \dr-1}$ and
$\Ldens{d}_i = \tilde{\Ldens{b}}_i^{\cconv \dr-1}$.
Then using properties 
(\ref{lem:blmetricconvexity}) and (\ref{lem:blmetricregularcconv}) of Lemma~\ref{lem:blmetric}
we see that
\begin{align*}
d(\Ldens{c}_i, \Ldens{d}_i) & \stackrel{\text{(\ref{lem:blmetricregularcconv})}}{\leq} (\dr-1)
d(\tilde{\Ldens{a}}_i, \tilde{\Ldens{b}}_i) \stackrel{\text{(\ref{lem:blmetricconvexity})}}{\leq} (\dr-1) \kappa.
\end{align*}
Using once again property (\ref{lem:blmetricconvexity}) of Lemma~\ref{lem:blmetric}
\begin{align*}
d(\frac1{w} \sum_{i=0}^{w-1} \Ldens{c}_i, \frac1{w} \sum_{i=0}^{w-1} \Ldens{d}_i) \leq (\dr-1) \kappa.
\end{align*}
Finally, using property (\ref{lem:blmetricregularvconv}) of Lemma~\ref{lem:blmetric}
\begin{align*}
& d(\Ldens{c}\vconv g(\Ldens{a}_{-w+1}, \dots, \Ldens{a}_{w-1}), 
\Ldens{c}\vconv g(\Ldens{b}_{-w+1}, \dots, \Ldens{b}_{w-1})) \\
& =  
d(\Ldens{c}\vconv (\frac1{w} \sum_{i=0}^{w-1} \Ldens{c}_i)^{\vconv \dl-1}, 
\Ldens{c}\vconv (\frac1{w} \sum_{i=0}^{w-1} \Ldens{d}_i)^{\vconv \dl-1}) \\ 
& \leq  2 (\dl\!-\!1) (\dr\!-\!1) \kappa.
\end{align*}
\end{proof}

\subsection{Fixed Points and Admissible Schedules}\label{sec:fpandschedules}
\begin{definition}[FPs of Density Evolution]\label{def:fixedpoints}
Consider DE for the $(\dl, \dr, \Lc, w)$ ensemble.
Let $\Ldens{\x}=(\Ldens{x}_{-\Lc}, \dots, \Ldens{x}_{\Lc})$. 
We call $\Ldens{\x}$ the {\em constellation} (of $L$-densities). 
We say that $\Ldens{\x}$ forms a FP
of DE with channel $\Ldens{c}$ if $\Ldens{\x}$ fulfills (\ref{eq:densevolxi})
for $i \in [-\Lc, \Lc]$.  As a short hand we say that $(\Ldens{c}, \Ldens{\x})$
is a FP.  We say that $(\Ldens{c}, \Ldens{\x})$ is a {\em non-trivial}
FP if $\Ldens{x}_i \neq \Delta_{+\infty}$ for at least one $i \in [-\Lc, \Lc]$.
Again, for $i\notin [-\Lc,\Lc]$, $\Ldens{x}_i = \Delta_{+\infty}$.
\qed
\end{definition}

\begin{definition}[Forward DE and Admissible Schedules]\label{def:forwardDE} 
Consider {\em forward} DE for the $(\dl, \dr, \Lc, w)$ ensemble.  More
precisely, pick a channel $\Ldens{c}$. Initialize 
$\Ldens{\x}^{(0)}=(\Delta_0, \dots, \Delta_0)$. Let $\Ldens{\x}^{(\ell)}$ be the result of
$\ell$ rounds of DE. This means that $\Ldens{\x}^{(\ell+1)}$ is generated from
$\Ldens{\x}^{(\ell)}$ by applying the DE equation \eqref{eq:densevolxi} to each
section $i\in [-\Lc,\Lc]$,
\begin{align*}
\Ldens{x}_i^{(\ell+1)} & = \Ldens{c}\vconv g(\Ldens{x}_{i-w+1}^{(\ell)},\dots,\Ldens{x}_{i+w-1}^{(\ell)}).
\end{align*}
We call this the {\em parallel} schedule.

More generally, consider a schedule in which in step $\ell$ an
arbitrary subset of the sections is updated, constrained only by the
fact that every section is updated in infinitely many steps. We call
such a schedule {\em admissible}. We call $\Ldens{\x}^{(\ell)}$
the resulting sequence of constellations.  
\qed
\end{definition}

\begin{lemma}[FPs of Forward DE]\label{lem:forwardDE}
Consider forward DE for the $(\dl, \dr, L, w)$ ensemble.  Let
$\Ldens{\x}^{(\ell)}$ denote the sequence of constellations under an admissible
schedule.  Then $\Ldens{\x}^{(\ell)}$ converges to a FP of DE, with each component being a symmetric $L$-density and this
FP is independent of the schedule.  In particular, it is equal
to the FP of the parallel schedule.
\end{lemma}
\begin{IEEEproof}
Consider first the parallel schedule. We claim that the vectors
$\Ldens{\x}^{(\ell)}$ are ordered, i.e., $\Ldens{\x}^{(0)}\succ \Ldens{\x}^{(1)}\succ \dots
\succ \underline{0}$ (the ordering is section-wise and $\underline{0}$ is the vector of $\Delta_{+\infty}$). This is true since
$\Ldens{\x}^{(0)}=(\Delta_0, \dots, \Delta_0)$, whereas $\Ldens{\x}^{(1)}\prec (\Ldens{c},
\dots, \Ldens{c}) \prec (\Delta_0, \dots, \Delta_0) =\Ldens{\x}^{(0)}$. It now
follows by induction on the number of iterations and the monotonicity
of the function $g(\cdot)$ that the sequence
$\Ldens{\x}^{(\ell)}$ is monotonically decreasing. More precisely, we have
$\Ldens{\x}^{(\ell+1)}_i \prec \Ldens{\x}^{(\ell)}_i$. Hence, from Lemma 4.75 in
\cite{RiU08}, we conclude that each section converges to a limit density
which is also symmetric.  Call the limit $\Ldens{\x}^{(\infty)}$.  Since the DE equations
are continuous it follows that $\Ldens{\x}^{(\infty)}$ is a FP of DE
\eqref{eq:densevolxi} with parameter $\Ldens{c}$. We call $\Ldens{\x}^{(\infty)}$
the {\em FP of forward DE}.

That the limit (exists in general and that it) does not depend on the
schedule follows by standard arguments and we will be brief.  The idea
is that for any two admissible schedules the corresponding computation
trees are nested. This means that if we look at the computation graph
of schedule lets say 1 at time $\ell$ then there exists a time $\ell'$
so that the computation graph under schedule $2$ is a superset of the
first computation graph. To be able to come to this conclusion we have
crucially used the fact that for an admissible schedule every section is
updated infinitely often. This shows that the performance under schedule
2 is at least as good as the performance under schedule 1.  Since the
roles of the schedules are symmetric, the claim follows.  \end{IEEEproof}

\subsection{Entropy, Error and Battacharyya Functionals for Coupled Ensemble}
\begin{definition}[Entropy, Error, and Battacharyya]\label{def:entropy}
Let $\Ldens{\x}$ be a constellation. 
Let $F(\cdot)$ denote either the $\entropy(\cdot)$ (entropy), $\perr(\cdot)$ (error probability), 
or $\batta(\cdot)$ (Battacharyya) functional defined in Section~\ref{sec:BPandDE}.

We define the (normalized) {\em entropy} , {\em error} and  
{\em Battacharyya} functionals of the constellation $\Ldens{\x}$ to be 
\begin{align*}
F(\Ldens{\x}) & = \frac1{2\Lc+1}\sum_{i=-\Lc}^{\Lc} F(\Ldens{x}_i).
\end{align*}
\qed
\end{definition}

\subsection{BP GEXIT Curve for Coupled Ensemble}\label{sec:ebpgexituncoupled}
We now come to a key object, the BP GEXIT curve for the coupled
ensemble.  We have discussed how to compute BP GEXIT curves for
uncoupled ensembles in Section~\ref{sec:ebpgexituncoupled}. For
coupled ensembles the procedure is similar.

In Section \ref{sec:fpandschedules} we have seen that for coupled systems
FPs of forward DE are well defined and that they can be computed by
applying a parallel schedule. This procedure allows us to compute {\em
some} FPs. 

But we can also use DE at fixed entropy, as discussed in
Section~\ref{sec:review}, to compute further FPs (in particular
unstable ones).  More, precisely, fix the desired average entropy
of the constellation, call it $\ent$.  Start with the initialization
$\Ldens{\x}^{(0)}=\underline{\Delta}_0$, the vector of all $\Delta_0$.
In each iteration proceed as follows. Perform one round of DE without
incorporating the channel, i.e., set
\begin{align*}
\Ldens{x}_{i}^{(\ell)} = 
g(\Ldens{x}_{i-w+1}^{(\ell-1)}, \cdots, \Ldens{x}_{i+w-1}^{(\ell-1)}).
\end{align*}
Now find a channel $\Ldens{c}_\sigma \in \{\Ldens{c}_\sigma\}$, assuming
it exists, so that after the convolution with this channel the average
entropy of the constellation is equal to $\ent$. Continue this procedure 
until the constellation has converged (under some suitable metric).

Assume that we have computed (via the above procedure) a complete family 
$\{\Ldens{c}_{\sigma}, \Ldens{\x}_{\sigma} \}$ of FPs of DE,
i.e., a family so that for each $\ih\in [0,1]$, there exists a parameter $\sigma$ so that
$\ih = \frac{1}{2 \Lc+1} \sum_{i=-\Lc}^{\Lc} \entropy(\Ldens{x}_{\sigma, i})$. Then
we can derive from it a BP GEXIT curve by projecting it onto 
\begin{align*}
\Big\{\entropy(\Ldens{c}_{\sigma}), 
\frac{1}{2 \Lc+1} \sum_{i=-\Lc}^{\Lc} G(\Ldens{c}_\sigma, 
\hat{g}(\Ldens{x}_{\sigma, i\!-\!w\!+\!1}, \!\dots\!,  \Ldens{x}_{\sigma, i\!+\!w\!-\!1})
)\Big\},
\end{align*}
where $\hat{g}(\cdot)$ was introduced in Section~\ref{sec:decoupled}, and 
$ \frac{1}{2 \Lc+1} \sum_{i=-\Lc}^{\Lc} G(\Ldens{c}_\sigma, 
\hat{g}(\Ldens{x}_{\sigma, i\!-\!w\!+\!1}, \!\dots\!,  \Ldens{x}_{\sigma, i\!+\!w\!-\!1})
)$ is the (normalized) GEXIT function of the constellation $\Ldens{\x}_{\sigma}$. 
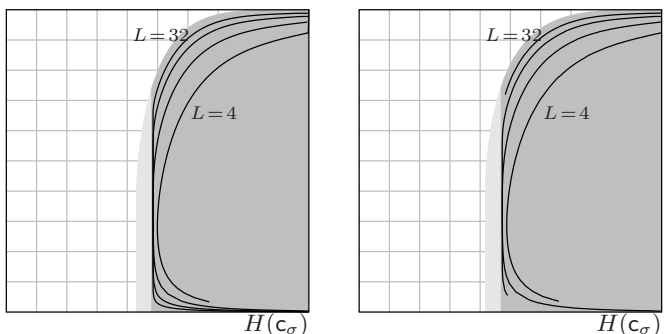
\begin{figure}[htp]
\begin{centering}
\input{ps/lrLexit}
\caption{BP GEXIT curves of the ensemble $(\dl=3, \dr=6, \Lc)$ for $\Lc=
4, 8, 16,$ and $32$ and transmission over the BAWGNC
(left) and the BSC (right).  The BP thresholds are 
$\ent^{\BPsmall}_{\BAWGNC/\BSC}(3, 6, 4)=0.4992/0.4878$, 
$\ent^{\BPsmall}_{\BAWGNC/\BSC}(3, 6, 8)=0.4850/0.47303$, 
$\ent^{\BPsmall}_{\BAWGNC/\BSC}(3, 6, 16)=0.4849/0.4729$, 
$\ent^{\BPsmall}_{\BAWGNC/\BSC}(3, 6, 32)=0.4849/0.4729$.
The light/dark gray areas mark the interior of the
BP/MAP GEXIT function of the underlying $(3, 6)$-regular ensemble,
respectively.}
\label{fig:lrLexit}
\end{centering}
\end{figure}
Figure~\ref{fig:lrLexit} shows the result of this numerical computation
when transmission takes place over the
BAWGNC (left-hand side) and the BSC (right-hand side).  Note that the resulting curves look  
similar to the curves when transmission takes place over the BEC, see
\cite{KRU10}. For small values of $\Lc$ the curves are far to the
right due to the significant rate loss that is incurred at the boundary.
For $\Lc$ around $10$ and above, the BP threshold of each ensemble is 
close to the area threshold of the underlying $(3, 6)$-regular ensemble,
namely $0.4792$ for the BAWGNC and
$0.4680$ for the BSC
(see the values in
Table~\ref{tab:areathreshold}).  The picture suggests that the threshold
saturation effect which was shown analytically to hold for the BEC in
\cite{MMRU09} also occurs for general BMS channels.

The aim of this paper is to prove rigorously that the situation
is indeed as indicated in Figure~\ref{fig:lrLexit}, i.e., that the BP
threshold of coupled ensembles is essentially equal to the area threshold
of the underlying uncoupled ensemble.

\subsection{Review for the BEC}
Let us briefly recall the main result of \cite{KRU10} which deals with
transmission over the BEC.
Let
$\epsilon^{\BPsmall}_{\BEC}(\dl, \dr, \Lc, w)$ and $\epsilon^{\MAPsmall}_{\BEC}(\dl, \dr,
\Lc, w)$ denote the BP threshold and the MAP threshold of the $(\dl, \dr,
\Lc, w)$ ensemble. Also, let $\epsilon^{\MAPsmall}_\BEC(\dl, \dr)$ denote the
MAP threshold of the underlying $(\dl, \dr)$-regular LDPC ensemble. Then
the main result of \cite{KRU10} states that
\begin{align*}
\lim_{w\to \infty}\!\lim_{\Lc\to \infty}\!\epsilon^{\BPsmall}_\BEC(\dl,\!\dr,\!\Lc,\!w)
& =\! 
\lim_{w\to \infty}\!\lim_{\Lc\to \infty}\!\epsilon^{\MAPsmall}_\BEC(\dl,\!\dr,\!\Lc,\!w) \\
& = \!\epsilon^{\MAPsmall}_\BEC(\dl, \dr).
\end{align*}
Also, (see \cite{RiU08}) as $\dl, \dr \to \infty$, with the ratio $\dl/\dr$
fixed, $\epsilon^{\MAPsmall}_\BEC(\dl,\dr) \to \dl/\dr$. 
Thus,
with increasing degrees, $(\dl, \dr, \Lc, w)$ ensembles under BP decoding
achieve the Shannon capacity for the BEC.

\subsection{First Result}\label{sec:firstresult}
Before we state and prove our main result (namely that coupled codes can
achieve capacity also for general BMS channels), let us first quickly
discuss a simple argument which shows that spatial coupling of codes
does have a non-trivial effect.

First consider the uncoupled case. We have seen in
Lemma~\ref{lem:bpboundsuncoupled} that when we fix the design rate
$1-\dl/\dr$ and increase the degrees the BP threshold converges to
$0$.  What happens if we couple such ensembles? We know that for
the BEC such ensembles achieve capacity. The next lemma asserts
that this implies a non-trivial BP threshold also for general BMS
channels.  
\blemma[Lower Bound on Coupled BP Threshold] Consider
transmission over an ordered and complete family $\{\Ldens{c}_\ent\}$
of BMS channels using a $(\dl, \dr, \Lc, w)$ ensemble and BP
decoding.

Let $\ent^{\BPsmall}=\ent^{\BPsmall}(\dl, \dr, \Lc, w, \{\Ldens{c}_\ent\})$ denote the
corresponding BP threshold and let
$\epsilon^{\BPsmall}=\epsilon^{\BPsmall}(\dl, \dr, \Lc, w)$ denote
the corresponding  BP threshold for transmission over the BEC.
Then
\begin{align}\label{equ:bmsversusbec}
\batta(\Ldens{c}_{\ent^{\BPsmall}(\dl, \dr, \Lc, w, \{\Ldens{c}_\ent\})}) \geq \epsilon_{\BPsmall}.
\end{align}
In particular, for every $\delta>0$ there exists a $w \in \naturals$ and a dd 
pair $(\dl, \dr)$ with $\dl/\dr$ fixed,  so that
\begin{align*}
\batta(\Ldens{c}_{\ent^{\BPsmall}(\dl, \dr, \Lc, w, \{\Ldens{c}_\ent\})}) \geq \dl/\dr-\delta.
\end{align*}
\elemma
\begin{IEEEproof}
Consider DE of the coupled ensemble (cf.
\eqref{eq:densevolxi}). Applying the Battacharyya functional, we get
\begin{align}\label{eq:battaxi}
\batta(\Ldens{x}_i) 
& =\batta(\Ldens{c}_\ent)  
 \Bigg(\batta\Bigl(\frac{1}{w} \sum_{j=0}^{w-1} \bigl(\frac{1}{w} \sum_{k=0}^{w-1} 
\Ldens{x}_{i+j-k} \bigr)^{\cconv \dr-1} \Bigr)\Bigg)^{\dl-1},
\end{align}
where we use the multiplicative property of the Battacharyya functional
at the variable node side.

Using the linearity of the Battacharyya functional and extremes of
information combining bounds for the check node convolution (\cite[Chapter
4]{RiU08}) we get
\begin{align}\label{eq:extremeinfoxi}
\batta(\Ldens{x}_i) 
& \leq \batta(\Ldens{c}_\ent) \Bigl(1 \!-\!\frac1{w}  
\sum_{j=0}^{w-1} \bigl(1 \!-\! \frac{1}{w} \sum_{k=0}^{w-1} 
\batta(\Ldens{x}_{i+j-k})\bigr)^{\dr-1} \Bigr)^{\dl-1}.
\end{align}
The preceding set of equations is formally equivalent to the DE equations
for the same spatially coupled ensemble and the BEC. Therefore, if
$\batta(\Ldens{c}_\ent) < \epsilon^{\BPsmall}(\dl, \dr, \Lc, w)$ then
the DE recursions, initialized with $\Ldens{c}_\ent$ must converge to
$\Delta_{+\infty}$, which implies (\ref{equ:bmsversusbec}).

Further, from \cite{KRU10} we know that for sufficiently large degrees $(\dl, \dr)$, with
their ratio fixed, and with $w$ sufficiently large, $\epsilon^{\BPsmall}(\dl, \dr, \Lc, w)$ approaches
$\dl/\dr$ arbitrarily closely (see the discussion in the preceding section), which proves
the final claim.  \end{IEEEproof}

\begin{example}[$(3, 6)$ Ensemble and BSC$(p)$]
Let us specialize to the case of transmission over the BSC using $(3,6)$-regular ensemble. Then we have
$\batta(\Ldens{c}) = 2\sqrt{p(1-p)}$. Using the above
argument and solving for $\epsilon$ in $2\sqrt{\epsilon(1-\epsilon)}
> \frac12,$ we  conclude that by a proper choice of $w$ and $(\dl, \dr)$ we can
transmit reliably at least up to an error probability of
$0.067$.
\end{example}

Combining the above result with Lemma~\ref{lem:entropyvsbatta} we conclude that
the BP threshold of the coupled ensemble is at least $(\dl/\dr)^2 - \delta$.
In summary, for general BMS channels and regular ensembles of fixed rate and
increasing degrees, their uncoupled BP threshold tends to $0$ but their coupled
BP threshold is lower bounded by a non-zero value.  We conclude that coupling
changes the performance in a fundamental way.  In the rest of the paper we will
strengthen the above result by showing that this non-zero value is in fact the
area threshold of the underlying ensemble and as degrees become large, this
will tend to the Shannon threshold, $\dl/\dr$.

\section{Main Results}\label{sec:main}
\subsection{Admissible Parameters}
In the sequel we will impose some restrictions on the
parameters. Rather than repeating these restrictions in each
statement, we collect them once and for all and give them a name.
\begin{definition}[Admissible Parameters]\label{def:admissible}
Fix the design rate $r$ of the uncoupled system. 
We say that the parameters $(\dl, \dr)$ and $w$ are {\em admissible} if the following conditions 
are fulfilled with $r=1-\frac{\dl}{\dr}$:
\renewcommand\theenumi{\roman{enumi}} 
\renewcommand{\labelenumi}{(\roman{enumi})}
\begin{enumerate}
\item \label{equ:admissibleone} $\dr \geq \sqrt{3} b \ln(b)$, $b=\frac{6}{\ln(4/3)(1-r)}$,
\item \label{equ:admissiblenine} $2(\dl-1)(\dr-1)(1-c^2)^{\frac{\dr-2}2} < 1$, $c=(1-r)(1-\dr e^{-4 (\dr-1) (\frac{2 (1-r)}{11 e })^{\frac{4}{3}}})-\frac{1}{\dr}.$
\item \label{equ:admissibletwo} $\entLE(\dl, \dr, \{\Ldens{c}_{\ent}\}) \leq (1-r)(1-\dr e^{-4 (\dr-1) (\frac{2 (1-r)}{11 e })^{\frac{4}{3}}})-\frac{1}{\dr}$, 
where $\entLE(\dl, \dr, \{\Ldens{c}_{\ent}\})$ is the bound stated in Lemma~\ref{lem:continuityforlargeentropy},
\item \label{equ:admissiblethree} $w > 2 \dl^3 \dr^2$,
\item \label{equ:admissiblefour} $w > 2 (\dl-1)(\dr-1) (\frac{16 \sqrt{2} \dr}{\ln(2) \dl})^2$,
\item \label{equ:admissiblefive} $w > 2 (\dl-1) (\dr-1) \dr^2  (4(\sqrt{2}+\frac2{\ln 2}\dl (\dr-1)))^2$,
\end{enumerate}
We say that the ensemble $(\dl, \dr, \Lc, w)$ is admissible if the
parameters $(\dl, \dr)$ and $w$ are admissible.  If we are only
concerned about the conditions on $(\dl, \dr)$, then we will say
that $(\dl, \dr)$ is admissible.
\qed
\end{definition}
{\em Discussion:}
Conditions  (\ref{equ:admissibleone}), (\ref{equ:admissiblenine})
and  (\ref{equ:admissibletwo}) are fulfilled if we take the degrees
sufficiently large.  Conditions (\ref{equ:admissiblethree}),
(\ref{equ:admissiblefour}), and (\ref{equ:admissiblefive}) can all
be fulfilled by picking a sufficiently large connection width $w$.

Why do we impose these conditions?  At several places we use simple
extremes of information combining bounds and these bounds are loose
and require, for the proof to work, the above conditions.  We believe
that with sufficient effort these bounds can be tightened and so
the restrictions on the degrees can be removed or at least significantly
loosened. We leave this as an interesting open problem.

Numerical experiments indicated that for any $3 \leq \dl \leq \dr$
and $w \geq 2$ the threshold saturation phenomenon happens, with a
``wiggle-size'' which vanishes exponentially in $w$.

Note that the above bounds imply the following bounds which we will
need at various places:
\renewcommand\theenumi{\roman{enumi}} 
\renewcommand{\labelenumi}{(\roman{enumi})}
\begin{enumerate}
\setcounter{enumi}{6}
\item \label{equ:admissiblesix} $\dr \geq \frac1{1-r}(1+\frac{2}{\ln(4/3)} \ln(2(\dr-1)^3))$,
\item \label{equ:admissibleseven} $\dr \geq 1 + 5 (\frac1{1-r})^{\frac43}$.
\end{enumerate}
Instead of condition (\ref{equ:admissibletwo}) above we can impose
the stronger but somewhat easier to check condition $\bar{\ih} \leq
(1-r)(1-\dr e^{-4 (\dr-1) (\frac{2 (1-r)}{11 e
})^{\frac{4}{3}}})-\frac{1}{\dr}$, where $\bar{\ih}$ is the upper
bound stated in Lemma~\ref{lem:FPforlargeentropy}, or even further
strengthen the condition to $\frac{e^{\frac14} \sqrt{2}}{(\dr-2)^{\frac14}}
\leq (1-r)(1-\dr e^{-4 (\dr-1) (\frac{2 (1-r)}{11 e
})^{\frac{4}{3}}})-\frac{1}{\dr}$. The last condition can be easily checked to
be satisfied for sufficiently large degrees.   

\subsection{Main Result}
\begin{theorem}[BP Threshold of the $(\dl, \dr, \Lc, w)$
Ensemble]\label{the:main} Consider transmission over a complete,
smooth, and ordered family of BMS channels, denote it by
$\{\Ldens{c}_\ent\}$, using the admissible ensemble $(\dl, \dr,
\Lc, w)$.  Let $\ent^{\BPsmall}(\dl, \dr, \Lc, w, \{\Ldens{c}_\ent\})$
and $\ent^{\MAPsmall}(\dl, \dr, \Lc, w,  \{\Ldens{c}_\ent\})$ denote
the corresponding BP and MAP threshold.  Further, let  $R(\dl, \dr,
\Lc, w)$ denote the design rate of this ensemble and set $r=1-\dl/\dr$.
Finally, let $\ent^{A}(\dl, \dr,  \{\Ldens{c}_\ent\})$  denote  the
area threshold of the underlying $(\dl, \dr)$-regular ensemble and
the given channel family.  Then
\begin{align}
& \ent^{A}(\dl, \dr, \{\Ldens{c}_\ent\}) \!-\!  f(\dl, \dr, w) \nonumber \\
& \leq \ent^{\BPsmall}(\dl, \dr, \Lc, w,  \{\Ldens{c}_\ent\}) \label{lem:bplowerbound} \\
& \leq \ent^{\MAPsmall}(\dl, \dr, \Lc, w,  \{\Ldens{c}_\ent\}) \nonumber \\
& \leq \ent^{A}(\dl, \dr,  \{\Ldens{c}_\ent\}) \!+\! \frac{(w-1)(\dr-1)^3}{\Lc} \label{equ:mapupperbound},
\end{align}
where
$f(\dl, \dr, w) = 8 (\dr-1)^3\big(\sqrt{2}+\frac2{\ln 2}\dl (\dr-1)\big) \sqrt{\frac{2 (\dl-1)(\dr-1)}{w}}$. 
Note that $f(\dl, \dr, w)$ depends only on the dd $(\dl, \dr)$ and
$w$ but is universal wrt the channel family $\{\Ldens{c}_\ent\}$.
Furthermore,
\begin{align}
\lim_{w \rightarrow \infty}
\lim_{\Lc \rightarrow \infty}
R(\dl, \dr, \Lc, w) & = 1 - \frac{\dl}{\dr}. \label{equ:ratelimit}
\end{align}
\end{theorem}
{\em Discussion}:                                                  
\begin{enumerate}[(i)]                                                    
\item The bound $ \ent^{\BPsmall} \leq \ent^{\MAPsmall} $ is trivial
and only listed for completeness.  Consider the upper bound on
$\ent^{\MAPsmall}$ stated in (\ref{equ:mapupperbound}). Start with
the circular ensemble stated in Definition~\ref{def:circularensemble}.
The original ensemble is recovered by setting the $w-1$ consecutive
positions in $[\Lc, \Lc+w-1]$ to $0$.  Define $K=2\Lc+w$.  We first
provide a lower bound on the conditional entropy for the circular
ensemble when transmitting over a BMS channel with entropy $\ent$.
We then show that setting $w-1$ sections to $0$ does not significantly
decrease this entropy. Overall this gives an upper bound on the MAP
threshold of the coupled ensemble in terms of the area threshold
of the underlying ensemble.

It is not hard to see that the BP GEXIT curve is the same for both the
$(\dl, \dr)$-regular ensemble and the circular ensemble (when all sections have the standard channel).  Indeed, 
forward DE (see Definition~\ref{def:forwardDE}) converges to the same FP
for both ensembles. Consider the circular ensemble and let
$\ent \in (\ent^{A}, 1]$.  The conditional entropy when
transmitting over the BMS channel with entropy $\ent$ is at least equal
to $1-\dl/\dr$ minus the area under the BP EXIT curve of $[\ent, 1]$
(see Theorem 3.120 in \cite{RiU08}).  Indeed, from the proof of Theorem 4.172 
in \cite{RiU08}, we have
$$
\liminf_{n\to \infty} \mathbb{E}[\entropy(X_1^n \mid Y_1^n(\ent))]/n \geq 1 - \frac{\dl}{\dr} - G(\{\Ldens{c}_{\ih},
\Ldens{x}_{\ih} \}_{\ih}^{1}).
$$

Note that the above integral, $G(\{\Ldens{c}_{\ih}, \Ldens{x}_{\ih}
\}_{\ih}^{1})$ is evaluated at the BP FPs. From
Lemmas~\ref{lem:FPforlargeentropy} and
\ref{lem:areathresholdapproachesshannon}, the BP FP densities $\Ldens{x}_{\ih}$
exist and the GEXIT integral is well-defined for all $\ih \geq \ih^A \geq
\entLE$.

Here, the entropy is normalized by $n=K M$, where $K$ is the length
of the circular ensemble and $M$ denotes the number of variable
nodes per section.  Assume that we set $w-1$ consecutive sections
of the circular ensemble to $0$ in order to recover the original
ensemble.  As a consequence, we ``remove'' an entropy (degrees of
freedom) of at most $(w-1)/K$ from the circular system. The remaining
entropy is therefore positive (and hence we are above the MAP
threshold of the coupled ensemble) as long as
$1-\dl/\dr-(w-1)/K-G(\{\Ldens{c}_{\ih}, \Ldens{x}_{\ih} \}_{\ih}^{1})
> 0$. From Lemmas~\ref{lem:areaunderBPGEXIT} and
\ref{lem:areathresholdapproachesshannon} we have $G(\{\Ldens{c}_{\ih},
\Ldens{x}_{\ih} \}_{\ih^A}^{1})=1-\dl/\dr$, so that the condition
becomes $ G(\{\Ldens{c}_{\ih}, \Ldens{x}_{\ih} \}_{\ih^A}^{1}) -
G(\{\Ldens{c}_{\ih}, \Ldens{x}_{\ih} \}_{\ih}^{1}) < (w-1)/K$. For all
channels with $\ent \geq \ent^{A}$ we have $ G(\Ldens{c}_{\ih},
\Ldens{x}_{\ih}) \geq \frac1{2(\dr-1)^3}$. For a derivation of this
statement we refer the reader to the proof of part (vi) of
Theorem~\ref{thm:existenceintermediateform}. This implies that
$ G(\{\Ldens{c}_{\ih},
\Ldens{x}_{\ih} \}_{\ih^A}^{\ih}) \geq (\ih - \ih^A)/(2(\dr-1)^3)$.
Furthermore,
$G(\{\Ldens{c}_{\ih}, \Ldens{x}_{\ih} \}_{\ih}^{1}) \leq
G(\{\Ldens{c}_{\ih}, \Ldens{x}_{\ih} \}_{\ih^A}^{1})$. This follows
from the definition of area threshold, which implies that for
$\ih>\ih^A$, $A(\Ldens{x}_{\ih}, \dl, \dr) > 0$ (cf.
Lemma~\ref{lem:areaunderBPGEXIT}) and then combining with
Lemma~\ref{lem:areaunderBPGEXIT}.  Putting things together we get
$$
G(\{\Ldens{c}_{\ih}, \Ldens{x}_{\ih}
\}_{\ih^A}^{1}) - G(\{\Ldens{c}_{\ih}, \Ldens{x}_{\ih}
\}_{\ih}^{1})
> \frac{\ent-\ent^{A}}{2(\dr-1)^3}. 
$$
We get the stated condition on $\ent^{\MAPsmall}$
by lower bounding $K$ by $2 \Lc$.

\item                                               
The lower bound  on $\ent^{\BPsmall}(\dl,\dr,\Lc,w,\{\Ldens{c}_{\ih}\})$ expressed in
\eqref{lem:bplowerbound} is the main result of this paper. It shows
that, up to a term which tends to zero when $w$ tends to infinity,
the BP threshold of the coupled ensemble is at least as large as the area threshold 
of the underlying ensemble.

Empirical evidence suggests that the convergence speed wrt $w$ is
exponential.  Our bound only guarantees a convergence speed of order
$\sqrt{1/w}$.
\end{enumerate}

Let us summarize. In order to prove  Theorem~\ref{the:main} we
``only'' have to prove the lower bound on $\ih^{\BPsmall}$. 
Not surprisingly, this is also the most difficult
to accomplish.  The remainder of this paper is dedicated to this
task.

\subsection{Extensions}
In Theorem~\ref{the:main} we start with a smooth, complete and
ordered channel family.  But it is straightforward to convert this
theorem and to apply it directly to single channels or to a collection
of channels. The next statement makes this precise.  
\begin{corollary}[$(\dl,\dr,\Lc,w)$ Universally Achieves Capacity]\label{cor:capachieving}
The $(\dl, \dr, \Lc, w)$ ensemble is universally capacity achieving
for the class of BMS channels. More precisely, assume we are given
$\epsilon>0$ and a target rate $R$. Let ${\mathcal C}(R)$
denote the set of BMS channels of capacity at least $R$.
To each $\Ldens{c} \in {\mathcal C}(R)$ associate the
family $\{\Ldens{c}_\ent\}_{\ent =0}^1$, by
defining 
\begin{align*}
\Ldens{c}_\ent = 
\begin{cases}
\frac{1}{\entropy(\Ldens{c})} [(\entropy(\Ldens{c})-\ent) \Delta_{+\infty} + \ent \Ldens{c}], & 0 \leq \ent \leq 
\entropy(\Ldens{c}), \\
\frac{1}{1-\entropy(\Ldens{c})} [(\ent- \entropy(\Ldens{c}))\Delta_0 + (1-\ent)) \Ldens{c}], & \entropy(\Ldens{c}) \leq \ent \leq 1. 
\end{cases}
\end{align*}
Then there
exists a set of parameters $(\dl, \dr, \Lc, w)$ so that
\begin{align*}
R(\dl, \dr, \Lc, w) \geq R-4\epsilon, \\
\inf_{\Ldens{c} \in {\mathcal C}(R)} 
\ent^{\BPsmall}(\dl, \dr, \Lc, w, \{\Ldens{c}_\ent\}) \geq 1-R+\epsilon.
\end{align*}
Since for each $\Ldens{c} \in {\mathcal C}(R)$ the associated family
$\{\Ldens{c}_\ent\}_{\ent =0}^1$ is ordered by
degradation, this implies that we can transmit with this ensemble
reliably over each of the channels in ${\mathcal C}(R)$ at a rate
of at least $R-4 \epsilon$, i.e., arbitrarily close to the Shannon limit.
\end{corollary}
\begin{IEEEproof}
Fix the ratio of the degrees so that $R-3\epsilon \leq 1-\dl/\dr \leq R-2\epsilon$.
Note that for each $\Ldens{c} \in {\mathcal C}(R)$ the constructed
family $\{\Ldens{c}_\ent\}$ is piece-wise smooth, ordered and
complete.  By applying Theorem~\ref{the:main} to each such channel
family we conclude that for admissible parameters (i.e., as long
as we choose the degrees and the connection width sufficiently
large) the threshold of the ensemble $(\dl, \dr, w, L)$ for the
given channel family is at least $\ent^A(\dl, \dr,
\{\Ldens{c}_\ent\})-f(\dl, \dr, w)$, where $\ent^A(\dl, \dr,
\{\Ldens{c}_\ent\})$ is the area threshold and $f(\dl, \dr, w)$ is
a universal quantity, i.e., a quantity which does not depend on the
channel family and which converges to $0$ when $w$ tends to infinity.
Further, we know from Lemma~\ref{lem:areathresholdapproachesshannon}
that the area threshold $\ent^A(\dl, \dr, \{\Ldens{c}_\ent\})$
approaches the Shannon threshold uniformly over all BMS channels
for increasing degrees. By our choice of $(\dl, \dr)$ the Shannon threshold is 
$1-(1-\dl/\dr) \geq 1-R+2\epsilon$.
Therefore, by first choosing sufficiently
large degrees $(\dl, \dr)$, and then a sufficiently large connection
width $w$, we can ensure that the BP threshold is at least
$1-R+\epsilon$.  Finally, by choosing the constellation length $\Lc$
sufficiently large, we can ensure that the rate loss we incur with
respect to the design rate the underlying ensemble is sufficiently
small so that the design rate of the coupled ensemble is at least
$R-4 \epsilon$.  \end{IEEEproof}

\begin{corollary}[Universally Capacity Achieving
Codes]\label{cor:capachievingcodes} Assume we are given $\epsilon>0$
and a target rate $R$. Let ${\mathcal C}(R)$ denote the
set of BMS channels of capacity at least $R$.  Then there
exists a set of parameters $(\dl, \dr, \Lc, w)$ of rate at least
$R-5 \epsilon$ with the following property.  Let $C(n)$ be an element of $(\dl,
\dr, \Lc, w)$ with blocklength $n$, where we assume that $n$ only
goes over the subsequence of admissible values. Then
\begin{align*}
\lim_{n \rightarrow \infty} 
\mathbb{E}_{C(n) \in (\dl, \dr, \Lc, w)}[\ind_{\{ \sup_{\Ldens{c} \in {\mathcal C}(R)} P_b^{\BPsmall}(C(n), \Ldens{c}) \leq \epsilon \}}] = 1.
\end{align*}
In words, almost all codes in
$(\dl, \dr, \Lc, w)$ of sufficient length are good for all channels in 
${\mathcal C}(R)$.
\end{corollary}
\begin{proof}
Note that according to (\ref{lem:blmetricpolish}) in
Lemma~\ref{lem:blmetric} the space of $|D|$ distributions endowed with
the Wasserstein metric is compact, and hence so is ${\mathcal
C}(R)$. Hence there exists a finite set of channels, denote
it by $\{\Ldens{c}_i\}_{i=1}^{I(\delta)}$, so that each channel in
${\mathcal C}(R)$ is within a (Wasserstein) distance at most $\delta$
from the set $\{\Ldens{c}_i\}$. We will fix the value of $\delta$
shortly.

Let us modify the set
$\{\Ldens{c}_i\}$ so that ${\mathcal C}(R)$ is not only close to
$\{\Ldens{c}_i\}$ but is also ``dominated'' by it.
For each $\Ldens{c} \in \{\Ldens{c}_i\}$,
define  
\begin{align*}
\tilde{\absDdist{c}}(y) =
\begin{cases}
\sqrt{\delta}+(1-\sqrt{\delta})\absDdist{c}(y), & 0 \leq y \leq z^*(\absDdist{c}), \\
1, & z^*(\absDdist{c}) \leq y \leq 1, \\
\end{cases}
\end{align*}
where  $z^*(\absDdist{c})$ is the supremum of all $z$ so that
$\int_z^1 (1-\absDdist{c}(y)) \dee y=\sqrt{\delta}$. If no such $z
\in [0, 1]$ exists then set $z^*(\absDdist{c})=0$. We claim that
for any $\Ldens{a}$ so that $d(\Ldens{a}, \Ldens{c}) \leq \delta$,
$\Ldens{a} \prec \tilde{\Ldens{c}}$. In other words we claim that
$\int_z^1 \absDdist{a}(y) \dee y \leq \int_z^1 \tilde{\absDdist{c}}(y)
\dee y$ for any $z \in [0, 1]$ (cf. \eqref{equ:degradationcdfs}).

For $z^*(\absDdist{c}) \leq z \leq 1$,  $\int_z^1 \tilde{\absDdist{c}}(y)
\dee y =1-z$, the maximum possible, and hence this integral is at least as large as
$\int_z^1 \absDdist{a}(y) \dee y$. Consider therefore the range $0
\leq z \leq  z^*(\absDdist{c})$. In this case
\begin{align*}
& \int_z^1 \tilde{\absDdist{c}}(y) \dee y 
\stackrel{\text{(a)}}{\geq} \sqrt{\delta} (1-z) + (1-\sqrt{\delta}) \int_z^1 \absDdist{c}(y) \dee y \\
& \stackrel{\text{(b)}}{=} \int_z^1 \absDdist{c}(y) \dee y + \sqrt{\delta} \int_z^1 (1-\absDdist{c}(y)) \dee y 
 \geq \int_z^1 \absDdist{c}(y) \dee y + \delta  \\
& \stackrel{\text{(c)}}{\geq} \int_z^1 \absDdist{c}(y) \dee y +   \int_z^1 |\absDdist{a}(y) - \absDdist{c}(y) | \dee y  
 \geq \int_z^1 \absDdist{a}(y) \dee y.
\end{align*}
In (a) we use the definition of $\tilde{\absDdist{c}}(y)$. To obtain
(b) we use that for $z\leq z^*(\absDdist{c})$ we have $\int_z^1
(1-\absDdist{c}(y))
 \geq \int_{ z^*(\absDdist{c})}^1 (1-\absDdist{c}(y)) = \sqrt{\delta}$.
Finally, in (c) we use the alternative definition of the Wasserstein
distance in Lemma~\ref{lem:blmetric}.

Further,
\begin{align*}
& d( \tilde{\Ldens{c}}, \Ldens{a})  \leq d( \tilde{\Ldens{c}}, \Ldens{c})+ d( \Ldens{c}, \Ldens{a}) 
\leq \int_0^1 |\tilde{\absDdist{c}} - \absDdist{C}(y) | \dee y  +\delta \\
\leq & \int_0^{z^*} |\sqrt{\delta}(1-\absDdist{c}(y))| \dee y  + \int_{z^*}^1 (1-\absDdist{c}(y)) \dee y +\delta  \leq 3 \sqrt{\delta}.
\end{align*}
In words, any density $\Ldens{a}$ which was close to $\Ldens{c}$
is still close to $\tilde{\Ldens{c}}$. We have therefore the set
$\{\tilde{\Ldens{c}}_i\}_{i=1}^{I(\delta)}$ of channels which
``cover'' and ``dominate'' the set of channels ${\mathcal C}(R)$ in
the sense that for every $\Ldens{a} \in {\mathcal C}(R)$ there exists
an element $\tilde{\Ldens{c}}_i \in
\{\tilde{\Ldens{c}}_i\}_{i=1}^{I(\delta)}$ so that $d(\Ldens{a},
\tilde{\Ldens{c}}_i) \leq 3 \sqrt{\delta}$ and $\Ldens{a} \prec
\tilde{\Ldens{c}}_i$. This implies in particular that $\min_i
1-\entropy(\tilde{\Ldens{c}}_i) \geq R-h_2( \frac32\sqrt{\delta})
\geq R-\epsilon$, where in the last step we use the relation between the
Wasserstein distance and entropy given by (ix) in Lemma~\ref{lem:blmetric}, 
also we assumed that we fixed
$\delta$ so that $h_2( \frac32\sqrt{\delta}) \leq \epsilon$. In
words, all channels in $\{\tilde{\Ldens{c}}_i\}_{i=1}^{I(\delta)}$
have capacity at least $R-\epsilon$.

From Corollary~\ref{cor:capachieving} we know that, given a finite
set of channels from ${\mathcal C}(R-\epsilon)$, there exists a set
of parameters $(\dl, \dr, \Lc, w)$ which has rate at least $R-5
\epsilon$ and BP threshold at least $1-R+2\epsilon$ universally for
the whole family. Since each element of
$\{\tilde{\Ldens{c}}_i\}_{i=1}^{I(\delta)}$ is an element of
${\mathcal C}(R-\epsilon)$ this ensemble ``works'' in particular
for all channels $\{\tilde{\Ldens{c}}_i\}_{i=1}^{I(\delta)}$ and
these channels ``dominate'' all channels in ${\mathcal C}(R)$ in
the sense that for element of $\Ldens{c} \in {\mathcal C}(R)$ there
is an element of $\{\tilde{\Ldens{c}}_i\}_{i=1}^{I(\delta)}$ which
is degraded wrt $\Ldens{c}$.

For each element $\tilde{\Ldens{c}}_{i}$ we know by standard
concentration theorems that ``almost all'' elements of the ensemble 
have a bit error rate of the BP decoder going to zero
\cite{RiU08, RiU01}.  Since the ``almost all'' means all but an
exponentially (in the blocklength) small subset and since we only
have a finite number of channel families, this implies that almost
all codes in the ensemble work for all the channels in the finite
subset. But since the finite subset dominates all channels in
${\mathcal C}(R)$ this implies that almost all codes work
for all channels in this set.  \end{proof}

\subsection{Proof of Main Result -- Theorem~\ref{the:main}}\label{sec:proof}
We start by proving some basic properties which any spatial FP has
to fulfill.
Since we are considering a symmetric ensemble (in terms of the
spatial arrangement) it will be useful to consider ``one-sided" FPs.

\begin{definition}[FPs of One-Sided DE] We say that $\Ldens{\x}$
is a one-sided FP (of DE) with channel $\Ldens{c}$ if
(\ref{eq:densevolxi}) is fulfilled for $i \in [-\Lfp, 0]$
with $\Ldens{x}_{i}= \Delta_{+\infty}$  for $ i < -\Lfp$. We say that the FP has a
{\em free} boundary condition if $\Ldens{x}_{i}= \Ldens{x}_0$ for $i > 0$.
We say that it has a {\em forced} boundary condition if $\Ldens{x}_{i}=
\Delta_{0}$ for $i > 0$. Lastly, we say that it has an {\em increasing} boundary condition 
if $\Ldens{x}_{i-1}\prec \Ldens{x}_{i}$ for $i > 0$, where $\Ldens{x}_i$, for $i\geq 1$, are fixed
 but arbitrary symmetric densities. 
\qed
\end{definition}

\begin{definition}[Proper One-Sided FPs]
We say that $\Ldens{\x}$ is {\em non-decreasing} if $\Ldens{x}_i \prec
\Ldens{x}_{i+1}$ for $i=-\Lfp,\dots, -1$. 
Let $(\Ldens{c}, \Ldens{\x})$ be a {\em non-trivial} and {\em
non-decreasing} one-sided FP (with any boundary condition).  As a
short hand, we then say that $(\Ldens{c}, \Ldens{\x})$ is a {\em proper
one-sided FP}.
Figure~\ref{fig:one-sided_fixed_point} shows an example.
\qed
\end{definition}

\begin{definition}[One-Sided Forward DE and Schedules]
Similar to Definition~\ref{def:forwardDE}, one can define  {\em
one-sided forward DE} by initializing all sections with $\Delta_0$
and by applying DE according to an admissible schedule. 
\qed
\end{definition}

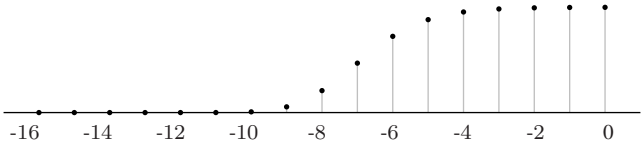
\begin{figure}[htp]
\begin{centering}
\input{ps/one-sided_fixed_point_arxiv}
\caption{\protect{A proper one-sided FP $(\Ldens{c}, \Ldens{\x})$ with free boundary condition
for the ensemble $(\dl=3,\dr=6,\Lfp=16,w=3)$ and the
channel $\Ldens{c}=$BAWGNC($\sigma$) with $\sigma=1.03978$. 
We have $\entropy(\Ldens{c})=0.46940$ and $\entropy(\Ldens{\x})=0.17$.
The height of the vertical bar at section $i$ is equal to $\entropy(\Ldens{x}_i)$.
}}
 \label{fig:one-sided_fixed_point}
\end{centering}
\end{figure}

There are two key ingredients of the proof.  The first ingredient
is to show that any one-sided spatial FP which is increasing,
``small'' on the left, and ``not too small'' and ``flat'' on the
right must have a channel parameter very close to the area threshold
$\ih^{\small A}$. This is made precise in (the Saturation)
Theorem~\ref{thm:existenceintermediateform}.

The second key ingredient is to show the existence of a such a one-sided FP $(\Ldens{c}^*,
\Ldens{\x}^*)$. Figure~\ref{fig:accordeonfp} shows a typical (two-sided)
such example.  This is accomplished in
(the Existence) Theorem~\ref{thm:existencebasicform}.
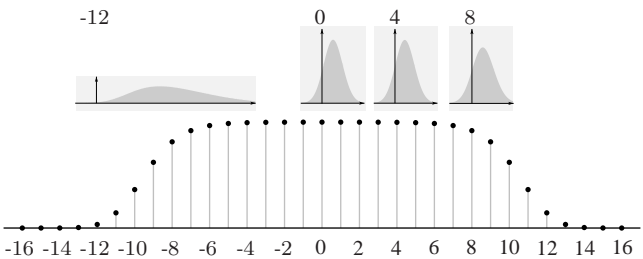
\begin{figure}[htp]
\begin{centering}
\input{ps/accordeonfp}
\caption{Unimodal FP of the $(\dl=3, \dr=6, \Lc=16, w=3)$ ensemble for the
BAWGNC($\sigma$) with $\sigma=0.9480$ (channel entropy $\approx 0.4789$).
The constellation has entropy equal to $0.2$. The bottom figure plots the
entropy of the density at each section. Notice the small values towards
the boundary, a fast transition, and essentially constant values in the
middle. The top figure shows the actual densities at sections $\pm12,
\pm 8, \pm 4, 0$.  Notice that for densities towards the boundary the
mass shifts towards the ``right,'' indicating a high reliability.  
Also plotted in the middle figure
relating to section $0$ is the BP forward DE density of the  $(3,6)$-regular ensemble
at $\sigma=0.9480$. The density is right on the top of the density at
section $0$ of the coupled-code ensemble, i.e., these two densities
are visually indistinguishable. The
density in section $\pm 4$ is also ``close'' to the density at
section $0$.  Thus in the flat part, the densities become close
to the BP density of the underlying ensemble. 
\label{fig:accordeonfp}}
\end{centering} \end{figure} 
Once these two theorems have been established, the proof of 
our main theorem is rather short and straightforward.

\begin{theorem}[Saturation]\label{thm:existenceintermediateform}
Fix $r\in (0,1)$ and let $(\dl, \dr, w)$ be admissible, with $r=1 -
\frac{\dl}{\dr}$, in the sense of conditions (\ref{equ:admissiblenine}), (\ref{equ:admissibletwo}), (\ref{equ:admissiblefour}), (\ref{equ:admissiblefive}), (\ref{equ:admissiblesix}) and (\ref{equ:admissibleseven}) of Definition~\ref{def:admissible}.
Let $(\Ldens{c}^*, \Ldens{\x}^*)$ be a proper one-sided FP on $[-\Lfp, 0]$, with
forced boundary condition, so that for some $\delta>0$, $2(w-1) \leq \Lc$, and $\Lc+w \leq \Msat \leq \Lfp$
the following conditions hold.
\begin{enumerate}[(i)]
\item
{\em Constellation is close to $\Delta_{+\infty}$ ``on the left''}:
\begin{align*}
\batta(\Ldens{x}^*_{-\Lfp+\Lc}) \leq \delta.
\end{align*}

\item {\em Constellation is not too small ``on the right'':}
\begin{align*}
\batta(\Ldens{x}^*_{-\Msat}) \geq \xunstab(1).
\end{align*}
\end{enumerate}

Then
\begin{align*}
\vert \entropy(\Ldens{c}^*) -  \ent^A(\dl, \dr, \{\Ldens{c}_\ent\}) \vert \leq  c(\dl, \dr, \delta, w, \Msat, \Lc).
\end{align*}
Here $c(\dl, \dr, \delta, w, \Msat, \Lc)$ is a function which can be made
arbitrarily small by choosing $\delta$ sufficiently small, $w$
sufficiently large, and $\Lc$ and $\Msat$ sufficiently large compared to $w$.
(This implies of course that the constellation length $\Lfp$ is also chosen sufficiently large.)
More precisely,
\begin{align*}
f(\dl, & \dr, w) = \lim_{\delta \rightarrow 0} 
\lim_{\Lc, \Msat \rightarrow \infty} 
c(\dl, \dr, \delta, w, \Msat, \Lc) \\
& = 8 (\dr-1)^3(\sqrt{2}\!+\!\frac2{\ln 2}\dl(\dr\!-\!1))\sqrt{\frac{2 (\dl\!-\!1)(\dr\!-\!1)}{w}}.
\end{align*}
\end{theorem}
The proof of Theorem~\ref{thm:existenceintermediateform} can be
found in Appendix~\ref{app:areathmsforpartialapproxfamily}.  The
proof of the following Theorem~\ref{thm:existencebasicform} is
contained in Apendix~\ref{sec:existencebasicform}.

\begin{theorem}[Existence of FP]\label{thm:existencebasicform}
Fix $r \in (0,1)$ and let $(\dl, \dr, w)$ be admissible in the sense
of conditions (\ref{equ:admissibleone}), (\ref{equ:admissiblenine}), (\ref{equ:admissibletwo}), (\ref{equ:admissiblethree}), (\ref{equ:admissiblefour}), (\ref{equ:admissiblefive}) in Definition~\ref{def:admissible}
with $r = 1-\frac{\dl}{\dr}$.  Let $\{\Ldens{c}_\sigma\}_{\sigma=0}^{1}$
be a smooth, ordered and complete channel family.

In the sequel, $\Lfp(\dl, \dr, w)$ is a positive constant which depends
on the ensemble but not the channel or the channel family and $c(\dl,
\dr)$ is a positive constant which depends on $\dl$ and $\dr$, but not
on the channel $\Ldens{c}$, the channel family, $\Lfp$ or $w$.

For any $\Lfp > \Lfp(\dl, \dr, w)$ and $0<\delta<\frac{\xunstab(1)}4$,
there exists a proper one-sided FP $(\Ldens{c}^*,
\Ldens{\x}^*)$ on $[-\Lfp, 0]$ with parameters $(\dl, \dr, w)$
and with forced boundary condition so that the following conditions are fulfilled: 
\begin{enumerate}[(i)]
\item
{\em Constellation is close to $\Delta_{+\infty}$ ``on the left''}:
Let $$\Lfp_1 = (\Lfp+1)\Bigl(\frac12 - \frac{wc(\dl,
\dr)}{(\Lfp+1)\delta}\Bigr). $$ Then $\batta(\Ldens{x}^*_i) 
\leq \delta$ for $i\in [-\Lfp, -\Lfp+\Lfp_1-1]$.
\item 
{\em Constellation is not too small ``on the right''}: 
Let $$\Lfp_2 = (\Lfp+1) \Bigl(\frac{\xunstab(1)}{4} - \frac{wc(\dl, \dr)}{(\Lfp+1)\delta}\Bigr).$$ 
Then $\batta(\Ldens{x}^*_i) \geq \xunstab(1)$ for $i \in [-\Lfp_2, 0]$.
\end{enumerate}
\end{theorem}
{\em Discussion:} In words, the theorem says that for any fixed $w
\in \naturals$ and $\delta>0$, if we pick $\Lfp$ sufficiently large,
we can construct a FP constellation which is small on the left for
a linear fraction of the total length and reasonably large on the right,
also for a linear fraction of the total length.

{\em Proof of Theorem~\ref{the:main}:} We are ready to prove the remaining
statement of our main theorem, i.e., \eqref{lem:bplowerbound}.  Let $(\dl, \dr)$ and $w$ be admissible in the
sense of conditions (\ref{equ:admissibleone}), (\ref{equ:admissiblenine}),
(\ref{equ:admissibletwo}), (\ref{equ:admissiblethree}),
(\ref{equ:admissiblefour}), (\ref{equ:admissiblefive}) in
Definition~\ref{def:admissible} and set $r = 1-\frac{\dl}{\dr}$.
We want to  show that $\ent^{\BPsmall} \geq
\ent^{A} - 8 (\dr-1)^3(\sqrt{2}+\frac2{\ln 2}\dl(\dr-1))\sqrt{\frac{2 (\dl-1)(\dr-1)}{w}}$.

First note that $\ent^{\BPsmall}$ is a decreasing
function of $\Lc$. This follows by comparing DE for two constellations
of increasing size and verifying that DE of the larger constellation
``dominates'' (in the sense of degradation) DE of the smaller
constellation. In the ensuing arguments we will take advantage of
this fact -- if we can lower bound the threshold for a particular
constellation size then we will have automatically lower bounded
also the threshold for all smaller constellation sizes. This is
convenient since at several steps we will need to pick $\Lc$
``sufficiently'' large, where the restrictions on the constellation
size stem from our use of simple extremes of information combining
bounds.

Choose a channel, call it $\Ldens{c}$, from the channel family
$\{\Ldens{c}_\ent\}$ with $\entropy(\Ldens{c})<\ent^{A}- 8 (\dr-1)^3
(\sqrt{2}+\frac2{\ln 2}\dl(\dr-1)) \sqrt{\frac{2 (\dl-1)(\dr-1)}{w}}$.  We will
show that for any admissible ensemble $(\dl, \dr, \Lc, w)$, where $\Lc$ is
chosen ``sufficiently large,'' the forward DE process converges to
the trivial FP. By our remarks above concerning the monotonicity
of the threshold in terms of $\Lc$, this implies that for {\em any}
length $\Lc$, DE converges to the trivial FP, hence proving our
main statement.

As stated in Theorem~\ref{thm:existenceintermediateform}, $f(\dl,\dr,
w)$ is the limit of $c(\dl, \dr, \delta, w, \Msat, \Lc)$ when first $\Lc$
and $\Msat$ tend to infinity and then $\delta$ tends to zero.  
We claim that, for the fixed parameters $(\dl, \dr, w)$, for any $\delta>0$ there exist 
$\Lc, \Msat, \Lfp \in \naturals$, sufficiently large, so that
\begin{align}
& \Lfp(\dl, \dr, w) \leq \Lfp, \label{equ:basicexistencehypothesis0} \\
& 2(w-1) \leq \Lc, \label{equ:conditionwandL} \\
& \Lc  \leq (\Lfp+1)\Bigl(\frac12 - \frac{wc(\dl,
\dr)}{(\Lfp+1)\delta}\Bigr), \label{equ:basicexistencehypothesis1}\\
& \Lc\!+\!w\leq \Msat  \leq (\Lfp\!+\!1) \Bigl(\frac{\xunstab(1)}{4}\! -\! \frac{wc(\dl,
\dr)}{(\Lfp\!+\!1)\delta}\Bigr) \leq \Lfp \!-\! \Lc, \label{equ:basicexistencehypothesis2}\\
& \entropy(\Ldens{c})  < \ih^A-c(\dl, \dr, \delta,w, \Msat, \Lc),
\label{equ:entropylessthanhAminusa}
\end{align}
where $\Lfp(\dl, \dr, w)$ and $c(\dl, \dr)$ are the constants given
in Theorem~\ref{thm:existencebasicform}.  To fulfill
(\ref{equ:entropylessthanhAminusa}), as discussed in
Theorem~\ref{thm:existenceintermediateform}, $c(\dl, \dr, \delta,
w, \Msat, \Lc)$ is a continuous function in its parameters which
converges to $8 (\dr-1)^3 (\sqrt{2}+\frac2{\ln 2}\dl(\dr-1)) \sqrt{\frac{2
(\dl-1)(\dr-1)}{w}}$ if we let $\delta$ tend to $0$ and let $\Msat$
and $\Lc$ tend to infinity.  Therefore, by choosing $\delta$
sufficiently small, and $\Lc$ and $\Msat$ sufficiently large we
fulfill (\ref{equ:entropylessthanhAminusa}).  By a proper such
choice we also fulfill (\ref{equ:conditionwandL}) and the first
inequality of (\ref{equ:basicexistencehypothesis2}).  Now note that
increasing $\Lfp$ loosens all above conditions.  In particular, for
any $\delta>0$ and $\Msat, \Lc , w \in \naturals$, by choosing
$\Lfp$ sufficiently large we fulfill (\ref{equ:basicexistencehypothesis0}),
(\ref{equ:basicexistencehypothesis1}), and the last two inequalities
of (\ref{equ:basicexistencehypothesis2}).  We have now fixed all
parameters.

Let $(\Ldens{c}^*, \Ldens{\x}^*)$ be the proper one-sided FP on
$[-\Lfp, 0]$ whose existence is promised by
Theorem~\ref{thm:existencebasicform}.  Recall that it has a forced
boundary condition, i.e., it is a FP if we assume that
$\Ldens{x}^*_i=\Delta_0$ for $i > 0$. Furthermore, from
\eqref{equ:basicexistencehypothesis1} and
\eqref{equ:basicexistencehypothesis2}, and since  $(\Ldens{c}^*,
\Ldens{\x}^*)$ is a proper one-sided FP, we satisfy the conditions
of Theorem~\ref{thm:existenceintermediateform}. Thus we conclude
that $\entropy(\Ldens{c}^*) \geq \ih^A - c(\dl, \dr, \delta, w,
\Msat, \Lc)$.

Next, create from the FP $(\Ldens{c}^*, \Ldens{\x}^*)$ on $[-\Lfp,
0]$ the constellation $\Ldens{\x}$ on $[-\Lfp, \Lfp]$ by appending
to $\Ldens{\x}^*$, $\Lfp$ densities $\Delta_0$ on the right which
are part of the constellation and by defining $\Ldens{x}_i=\Delta_0$
for $i > \Lfp$ (forced boundary condition). Note that this redefined constellation
$(\Ldens{c}^*, \Ldens{\x})$ is not a FP since it does not
fulfill the FP equations for the positions $i \in [1, \Lfp]$.

Initialize DE with $\Ldens{\x}$, i.e., set $\Ldens{\x}^{(0)}=\Ldens{\x}$.
Apply forward DE to $\Ldens{\x}$ with the channel $\Ldens{c}$ as
chosen previously (cf. \eqref{equ:entropylessthanhAminusa}).  Call
the resulting constellation, after $\ell$ steps of DE,
$\Ldens{\x}^{(\ell)}$.

We claim that for all $\ell \geq 0$,  $\Ldens{\x}^{(\ell)}$ is spatially
monotonically increasing, i.e., $\Ldens{x}_i^{(\ell)} \prec
\Ldens{x}_{i+1}^{(\ell)}$, for all $i \in [-\Lfp, \Lfp]$, and that
$\Ldens{\x}^{(\ell)}$ is monotonically decreasing as a function of
$\ell$, i.e., $\Ldens{\x}^{(\ell+1)} \prec \Ldens{\x}^{(\ell)}$.

To prove the first claim recall that $\Ldens{\x}^{(0)} = \Ldens{\x}$,
which is monotonically increasing and has forced boundary condition
on the right.  But DE preserves the monotonicity
so that for every $\ell \geq 0$, $\Ldens{x}_i^{(\ell)}
\prec \Ldens{x}_{i+1}^{(\ell)}$, for all $i \in [-\Lfp, \Lfp]$.

Consider now the second claim. Assume we run one step of DE on
$\Ldens{\x}^{(0)}$ with the channel $\Ldens{c}^*$. Then for $i \in
[-\Lfp, 0]$, $\Ldens{x}_i^{(1)}=\Ldens{x}_i^{(0)}$ by construction.
For $i \in [1, \Lfp]$, $\Ldens{x}_i^{(1)} \prec \Ldens{c}^* \prec
\Delta_0 = \Ldens{x}_i^{(0)}$. In words, for each $i \in [-\Lfp,
\Lfp]$ the constellation is decreasing. It is therefore also
decreasing if we run one step of DE with the channel $\Ldens{c}
\prec \Ldens{c}^*$.  As a consequence, since DE preserves the order
imposed by degradation, we must have $\Ldens{\x}^{(\ell+1)} \prec
\Ldens{\x}^{(\ell)}$ for all  $\ell \geq 0$. Thus the process must
converge to a FP of DE with forced boundary condition. Call this
resulting FP $\Ldens{\x}^{(\infty)}$.

We claim that $\batta(\Ldens{x}^{(\infty)}_{\Lc+1}) < \xunstab(1)$.
Assume to the contrary that this is not true. Then we can apply
Theorem~\ref{thm:existenceintermediateform} to $(\Ldens{c},
\Ldens{\x}^{(\infty)})$ to arrive at a contradiction. Let us discuss
this point in detail. Since $\Ldens{x}^{(\ell)}_i \prec
\Ldens{x}^{(\ell)}_{i+1}$ for all $\ell$ we must have
$\Ldens{x}^{(\infty)}_i \prec \Ldens{x}^{(\infty)}_{i+1}$ for all
$i \in [-\Lfp, \Lfp]$.  Combined with the fact that
$\Ldens{x}_i^{(\infty)}=\Delta_0$ for $i>\Lfp$, we conclude that
$(\Ldens{c},\Ldens{\x}^{(\infty)})$ is a proper one-sided FP on
$[-\Lfp, \Lfp]$ with forced boundary condition.  Furthermore, from
\eqref{equ:basicexistencehypothesis0}, \eqref{equ:conditionwandL},
\eqref{equ:basicexistencehypothesis1} and
\eqref{equ:basicexistencehypothesis2} we see that $\Ldens{\x}^{(\infty)}$
satisfies all hypotheses of Theorem~\ref{thm:existenceintermediateform}.
More precisely, by assumption the constellation is large
for the last $\Lfp-\Lc$ sections. Hence from the choice of $\Msat$ as
given by \eqref{equ:basicexistencehypothesis2} we must have
$\batta(\Ldens{x}^{(\infty)}_{\Lfp-\Msat})\geq \xunstab(1)$.  From
\eqref{equ:basicexistencehypothesis1} it is clear that
$\batta(\Ldens{x}^{(\infty)}_{-\Lfp+\Lc})\leq \delta$.
 As a consequence, from
the Theorem~\ref{thm:existenceintermediateform} we conclude that
$\entropy(\Ldens{c}) \geq \ih^A - c(\dl, \dr, \delta, w, \Msat,
\Lc)$. But this contradicts our initial assumption on $\entropy(\Ldens{c})$
(cf. \eqref{equ:entropylessthanhAminusa}).

We are now ready to prove our main claim.  Consider a coupled
ensemble on $[1, \Lc+1]$ with parameters $(\dl, \dr, w)$.  More
precisely, the coupled ensemble has sections from $[1, \Lc+1]$ with
$i\notin [1, \Lc+1]$ set to $\Delta_{+\infty}$.  Initialize all
sections in $[1, \Lc+1]$  to $\Delta_0$.  Call this constellation
$\Ldens{\y}^{(0)}$.  Run forward DE with the channel $\Ldens{c}$
on $\Ldens{\y}^{(0)}$, call the result $\{\Ldens{\y}^{(\ell)}\}$,
and let $\Ldens{\y}^{(\infty)}$ denote the limit, which is a FP.
We have $\Ldens{y}^{(\ell)}_i
\prec \Ldens{x}^{(\ell)}_i$, $i \in [1, \Lc+1]$, since $\Ldens{y}_i^{(0)}
= \Ldens{x}_i^{(0)}$ for $i \in [1, \Lc+1]$ and $\Ldens{y}^{(0)}_i =
\Delta_{+\infty} \prec \Ldens{x}_i^{(0)}$ for $ i \not \in [1,
\Lc+1]$ and DE preserves the ordering.
Therefore $\batta(\Ldens{y}_i^{(\infty)}) \leq \batta(\Ldens{x}^{(\infty)}_{\Lc+1}) < \xunstab(1)$, for all $i \in [1, \Lc+1]$.  
Let $\batta_j$, for some $j\in [1,\Lc+1]$, denote the maximum of the Battacharyya
parameter over all sections of $\Ldens{\y}^{(\infty)}$. 
From extremes of information combining we have
\begin{align*} 
 \batta_j = \batta(\Ldens{y}^{(\infty)}_j)   & \leq \batta(\Ldens{c}
) (1 - (1 - \batta_j)^{\dr-1})^{\dl-1} \\ & \leq (1 - (1 -
\batta_j)^{\dr-1})^{\dl-1} .  
\end{align*} 
The last inequality implies that $\batta_j = 0$ since $\batta_j
\in [\xunstab(1), 1]$ is excluded. From
property~(\ref{lem:blmetricbattaboundswasser}) of Lemma~\ref{lem:blmetric}
we conclude that $d(\Ldens{y}^{(\infty)}_i, \Delta_{+\infty}) \leq
\batta(\Ldens{y}^{(\infty)}_i) \leq \batta_j=0$, for all $i\in
[1,L+1]$. In other words, $\Ldens{\y}^{(\infty)} =
\underline{\Delta}_\infty$, as claimed.

\subsection{Conclusion and Outlook}\label{sec:conclusion}
We have shown one can construct low-complexity coding schemes which
are universal for the whole class of BMS channels by spatially
coupling regular LDPC ensembles.  Thus, we resolve a long-standing
open problem of whether there exist low-density parity-check ensembles
which are capacity-achieving using BP decoding.  These ensembles
are not only attractive in an asymptotic setting but also for
applications and standards since they can easily be designed to
have both, good thresholds and low error floors.  In addition, these
ensembles are {\em universal} in the sense that one and the same
ensemble is good for the whole class of BMS channels, assuming that
the channel is known at the receiver. In fact, we have shown the
stronger statement that almost all codes in such an ensemble are
good for all channels in this class.

Let us discuss some open questions.

\begin{itemize}
\item[]{\em Maxwell Conjecture:}
As a byproduct of our proof, we know that the MAP threshold of
coupled ensembles is essentially equal to the area threshold of the
uncoupled ensemble. In addition we know that the MAP threshold of
the uncoupled ensemble is also upper bounded by the area threshold.
The Maxwell conjecture states that in fact the MAP threshold of the
uncoupled ensemble is {\em equal} to the area threshold. So if one
can establish that the MAP threshold of the uncoupled ensemble is
at least as large as the MAP threshold of the coupled ensemble,
then the Maxwell conjectured would be proved. A natural approach
to resolve this issue is to use interpolation techniques and it is
likely that the Maxwell conjecture can be proved in a way similar
as this was done in \cite{HMU11a} for other graphical models.

\item[]{\em Convergence Speed:}
As discussed previously, we only give weak bounds on the speed of
convergence of the ensemble to the Shannon capacity (as a function
of the degrees, the constellation length $\Lc$, as well as the
coupling width $w$). Numerical evidence suggests much stronger
results.  Settling the question of the actual convergence speed
is both challenging and interesting.

\item[]{\em Lifting of Restrictions:}
Our results apply only to sufficiently large degrees whereas
numerical calculations indicate that the threshold saturation
effect equally shows up for small degrees. This is a consequence
of the fact that at many places we have used simple extremes
of information combining bounds. With sufficient effort it is likely
that one can extend the proof to many dds 
which are currently not covered by our statement.

\item[]{\em General Ensembles:}
In a similar vein, we restricted our investigation to regular
ensembles to keep things simple, but the same technique applies in
principle also to irregular or even structured ensembles. Again,
depending on the structure of the underlying ensemble, much effort
might be required to derived the necessary bounds.

\item[]{\em Wiggle Size:}
Perhaps the weakest link in our derivation is the treatment of the
connection width $w$. In our current statements this connection
width has to be chosen large. Empirically, small such
lengths, such as the extreme case $w=2$ give already excellent
results and by increasing $w$ the convergence to the area threshold
seems to happen exponentially fast. How to derive practically relevant
bounds for such small values of $w$ is an important open problem.

\item[]{\em Scaling:}
More generally, from a practical point of view, what is needed is
a firm understanding of how the performance of such codes scale
in each of the parameters in $\dl$, $\dr$, $\Lc$, $M$, as well as
$w$.  Only then will it be possible to design codes in a principled
fashion. 

\item[]{\em Practical Issues:} Further important topics are, the
design of good termination schemes which mitigate the rate-loss, a
systematic investigation of how structure in the interconnection
pattern as well as the codes influences the performance, and how
to optimally choose the scheduling (e.g., windowed decoding) to
control the complexity of the decoder \cite{KMRU10}.

\item[]{\em General Models:} As was discussed briefly in the
introduction, the threshold saturation phenomenon has been empirically
found to hold in a large variety of systems. This suggests that one
should be able to formulate a rather general theory rather than
finding a separate proof for each of these cases.  For all
one-dimensional systems this has recently been accomplished in
\cite{KRU12}. For higher-dimensional or infinite-dimensional systems
this is a challenging open problem.
\end{itemize}

\section{Acknowledgments} We would like to thank H. Hassani, S. Korada, N.
Macris, C. M{\'e}asson, and A. Montanari for interesting discussions
on this topic and H. Hassani for his feedback on an early draft.  S. Kudekar would like to thank Misha Chertkov,
Cyril M{\'e}asson, Jason Johnson, Ren\'e Pfitzner and Venkat Chandrasekaran for their encouragement and Bob Ecke for hosting him
in the Center for Nonlinear Studies, Los Alamos National Laboratory
(LANL), where most of his work was done. He also gratefully
acknowledges his support from the U.S. Department of Energy at Los
Alamos National Laboratory under Contract No.  DE-AC52-06NA25396
as well as from NMC via the NSF collaborative grant CCF-0829945 on
``Harnessing Statistical Physics for Computing and Communications.''
The work of R. Urbanke was supported by the European project STAMINA,
265496.

\begin{appendices}

\section{Entropy versus Battacharyya -- Lemma \protect{\ref{lem:entropyvsbatta}}}\label{sec:someusefulfacts}
\begin{lemma}[Bounds on Binary Entropy Function]\label{lem:boundsonent}
Let $h_2(x)=-x \log_2(x)-(1-x) \log_2(1-x)$. 
Then for $x \in [0, 1/2]$,
\begin{align}
h_2(x)& \geq 1 - (1-2x)^2 \label{eq:lowerboundbinaryentropy}, \\
h_2(x)& \leq 2 \sqrt{x(1-x)}, \label{eq:upperboundbinaryentropyone} \\
h_2(x)& \leq \frac{11}{4} x^\frac{3}{4}. \label{eq:upperboundbinaryentropytwo}
\end{align}
\end{lemma}
\begin{IEEEproof}
To prove (\ref{eq:lowerboundbinaryentropy}), write
\begin{align*}
h_2(x) & \stackrel{\text{\cite[Lemma II.1]{WiS07}}}{=} 1 - \frac1{2\ln 2}\sum_{n=1}^{\infty} \frac{(1-2x)^{2n}}{n(2n-1)} \\
& 
\geq 1 \!-\! (1\!-\!2x)^{2}\underbrace{\frac1{2\ln 2}\sum_{n=1}^{\infty} \frac{1}{n(2n-1)}}_{=1} 
 = 1 - (1-2x)^2.
\end{align*}
Consider now (\ref{eq:upperboundbinaryentropyone}). 
Set $g(z)=2 \sqrt{(1-x)x}-h_2(x)\mid_{x=(1-z)/2}= \sqrt{1-z^2}-h_2(\frac{1-z}{2})$.
We want to show that $g(z) \geq 0$ for $z \in [0, 1]$.
We have
\begin{align*}
g'(z) & = -\frac{z}{\sqrt{1-z^2}} + \frac1{2\ln 2}\ln\Big(\frac{1+z}{1-z}\Big), \\
g''(z) & = -\frac1{(1-z^2)^{3/2}} + \frac1{(1-z^2)\ln2}.
\end{align*}
The following claims are straightforward to verify using the
explicit formulae for $g(z)$, $g'(z)$, and $g''(z)$: (i) $g(0) = g(1) = 0$, (ii) $g'(0) = 0$, (iii)
$g''(0) > 0$, (iv) $g''(z) = 0$ has exactly one solution in $[0,
1]$.

Suppose there exists a $w$,  $0<w<1$, so that $g(w) < 0$. Then from
(i), (ii) and (iii) we must have $g(z) = 0$ for at least three
distinct elements of $[0,1]$.  Rolle's theorem then implies that
$g'(z)=0$ has at least two distinct solutions in $(0,1)$ and hence
at least three distinct solutions in $[0,1]$ (since by (i) $g'(0)=0$).
Using Rolle's theorem again, this implies that $g''(z)=0$ has at
least two solutions in $[0,1]$, contradiction (iv).

We prove (\ref{eq:upperboundbinaryentropytwo}) along similar lines. Consider $g(x) =
\frac{11}{4} x^{\frac34} - a \frac{x}{\ln2} + x \log_2(x)$, where
$a=4 (1 + \ln(\frac{11 \ln(2)}{16})) \approx  1.035 > 1$. Note that
$g(x) \leq \frac{11}{4} x^{\frac34}- h_2(x)$ for $x \in [0, \frac12]$ 
(to verify this, upper bound the term $-(1-x)
\log_2(1-x)$ of the entropy function by $x/\ln(2)$).  So if we can
prove that $g(x) \geq 0$ for $x \in [0, \frac12]$ then we are done.

Direct inspections of the quantities shows that $g(0)=0$,
$g'(0+)=+\infty$, $g(x^*)=g'(x^*)=0$, where $x^* = \frac{14641
\log(2)^4}{65536} \approx 0.05157$, and $g(\frac12)>0$.

It follows that if there exists an $x \in [0, \frac12]$ so that
$g(x) <0$ then $g(x)$ must have at least $4$ roots in this range,
therefore by Rolle $g'(x)$ must have at least $3$ roots, and again
by Rolle $g''(x)$ must have at least $2$ roots. But an explicit
check shows that $g''(x)=-\frac{33}{64 x^{\frac54}} + \frac{1}{x
\ln(2)}=0$. So $g''(x)=0$ can only have a single solution.  \end{IEEEproof}

{\em Proof of Lemma~\ref{lem:entropyvsbatta}}: 
Let $\absDdens{a}$ denote the density in the $|D|$-domain. Then
\begin{align*}
\sqrt{\entropy(\absDdens{a})} & = \sqrt{\int_0^1 \!\!\!h_2\big(\frac{1-z}2\big) \absDdens{a}(z)
\dee z} 
 \stackrel{\eqref{eq:lowerboundbinaryentropy}}{\geq}
\sqrt{\int_0^1 \!\!\!(1 - z^2)\absDdens{a}(z) \dee z} \\
& \stackrel{\text{Jensen}}{\geq} \int_0^1 \sqrt{1 - z^2}\absDdens{a}(z) \dee z = \batta(\absDdens{a}). 
\end{align*} 
This proves that $\batta(\absDdens{a})^2$ lower bounds
$\entropy(\absDdens{a})$.  For the upper bound we have
\begin{align*}
\batta(\absDdens{a}) & = \int_0^1 \sqrt{1-z^2} \absDdens{a}(z) \dee z \nonumber \\
& = \int_0^1 \underbrace{\Big(\sqrt{1\!-\!z^2} \!-\!  h_2(\frac{1\!-\!z}2)\Big)}_{
	\geq 0 \; \text{by (\ref{eq:upperboundbinaryentropyone}) with $x=\frac{1-z}{2}$}} 
 \absDdens{a}(z) \dee z + \entropy(\absDdens{a}).
\end{align*}
\qed

\section{Upper Bound on BP Threshold -- Lemma~\ref{lem:bpboundsuncoupled}}\label{sec:proofofbpboundsuncoupled}
\begin{IEEEproof}
We use ideas from extremes of information combining.  We get an upper
bound on the BP threshold by assuming that the densities at check nodes
are from the BSC family and that densities at variable nodes are from the
BEC family.

Let $x$ represent the entropy of the variable-to-check
message and let $c$ denote the entropy of the channel. 
If for any $x \in [0, c]$
\begin{align}\label{equ:necessarycondition}
h_2((1 - (1 - 2 h_2^{-1}(x))^{\dr - 1})/2) > (x/c)^\frac{1}{\dl-1},
\end{align}
then DE will not converge to the perfect decoding FP.  The left-hand
side represents the minimum entropy at the output of a check node
which we can get if the input entropy is $x$ (and this minimum is
achieved if the input density is from the BSC family).  The right-hand
side represents the maximum input entropy which we can have at the
input of a variable node if we want an output entropy equal to $x$
(and this minimum is achieved if the input density is from the BEC
family). Note that we can extend the inequality
(\ref{equ:necessarycondition}) to {\em all} $x \in [0, 1]$ without
changing the condition since for $x \in (c, 1]$, the right hand
side is strictly bigger than $1$, whereas the left-hand side is
always bounded above by $1$.

The preceding condition is equivalent to saying that in order for DE to succeed,
we must have 
\begin{align*}
c \leq \frac{x}{(h_2((1 - (1 - 2 h_2^{-1}(x))^{\dr - 1})/2))^{\dl-1}},
\end{align*}
for all $x \in [0, 1]$.
We can also write this as
\begin{align*}
c \leq \frac{h_2(x)}{(h_2((1 - (1 - 2 x)^{\dr - 1})/2))^{\dl-1}},
\end{align*}
where $x \in [0, \frac12]$.

We want to show that $c$ cannot be too large, i.e., we are looking for an upper bound
on $c$. Note that any value of $x$ gives a bound.
Let us choose $x=\frac{1}{2 \sqrt{\dr-1}}$. This gives the bound
\begin{align*}
c \leq \frac{h_2(\frac{1}{2\sqrt{\dr-1}})}{(h_2(\frac{1-e^{-\sqrt{\dr-1}}}{2}))^{\dl-1}}.
\end{align*}
To obtain the above inequality we first write $(1-2x)^{\dr-1}$ as
$\text{exp}((\dr-1)\log(1-2x))$. For $x \in [0, \frac12]$ we use
the Taylor expansion $$\log(1-2x) = -2x - \frac{(2x)^2}{2} -
\frac{(2x)^3}3... \leq -2x=-\frac{1}{\sqrt{\dr-1}}.$$ Thus
$\text{exp}((\dr-1)\log(1-2x)) \leq \text{exp}(-\sqrt{\dr-1})$ and
$h_2((1 - (1 - 2 x)^{\dr - 1})/2) \geq h_2(\frac{1-e^{-\sqrt{\dr-1}}}{2})$.
We want to simplify the expression even further.  Using \cite[Lemma
II.1]{WiS07} and bringing out the first term in the summation,
\begin{align}\label{eq:taylorexpansionforbinaryentropy}
h_2(x) & = 1 - \frac1{2\ln 2}(1-2x)^2 - \frac1{2\ln 2}\sum_{n=2}^{\infty} \frac{(1-2x)^{2n}}{n(2n-1)} \nonumber \\
& \geq 1 - \frac1{2\ln 2}(1-2x)^2 - \frac1{2\ln 2}\sum_{n=2}^{\infty} (1-2x)^{2n} \nonumber \\
& = 1 - \frac1{2\ln 2}(1-2x)^2 - \frac{(1-2x)^4}{2\ln 2}\sum_{n=0}^{\infty} ((1-2x)^{2})^n \nonumber \\
& = 1 - \frac2{\ln 2}(x-\frac12)^2 - \frac{8(x-1/2)^4}{\ln(2)(1-4(x-1/2)^2)}. 
\end{align}
Substituting $x = (1 - e^{-\sqrt{\dr-1}})/2$ we have
\begin{align*}
h_2(\frac{1\!\!-\!\!e^{\!-\!\sqrt{\dr\!-\!1}}}{2})^{\dl\!-\!1} & \!\geq\! (1 \!\!-\!\! \frac{e^{-2\sqrt{\dr\!-\!1}}}{2\ln 2}\!-\! \frac{e^{-4\sqrt{\dr\!-\!1}}}{2\ln 2}\frac1{1\!\!-\!\!e^{-2\sqrt{\dr\!-\!1}}})^{\dl \!-\!1} \\
& \!\geq\! 1 \!-\! \frac{(\dl-1)}{2\ln 2}\left(e^{-2\sqrt{\dr\!-\!1}}\!+\! \frac{e^{-4\sqrt{\dr\!-\!1}}}{1\!-\!e^{-2\sqrt{\dr\!-\!1}}}\right). 
\end{align*}
We conclude that
\begin{align*}
c  \leq 
\frac{h_2(\frac{1}{2 \sqrt{\dr-1}})}{1 \!-\! \frac{(\dl-1)}{2\ln 2}\left(e^{-2\sqrt{\dr\!-\!1}}\!+\! \frac{e^{-4\sqrt{\dr\!-\!1}}}{1\!-\!e^{-2\sqrt{\dr\!-\!1}}}\right)}  
 \leq \frac{h_2(\frac{1}{2 \sqrt{\dr-1}})}{1-\dl e^{-2\sqrt{\dr-1}}}.
\end{align*}
\end{IEEEproof}

\section{Basic Properties of the Wasserstein Metric -- 
Lemma~\ref{lem:blmetric}}\label{sec:blmetric}
{\em Proof}:
\renewcommand\theenumi{\roman{enumi}} 
\renewcommand{\labelenumi}{(\roman{enumi})}
\begin{enumerate}
\item {\em Alternative Definitions}:
The equivalence of the basic definition (cf.
Definition~\ref{def:wasserstein}) and the first alternative description
is shown in (6.2) and (6.3) in \cite{Villani09}.  The equivalence
of the first and second alternative descriptions is shown in
\cite{Val73}.

\item{\em Boundedness}:
Follows directly from either of the two alternative descriptions.

\item {\em Metrizable and Weak
Convergence}: See \cite[Theorem 6.9]{Villani09}.

\item {\em Polish Space}: See \cite[Theorem 6.18]{Villani09}.

\item {\em Convexity}:
We have
\begin{align*}
& \Big\vert \int_{0}^{1} \!\!f(x)(\alpha \absDdens{a}(x) + 
\bar{\alpha}\absDdens{b}(x) - \alpha \absDdens{c}(x) - 
\bar{\alpha}\absDdens{d}(x)) \,\dee x \Big \vert \leq \\
&  \alpha \Big\vert\!\! \int_{0}^{1} \!\!\!\!\!\!f(x)(\absDdens{a}(x) \!-\! \absDdens{c}(x)) \dee x \Big\vert \!\!+\!\!  
\bar{\alpha} \Big\vert \!\!\int_{0}^{1} \!\!\!\!\!\!f(x)(\absDdens{b}(x) \!-\! \absDdens{d}(x)) \dee x \Big\vert.
\end{align*}

\item {\em Regularity wrt $\vconv$}: Let $\tilde{f}(\cdot)$ be
$\Lip(1)[0, 1]$. Without loss of generality assume that $\tilde{f}(0)=0$.
Indeed, since we consider the difference of densities, subtracting
 a constant does not affect the integral. Define $f(x)$ for $x \in
 [-1, 1]$ by setting
$f(x)=\tilde{f}(x)$ for $x \in [0, 1]$ and $f(x)=\tilde{f}(-x)$ for
$x \in [-1, 0]$. Then $f(x)$ is $\Lip(1)[-1, 1]$ and also $f(0)=0$.

Let $\Ddens{d}=\Ddens{a}\vconv \Ddens{c}$ and  $\Ddens{e}=\Ddens{b}\vconv \Ddens{c}$ be the $D$-domain representation. 
Thus $d(\Ddens{d}, \Ddens{e})$ is characterized by
\begin{align*}
& \Big\vert \int_{0}^{1} \tilde{f}(z) (\absDdens{d}(z) -\absDdens{e}(z)) \,\dee z \Big\vert
\nonumber \\ 
\stackrel{\text{(i)}}{=} & \Big\vert \int_{-1}^{1} f(z) (\Ddens{d}(z) -\Ddens{e}(z)) \,\dee z \Big\vert
\nonumber \\ 
\stackrel{\text{(ii)}}{=} &
\Bigl\vert \int_{-1}^{1} \int_{-1}^{1}  (\Ddens{a}(x)\Ddens{c}(y)-\Ddens{b}(x)\Ddens{c}(y))
f(g(x, y)) \,\dee x \dee y  \Bigr\vert \\
\stackrel{\text{(iii)}}{=} &
\Bigl\vert \int_{0}^{1} \absDdens{c}(y) \,\dee y\int_{0}^{1}  (\absDdens{a}(x)-\absDdens{b}(x))
h( x, y) \,\dee x \Bigr\vert.
\end{align*}
In step (i) we use the construction of $f(z)$ along with the relation
between $D$ and $\vert D\vert$ domains given by \eqref{eq:Dvs|D|}.
We defined $g(x, y)=\tanh(\tanh^{-1}(x)+\tanh^{-1}(y))=\frac{x+y}{1+x
y}$ and step (ii) follows by explicitly writing the variable node
convolution in the $D$-domain.  In step (iii) we defined
\begin{align*}
h(x, y)= \frac14 \sum_{i \in \{\pm1\}} \sum_{j \in \{\pm 1\}} f(g(i x, j y))(1+i x)(1+j y).
\end{align*}
To obtain this equivalent formulation of the integral in 
step (iii) we make use of the symmetry conditions of $D$-densities and the implied
relationship between $D$ and $|D|$ densities for $y\in [0, 1]$,
\begin{align}\label{eq:Dvs|D|}
\Ddens{a}(-y)=\Ddens{a}(y) \frac{1\!-\!y}{1\!+\!y}, \,\,
\Ddens{a}(y)=\absDdens{a}(y) \frac{1\!+\!y}{2}.
\end{align}

We claim that $h( x,y)$ is $\Lip(2)[0, 1]$ (as a function $x$). This
will settle the proof of the lemma.  Notice that $h(x,y)$ is a linear combination
of four functions. Let us consider a generic term. Writing $g( \cdot, \cdot)$
explicitly, we have
\begin{align*}
& \vert f(g(i x, j y))(1\!+\!i x)\!-\! f(g(i z, j y))(1\!+\!i z)\vert (1\!+\!j y) \\ 
&\!=\!\vert f(\frac{ix\!+\!jy}{1\!+\!ijxy})(1\!+\!i x)\!-\! 
f(\frac{iz\!+\!jy}{1\!+\!ijzy})(1\!+\!i z)\vert 
(1\!+\!j y)
\\ 
&\!\leq\!\vert f(\frac{ix\!+\!jy}{1\!+\!ijxy})(1\!+\!i x)\!-\!
f(\frac{iz\!+\!jy}{1\!+\!ijzy})(1\!+\!i x) \vert (1\!+\!j y)
\\ 
&\!+\!\vert f(\frac{iz\!+\!jy}{1\!+\!ijzy})(1\!+\!i x)\!-\!
f(\frac{iz\!+\!jy}{1\!+\!ijzy})(1\!+\!i z)\vert
(1\!+\!j y)
 \\ 
&\stackrel{\text{(i)}}{\leq} (1+i x)(1+j y) \frac{(1-y^2)}{(1+ijxy)(1+ijzy)} \vert(ix - iz) \vert \\ 
& +(1+j y) \vert (ix -iz)  \vert  
\end{align*}
In (i) we use the Lipschitz continuity of $f(\cdot)$ and $i^2=j^2=1$
to obtain the first term. We use $\vert f(\cdot) \vert\leq 1$ to obtain
the second term in (i). Indeed, since $\tilde{f}$ is $\Lip(1)[0,1]$ and $\tilde{f}(0)=0$
 we must have $\vert f(x) \vert = \vert \tilde{f}(|x|) \vert = \vert \tilde{f}(|x|) - \tilde{f}(0) \vert \leq \vert x \vert \leq 1$. 
Also,
in the above expression, we can replace $\vert (ix -iz)  \vert$ by $\vert
x -z \vert$.

Now we sum over all possible $i,j$ and divide by 4 to get
\begin{align*}
& \vert h(x,y) - h(z,y) \vert \leq \frac14\vert x - z\vert \times \Big(\sum_{i\in \{\pm 1\}, j\in \{\pm 1\}} (1+j y) \\ 
& + \sum_{i\in \{\pm 1\}, j\in \{\pm 1\}}
(1+i x)(1+j y) \frac{(1-y^2)}{(1+ijxy)(1+ijzy)} \Big).
\end{align*}
Since $\sum_{j\in \{\pm 1\}} j y = 0$ we have
$$\sum_{i\in \{\pm 1\}, j\in \{\pm 1\}} (1+j y) = 4.$$ 

Let us consider the other term. We split the sum into two parts,
one sum over $ij>0$ and the other over $ij < 0$.  We have
\begin{align*}
\sum_{ij < 0} (1+i x)(1+j y) \frac{(1-y^2)}{(1+ijxy)(1+ijzy)}& = 2 \frac{1-y^2}{(1-zy)}, \\
\sum_{ij > 0} (1+i x)(1+j y) \frac{(1-y^2)}{(1+ijxy)(1+ijzy)}& = 2 \frac{1-y^2}{(1+zy)}.
\end{align*}
Adding the two we get the total contribution 
$$
2(1-y^2)\Big(\frac{1}{1+zy}+\frac{1}{1-zy} \Big) = 4\frac{1-y^2}{1-z^2y^2} \leq 4.
$$
Putting everything together we get
\begin{align*}
& \vert h(x,y) - h(z,y) \vert \leq 2\vert x - z\vert. 
\end{align*}

To get a good bound on $d(\Ldens{a}^{\vconv i}\vconv\Ldens{c}, \Ldens{b}^{\vconv i}\vconv\Ldens{c})$ in terms
of $d(\Ldens{a}, \Ldens{b})$ for $i \geq 2$ consider
\[
\Ldens{c'} = \frac{1}{i}
\sum_{j=1}^{i}\Ldens{a}^{\vconv i-j}\vconv \Ldens{b}^{\vconv j-1},
\]
and note that the Wasserstein metric can be expressed directly in the L-domain as
\[
d(\Ldens{a},\Ldens{b}) = 
\int_0^\infty \Big| \int_{-x}^x (\Ldens{a}(y)-\Ldens{b}(y)) \dee y \Big|\, 
\frac{2e^{-x}}{(1+e^{-x})^2}\,\dee x
\]
Applying this representation we observe that
\[
d(\Ldens{a}\vconv \Ldens{c}\vconv \Ldens{c'} , \Ldens{b}\vconv\Ldens{c}\vconv \Ldens{c'}) =
\frac{1}{i} d(\Ldens{a}^{\vconv i}\vconv \Ldens{c}, \Ldens{b}^{\vconv i}\vconv \Ldens{c})
\]
which yields
\[
 d(\Ldens{a}^{\vconv i}\vconv \Ldens{c}, \Ldens{b}^{\vconv i}\vconv \Ldens{c})
\le 2 i d(\Ldens{a} , \Ldens{b})\,.
\]

\item {\em Regularity wrt $\cconv$}:
Let $f(x)$ be $\Lip(1)[0, 1]$.  Let $\Ddens{d}=\Ddens{a}\cconv
\Ddens{c}$ and  $\Ddens{e}=\Ddens{b}\cconv \Ddens{c}$ be the $D$-domain representation. 
\begin{align*}
& \Big\vert \int_{0}^{1}  f(z) (\absDdens{d}(z)-\absDdens{e}(z)) \,\dee z \Big\vert \nonumber \\ 
& \stackrel{\text{(a)}}{=} \Bigl\vert \int_0^1\int_0^1 (\absDdens{a}(x)\absDdens{c}(y)-\absDdens{b}(x)\absDdens{c}(y)) f(x y) 
\,\dee x \dee y \Bigr\vert \\
& \leq  \int_{0}^{1}  \dee y \,\Ddens{c}(y)  
\Bigl\vert \int_{0}^{1} f(x y) (\absDdens{a}(x)  -\absDdens{b}(x)) \,\dee x \Bigr\vert,
\end{align*}
where step (a) follows since in the $|D|$-domain, check-node convolution
corresponds to a multiplication of the values.

But note that if $f(x)$ is $\Lip(1)[0, 1]$ then $f(x y)$ is $\Lip(|y|)[0,
1]$. Hence,
\begin{align*}
d(\Ldens{a}\cconv\Ldens{c}, \Ldens{b}\cconv\Ldens{c}) 
& \le d(\Ldens{a},\Ldens{b}) \int_{0}^{1}  \dee y \,\absDdens{c}(y) y \\
& \stackrel{\perr(\Ldens{c}) = \int_0^1 \frac{(1-y)}2 \absDdens{c}(y) \dee y}{=} 
d(\Ldens{a},\Ldens{b})(1-2\perr(\Ldens{c})) \\
& \stackrel{\batta(\Ldens{c}) \leq 2 \sqrt{\perr(\Ldens{c})(1-\perr(\Ldens{c}))}}{\leq} 
d(\Ldens{a},\Ldens{b})\sqrt{1 - \batta^2(\Ldens{c})}.
\end{align*} 
Above, the relation between the Battacharyya and error parameters can be obtained
via extremes of information combining (see \cite{RiU08}).
Let us focus on the last part. 
To get a good bound on $d(\Ddens{a}^{\cconv i}, \Ddens{b}^{\cconv i})$ in terms
of $d(\Ddens{a}, \Ddens{b})$ for $i \geq 2$, consider
\[
\Ddens{c} = \frac{1}{i}
\sum_{j=1}^{i}\Ddens{a}^{\cconv i-j}\cconv \Ddens{b}^{\cconv j-1},
\]
and note that the Wasserstein metric can be expressed directly in the D-domain as
\[
d(\Ddens{a},\Ddens{b}) = 
\int_0^1 \Big| \int_{-x}^x (\Ddens{a}(y)-\Ddens{b}(y)) \dee y \Big|\, 
\dee x
\]
Applying this representation, we observe that
\[
d(\Ddens{a}\cconv \Ddens{c} , \Ddens{b}\cconv\Ddens{c}) =
\frac{1}{i} d(\Ddens{a}^{\cconv i}, \Ddens{b}^{\cconv i}).
\]
This yields
\begin{align*}
d (\Ddens{a}^{\cconv i}, \Ddens{b}^{\cconv i}) 
& \le i d(\Ddens{a} , \Ddens{b}) (1-2\perr(\Ddens{c})) \\
& = d(\Ddens{a} , \Ddens{b}) \sum_{j=1}^i
(1\!-\!2\perr(\Ddens{a}^{\cconv i-j} \cconv \Ddens{b}^{\cconv j-1})) \\
& = d(\Ddens{a} , \Ddens{b}) \sum_{j=1}^i
(1\!-\!2\perr(\Ddens{a}))^{i-j}
(1\!-\!2\perr(\Ddens{b}))^{j-1} \\
& \le d(\Ddens{a} , \Ddens{b}) \sum_{j=1}^i
(1\!-\!\batta^2(\Ddens{a}))^{\frac{i-j}{2}}
(1\!-\!\batta^2(\Ddens{b}))^{\frac{j-1}{2}}.
\end{align*}

\item {\em Regularity wrt DE}: Follows from properties
(\ref{lem:blmetricregularvconv}) and (\ref{lem:blmetricregularcconv}).

\item {\em Wasserstein Bounds Battacharyya and Entropy:}
Let $g$ be a positive function on $[0,1]$ and let $f$ be a $C^2$ concave
decreasing function on $[0,1].$ Then, for any $c \ge |g|_\infty,$
\[
-\int_0^1 f'(x) g(x) \dee x \le c \Bigl(
f(1-\frac{1}{c}\int_0^1 g(z) \dee z) - f(1)
\Bigr)\,.
\]
Before proving the inequality let us use it to establish the stated bounds.
Set $g(z) = |\absDdens{B}(z) - \absDdens{A}(z)|.$ Then
$|g|_\infty \le 1$ and 
$\int_0^1 g(z) \text{d}z = d(\Ddens{a},\Ddens{b}).$
Now, for the Battacharyya bound let $f(z) = \sqrt{1-z^2}$ and note
\begin{align*}
|\batta(\Ddens{b}) - \batta(\Ddens{a})|
&=
\Big|\int_0^1 f(z) (\Ddens{b}(z) - \Ddens{a}(z)) \text{d}z\Big| \\
&=
\Big|-\int_0^1 f'(z) (\absDdens{B}(z) - \absDdens{A}(z)) \text{d}z\Big| \\
&\le
-\int_0^1 f'(z) g(z)\text{d}z\,.
\end{align*}
We obtain
\begin{align*}
|\batta(\Ddens{b}) - \batta(\Ddens{a})|
& \le
\sqrt{1-(1-d(\Ddens{a},\Ddens{b}))^2} \\
& =
\sqrt{d(\Ddens{a},\Ddens{b})}
\sqrt{2-d(\Ddens{a},\Ddens{b})}\,.
\end{align*}
For the entropy case we set $f(z) = h_2(\frac{1-z}{2}).$
The same argument as above yields
\begin{align*}
|\entropy(\Ddens{b}) - \entropy(\Ddens{a})|
& \le
h_2(\frac{d(\Ddens{a},\Ddens{b})}{2}) \\
& \le
\frac{1}{\ln 2}\sqrt{d(\Ddens{a},\Ddens{b})}
\sqrt{2-d(\Ddens{a},\Ddens{b})}\,.
\end{align*}

We prove the stated inequality.
Let us define
\[
\hat{g}(z) =c \ind_{\{z \ge 1 - \frac{1}{c} \int_0^1 g(x)\text{d}x\}}\,,
\]
where $c \ge |g|_\infty.$
For each $z \in [0,1]$ we have
\(
\int_0^1 
(g(z) - \hat{g}(z)) \text{d}z\ge 0
\)
with equality at $z=1.$ Hence,
\begin{align*} 
0 & \ge
\int_0^1 f''(z)
\Bigl(
\int_0^z
(g(x) - \hat{g}(x))
\text{d}x\Bigr)
\text{d}z \\
& =
-\int_0^1  f'(z)
(g(z) - \hat{g}(z))\text{d}z.
\end{align*}
This yields
\begin{align*}
-\int_0^1  f'(z)
g(z)\text{d}z
& \le
-\int_0^1 f'(z)
\hat{g}(z))\text{d}z \\
& =
c \bigl(
f(1-\frac{1}{c}\int_0^1 g(x) \text{d}x) - f(1)
\bigr)\,.
\end{align*}

\item {\em Battacharyya Sometimes Bounds Wasserstein}:
Since the cumulative $|D|$-distribution of $\Delta_0$ is equal to $1$ on $[0, 1]$, 
the maximum possible value, we have
\begin{align}
d(\Ddens{a},\Delta_0)
& = \int_0^1 (1-\absDdist{A}(z))\text{d}z \nonumber \\
& = 1 - 2\perr (\Ldens{a}) 
\le \sqrt{1-\batta (\Ddens{a})^2}\,.\label{eq:disttodelta0}
\end{align}
Similarly, since the cumulative $|D|$-distribution of $\Delta_1$ is $0$
on $[0,1),$ we have
\begin{align}
d(\Ddens{a},\Delta_1)
 = \int_0^1 \absDdist{A}(z)\text{d}z 
& = 2\perr (\Ddens{a})  \le \batta (\Ddens{a})\,. \label{eq:disttodelta1}
\end{align}

\end{enumerate}

\section{Wasserstein Metric and Degradation -- 
Lemma~\ref{lem:degradationandwasserstein}}\label{sec:blmetricversusdegradation}
{\em Proof}:
\renewcommand\theenumi{\roman{enumi}} 
\renewcommand{\labelenumi}{(\roman{enumi})}
\begin{enumerate}
\item{\em Wasserstein versus Degradation}: 
Let $f$ be a function of bounded total variation on $[0,1].$
(This implies that $f$ has left and right limits.)
Note that we include $|f(0-)|$ and $|f(1+)|$ in the definition of total variation,
which we denote by $\int_0^1 | f'(x) | \text{d}x.$
Define $F(x) = \int_0^x f(z) \text{d}z.$
We claim that if $F \ge 0$ then
\[
\Bigl(\int_0^1 F(x) \text{d}x\Bigr)\Bigl(\int_0^1 | f'(x) | \text{d}x\Bigr) 
\ge \frac{1}{2} \Bigl(\int_0^1 | f(x) | \text{d}x\Bigr)^2
\]
This claim implies statement (\ref{lem:blmetricdegradation}) by setting
$f(z) = (\absDdist{b}(1-z)-\absDdist{a}(1-z))$ and noting that, in this case,
$\int_0^1 |f'(z)| \text{d}z \le 2.$

We now prove the claim.  Let $S$ be the set of points $x$ in $[0,1],$ including the endpoints,
where $f(x-) f(x+) \le 0.$ Note that $S$ is closed and we may assume $f=0$ on $S.$
The complement of $S$ is a collection of disjoint open intervals such that $f$ is either
strictly positive or strictly negative in each interval.
Consider the subset of intervals on which $f$ is strictly negative.
Without loss of generality we may take this collection to be finite.
Indeed, suppose there are countably infinitely many such intervals  $J_1, J_2, ...$
Define an approximation $f_k$ by setting $f_k(x) = -f(x)$ for $x \in \cup_{i=k+1}^\infty J_i$
and $f_k(x) = f(x)$ otherwise.
Then $F_k(x) = \int_0^x f_k(z) \text{d}z \ge F(x) \ge 0$ and
$F_k \rightarrow F$ uniformly.  Furthermore, 
$\int_0^1 |f_k(x)| = \int_0^1 |f(x)|$ and
$\int_0^1 |f'_k(x)|$ converges to $\int_0^1 |f'(x)|$ from below.

By taking unions of intervals as necessary we can find an increasing
sequence $0=x_1,x_2,...,x_{2k},x_{2k+1}=1$ such that on $I_i = [x_i,x_{i+1}]$ we have
$f \ge 0$ for $i$ odd and $f \le 0$ for $i$ even.
The sequence of points $x_i$ is strictly increasing except possibly for the last pair
which may coincide at $1$.
Define
\begin{align*}
h_i & = \max_{x \in I_i} |f(x)|\,, \\
w_i & = |\int_{I_i} f(x) \text{d}x|/h_i = \int_{I_i} |f(x)| \text{d}x /h_i\,,
\end{align*}
where $w_i = 0$ if $h_i = 0.$  Note that $w_i \le |I_i|.$
We have
\begin{align*}
\int_0^1 |f'(x)| \text{d}x  & \ge 2 \sum_{i=1}^{2k} h_i\, \int_0^1 |f(x)| \text{d}x  =  \sum_{i=1}^{2k} h_i w_i\,.
\end{align*}
We claim in addition that
\[
2 \int_0^1 F(x) \text{d}x \ge  \sum_{i=1}^{2k} h_i w^2_i\,.
\]
The desired result then follows from Jensen's inequality
\[
\frac{\sum_{i=1}^{2k} h_i w^2_i}
{\sum_{i=1}^{2k} h_i }
\ge
\Biggl(\frac{\sum_{i=1}^{2k} h_i w_i}
{\sum_{i=1}^{2k} h_i }\Biggr)^2\,.
\]

Now, note that
\[
\int_0^1 F(x) = \int_0^1 (1-x) f(x) \text{d}x\,.
\]
It is straightforward to show that for odd $i$ we have
\[
 \int_{I_i} (1-x) f(x) \text{d}x \ge \frac{1}{2}((\bar{x}_{i+1}+ w_{i})^2 - \bar{x}^2_{i+1}) h_i
\]
and for even $i$ we have 
\[
\int_{I_i} (1-x) f(x) \text{d}x \ge -\frac{1}{2}(\bar{x}_{i}^2 - (\bar{x}_{i} - w_i)^2) h_i
\]
where $\bar{x}=1-x.$
Indeed, for odd $i$ we have 
$\int_{x_i}^{z} (f(x)-h_i 1_{\{x \ge x_{i+1}-w_i\}}) \text{d}x \ge 0$
for all $z\in [x_i,x_{i+1}]$ with equality at $z=x_{i+1}.$
Hence
\begin{align*}
\int_{x_i}^{x_{i+1}} \!\!\!\!\!& (1-x) (f(x)-h_i 1_{\{x \ge x_{i+1}-w_i\}})\text{d}{x} \\
&  =
\int_{x_i}^{x_{i+1}} 
\Bigl(\int_{x_i}^z (f(x)-h_i 1_{\{x \ge x_{i+1}-w_i\}}) \text{d}x\Bigr)
\text{d}z \ge 0,
\end{align*}
which gives
\begin{align*}
\int_{x_i}^{x_{i+1}} & (1-x) f(x)\text{d}{x} \\
&  \ge
\int_{x_i}^{x_{i+1}}  (1-x) h_i 1_{\{x \ge x_{i+1}-w_i\}})\text{d}{x} \\
& = -\frac{1}{2} h_i (\bar{x}^2_{i+1} - (\bar{x}_{i+1} + w_i)^2)\,.
\end{align*}
The argument for even $i$ is similar.
We obtain
\begin{align*}
& 2 \int_{I_{2i-1} \cup I_{2i}} (1-x) f(x) \text{d}x \ge \\
& h_{2i-1} w_{2i-1}^2 + h_{2i} w_{2i}^2 + 
2(h_{2i-1} w_{2i-1} - h_{2i} w_{2i}) \bar{x}_{2i}
\end{align*}
Defining $\bar{x}_{2k+2}=0$ for notational convenience,
we can write
\begin{align*}
2 & \int_0^1 (1-x) f(x) \text{d}x  - \sum_{i=1}^{2k} h_i w_i^2 \\
& \ge 
2 \sum_{i=1}^k (h_{2i-1} w_{2i-1} - h_{2i} w_{2i}) \bar{x}_{2i} \\
& = 
2 \sum_{i=1}^k \Biggl( \sum_{j=1}^i (h_{2j-1} w_{2j-1} - h_{2j} w_{2j}) \Biggr)
(\bar{x}_{2i}-\bar{x}_{2(i+1)}) \\
& = 
2 \sum_{i=1}^k F(x_{2i+1}) (\bar{x}_{2i}-\bar{x}_{2(i+1)}) 
 \ge 0,
\end{align*}
and the proof is complete.

\item
{\em Entropy and Battacharyya Bound Wasserstein}:
Let us first focus on the inequality between the Wasserstein distance and the Battacharyya parameter. 
From point (i) we know that
\begin{align*}
d(\Ldens{a}, \Ldens{b}) & \leq 2 \sqrt{\int_0^1 z (\absDdist{B}-\absDdist{A}) \text{d}z} \\ 
& = 2 \sqrt{
\int_0^1 \Bigl(\int_z^1 (\absDdist{B}(x)-\absDdist{A}(x))\text{d}x \Bigr) \text{d}z
}. 
\end{align*}
By integrating by parts twice we have
\begin{align}\label{eq:battasecondderivative}
\batta(\Ddens{a}) & =
\int_0^1 \sqrt{1-z^2} \absDdens{a}(z) \text{d}z \nonumber \\
& = 
\int_0^1 (1-z^2)^{-\frac{3}{2}}\Bigl(\int_z^1 \absDdist{A}(x) \text{d}x \Bigr) \dee z,
\end{align}
and
\begin{align*}
\entropy(\Ddens{a}) & =
\int_0^1 h_2\Bigl(\frac{1-z}{2}\Bigr) \absDdens{a}(z) \text{d}z \\
& = 
\frac{1}{\ln 2}\int_0^1 (1-z^2)^{-1}\Bigl(\int_z^1 \absDdist{A}(x) \text{d}x \Bigr) \dee z.
\end{align*}
Thus we obtain
\begin{align*}
\int_0^1 \!\!\!\!z (\absDdist{B}\!-\!\absDdist{A}) \text{d}z 
&\leq (\ln 2) (\entropy(\Ldens{b}) \!-\! \entropy(\Ldens{a})) \leq \batta(\Ldens{b}) \!-\! \batta(\Ldens{a})\,. 
\end{align*}
This yields
\begin{align*}
d(\Ldens{a}, \Ldens{b}) 
& \leq 2 \sqrt{(\ln 2) (\entropy(\Ldens{b}) - \entropy(\Ldens{a}))} 
\leq 2 \sqrt{\batta(\Ldens{b}) - \batta(\Ldens{a})}.
\end{align*}
For the final inequality first note that
\(
g(z) = (1-z^2)^{-1}\Bigl(\int_z^1 (\absDdist{B}(x)-\absDdist{A}(x)) \text{d}x \Bigr)
\leq  1\,.
\)
Let 
\(
v = \int_0^1 g(z) \text{d}z  = (\ln 2) {(\entropy(\Ldens{b}) - \entropy(\Ldens{a}))}\,.
\)
It follows that
\begin{align*}
{\batta(\Ldens{b}) - \batta(\Ldens{a})} 
& = \int_0^1 \frac{1}{\sqrt{1-z^2}} g(z) \dee z \\
&\leq \int_{1-v}^1 \frac{1}{\sqrt{1-z^2}}  \dee z 
 = \arccos (1-v) \\
& \leq \frac{\pi}{2}\sqrt{v} 
 = \frac{\pi}{2}\sqrt{\ln 2 {(\entropy(\Ldens{b}) - \entropy(\Ldens{a}))} } \\
& \leq \sqrt{2 {(\entropy(\Ldens{b}) - \entropy(\Ldens{a}))} }\,.
\end{align*}

\item {\em Continuity for Ordered Families}: 
Assume that $\Ldens{a} \prec \Ldens{b}$. 
From point (\ref{lem:blmetricentropy}) we know that
\begin{align*}
d(\Ldens{a}, \Ldens{b}) & \leq 2 \sqrt{\batta(\Ldens{b}) - \batta(\Ldens{a})},
\end{align*}
and the continuity follows from the continuity of the Battacharyya parameter for smooth channel families. 
\end{enumerate}

\section{Sufficient Condition for Continuity -- Lemma~\ref{lem:condforcontinuity}, 
Continuity for Large Entropies -- Lemma~\ref{lem:continuityforlargeentropy},  
Universal Bound on Continuity Region -- Lemma~\ref{lem:FPforlargeentropy}} \label{app:bounds}
\blemma[Bound on $\batta$]\label{lem:bbound}
Consider two $L$-densities $\Ldens{a}_1 \prec \Ldens{a}_2$.
Then, for any degree distribution $\redge(\cdot)$,
\begin{align*} 
(\batta(\rho(\Ldens{a}_{2}))-\batta(\rho(\Ldens{a}_{1}))) & \leq
(\batta(\Ldens{a}_2)\! -\! \batta(\Ldens{a}_1)) 
\rho'(1 - \batta^2({\Ldens{a}_1})).
\end{align*}
\end{lemma}
\begin{IEEEproof}
Let $\Ldens{a}$ be a density and let $U$ be distributed according to the
corresponding $|D|$-distribution.  By Jensen's inequality we have
\begin{align*}
\batta(\Ldens{a}) &= \expectation [(1-U^2)^\frac{1}{2}]
\le (\expectation [1-U^2])^\frac{1}{2}
= (1-\moment_{\Ldens{a}, 1})^\frac{1}{2},
\end{align*}
where we have introduced the notation $\moment_{\Ldens{a}, k} = \expectation[U^{2 k}].$
The Taylor expansion 
\(
(1-u^2)^\frac{1}{2} = 1-\sum_{k=1}^\infty \alpha_{k} u^{2k}
\)
gives
\[
\batta(\Ldens{a}) = 1-\sum_{k=1}^\infty \alpha_{k} \moment_{\Ldens{a}, k}\,
\]
where $\alpha_{k}$ is positive for each $k$.
The functionals $\moment_{\Ldens{a}, k}$ have the important (Fourier) property
$\moment_{\rho(\Ldens{a}), k} = \rho(\moment_{\Ldens{a}, k})$ \cite{RiU08}.\footnote{We introduced here only
the even moments, since only these are needed. The odd moments are multiplicative as well.}
Since $u^k$ is convex and increasing for $k\ge 1,$
we have $\moment_{\Ldens{a}_1, k} \ge m_{\Ldens{a}_2, k}$.
Hence,
\begin{align*}
&\quad\batta(\rho(\Ldens{a}_2))-\batta(\rho(\Ldens{a}_1)) \\
&=
\sum_{k=1}^\infty \alpha_{k} \bigl(\rho(\moment_{\Ldens{a}_1, k})-\rho(\moment_{\Ldens{a}_2, k})\bigr) \\
&\le
\sum_{k=1}^\infty \alpha_{k} \rho'(\moment_{\Ldens{a_1}, k})
\bigl(\moment_{\Ldens{a}_1, k}-\moment_{\Ldens{a}_2, k}\bigr) \\
&\le
\rho'(\moment_{\Ldens{a_1}, 1})\bigl(\sum_{k=1}^\infty \alpha_{k} 
(\moment_{\Ldens{a}_1, k}-\moment_{\Ldens{a}_2, k})\bigr) \\
&\le
\rho'(1-\batta^2(\Ldens{a_1}))\bigl(\sum_{k=1}^\infty \alpha_{k} 
(\moment_{\Ldens{a}_1, k}-\moment_{\Ldens{a}_2, k})\bigr) \\
&=
\rho'(1-\batta^2(\Ldens{a_1}))
\bigl(\batta(\Ldens{a}_2)-\batta(\Ldens{a}_1)\bigr).
\end{align*}
\end{IEEEproof}

\blemma[Bound on Derivative of $\batta$]\label{lem:dbbound}
Consider two $L$-densities $\Ldens{a}_1 \prec \Ldens{a}_2$.  
Let $0 \leq \ih_1\leq \ih_2 \leq 1$ and let 
$\Ldens{c}_{\ent_1}$ and $\Ldens{c}_{\ent_2}$ 
denote the two corresponding channels from an ordered family $\{\Ldens{c}_\ent\}$.
Set $B_{\ih_i} = \batta(\Ldens{c}_{\ent_i}\}$ for $i=1,2$.
Then, for any dd pair $(\ledge, \redge)$
\begin{align*}
|\batta(\Tc_{\ih_1}(\Ldens{a}_{1}))-&\batta(\Tc_{\ih_2}(\Ldens{a}_{2}))|
\le \\ &\alpha\, 
|\batta(\Ldens{a}_{1})-\batta(\Ldens{a}_{2})|
+|B_{\ih_1}-B_{\ih_2}|\, ,\nonumber
\end{align*}
where
$\alpha = B_{\ih_1}\lambda'(1) \rho'(1 - \batta^2({\Ldens{a}_1}))$.
\end{lemma}

\begin{IEEEproof}
First, since $\batta(\Ldens{a} \vconv \Ldens{b})=\batta(\Ldens{a})
\batta(\Ldens{b})$, $\batta(\Tc_{\ent}(\Ldens{a}))= B_{\ent}
\lambda(\batta(\rho(\Ldens{a})))$.
Second, since 
$0\le \lambda(x)\le 1$ and $\lambda'(x)\le \lambda'(1)$,
$|\lambda(x_1)-\lambda(x_2)| \leq \lambda'(1) |x_1-x_2|$ for all $x_1, x_2 \in [0, 1]$.
This implies that $|\batta(T_{\ih_1}(\Ldens{a}_{1}))\!-\!\batta(T_{\ih_1}(\Ldens{a}_{2}))|$
is upper bounded by $\lambda'(1)B_{\ih_1} 
|\batta(\rho(\Ldens{a}_{1}))-\batta(\rho(\Ldens{a}_{2}))|$.
Using the triangle inequality, we get
\begin{align}
& |\batta(T_{\ih_1}(\Ldens{a}_{1}))-\batta(T_{\ih_2}(\Ldens{a}_{2}))| \nonumber \\
& \leq
|\batta(T_{\ih_1}(\Ldens{a}_{1}))\!-\!\batta(T_{\ih_1}(\Ldens{a}_{2}))|\!+\! 
|\batta(T_{\ih_1}(\Ldens{a}_{2}))\!-\!\batta(T_{\ih_2}(\Ldens{a}_{2}))| \nonumber \\
& \leq \label{eq:Contraction0} \lambda'(1)B_{\ih_1} 
|\batta(\rho(\Ldens{a}_{1}))-\batta(\rho(\Ldens{a}_{2}))|
+|B_{\ih_1}-B_{\ih_2}|.
\end{align}
The first term above can be bounded using Lemma~\ref{lem:bbound}. 
\end{IEEEproof}

{\em Proof of Lemma~\ref{lem:condforcontinuity}}: Denote by
$\Ldens{x}_{\ih}$ the BP FP for the channel $\Ldens{c}_{\ih}$ and
notice that any other FP $\Ldens{x}'_{\ih}$ for the same channel
is necessarily upgraded with respect to $\Ldens{x}_{\ih}$, i.e.,
$\Ldens{x}'_{\ih}\prec \Ldens{x}_{\ih}$.  Indeed, $\Ldens{x}'_{\ih}
\prec \Delta_0$. By applying the density evolution operator, we
deduce that $\Ldens{x}'_{\ih} \prec \Ldens{x}^{(\ell)}_{\ih}$, where
$\Ldens{x}^{(\ell)}_{\ih}$ is the density after $\ell$ iterations
of BP.  By taking the limit $\ell\to\infty$ we get $\Ldens{x}'_{\ih}
\prec \Ldens{x}_{\ih}$.  We conclude that if $\Ldens{x}_\ih$ does
not satisfy (\ref{eq:UniquenessCond}) then neither can any other
FP for the same channel.

Assume on the other hand that $\Ldens{x}_\ih$ satisfies
(\ref{eq:UniquenessCond}) and that there exists a distinct FP for
the same channel, necessarily upgraded with respect to $\Ldens{x}_{\ih}$,
also satisfying (\ref{eq:UniquenessCond}). Call this density
$\Ldens{x}_{\ih}'$.  In this case,
\begin{align*}
|\batta(\Ldens{x}_{\ih})-\batta(\Ldens{x}'_{\ih})|  
& \stackrel{\text{$\Ldens{x}_{\ih}, \Ldens{x}'_{\ih}$ are FPs}}{=} |\batta(T_{\ih}(\Ldens{x}_{\ih}))-\batta(T_{\ih}(\Ldens{x}'_{\ih}))| \\
& \stackrel{\text{Lemma~\ref{lem:dbbound}}}{\leq} (1-\delta) |\batta(\Ldens{x}_{\ih})-\batta(\Ldens{x}'_{\ih})|,
\end{align*}
a contradiction since $\delta>0$.
The above argument shows that there can be at most one FP with this
property and that this FP must be the forward DE one.

Let us now prove Lipschitz continuity, c.f. (\ref{eq:Lipschitz}).
Under our hypotheses,
the two FPs $\Ldens{x}_{\ih_1}$ and $\Ldens{x}_{\ih_2}$ are
the BP FPs for channels $\Ldens{c}_{\ih_1}$ and $\Ldens{c}_{\ih_2}$.
Consider therefore the respective BP sequences (starting with $\Delta_0$)
$\{\Ldens{x}^{(\ell)}_{\ih_1}\}_{\ell\ge 0}$,
$\{\Ldens{x}^{(\ell)}_{\ih_2}\}_{\ell\ge 0}$. For each $\ell$,
$\Ldens{x}^{(\ell)}_{\ih_1}$ (respectively $\Ldens{x}^{(\ell)}_{\ih_2}$)
is degraded with respect to $\Ldens{x}_{\ih_1}$ (respectively
$\Ldens{x}_{\ih_2}$),
and therefore satisfies the  condition (\ref{eq:UniquenessCond}),
since the latter does.
Furthermore, assuming without loss of generality
$\ent_2>\ent_1$, we have $\Ldens{x}^{(\ell)}_{\ih_2}\succ 
\Ldens{x}^{(\ell)}_{\ih_1}$.
Let $\delta_{(\ell)} \defas
|\batta(\Ldens{x}^{(\ell)}_{\ih_1})-\batta(\Ldens{x}^{(\ell)}_{\ih_2})|$.
Since DE is initialized with $\Delta_0$, we have $\delta_0=0$. By
applying Lemma \ref{lem:dbbound} we get $\delta_{\ell+1}\le
(1-\delta)\, \delta_{\ell} +|B_{\ih_1}-B_{\ih_2}|$, and therefore
\begin{align*}
\delta_{\ell}\le (1+(1-\delta)+(1-\delta)^2&+\cdots+(1-\delta)^{\ell-1})\,	 
|B_{\ih_1}-B_{\ih_2}| \\ & \le \frac{1-(1-\delta)^{\ell}}{1-(1-\delta)}\, |B_{\ih_1}-B_{\ih_2}|\, .
\end{align*}
The thesis follows by taking the $\ell\to\infty$ limit.\qed

{\em Proof of Lemma~\ref{lem:continuityforlargeentropy}}:
For $\beta \in [0, 1]$ define 
\begin{align} \label{equ:bhbound}
g(\beta) 
& = \frac{\beta}{(1-(1-\beta^2)^{\dr-1})^{\frac{\dl-1}{2}}}.
\end{align}
Note that $g(1)=1$ and that $g(\beta)$ is continuous.

Assume that we run forward DE with the channel $\Ldens{c}$ and that
$\batta(\Ldens{c})=g(\beta)$, for some $\beta \in [0, 1]$.  We then claim that
for the resulting FP $\Ldens{x}$, $\batta(\Ldens{x}) \geq \beta$.
To see this, let $\{\Ldens{x}^{(\ell)}\}$ denote the sequence of densities with
$\Ldens{x}^{(0)}=\Delta_0$. Using the Battacharyya functional on
the DE equations and then extremes of information combining bounds
we see that
\begin{align*}
\batta(\Ldens{x}^{(\ell)})  \geq \batta(\Ldens{c}) \Bigl(1-(1-\batta(\Ldens{x}^{(\ell-1)})^2)^{\dr-1} \Bigr)^{\frac{\dl-1}{2}}.
\end{align*}
Note that if $\batta(\Ldens{x}^{(\ell-1)}) \geq \beta$ then
\begin{align*}
\batta(\Ldens{x}^{(\ell)}) & \geq  
 \batta(\Ldens{c}) \Bigl(1-(1-\batta(\Ldens{x}^{(\ell-1)})^2)^{\dr-1} \Bigr)^{\frac{\dl-1}{2}} \\
 & \geq  g(\beta) 
 \Bigl(1-(1-\beta^2)^{\dr-1} \Bigr)^{\frac{\dl-1}{2}} = \beta.
\end{align*}
The induction is anchored by noting that $1=\batta(\Delta_0) \geq
\beta$ since we assumed that $\beta \in [0, 1]$.  In summary, for
each $\beta \in (0, 1]$, equation~\eqref{equ:bhbound} gives us the
lower bound $\batta(\Ldens{x}) \geq \beta$ for the FP $\Ldens{x}$
of forward DE with the channel $\batta(\Ldens{c})=g(\beta)$. Another
way of interpreting \eqref{equ:bhbound} is that it gives us an upper
bound on $\batta(\Ldens{c})$ if we fix $\batta(\Ldens{x})=\beta$.

According to Lemma~\ref{lem:condforcontinuity}, the GEXIT curve is
Lipschitz continuous (in the Battacharyya parameter) at the FP
$(\Ldens{c}_{\ih}, \Ldens{x}_{\ih})$ if
\begin{align} \label{equ:lipschitz}
\batta(\Ldens{x}_{\ih}) & \geq \sqrt{1-(\batta(\Ldens{c}_{\ih}) (\dl-1)(\dr-1))^{-\frac{1}{\dr-2}}}.
\end{align}
Note that (\ref{equ:bhbound}) as well as \eqref{equ:lipschitz}
(if we interpret the inequality as an equality) give rise to 
curves in the $(\batta(\Ldens{c}), \batta(\Ldens{x}))$ space.
Inserting (\ref{equ:bhbound}) into \eqref{equ:lipschitz} gives
us the points where these two curves cross.  If we set
$\sqrt{x}=\batta(\Ldens{x}_{\ih})$, massage the resulting expression,
and set it to $0$, we get  \eqref{equ:existence}.
As shown in the subsequent Lemma~\ref{lem:technicallemma1}, 
\eqref{equ:existence} has a unique positive solution in
$(0, 1]$ (i.e., the two curves only cross once), $b(x)<a(x)$ after this solution, and $g(\beta)$ is an increasing function
above this solution. The situation
is shown in Figure~\ref{fig:crossings}.
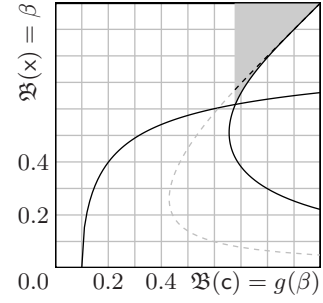
\begin{figure}[htp]
\centering
\input{ps/crossings_arxiv}
\caption{\label{fig:crossings} 
Consider the $(3, 6 )$-regular ensembles.
The $C$-shaped curve on the right is (\ref{equ:bhbound}). This curve
has two branches. The top branch gives a tighter bound and pairs
$(\batta(\Ldens{c}), \batta(\Ldens{x}))$ generated by DE must lie
above this branch.  The second curve, given by \eqref{equ:lipschitz},
denotes the region (above the curve) where there can be at most one
FP. The GEXIT curve for the BEC is shown as a dashed curve.
The portion of this GEXIT curve starting at $(1, 1)$ which is
contained in the gray area is guaranteed to be smooth.  
} \end{figure}

Inserting this solution back into  \eqref{equ:bhbound} gives us a
value of $\batta(\Ldens{c}_{\ih})$ so that for all channels with
larger Battacharyya constant the densities generated by forward DE
are non-trivial and are Lipschitz continuous.  This insertion is
equivalent to evaluating $c(x)$ at $x=\xLE$.

Let us finish the proof by showing that $\batta(\Ldens{x}_{\ih}) \geq \xunstab(1)$
for all $\ih > \entLE$. Indeed, from the extremes of 
information combining we have
$$
\batta(\Ldens{x}_{\ih}) \leq (1 - (1 -
\batta(\Ldens{x}_{\ih})^{\dr-1}))^{\dl-1},
$$
where above we have replaced $\batta(\Ldens{c}_{\ih}) \leq 1$. Above inequality
implies that either $\batta(\Ldens{x}_{\ih})=0$ or $\batta(\Ldens{x}_{\ih})\in
[\xunstab(1), 1]$. From the
above discussion we know that for $\ih > \entLE$ the densities generated by forward DE are
non-trivial. Putting things together we conclude that
$\batta(\Ldens{x}_{\ih})\geq \xunstab(1)$. 

\begin{lemma}[Unique Zero]\label{lem:technicallemma1}
For $\dr \geq \dl\geq 3$ let
\begin{align*}
a(x) & = (1-(1-x)^{\dr-1})^{\dl-1}, \\
b(x) & = (\dl-1)^2(\dr-1)^2 x (1-x)^{2(\dr-2)}, \\
c(x) & = \sqrt{x/a(x)}.
\end{align*}
Then there is a unique solution of $a(x) = b(x)$ in the interval $(0,1]$, call it $\tilde{x}$.
Further, $c(x)$ is increasing for $x \in [\tilde{x}, 1]$.
\end{lemma}
\begin{IEEEproof}
Set $L=\dl-1$ and $R=\dr-1$, multiply the equation by $1/L^2$ and set
$y=(1-x)^R$. This gives the equivalent equation $A(y)=B(y)$, where $A(y)  =
(1-y)^L/L^2$, and $B(y)  = R^2 (y^{2-\frac2{R}}-y^{2-\frac1{R}})$.

The function $A(y)$ is (i) decreasing and convex for $L \geq 2$, (ii)
$A(0)=1/L^2>0$, (iii) $A(1)=0$. The function $B(y)$ is (i) increasing
for $y \in [0, y_1=(\frac{2R-2}{2R-1})^R]$, (ii) decreasing for $y \in
[y_1, 1]$, (iii) concave for $y \in [y_2=(\frac{(2R-2)(R-2)}{(2R-1)(R-1)})^R, 1]$,
and (iv) $B(0)=B(1)=0$. Note that $0 \leq y_2 < y_1$ since we assumed that $R \geq 2$.

We conclude that in the region $[0, y_1]$ there is exactly one solution, call it $\tilde{y}$:
there is at least one since $1/L^2=A(0)>B(0)=0$, whereas $A(y_1)<1/L^2 \stackrel{R\geq L\geq 2}{\leq}
R/8 <
R 2^{-3+\frac1R} <
R^2 2^{-2+\frac1R} (2^{\frac1{R}}-1) = 
B(\frac12)
\leq B(y_1)$ (since $y_1$ is the position where $B(y)$ is maximized);
and there is only one solution since in $[0, y_1]$, $A(y)$ is strictly decreasing,
whereas $B(y)$ is increasing. 

In the region, $y \in [y_1, 1)$ there can be no further solution
since $A(y_1)< B(y_1)$, $A(1)=B(1)=0$, and $A(y)$ is convex whereas
$B(y)$ is concave.

Note that $b(x)$ starts at $0$, then increases until it reaches its
maximum, and then decreases back to $0$, which it reaches at $x=1$.
Let $\hat{x}$ be the largest value within $[0, 1]$ so that $b(\hat{x})=1$
(we will verify shortly that this is well defined).
Since $b(\hat{x})=1$ but $a(x)\leq 1$ for all $x \in [0, 1]$, it
is clear that $\hat{x} \leq \tilde{x}$.  Note that $\tilde{x}$ is obtained from $\tilde{y}$. 
Recall that we want to
show that $c(x)$ is increasing for $x \in [\tilde{x}, 1]$.  We will
show the stronger statement that $c(x)$ is increasing for $x \in
[\hat{x}, 1]$.  This is equivalent to showing that $x/a(x)$ is
increasing in this range. Note that $(x/a(x))'=p(x) q(x)$, where
\begin{align*}
q(x)=1 - (1 - x)^{\dr - 2} ((\dl \dr -\dl- \dr) x+1),
\end{align*}
and $p(x)\geq 0$ for $x \in [0, 1]$.  The factor $q(x)$ can be
written as $y^{\dr - 1} (\dl \dr -\dl- \dr)- y^{\dr - 2} ((\dl \dr
-\dl- \dr+1)+1)$, where $y=1-x$.  This polynomial has two sign changes
and hence by Descarte's rule of signs at most two positive roots.
It follows that $q(x)$ has at most $2$ roots for $x \leq 1$. Since
$q(0)=0$ and $q(1)=1$, there must be exactly one root of $q(x)$ in
$(0, 1]$ and once the function is positive, it stays so within $[0,
1]$. It therefore suffices to prove that $q(\hat{x}) \geq 0$. By definition of $\hat{x}$
we have $(1-\hat{x})^{\dr-2}=\frac{1}{(\dl-1)(\dr-1) \sqrt{\hat{x}}}$.
We therefore have $q(\hat{x})=r(z)\mid_{z=\hat{x}}$, where
\begin{align*} 
r(z)
& = 1 \!-\! 
\frac{(\dl \dr \!-\!\dl \!-\!\dr) \sqrt{z}}{(\dl\!-\!1)(\dr\!-\!1)} -
\frac{1}{(\dl\!-\!1)(\dr\!-\!1) \sqrt{z}}.
\end{align*}
A quick check shows that $r(z)\geq 0$ for $z \in [\frac{1}{(\dl
\dr-\dl - \dr)^2}, 1]$.  The proof will be complete if we can show
that $\hat{x} \in [\frac{1}{(\dl \dr-\dl - \dr)^2}, 1]$. We do this
in two steps.  We claim that $\hat{x} \geq \breve{x}=\frac{c
\ln\sqrt{(\dl-1)(\dr-1)}}{\dr-2}$, where
$c=\frac{1}{1+\frac{\ln\sqrt{(\dl-1)(\dr-1)}}{\dr-2}}$, and that
$\breve{x} \in [\frac{1}{(\dl \dr-\dl - \dr)^2}, 1]$. The second claim is
immediate. To see the first,
\begin{align*}
b(\breve{x}) 
& \geq  (\dl\!-\!1)^2(\dr\!-\!1) c  \ln\sqrt{(\dl\!-\!1)(\dr\!-\!1)} e^{2(\dr\!-\!2) \ln(1-\breve{x})} \\
& \geq  (\dl\!-\!1)^2(\dr\!-\!1) c  \ln\sqrt{(\dl\!-\!1)(\dr\!-\!1)} e^{-\frac{2(\dr\!-\!2) \breve{x}}{1-\breve{x}}} \\
& =  \frac{(\dl\!-\!1)(\dr\!-\!2)  \ln\sqrt{(\dl\!-\!1)(\dr\!-\!1)}}{\dr-2+\ln\sqrt{(\dl-1)(\dr-1)}}  \geq 1 = b(\hat{x}).
\end{align*}
This shows that that the maximum of $b(x)$ in $[0, 1]$ is above $1$ and so $\hat{x}$ is
well defined. Since further, $b(x)$ is a unimodal function and $\hat{x}$ was defined to be the
largest value of $x \in [0, 1]$ so that $b(\hat{x})=1$ it follows that $\hat{x} \geq \breve{x}$, as claimed.
\end{IEEEproof}

{\em Proof of Lemma~\ref{lem:FPforlargeentropy}}: 
Let $a(x), b(x)$ and $c(x)$ be as defined in
Lemma~\ref{lem:continuityforlargeentropy}. We will provide an upper
bound on the unique solution of $a(x) = b(x)$. 
Notice that $a(x)$ represents the DE equations for a
BEC with parameter $\epsilon=1$.  Therefore, we know that for $x\geq
\xunstab(1)$, $a(x) \geq x$. We claim that $b(x)$ and $l(x)=x$
intersect only at one point in $(0, 1]$.  Indeed $b(x)=x$, $x \in (0, 1]$, is equivalent to 
\begin{align*}
x = 1 - ((\dl-1)(\dr-1))^{-\frac1{\dr-2}} \triangleq \overline{x}.
\end{align*}
Since $b(1)=0$, whereas $l(1)=1$, we conclude that for $x \in
[\overline{x}, 1]$, $b(x) \leq x$.

We further claim that $\overline{x} \geq \xunstab(1)$. Let us assume
this for a moment. Then we have $a(x) \geq x \geq b(x)$ for $x \in
[\overline{x}, 1]$.  We conclude that the unique solution of
$a(x)=b(x)$ in $(0, 1]$ is upper bounded by $\overline{x}$.

We finish the lemma by proving $\overline{x} \geq \xunstab(1)$. Indeed, since
$\overline{x} \neq 0$, all we need to show is that $ (1 - (1 -
\overline{x})^{\dr-1})^{\dl-1} \geq \overline{x}$\footnote{Recall that for the
BEC(1), the DE equation is given by $x=(1 - (1-x)^{\dr-1})^{\dl-1}$.
Furthermore, there are 3 FPs namely, 0, $\xunstab(1)$ (unstable) and $1$
(stable). Finally,  we have that $(1 - (1-x)^{\dr-1})^{\dl-1} \geq x$ if and
only if $x=0$ or $x\in [\xunstab(1), 1]$. See Chapter 3 in \cite{RiU08} for more details.}.
For $3=\dr=\dl$ one can verify the validity of the claim directly.
In general, we have
\begin{align*}
(1 - (1 & - \overline{x})^{\dr-1})^{\dl-1}  \geq (1 - (1 - \overline{x})^{\dr-2})^{\dl-1} \\
& =  \Big(1 - \frac1{(\dl-1)(\dr-1)}\Big)^{\dl-1} 
\geq 1 - \frac{1}{\dr-1} \\
& \geq 1 - \Big(\frac1{(\dl-1)(\dr-1)}\Big)^{\frac1{\dr-2}}  = \overline{x},
\end{align*}
where the last inequality follows since 
$ (\frac1{(\dl-1)(\dr-1)})^{\frac1{\dr-2}} \stackrel{\dr\geq 4}{\geq} \frac1{\dr-1}$.

The Battacharyya parameter of the channel is thus upper bounded by
$\sqrt{\overline{x}/a(\overline{x})}$. Using the upper bound on the
entropy in Lemma~\ref{lem:entropyvsbatta},  we get the claimed bound.

It remains to show that this bound converges to $0$ when we fix
the rate and let the dds tend to infinity. To
simplify our notation, let $L=\dl-1$ and $R=\dr-1$.  We have
\begin{align*}
\overline{\ih} & =  \sqrt{\overline{x}/a(\overline{x})} 
 =\sqrt{\left(1-\left(LR \right)^{-\frac{1}{R-1}}\right)\left(1-\left(LR\right)^{-\frac{R}{R-1}}\right)^{-L}} \\
&\stackrel{\text{(a)}}{\leq} e^{\frac{1}{4}}\sqrt{1-\left(L R\right)^{-\frac{1}{R-1}}} = e^{\frac{1}{4}}\sqrt{1-e^{-\frac{\ln(R L)}{R-1}}} \\ 
& \leq e^{\frac{1}{4}}\sqrt{1-e^{-\frac{2}{\sqrt{R-1}}}}  \leq \frac{e^{\frac{1}{4}}\sqrt{2}}{(\dr-2)^{\frac{1}{4}}},
\end{align*}
where (a) is obtained by using the following sequence of inequalities,
\begin{align*}
&\sqrt{\left(1-\left(LR\right)^{-\frac{R}{R-1}}\right)^{-L}} 
 = \sqrt{e^{-L \ln(1-(LR)^{-\frac{R}{R-1}})}} \\
&\!\!\! \stackrel{\substack{\text{Taylor Expansion}\\\text{for}\ln(1-x)}}{\leq}\!\!\! \sqrt{e^{\frac{L (LR)^{-\frac{R}{R\!-\!1}}}{1\!-\!(LR)^{-\frac{R}{R\!-\!1}}}}} 
\!\!\leq \!\!\sqrt{e^{\frac{ (L(LR)^{-1})^{\frac{R}{R\!-\!1}}}{1\!-\!(LR)^{-1}}}} \! \leq\! e^{\frac{1}{2(\dr-2)}} \!\!\stackrel{\dr \geq 4}{\leq}\!\! e^{\frac{1}{4}}.
\end{align*}

We finish the proof by showing that $\entLE \leq \overline{\ih}$
and $\batta(\Ldens{x}_{\ih}) \geq \xunstab(1)$ for $\ih \geq
\overline{\ih}$.  Let us first show that $\entLE \leq \overline{\ih}$.
Note that $\entLE = \ent_{\BMS}(c(\xLE))$, where recall that
$\ent_{\BMS}(\cdot)$ is the function which maps the Battacharyya
constant of an element of the family to the corresponding entropy.
Thus we have $\entLE \leq c(\xLE)$. The proof is now complete by
observing that $c(\xLE) \leq c(\overline{x})$, due to the monotonicity
of the function $c(x)$ for $x \geq \tilde{x}$, as shown in
Lemma~\ref{lem:technicallemma1}.  \qed

\section{Entropy Product Inequality -- Lemma~\ref{lem:magic}}\label{app:magic}
By definition, we have
\[
\entropy(\Ldens{a}\vconv\Ldens{b})
=\int_{0}^{1} \int_{0}^{1}
\absDdens{a}(x)
\absDdens{b}(y)
k(x,y)
\text{d}x
\text{d}y,
\]
with the kernel as given in the statement.
Differentiating, we have
\begin{align*}
k_y(x,y)
& =
-\frac{1}{2 \ln(2)}\Bigl(
\ln \frac{1+y}{1-y}-
x \ln \frac{1+xy}{1-xy}
\Bigr), \\
k_{yy}(x,y)
& =
-\frac{1-x^2}{\ln(2) (1-y^2) (1-x^2y^2)},
\\
k_{xxyy}(x,y) & =
\frac{2}{\ln(2)} \frac{1+3x^2y^2}{(1-x^2y^2)^3}.
\end{align*}
Integrating by parts twice for each dimension, we see that
\begin{align*}
\entropy(\Ldens{a}\vconv\Ldens{b})
&=\int_{0}^{1} \int_{0}^{1}
\absDdens{a}(x)
\absDdens{b}(y)
k(x,y)
\dee x
\dee y
\\
&=\int_{0}^{1} \int_{0}^{1}
\tilde{\absDdist{A}}(x)
\tilde{\absDdist{B}}(y)
k_{xxyy}(x,y)
\dee x
\dee y.
\end{align*}
This proves the alternative representation of this integral.

Note that the bound (\ref{equ:boundonk}) is implied by
\(
\frac{1}{(1-x^2y^2)^3} \le (1-x^2)^{-\frac{3}{2}}(1-y^2)^{-\frac{3}{2}}\,.
\)
Let $u=(1-x^2)^{-1}$ and $v=(1-y^2)^{-1}.$  Then the desired inequality is equivalent to
\(
\frac{1}{(1/u+1/v-1/uv)^3} \le u^{\frac{3}{2}}v^{\frac{3}{2}}\,
\)
for $u,v \ge 1.$
Raising both sides to the power of $\frac{2}{3}$ this becomes
\(
\frac{1}{(1/u+1/v-1/uv)^2} \le uv\,.
\)
Multiplying both sides by $\frac{1}{(uv)^2}$ this can be written as
\(
uv \le (v+u-1)^2 
\)
which is equivalent to
\(
0 \le (v-1)^2 + (u-1)^2 + uv-1\,,
\)
proving the claim.

This bound $k_{xxyy}(x, y) \leq \frac{8}{\ln(2)} (1-x^2)^{-\frac32}
(1-y^2)^{-\frac32}$ immediately gives rise to the claim
(\ref{equ:degradedcase}): the right-hand side factorizes and,
excluding the constant $8/\ln(2)$, each factor is just the Battacharyya
kernel in this representation ($(1-x^2)^{-\frac32}$ is the second derivative of $\sqrt{1-x^2}$, the
Battacharyya kernel in the $|D|$-domain, cf. \eqref{eq:battasecondderivative}).
Note that we can use the upper bound on $k_{xxyy}(x,y)$ to obtain (\ref{equ:degradedcase})
since by \eqref{equ:degradationcdfs}, the differences  $(\tilde{\absDdist{B'}}(y)- \tilde{\absDdist{B}}(y))$
and $(\tilde{\absDdist{A'}}(x)- \tilde{\absDdist{A}}(x))$ are non-negative.

It remains to prove the claim (\ref{equ:partiallydegradedcase}).
We claim that if $d(\Ldens{b}', \Ldens{b}) \leq \delta$ then
$|\tilde{\absDdist{B'}}(y)- \tilde{\absDdist{B}}(y)| \leq \min\{\delta,
1-y\}$. The second bound is immediate since $0 \leq \vert\absDdist{B'}(y) - \absDdist{B}(y) \vert 
\leq 1$ so that $|\tilde{\absDdist{B'}}(y)- \tilde{\absDdist{B}}(y)| \leq \int_y^1 \dee y = 1-y$. To see that the difference is less than $\delta$
 we have $|\tilde{\absDdist{B'}}(y)- \tilde{\absDdist{B}}(y)| \leq \int_{y}^1
 \vert \absDdist{B'}(z)- \absDdist{B}(z)\vert \dee z \leq  \int_{0}^1 \vert
 \absDdist{B'}(z)- \absDdist{B}(z)\vert \dee z
 \stackrel{\text{(\ref{lem:alternative})},
 \text{Lemma}~\ref{lem:blmetric}}{=} d(\Ldens{b}', \Ldens{b})$.
We now have
\begin{align*}
& \entropy((\Ldens{a}'-\Ldens{a})\vconv(\Ldens{b}'-\Ldens{b})) \\
&\leq \int_{0}^{1} \!\!\!\!\int_{0}^{1} \!\!\!|\tilde{\absDdist{A'}}(x)-\tilde{\absDdist{A}}(x)| 
|\tilde{\absDdist{B'}}(y)-\tilde{\absDdist{B}}(y)|
k_{xxyy}(x,y)
\dee x \dee y \\
&\leq \frac{8}{\ln(2)} \batta(\Ldens{a'}-\Ldens{a}) \int_{0}^{1} \min\{\delta, 1-y\} (1-y^2)^{-\frac32} \dee y  \\
&\leq \frac{8}{\ln(2)} \batta(\Ldens{a'}-\Ldens{a}) \sqrt{2 \delta}, 
\end{align*}
where to obtain the second inequality we combine the upper bound on $k_{xxyy}(x,y)$ derived above with the
alternative representation of $\batta(\Ldens{a})$ as given in \eqref{eq:battasecondderivative}.
\section{Evaluation of GEXIT Integral -- Lemma~\ref{lem:areaunderBPGEXIT}}\label{app:areaunderBPGEXIT}
For the proof of Lemma~\ref{lem:areaunderBPGEXIT} it will be handy to have the following
two lemmas available.
\begin{lemma}[Entropy of Single-Parity Check Code]\label{lem:entropyofcheck}
Consider a single-parity check code of length $\dr$. Let $X$ denote
a codeword, chosen uniformly at random from this code. Let $Y$ denote
the result of passing the codeword through a BMS channel with 
density $\Ldens{x}$. Then 
\begin{align*}
H(X \mid Y) & = 
\dr \entropy(\Ldens{x}) - \entropy(\Ldens{x}^{\cconv \dr}).
\end{align*}
\end{lemma}
\begin{IEEEproof}
Let $X_1,...,X_{\dr}$ be uniform random bits and let $Z$ denote
their parity.  Suppose $X_i$ is transmitted through the BMS channel
with density $\Ldens{x}$.  Let the received vector be $Y$.

The entropy of the single parity check code is $\entropy (X | Z=0,Y).$
By symmetry we have $\entropy (X | Z=0,Y) = \entropy (X | Z=1,Y) = \entropy (X | Z,Y).$
Now
$\entropy (X,Z | Y)  = \entropy (X | Y) + \entropy (Z | X,Y) =  \entropy (X | Y) =
\sum_{i=1}^{\dr} \entropy(\Ldens{x})$,
but we also have
$\entropy (X,Z | Y)  = \entropy (Z | Y) + \entropy (X | Z,Y)
 =\entropy(\Ldens{x}^{\cconv \dr}) + \entropy (X | Z,Y)$.
Thus, the entropy of the single parity check code is
\[
\entropy (X | Z,Y) = \dr \entropy(\Ldens{x}) - \entropy(\Ldens{x}^{\cconv \dr})\,.
\]
Now consider the channel that transmits a  bit once through the channel with
density $\Ldens{a}$ and again through a channel with density $\Ldens{b}.$ The entropy
of the combined channel is $\entropy(\Ldens{a} \vconv \Ldens{b}).$
This is equivalent to the single parity check code of two bits. Hence
\[
\entropy(\Ldens{a} \vconv \Ldens{b}) = \entropy(\Ldens{a}) + \entropy(\Ldens{b}) - \entropy(\Ldens{a} \cconv \Ldens{b}),
\]
which proves (the Duality Rule of) Lemma~\ref{lem:dualityrule}.
\end{IEEEproof}

\begin{lemma}[Entropy of Tree Code]\label{lem:entropyoftree}
Consider the $(\dl, \dr)$-regular computation tree of height $2$ (see e.g.,
Figure~\ref{fig:computationgraph}).
This tree represents a code of length $1+\dl (\dr-1)$ containing
$2^{1+\dl(\dr-2)}$ codewords.  Let $X$ be chosen uniformly at random
from the set of codewords and let $Y$ be the result of sending the components of $X$
through independent BMS channels. The root node goes through the BMS channel
$\Ldens{c}$ and all leaf nodes are passed through the BMS channel
$\Ldens{x}$. 
Then,
\begin{align}\label{eq:entropyoftreecodewithoutFP}
\entropy(X \mid Y)  = \entropy(\tilde{\Ldens{x}})& + \dl(\dr-1) \entropy(\Ldens{x})  -
\entropy(\tilde{\Ldens{x}}\cconv\Ldens{x}^{\cconv \dr-1})  \nonumber \\ 
& - (\dl-1) \entropy(\Ldens{x}^{\cconv \dr-1}),
\end{align}
where $\tilde{\Ldens{x}} = \Ldens{c}\vconv (\Ldens{x}^{\cconv \dr-1})^{\vconv \dl-1}$. 
\end{lemma}
\begin{IEEEproof} Using the chain rule, rewrite $\entropy(X \mid Y)$ as
\begin{align*}
\entropy(X \mid Y) & = \entropy(X_1 \mid Y) + \entropy(X_{\sim 1} \mid X_1, Y_{\sim 1}),
\end{align*}
where $X_1$ corresponds to the root variable node and $X_{\sim 1}$ is
the set of all the leaf nodes.  The first term is computed by density
evolution by considering all the independent messages flowing from the 
leaf nodes into the root node. Indeed, we convolve the channel density $\Ldens{c}$
with the densities coming from the $\dl$ check nodes, each of which has density
$\Ldens{y} = \Ldens{x}^{\cconv \dr-1}$. Thus we get
\begin{align*}
\entropy(X_1 \mid Y) & = \entropy(\Ldens{c} \vconv \Ldens{y}^{\vconv \dl}) \stackrel{\tilde{\Ldens{x}} = \Ldens{c}\vconv \Ldens{y}^{\vconv \dl-1}}{=}  \entropy(\tilde{\Ldens{x}} \vconv \Ldens{x}^{\cconv \dr-1}) \\ 
&\stackrel{\text{Lemma~\ref{lem:dualityrule}}}{=}  \entropy(\tilde{\Ldens{x}})+
\entropy(\Ldens{x}^{\cconv \dr-1})-
\entropy(\tilde{\Ldens{x}}\cconv\Ldens{x}^{\cconv \dr-1}).
\end{align*}
Further,
\begin{align*}
\entropy(X_{\sim 1} \mid X_1=0, Y_{\sim 1}) & = 
\entropy(X_{\sim 1} \mid X_1=1, Y_{\sim 1}) \\ 
& = \dl [(\dr-1) \entropy(\Ldens{x})-\entropy(\Ldens{x}^{\cconv \dr-1})].
\end{align*}
Indeed, when we condition on the root node to take either $0$ or $1$,  we
split the code into $\dl$  codes, each of which is a single parity-check code of length
$\dr-1$. Using the previous Lemma~\ref{lem:entropyofcheck}, we obtain the above
expressions. Combining the above statements proves the claim.\footnote{ 
For completeness, although the exact marginal does not factor into the computation,
note that there are $2^{1+\dl(\dr-2)}$ codewords in the code.
Out of those, $2^{\dl (\dr-2)}$ have a $0$ in the root node. So the
marginal of $X_1=0/1$ is one-half.}
\end{IEEEproof}
\begin{remark}
We stress that in Lemma~\ref{lem:entropyoftree}, $(\Ldens{c},
\Ldens{x})$ need not form a FP pair. Thus $\Ldens{x}$ will be
different from $\tilde{\Ldens{x}}$, in general.  
We will use the above expression when $\tilde{\Ldens{x}}$ and
$\Ldens{x}$ are ``close'' (in the Wasserstein sense), i.e., $(\Ldens{c},\Ldens{x})$ forms
 an approximate FP pair. This will
allow us to give an estimate of the entropy of the tree code.
\end{remark}

{\em Proof of Lemma~\ref{lem:areaunderBPGEXIT}:} Note first that
the integral $G(\dl, \dr, \{\Ldens{c}_{\ih}, \Ldens{x}_{\ih}
\}_{\ih}^{1})$ is well defined.  This is true since we assumed that
$\ent \geq \entLE$. This implies that we are integrating over a
continuous function (cf.  Corollary~\ref{cor:continuityofBPGEXIT}).
Hence the integral exists.  All that remains to be shown is that
the value of this integral is indeed $1-\frac{\dl}{\dr}-A$, as
claimed.

To evaluate the integral we consider the code corresponding to the
$(\dl, \dr)$-regular computation tree of height $2$ as in
Lemma~\ref{lem:entropyoftree}. Let $X$ be chosen uniformly at random
from the set of codewords and assume that the component corresponding
to the root node is sent through the channel $\Ldens{c}_{\ih}$,
whereas all components corresponding to the leaf nodes are sent
through the channel $\Ldens{x}_{\ih}$. Let $Y$ be the received word.
Since $\{\Ldens{c}_{\ih}, \Ldens{x}_{\ih}\}_{\ih }$ is, by assumption,
a FP family, the density flowing from any check node into the root
node is $\Ldens{y}_{\ih}=\Ldens{x}_{\ih}^{\cconv \dr-1}$ and so the
total density seen by the variable node (excluding the observation
of the variable node itself) is $\Ldens{y}_{\ih}^{\vconv \dl}$.
Therefore, the GEXIT integral associated to the root of this tree
code is the desired integral. We will evaluate this integral by
first determining the sum of {\em all} the GEXIT integrals associated
to this tree and then by subtracting from it the GEXIT integrals
associated to the leaf nodes.

In the sequel we will perform manipulations, such as writing a total
derivative as the sum of its partial derivatives or writing a
function as the integral of its derivative. In a first pass we will
assume that all these operations are well defined.  In a second
step we will then see how to justify these steps by approximating
the desired integrals by a series of simple integrals.

Label the variable nodes of the tree with the set $\{1, \dots, 1+\dl(\dr-1)
\}$ so that the root has label $1$.  Note that by assumption
$\entropy(\Ldens{c}_{\ih})=\ih$, so that the entropy of the first
component of $Y$, call it $\ih_1$, is $\ih$. The entropy of the
remaining components, call them $\ih_i$, $i \in \{2, \dots,
1+\dl(\dr-1)\}$, are all equal and take on the value
$\entropy(\Ldens{x}_\ih)$. So we imagine that all components are
parameterized by $\ih$.

From Definition~\ref{def:gexitintegralbasic} we have,
\begin{align*}
G(\dl, \dr, \{\Ldens{c}_{\ih}, \Ldens{x}_{\ih} \}_{\ih^*}^{1}) 
& = \int_{\ih^*}^1 \frac{\partial \entropy(X_1\mid Y(\underline{\ih}))}{\partial \ih_1} \frac{\partial \ih_1(\ih)}{\partial \ih} \dee \ih.
\end{align*}
Note that
\begin{align}\label{eq:totalderivativeequaltosumofpartialderivatives}
\int_{\ih^*}^1 &\dee \ih \frac{\dee}{\dee \ih} \entropy(X\mid Y(\underline{\ih})) =\!\!
\int_{\ih^*}^1 \frac{\partial
\entropy(X_1\mid Y(\underline{\ih})}{\partial \ih_1}\frac{\partial \ih_1(\ih)}{\partial \ih} \dee \ih+ \nonumber \\
& + \underbrace{\sum_{i=2}^{1+\dl(\dr-1)} \int_{\ih^*}^1 \frac{\partial
\entropy(X_i\mid Y(\underline{\ih}))}{\partial \ih_i} \frac{\partial \ih_i(\ih)}{\partial \ih} \dee \ih}_{\text{GEXIT of leaf nodes}}.
\end{align}
The lhs evaluates to 
\begin{align*}
& \int_{\ih^*}^1 \!\!\!\!\dee \ih \frac{\dee}{\dee \ih} \entropy(X\mid Y)  = 
\entropy(X\mid Y(1)) - \entropy(X\mid Y(\ent^*)) \\
& = \Big(1 + \dl(\dr-1) - \dl\Big) -  \\ 
&  \Big(\entropy(\Ldens{x})(1 + \dl(\dr-1)) - 
\entropy(\Ldens{x}^{\cconv \dr}) -(\dl-1)\entropy(\Ldens{x}^{\cconv \dr-1})\Big).
\end{align*}
The last inequality is obtained by using Lemma~\ref{lem:entropyoftree}
for the two endpoints and recalling that we set $\Ldens{x} =
\Ldens{x}_{\ih^*}$.

Let us consider the leaf node contributions. By symmetry these
contributions are all identical. If we focus on a single check node,
then again due to symmetry, the GEXIT integrals of all leaf nodes
is the same. But the sum of all the GEXIT integrals is equal to the
change in entropy of a single-parity check code of length $\dr$.
Thus, using Lemma~\ref{lem:entropyofcheck}, we see that the integral
of any single GEXIT integral is equal to
\begin{align}\label{eq:gexitofasingleleafnode}
\frac1{\dr}\Big((\dr - 1) - (\dr\entropy(\Ldens{x}) - \entropy(\Ldens{x}^{\cconv \dr}))\Big).
\end{align}
Combining all these statements, we get
\begin{align*}
& G(\dl, \dr, \{\Ldens{c}_{\ih}, \Ldens{x}_{\ih} \}_{\ih^*}^{1}) 
= \Big(1 + \dl(\dr-1) - \dl\Big) -  \\ 
&  \Big(\entropy(\Ldens{x})(1 + \dl(\dr-1)) - \entropy(\Ldens{x}^{\cconv
\dr}) -(\dl-1)\entropy(\Ldens{x}^{\cconv \dr-1})\Big) \\
& - \frac{\dl(\dr-1)}{\dr}\Big((\dr - 1) - (\dr\entropy(\Ldens{x}) -
\entropy(\Ldens{x}^{\cconv \dr}))\Big) \\
& = 1 - \frac{\dl}{\dr} - A.
\end{align*}

It remains to justify the previous derivation. We proceed as follows.
Instead of working with $\{\Ldens{c}_\ent, \Ldens{x}_\ent\}$, we
will work with a simpler family which is piece-wise linear and
``close'' to the original family. Because it is piece-wise linear,
the operations are simple to justify. Because it is ``close'' to
the original family, the result is ``close'' to what we want to
show. By taking a sequence of such families which approximate the
original family closer and closer, we obtain the desired result.

Let us start by constructing a piece-wise linear family, call it
$\{\tilde{\Ldens{c}}_\ent, \tilde{\Ldens{x}}_\ent\}$, which
approximates the original family $\{\Ldens{c}_\ent, \Ldens{x}_\ent\}$.
Consider the channel family $\{\Ldens{c}_\ent\}$ and sample it
uniformly in $\ent$ with a spacing of $\Delta \ent$.  To be precise,
pick the samples (from the original family) at $i \Delta \ent$, for
an appropriate range of integers $i$.  By a suitable choice we can
ensure that $\ent^*=i \Delta \ent$ for some $i \in \naturals$.  In
general, $\ent=1$ will not be of the form $i \Delta \ent$. This
means that the last sample is not lying  on the lattice. But we can
ensure that also for the last sample the ``gap'' (in entropy) is
at most $\Delta \ent$.  This is all that is needed for the proof.
Hence, for notational convenience we will ignore this issue and
assume that all samples have the form $i \Delta \ent$.

Construct from this set of samples a family by constructing a
piece-wise linear interpolation, call the result
$\{\tilde{\Ldens{c}}_\ent\}$.  Note that since the entropy functional
is linear, this construction leads to a family so that
$\entropy(\tilde{\Ldens{c}}_\ent)=\ent$.  Further,
$\{\tilde{\Ldens{c}}_\ent\}$ is ordered and piece-wise smooth. We
claim that \begin{align*} d(\Ldens{c}_\ent, \tilde{\Ldens{c}}_\ent)
& = d(\Ldens{c}_\ent, \alpha \Ldens{c}_{i \Delta \ent} + \bar{\alpha}
\Ldens{c}_{(i+1) \Delta \ent}) \leq 2 \sqrt{\ln(2) \Delta \ent},
\end{align*} where $i = \lfloor \frac{\ent}{\Delta \ent} \rfloor$
and $\alpha \in [0, 1]$ is a suitable interpolation factor.  In the
last step we have made use of (\ref{lem:blmetricconvexity}) in
Lemma~\ref{lem:blmetric}, the convexity property of the Wasserstein
distance, and the fact that consecutive samples have an entropy
difference of (at most) $\Delta \ent$.  Further, since they are
ordered, i.e., $\Ldens{c}_{i \Delta \ent} \prec \Ldens{c}_\ent \prec \Ldens{c}_{(i+1) \Delta \ent}$, an entropy difference of at most $\Delta \ent$ implies a
Wasserstein distance of at most $2 \sqrt{\ln(2) \Delta \ent}$ (cf.
(ii) of Lemma~\ref{lem:degradationandwasserstein}).

To each $\Ldens{c}_{\ent=i \Delta \ent}$ corresponds a FP $\Ldens{x}_\ent$,
call it $\Ldens{x}_{i \Delta \ent}$.  Take the collection $\{\Ldens{x}_{i
\Delta \ent}\}$.  Since this collection is ordered we can construct
from it an ordered and piece-wise smooth family via a linear
interpolation of consecutive samples in the same manner as we have done
this for the channel family.  We have
\begin{align*}
d(\Ldens{x}_{(i+1) \Delta \ent}, \Ldens{x}_{i \Delta \ent}) 
& \stackrel{\text{(i)}}{\leq} 2 \sqrt{\batta(\Ldens{x}_{(i+1) \Delta \ent})-\batta(\Ldens{x}_{i \Delta \ent})} \\
& \stackrel{\text{(ii)}}{\leq} 2 \sqrt{\frac{1}{\delta} (\batta(\Ldens{c}_{(i+1) \Delta \ent})-\batta(\Ldens{c}_{i \Delta \ent})} \\
& \stackrel{\text{(iii)}}{\leq} \sqrt{\frac{8}{\delta}} (\Delta \ent)^{\frac14}.
\end{align*}
Step (i) follows from Lemma~\ref{lem:degradationandwasserstein},
property (\ref{lem:blmetricentropy}).  In step (ii) we made use of
the fact that $\ent^* > \tilde{\ent}(\dl, \dr, \{\Ldens{c}_\ent\})$, so that according
to Lemma~\ref{lem:condforcontinuity}, $\delta \geq
1-\batta(\Ldens{c}_{\ent^*}) (\dl-1)(\dr-1)
(1-\batta(\Ldens{x}_{\ent^*})^2)^{\dr-2}>0$. In step (iii) we used once more
Lemma~\ref{lem:degradationandwasserstein}, property
(\ref{lem:blmetricentropy}).
Now consider the distance $d(\Ldens{x}_\ent, \tilde{\Ldens{x}}_\ent)$. We have
\begin{align*}
d(\Ldens{x}_\ent, \tilde{\Ldens{x}}_\ent) & \leq \alpha d(\Ldens{x}_\ent, \Ldens{x}_{i\Delta\ent}) + \bar{\alpha} d(\Ldens{x}_\ent, \Ldens{x}_{(i+1)\Delta\ent}) 
 \leq \sqrt{\frac{8}{\delta}} (\Delta \ent)^{\frac14}.
\end{align*}
The last inequality above follows from considering the same steps as before, since the densities
are ordered and each of them are FPs at channels with entropy difference at most $\Delta \ent$. 
Recall that $\{\Ldens{c}_\ent, \Ldens{x}_\ent\}$ is a FP family, hence we can write
\begin{align*}
& d(\tilde{\Ldens{x}}_\ent, \tilde{\Ldens{c}}_\ent \vconv ((\tilde{\Ldens{x}}_\ent)^{\cconv \dr-1})^{\vconv \dl-1}) \\
\leq & d(\tilde{\Ldens{x}}_\ent, \Ldens{x}_\ent)\!+\!
d(\Ldens{c}_\ent \!\vconv\! ((\Ldens{x}_\ent)^{\cconv \dr-1})^{\vconv \dl-1},
\tilde{\Ldens{c}}_\ent \!\vconv\! ((\tilde{\Ldens{x}}_\ent)^{\cconv \dr-1})^{\vconv \dl-1}) \\
\leq & d(\tilde{\Ldens{x}}_\ent, \Ldens{x}_\ent)\!+\!2d(\tilde{\Ldens{c}}_\ent, \Ldens{c}_\ent)\!+\!2d((\Ldens{x}_\ent^{\cconv \dr-1})^{\vconv \dl-1},(\tilde{\Ldens{x}}_\ent^{\cconv \dr-1})^{\vconv \dl-1}) \\
\leq &  4 \sqrt{\ln(2) \Delta \ent} + (4 (\dl-1) (\dr-1)+1)  \sqrt{\frac{8}{\delta}} (\Delta \ent)^{\frac14}.
\end{align*}
In words, $\{\tilde{\Ldens{c}}_\ent, \tilde{\Ldens{x}}_\ent\}_{\ent \geq \ent^*}$ forms
an approximate FP family. Above, we have used properties (v) and (vi) of Lemma~\ref{lem:blmetric}. 

Let us now apply the family $\{\tilde{\Ldens{c}}_\ent,
\tilde{\Ldens{x}}_\ent\}_{\ent \geq \ent^*}$ to the depth-2 tree. More
precisely, we consider the depth-2 tree code where the root node is
passed through the channel $\tilde{\Ldens{c}}_{\ent}$ and the leaves are passed
through the channel $\tilde{\Ldens{x}}_{\ent}$.  We claim that all GEXIT
integrals are well defined and that their sum is indeed the difference of the
entropies. Let us prove this claim in steps. 

The root integral has the form
\begin{align*}
\sum_{i} \int_{i \Delta \ent}^{(i+1) \Delta \ent}
\entropy((\Ldens{c}_{(i+1) \Delta \ent}-\Ldens{c}_{i \Delta \ent}) \vconv \Ldens{z}_\ent) \frac{\dee \ent}{\Delta \ent},
\end{align*}
where $\tilde{\Ldens{x}}_\ent = (\frac{\ent}{\Delta \ent} -\lfloor
\frac{\ent}{\Delta \ent} \rfloor) \Ldens{x}_{\lceil \frac{\ent}{\Delta
\ent} \rceil \Delta \ent} +(\lceil \frac{\ent}{\Delta \ent} \rceil
-\frac{\ent}{\Delta \ent}) \Ldens{x}_{\lfloor \frac{\ent}{\Delta
\ent} \rfloor \Delta \ent}$ and $\Ldens{z}_\ent =
((\tilde{\Ldens{x}}_\ent)^{\cconv \dr-1})^{\vconv \dl}$.  If we
expand out $\Ldens{z}_\ent$ explicitly then we see that the 
segment from $i$ to $(i+1)$  has the form $\sum_{\alpha} (\frac{\ent}{\Delta
\ent} -\lfloor \frac{\ent}{\Delta \ent} \rfloor)^{j_\alpha} (\lceil
\frac{\ent}{\Delta \ent} \rceil -\frac{\ent}{\Delta \ent})^{k_\alpha}
\Ldens{b}_{i, \alpha}$ for some fixed densities $\Ldens{b}_{i, \alpha}$ which
are various convolutions of two consecutive densities $\Ldens{x}_{i
\Delta \ent}$ and $\Ldens{x}_{(i+1) \Delta \ent}$ and some
strictly positive integers $j_\alpha$ and $k_\alpha$.  Set
$\sigma=(\frac{\ent}{\Delta \ent} -\lfloor \frac{\ent}{\Delta \ent}
\rfloor)$, so that $\sigma$ goes from $0$ to $1$ in each segment.
Then in each segment the integral has the form 
\begin{align*}
& \int_{0}^{1}  
\entropy\Bigr( (\Ldens{c}_{(i+1) \Delta \ent}-\Ldens{c}_{i \Delta \ent})
\vconv \sum_{\alpha} \sigma^{j_\alpha} (1-\sigma)^{k_\alpha} \Ldens{b}_{i, \alpha} \Bigl) \dee \sigma  \\
= & \sum_{\alpha} \frac{j_\alpha! k_\alpha!}{(j_\alpha+k_\alpha+1)!} 
\entropy( (\Ldens{c}_{(i+1) \Delta \ent}-\Ldens{c}_{i \Delta \ent}) \vconv \Ldens{b}_{i, \alpha}).
\end{align*}
So the root integral is in fact well defined. The same argument can be repeated for
the leaf integrals to show that they are also well defined.

If we consider one segment and add all the contributions (which as
we saw can be written down explicitly) we can verify that the sum
of all the GEXIT integrals is indeed equal to the difference of the
entropy of the tree. This calculation is in principle straightforward
but somewhat tedious, so we skip the details.

If $\{\tilde{\Ldens{c}}_\ent, \tilde{\Ldens{x}}_\ent\}$ were a true
FP family then the GEXIT integral of the root node would be equal
to $1-\frac{\dl}{\dr} - A$. This follows by the same steps which
we used in our initial casual derivation: once we know that all
integrals exist and add up to the total change in the entropy of
the tree code, all that is needed to draw this conclusion is to
observe that for a true FP family we can use a symmetry argument
to compute the value of each leaf GEXIT integral.

However $\{\tilde{\Ldens{c}}_\ent, \tilde{\Ldens{x}}_\ent\}$ is only 
an approximate (in the Wasserstein distance) FP family.
But we know that by making $\Delta \ent$ sufficiently small, we can
make the approximation arbitrarily good.  It is intuitive that by
taking a sequence of such approximations which converges to a true
FP family the limiting value of the GEXIT integral of the root node
should again be $1-\frac{\dl}{\dr} - A$. Let us show this more 
precisely.

We have already established that the sum of the individual GEXIT
integrals is equal to the total change of the entropy of the tree
code. This change only depends on the endpoints but not on the
chosen path. In particular, the endpoints for $\{\tilde{\Ldens{c}}_\ent,
\tilde{\Ldens{x}}_\ent\}_{\ent =\ent^*}^1$ and $\{\Ldens{c}_\ent,
\Ldens{x}_\ent\}_{\ent=\ent^*}^1$  are the same.

All is left is therefore to prove that each leaf GEXIT integral has
a value which approaches (\ref{eq:gexitofasingleleafnode}) when
$\Delta \ent$ approaches $0$. We know that this would be true if
all the messages entering check nodes were $\tilde{\Ldens{x}}_\ent$
and so the GEXIT integral was $\int_{\ent^*}^{1} \entropy(\frac{\dee
\tilde{\Ldens{x}}_\ent}{\dee \ent} \vconv \tilde{\Ldens{x}}_\ent^{\cconv
\dr-1}) \dee  \ent$.  But the actual GEXIT integral is $\int_{\ent^*}^{1}
\entropy(\frac{\dee \tilde{\Ldens{x}}_\ent}{\dee \ent} \vconv
\Ldens{z}_\ent) \dee  \ent$,
where $\Ldens{z}_\ent$ is the density flowing from the ``interior'' of the tree into a leaf node. 
Let us now show that
\begin{align*} 
\int_{\ent^*}^{1} 
\Bigl(\entropy(\frac{\dee \tilde{\Ldens{x}}_\ent}{\dee \ent} \vconv \Ldens{z}_\ent) -
\entropy(\frac{\dee \tilde{\Ldens{x}}_\ent}{\dee \ent} \vconv \tilde{\Ldens{x}}_\ent^{\cconv \dr-1}) \Bigr) \dee  \ent
 \stackrel{\Delta \ent \rightarrow 0}{\rightarrow} 0.
\end{align*} 
In fact, let us show that
\begin{align*} 
\int_{\ent^*}^{1} 
|\entropy(\frac{\dee \tilde{\Ldens{x}}_\ent}{\dee \ent} \vconv \Ldens{z}_\ent) -
\entropy(\frac{\dee \tilde{\Ldens{x}}_\ent}{\dee \ent} \vconv \tilde{\Ldens{x}}_\ent^{\cconv \dr-1})| \dee  \ent
 \stackrel{\Delta \ent \rightarrow 0}{\rightarrow} 0.
\end{align*} 
Note that for any $\ih\in [\ih^*,1]$ we have
\begin{align*}
& d(\tilde{\Ldens{x}}_\ent^{\cconv \dr-1}, \Ldens{z}_\ent) = d(\tilde{\Ldens{x}}_\ent^{\cconv \dr-1}, \tilde{\Ldens{x}}_\ent^{\cconv \dr-2}\cconv \tilde{\Ldens{c}}_\ent \vconv (\tilde{\Ldens{x}}_\ent^{\cconv \dr-1})^{\vconv \dl-1}) \\
& \stackrel{\text{(\ref{lem:blmetricregularcconv})},\text{Lemma}~\ref{lem:blmetric}}{\leq}  d(\tilde{\Ldens{x}}_\ent, \tilde{\Ldens{c}}_\ent \vconv (\tilde{\Ldens{x}}_\ent^{\cconv \dr-1})^{\vconv \dl-1})
 \stackrel{\Delta \ent \rightarrow 0}{\rightarrow} 0.
\end{align*}
Using the same line of reasoning as in in the proof of
Corollary~\ref{cor:continuityofBPGEXIT}, we see that therefore for each $\ent$,
$\lim_{\Delta \ent \rightarrow 0} |\entropy(\frac{\dee \tilde{\Ldens{x}}_\ent}{\dee \ent} \vconv
\Ldens{z}_\ent) -\entropy(\frac{\dee \tilde{\Ldens{x}}_\ent}{\dee
\ent} \vconv \tilde{\Ldens{x}}_\ent^{\cconv \dr-1})  |=0$.
Since the integrand is also bounded, it follows by Lebesgue's dominated convergence theorem
that also the integral of this quantity over $\ent$
converges to $0$ when $\Delta \ent$ is taken to $0$.

The only thing which remains to be done is to prove that the GEXIT
integral of the root node when we use the linearized family converges
to the true GEXIT integral when we let $\Delta  \ent$ tend to $0$.
We will do this in several steps by considering the chain of integrals
\begin{itemize}
\item[(i)] $G(\dl, \dr, \{\Ldens{c}_{\ih}, \Ldens{x}_{\ih} \}_{\ih^*}^{1})$,
\item[(ii)] $G(\dl, \dr, \{\Ldens{c}_{\ih}, \hat{\Ldens{x}}_{\ih} \}_{\ih^*}^{1})$,
\item[(iii)] $G(\dl, \dr, \{\tilde{\Ldens{c}}_{\ih}, \hat{\Ldens{x}}_{\ih} \}_{\ih^*}^{1})$,
\item[(iv)] $G(\dl, \dr, \{\tilde{\Ldens{c}}_{\ih}, \tilde{\Ldens{x}}_{\ih} \}_{\ih^*}^{1})$,
\end{itemize}
and by showing that the value of consecutive such integrals is
arbitrarily close.  Here, $\{\hat{\Ldens{x}}_{\ih}\}$ is a family
which is piece-wise constant on each segment, taking on the value
of its left boundary.

First note that the integral in (i) is well defined, being the
integral over a continuous function.  That the integrals in (i) and
(ii) are close follows by the same line of arguments as we just
used above.  The same idea applies to prove that the integrals (iii)
and (iv) are close to each other.  Finally, the value of (ii) and
(iii) is in fact equal. This is true since $\{\hat{\Ldens{x}})_\ent\}$
is in fact constant on each segment and $\{\Ldens{c}_\ent\}$ agrees
with $\{\tilde{\Ldens{c}}_\ent\}$ at the endpoints of the segments.

\qed

\section{Negativity -- Lemma~\ref{lem:asymptoticnegativity}}\label{app:asymptoticnegativity}

We prove Lemma~\ref{lem:asymptoticnegativity} by showing the following slightly
stronger statement.
\begin{lemma}
Let $\Ldens{x}$ be an $L$-density and consider a degree-distribution
$(\dl, \dr)$ such that $\dr \geq 1+5 (\frac1{1-r})^{\frac43}$.  Define 
$I_1= [(\frac34)^{\frac{\dl-1}{2}}+\frac1{2(\dr-1)^3}, \frac{1}{2e} \frac{\dl}{\dr}]$, and
$I_2=[\frac1{2 e} \frac{\dl}{\dr}, \frac{\dl}{\dr} - 
\dl e^{-4 (\dr-1) (\frac{2 \dl}{11 e \dr})^{\frac43}}-\kappa]$,
where $\kappa>0$. 
\begin{enumerate}
\item
Assume that $\Ldens{x}$ is a $\delta$-approximate FP, i.e., $d(\Ldens{x},
\Ldens{c} \vconv (\Ldens{x}^{\cconv \dr-1})^{\vconv \dl-1}) \leq
\delta$, for some channel $\Ldens{c}$ and $\delta \leq (\frac{\ln(2) \dl}{16 \sqrt{2}\dr})^2$. Then
if $\entropy(\Ldens{x}) \in I_1$, $A \leq -\frac1{16 e} \frac{\dl}{\dr}$.
\item 
For $\entropy(\Ldens{x}) \in I_2$, 
$A \leq -\kappa$. 
\end{enumerate}
\end{lemma}
\begin{IEEEproof}
Set $\Ldens{y}=\Ldens{x}^{\cconv \dr-1}$.
Let us first characterize the area $A$ in a more convenient form. 
We have
\begin{align*}
A & = \entropy(\Ldens{x}) + (\dl - 1 - \dl/\dr) \entropy(\Ldens{x}^{\cconv \dr}) - 
(\dl - 1) \entropy(\Ldens{y}) \\
 & = \entropy(\Ldens{x})\!-\!\frac{\dl}{\dr} \entropy(\Ldens{y}) \!+\! 
 (\dl \!-\! 1 \!-\! \dl/\dr) (\entropy(\Ldens{x}^{\cconv \dr}) \!-\! \entropy(\Ldens{y})). 
\end{align*}
For the $L$-distributions $\Ldens{x}$ and $\Ldens{y}$ let $\absDdens{x}$ and $\absDdens{y}$ be
the associated $|D|$ distributions.
Following the lead of L. Boczkowski \cite{Boc11} we write
\begin{align}\label{eq:boczkowski}
\entropy(\Ldens{x}) 
& = \int_{0}^{1} \absDdens{x}(z) h_2(\frac{1-z}{2}) \dee z 
 \stackrel{\text{(a)}}{=} 
1- \int_{0}^{1} \absDdens{x}(z) \sum_{n \geq 1} \alpha_n z^{2 n} \dee z \nonumber \\
& = 1\!-\! \sum_{n \geq 1} \alpha_n \underbrace{\int_{0}^{1} \absDdens{x}(z) z^{2 n}}_{\moment_{\Ldens{x}, n}} \dee z 
\stackrel{\text{(b)}}{=} 1\!-\! \sum_{n \geq 1} \alpha_n \moment_{\Ldens{x}, n}.
\end{align}
In step (a) we have used the expansion of Lemma~\ref{lem:boundsonent},
where $\alpha_n=\frac{1}{2 \ln(2) n (2n-1)}$, $n \geq 0$. Note that
$\alpha_n \geq 0$ and that $\sum_{n \geq 1} \alpha_n=1$.  
Most importantly, as mentioned in the proof of Lemma~\ref{lem:bbound},
the moments $\moment_{\Ldens{x}, n}$ are multiplicative under $\cconv$. This implies that for
$d \geq 1$, $\entropy(\Ldens{x}^{\cconv d})  = 1- \sum_{n \geq
1} \alpha_n \moment_{\Ldens{x}, n}^d$.  E.g., for two distributions $\Ldens{x}$ and
$\Ldens{y}$ we have
\begin{align*}
& 1\!-\!\entropy(\Ldens{x} \cconv \Ldens{y}) =  1- \int \int \absDdens{x}(z_1) \absDdens{y}(z_2) 
h_2(\frac{1\!-\!z_1 z_2}{2}) \dee z_1 \dee z_2\\
& = \int \!\!\!\!\!\int  \absDdens{x}(z_1) \absDdens{y}(z_2) 
\sum_{n \geq 1} \alpha_n z_1^{2n} z_2^{2n} \dee z_1 \dee z_2 
 = \sum_{n \geq 1} \alpha_n \moment_{\Ldens{x}, n} \moment_{\Ldens{y}, n},
\end{align*}
where in the first equality we use that in the $|D|$-domain the check node
operation is simply a multiplication. 

Assume at first that $\entropy(\Ldens{x}) \in [(\frac34)^{\frac{\dl-1}{2}},
\frac{1}{2e} \frac{\dl}{\dr}+\frac{1}{2(\dr-1)^3}]$ and that
$\Ldens{x}=\Ldens{c} \vconv \Ldens{y}^{\vconv \dl-1}$ for some
channel $\Ldens{c}$.  Define $\psi(x)=(1-x)x^{\dr-1}$.  Then
\begin{align*}
A 
& = \entropy(\Ldens{x}) -\frac{\dl}{\dr} \entropy(\Ldens{y}) + 
(\dl - 1 - \dl/\dr)\sum_{n} \alpha_n \psi(\moment_{\Ldens{x}, n}) \\ 
& \stackrel{\text{(a)}}{\leq} \entropy(\Ldens{x}) -\frac{\dl}{\dr} \entropy(\Ldens{y}) + 
\underbrace{(\dl - 1 - \dl/\dr) \frac{(1-\frac{1}{\dr})^{\dr}}{\dr-1}}_{B} \\
& \stackrel{\text{(b)}}{\leq} \entropy(\Ldens{x}) -\frac{\dl}{\dr} \entropy(\Ldens{x})^{\frac{2}{\dl-1}} + 
\frac{\dl - 1 - \dl/\dr}{\dr-1} (1-\frac{1}{\dr})^{\dr}\\
 & \stackrel{\text{(c)}}{\leq} \frac1{2 e} \frac{\dl}{\dr}+ \frac1{2(\dr-1)^3}- \frac{\dl}{\dr} \frac34 + \frac{\dl}{\dr}\frac1e  
\leq - \frac1{8 e} \frac{\dl}{\dr}.
\end{align*}
In (a) we used the bound $\psi(x) \leq
\frac{(1-\frac{1}{\dr})^{\dr}}{\dr-1}$ so that $\sum_{n} \alpha_n
\psi(\moment_{\Ldens{x}, n}) \leq \frac{(1-\frac{1}{\dr})^{\dr}}{\dr-1}$. 
Consider step (b).
Set $\entropy(\Ldens{y})=h_2(p) \stackrel{\text{Lemma}~\ref{lem:boundsonent}}{\geq} 4 p \bar{p}$. Then
\begin{align*}
\entropy(\Ldens{x}) & = \entropy(\Ldens{c} \vconv \Ldens{y}^{\vconv \dl-1})  
\leq \entropy(\Ldens{y}^{\vconv \dl-1}) 
\leq \entropy(\Ldens{a}_{\BSC(p)}^{\vconv \dl-1})  \\
& \stackrel{\text{Lem.~\ref{lem:entropyvsbatta}}}{\leq} 
\batta(\Ldens{a}_{\BSC(p)}^{\vconv \dl-1}) = (4 p \bar{p})^{\frac{\dl-1}{2}} \leq  \entropy(\Ldens{y})^{\frac{\dl-1}{2}}.
\end{align*}
In step (c) we substituted the upper and lower bounds on
$\entropy(\Ldens{x})$ for the first and second expression respectively.
Also, in the last inequality, we have 
$\frac1{2 (\dr-1)^3} \leq \frac{\dl}{\dr}(\frac34 - \frac{13}{8e})$ since we assumed that
$\dr \geq 1+5 (\dr/\dl)^{\frac43} \geq 1+(2\frac{\dl}{\dr}(\frac34 - \frac{13}{8e}))^{-\frac43}$.

Let us summarize. If $\Ldens{x} = \Ldens{c} \vconv \Ldens{y}^{\vconv \dl-1}$ and if $\entropy(\Ldens{x})
\in  [(\frac34)^{\frac{\dl-1}{2}}, \frac{1}{2e} \frac{\dl}{\dr}+\frac1{2(\dr-1)^3}]$
then $A \leq - \frac1{8 e} \frac{\dl}{\dr}$.  Let us drop the 
condition  $\Ldens{x} = \Ldens{c} \vconv \Ldens{y}^{\vconv \dl-1}$
and assume instead that $d(\Ldens{x}, \Ldens{c} \vconv \Ldens{y}^{\vconv \dl-1}) \leq \delta$. 
Define $\tilde{\Ldens{x}}=\Ldens{c} \vconv \Ldens{y}^{\vconv \dl-1}$. Then
\begin{align*}
A & \leq \entropy(\tilde{\Ldens{x}}) -\frac{\dl}{\dr} \entropy(\Ldens{y}) +B +(\entropy(\Ldens{x}) -\entropy(\tilde{\Ldens{x}})) \\
  & \leq \entropy(\tilde{\Ldens{x}}) \!-\!\frac{\dl}{\dr} \entropy(\Ldens{y}) \!+\! B \!+\! 3 \sqrt{\delta} 
   \leq -\frac{1}{8e} \frac{\dl}{\dr} \!+\! 3 \sqrt{\delta} \leq -\frac{1}{16 e} \frac{\dl}{\dr}.
\end{align*}
The one-before last step follows since if $\entropy(\Ldens{x}) \in I_1$ then
$\entropy(\tilde{\Ldens{x}}) \in  [(\frac34)^{\frac{\dl-1}{2}},
\frac{1}{2e} \frac{\dl}{\dr}+\frac1{2(\dr-1)^3}]$ and so we can apply the previous procedure.
Also in the above computations we have used property
(\ref{lem:blmetricwasserboundsbatta}) of Lemma~\ref{lem:blmetric} to bound $ |
\entropy(\Ldens{x}) -\entropy(\tilde{\Ldens{x}})| \leq 3\sqrt{\delta}$.

For $\entropy(\Ldens{x}) \in [\frac1{2 e} \frac{\dl}{\dr},
\frac{\dl}{\dr}- \dl e^{-4 (\dr-1) (\frac{\dl}{16 e \dr})^{2}}-\kappa]$,
\begin{align*}
A & = \entropy(\Ldens{x}) -\frac{\dl}{\dr} \entropy(\Ldens{y}) + 
(\dl - 1 - \frac{\dl}{\dr}) \sum_{n} \alpha_n \psi(\moment_{\Ldens{x}, n}) \\
 & \stackrel{\text{(a)}}{\leq} \entropy(\Ldens{x}) -\frac{\dl}{\dr} \entropy(\Ldens{y}) + 
(\dl - 1 - \frac{\dl}{\dr}) \sum_{n} \alpha_n \moment_{\Ldens{x}, n}^{\dr-1} \\
 & \stackrel{\text{(b)}}{\leq} \entropy(\Ldens{x}) -\frac{\dl}{\dr} \entropy(\Ldens{y}) + 
(\dl - 1 - \frac{\dl}{\dr}) \sum_{n} \alpha_n \moment_{\Ldens{x}, 1}^{\dr-1} \\
 & \stackrel{\text{(c)}}{\leq} \entropy(\Ldens{x})\! -\!\frac{\dl}{\dr} \entropy(\Ldens{y}) \!+\! 
(\dl\! -\! 1 \!- \!\frac{\dl}{\dr})\!\! \sum_{n} \!\!\alpha_n (1\!-\!2 h_2^{\!-\!1}\!(\entropy(\Ldens{x})))^{2 (\dr\!-\!1)} \\
 & \stackrel{\text{(d)}}{\leq} h_2(p)- \frac{\dl}{\dr} (1-e^{-4 (\dr-1) p}) + (\dl\!-\!\frac{\dl}{\dr}) (1-2 p)^{2 (\dr-1)} \\
 & \stackrel{\text{(e)}}{\leq} \frac{\dl}{\dr}\!-\!\dl e^{-4 (\dr\text{-}1) (\frac{2 \dl}{11 e \dr})^{\frac{4}{3}}}
 \!\!\!\!-\!\kappa -\! \frac{\dl}{\dr} \!+\! \dl e^{-4 (\dr-1) (\frac{2 \dl}{11 e \dr})^{\frac{4}{3}}} \\
& \leq -\kappa.
\end{align*}
In (a) we upper bound $\psi(x)=(1-x) x^{\dr-1}$ by $x^{\dr-1}$,
$x \in [0, 1]$, and note that $\moment_{\Ldens{x}, n} \in [0, 1]$.
In (b) we use $\moment_{\Ldens{x}, n} \leq \moment_{\Ldens{x}, 1}$
(this is true since $x^{2n}$ is decreasing for each fixed $x \in
[0, 1]$ as a function of $n$) and that $x^{\dr-1}$ is increasing.
Step (c) is a consequence of the bound
$\moment_{\Ldens{x}, 1} \leq (1-2h_2^{-1}(\entropy(\Ldens{x})))^2$.
Let us prove this inequality. Equivalently, we
want to show
$\entropy(\Ldens{x}) \leq h_2\Big(\frac{1 - \sqrt{\moment_{\Ldens{x}, 1}}}2\Big)$.
By Jensen
\begin{align*}
\moment_{\Ldens{x}, n} & = \int \absDdens{x}(z) z^{2 n} \dee z
\geq  \Bigl(\int \absDdens{x}(z) z^2 \dee z \Bigr)^n =\moment_{\Ldens{x}, 1}^n.
\end{align*}
Using the above we have,
\begin{align*}
1- \sum_{n \geq 1} \alpha_n \moment_{\Ldens{x}, n} & \leq 1 - \sum_{n \geq 1} \alpha_n \moment^n_{\Ldens{x}, 1} 
= h_2\Big(\frac{1 - \sqrt{\moment_{\Ldens{x}, 1}}}2\Big).
\end{align*}
The claim is proven by noticing that the lhs above is equal to $\entropy(\Ldens{x})$.

Step (d) uses the following lower bound on
$\entropy(\Ldens{y})=\entropy(\Ldens{x}^{\cconv \dr-1})$.  Set $\entropy(\Ldens{x})=h_2(p)$. From
extremes of information combining we know that we get the lowest
entropy if we assume that $\Ldens{x}$ is a BSC density. Therefore,
\begin{align*}
\entropy(\Ldens{y}) &  \geq  h_2\bigl( \frac{1-(1-2 p)^{\dr-1}}{2} \bigr) 
\stackrel{\text{(\ref{eq:lowerboundbinaryentropy})}}{\geq} 1\!-\!(1\!-\!2p)^{2(\dr-1)} \\
& = 1-e^{2(\dr\text{-}1) \ln(1\text{-}2 p)}  \geq 1\!-\!e^{-4 (\dr-1) p}.
\end{align*}
Consider finally step (e).  We know that $h_2(p) \in I_2$.  Combined
with (\ref{eq:upperboundbinaryentropytwo}) and $(1\!-\!2p)^{2(\dr-1)} \leq e^{-4 (\dr-1) p}$  we conclude that $p
\geq (\frac{2}{11 e} \frac{\dl}{\dr})^{\frac43}$.  \end{IEEEproof}

\section{Spacing of FPs --Lemma~\ref{lem:spacing} and Transition Length of FPs --
Lemma~\ref{lem:transitionlength}}\label{sec:spacingandtransitionlength}

If we are given a proper one-side FP (with any boundary
condition) then consecutive elements of the FP cannot be too
different from each other. This is made precise in the following
lemma.  
\begin{lemma}[Spacing of FP]\label{lem:spacing}
Let $(\Ldens{c}, \Ldens{\x})$ be a proper one-sided FP on $[-\Lfp, 0]$, $\Lfp \geq 0$ with any 
boundary condition. 
\begin{enumerate}[(i)]
\item
For $i \in [-\Lfp+1, 0]$ 
\begin{align*}
d(\Ldens{x}_i, \Ldens{x}_{i-1}) \leq  \frac{\dl-1}{w}, \;\; \batta(\Ldens{x}_i)-\batta(\Ldens{x}_{i-1}) \leq  \frac{\dl-1}{w}.
\end{align*} 
\item
Let $\Ldens{\xavg}_i$ denote the weighted average $\Ldens{\xavg}_i=\frac1{w^2}\sum_{j, k=0}^{w-1} \Ldens{x}_{i+j-k}$. 
Then, for any $i \in [-\infty, \infty]$, 
\begin{align*}
d(\Ldens{\xavg}_i, \Ldens{\xavg}_{i-1})  \leq \frac{1}{w}, \;\;
\batta(\Ldens{\xavg}_i)-\batta(\Ldens{\xavg}_{i-1}) \leq \frac{1}{w}.  
\end{align*} 
\end{enumerate}
\end{lemma}
{\em Discussion:}
Each of these two claims states that consecutive distributions are
``close'' either wrt the Wasserstein distance or the 
Battacharyya parameter. Further, the difference is either for
the distributions themselves or their averages.
\begin{proof}
\begin{enumerate}[(i)]
\item
To simplify notation, for $i \in [-\Lfp+1, 0]$ fixed, let
$\Ldens{f}_{j} = \bigl(\frac{1}{w} \sum_{k=0}^{w-1} 
\Ldens{x}_{i+j-k-1} \bigr)^{\cconv \dr-1}$. 
Writing the DE equations explicitly, 
\begin{align*}
\Ldens{x}_i & = \Ldens{c} \vconv 
\Bigl(\frac{1}{w} \sum_{j=1}^{w} \Ldens{f}_{j} \Bigr)^{\vconv \dl-1}\!\!\!\!, \;\;
\Ldens{x}_{i-1} = \Ldens{c} \vconv 
\Bigl(\frac{1}{w} \sum_{j=0}^{w-1} \Ldens{f}_{j}\Bigr)^{\vconv \dl-1}.
\end{align*}
Note that the expressions for $\Ldens{x}_i$ and $\Ldens{x}_{i-1}$ are
similar. The only difference is that $\Ldens{x}_i$ contains $\Ldens{f}_w$
whereas $\Ldens{x}_{i-1}$ contains $\Ldens{f}_0$. Rewrite both expressions in the form 
\begin{align*}
\Ldens{x}_i = \Ldens{c} \vconv \Bigl(\frac{1}{w} \sum_{j=1}^{w} \Ldens{a}_{j} \Bigr)^{\!\vconv \dl\!-\!1}\!\!\!, & \;\;\;
\Ldens{x}_{i-1} = \Ldens{c} \vconv \Bigl(\frac{1}{w} \sum_{j=1}^{w} \Ldens{b}_{j} \Bigr)^{\!\vconv \dl\!-\!1},
\end{align*}
where $\Ldens{a}_i=\Ldens{b}_i = \Ldens{f}_{i-1}$, $i=2, \dots, w$,
$\Ldens{a}_1=\Ldens{f}_w$, and $\Ldens{b}_1=\Ldens{f}_0$.
Now expand $\Ldens{x}_i$ as well as $\Ldens{x}_{i-1}$ in the form 
\begin{align*}
\Ldens{x}_i & =     \sum_{d_1, \dots, d_w: d_1\!+\!\dots\!+\!d_w=\dl-1} 
\frac{\binom{\dl\!-\!1}{d_1, \dots, d_w}}{w^{-(\dl\!-\!1)}}
	 \Ldens{a}_1^{\vconv d_1}\! \vconv\! \Ldens{c}_{d_2, \dots, d_w}, \\
\Ldens{x}_{i-1} & 
 =     \sum_{d_1, \dots, d_w: d_1\!+\!\dots\!+\!d_w=\dl-1} 
\frac{\binom{\dl\!-\!1}{d_1, \dots, d_w}}{w^{-(\dl\!-\!1)}}
\Ldens{b}_1^{\vconv d_1} \!\vconv\! \Ldens{c}_{d_2, \dots, d_w},
\end{align*}
where $\Ldens{c}_{d_2, \dots, d_w}=\Ldens{a}_2^{\vconv d_2} \vconv \dots
\vconv \Ldens{a}_w^{\vconv d_w} \vconv \Ldens{c}$.
Note that the terms in the expansions of $\Ldens{x}_i$ and $\Ldens{x}_{i-1}$ 
with $d_1=0$ are identical. Therefore, if we consider
$\batta(\Ldens{x}_i)-\batta(\Ldens{x}_{i-1})$, these terms cancel. 
We can upper bound the difference by the Battacharyya constant of all those terms
of the expansion of $\Ldens{x}_i$ which correspond to $d_1\geq1$, i.e.,
\begin{align*}
& \batta(\Ldens{x}_i)-\batta(\Ldens{x}_{i-1}) \\
& \leq w^{-(\dl-1)} \!\!\!\!\!\!\!\!\!\!\!\!\!\!\!\!\! 
\sum_{\substack{d_1 \geq 1, \dots, d_w,\\\text{s.t.}\, d_1+\dots+d_w=\dl-1}} \!\!\binom{\dl-1}{d_1, \dots, d_w} \batta(
\Ldens{a}_1^{\vconv d_1} \vconv \Ldens{c}_{d_2, \dots, d_w}
)\\
& \leq w^{-(\dl-1)} \!\!\!\!\!\!\!\!\!\!\!\!\!\!\!\!\! 
\sum_{\substack{d_1 \geq 1, \dots, d_w,\\\text{s.t.}\, d_1+\dots+d_w=\dl-1}} \!\!\binom{\dl-1}{d_1, \dots, d_w} \\
& = 1 - (1-\frac{1}{w})^{\dl-1} \leq \frac{\dl-1}{w}.
\end{align*}
If we are interested in the Wasserstein distance instead, we can
proceed in an almost identical fashion. The only difference is
that in the last sequence of inequalities we use the convexity
property (\ref{lem:blmetricconvexity}) and the boundedness property
(\ref{lem:blmetricboundedness}) of (the Wasserstein metric)
Lemma~\ref{lem:blmetric}. 

\item
Using the convexity property (\ref{lem:blmetricconvexity}) of (the
Wasserstein metric) Lemma~\ref{lem:blmetric} and canceling common terms,
we get
\begin{align*}
& d(\Ldens{\xavg}_i, \Ldens{\xavg}_{i-1})  = 
d\bigl(\frac1{w^2}\sum_{j, k=0}^{w-1} \Ldens{x}_{i+j-k}, \frac1{w^2}\sum_{j, k=0}^{w-1} \Ldens{x}_{i+j-k-1}\bigr) \\
& =\frac1{w^2} d\bigl(\sum_{j=0}^{w-1} \Ldens{x}_{i+j}, \sum_{j=0}^{w-1} \Ldens{x}_{i-1-j}\bigr)
\leq \frac{1}{w}. 
\end{align*}
The proof for the Battacharyya parameter proceeds
in an identical fashion and uses the linearity of the Battacharyya
parameter.  
\end{enumerate}
\end{proof}

\begin{lemma}[Basic Bounds on FP]\label{lem:avgprop}
Let $(\Ldens{c},\Ldens{\x})$ be a proper one-sided FP on $[-\Lfp, 0]$, $\Lfp \geq 0$ with any boundary condition.
Let $\batta_i=\batta(\Ldens{x}_i)$ denote the Battacharyya parameter of
the density of the $i$-th section.  
Then for all $i\in [-\Lfp,0]$,
\begin{align*}
\batta_i \leq
\batta(\Ldens{c}) (1-(1-\frac1{w^2}\sum_{j, k=0}^{w-1}\batta_{i+j-k})^{\dr-1})^{\dl-1}.
\end{align*}
\end{lemma}
\begin{IEEEproof}
For all $i \in [-\Lfp, 0]$
\begin{align*}
\Ldens{x}_i
& = \Ldens{c} \vconv
\Bigl(\frac{1}{w} \sum_{j=0}^{w-1} \bigl(\frac{1}{w} \sum_{k=0}^{w-1}
\Ldens{x}_{i+j-k} \bigr)^{\cconv \dr-1} \Bigr)^{\vconv\dl-1}.
\end{align*}
Since the Battacharyya parameter is multiplicative in $\vconv$
and linear,
\begin{align*}
\batta(\Ldens{x}_i)
& = \batta(\Ldens{c})
\Bigl(\frac{1}{w} \sum_{j=0}^{w-1} \batta\Big(\bigl(\frac{1}{w} \sum_{k=0}^{w-1}
\Ldens{x}_{i+j-k} \bigr)^{\cconv \dr-1}\Big) \Bigr)^{\dl-1}.
\end{align*}
Further, recall from Lemma~\ref{lem:extremes}, property
(\ref{lem:extremesmaxcconv}), and the ensuing discussion, that
$ \batta(\Ldens{a}^{\cconv \dr-1}) \leq 1-(1-\batta(\Ldens{a}))^{\dr-1}$, so that
$$
\batta\Big(\bigl(\frac{1}{w} \sum_{k=0}^{w-1}
\Ldens{x}_{i+j-k} \bigr)^{\cconv \dr-1}\Big) \leq 1 - \bigl(1-\frac{1}{w}
\sum_{k=0}^{w-1} \batta_{i+j-k} \bigr)^{\dr-1}.
$$
Combining, we get
\begin{align*}
\batta_i \leq \batta(\Ldens{c})\Big( 1 - \frac{1}{w} \sum_{j=0}^{w-1} \bigl(1-\frac{1}{w}
\sum_{k=0}^{w-1} \batta_{i+j-k} \bigr)^{\dr-1} \Big)^{\dl-1}.
\end{align*}
Let $f(x) = (1-x)^{\dr-1}$, $x \in [0, 1]$. Since $f''(x) =
(\dr-1)(\dr-2)(1-x)^{\dr-3}\geq 0$, $f(x)$ is convex.
Let $y_{j} = \frac1w\sum_{k=0}^{w-1}\batta_{i+j-k}$. 
Then by Jensen,
\begin{align*}
\frac{1}{w} \sum_{j=0}^{w-1} f(y_{j}) \geq f(\frac1w\sum_{j=0}^{w-1}y_j),
\end{align*}
which proves the claim.
\end{IEEEproof} 

\begin{lemma}[Basic Properties of $h(x)$, \cite{KRU10}]\label{lem:propertyofh(x)}
Consider the $(\dl, \dr)$-regular ensemble with $\dl \geq 3$ and
let $\epsilon \in (\epsilon^{\BPsmall}, 1]$, where
$\epsilon^{\BPsmall}(\dl, \dr)$ is the BP threshold the regular
ensemble when transmitting over the BEC.  Define $h(x) = \epsilon
(1-(1-x)^{\dr-1})^{\dl-1} - x$.

\begin{enumerate}[(i)]
\item For $\epsilon>\epsilon^{\BPsmall}$, $h(x)=0$ has exactly three solutions, one
of them being 0 and the other two denoted by $\xunstab(\epsilon)$ and $\xstab(\epsilon)$
with $0<\xunstab(\epsilon) < \xstab(\epsilon)$. 
Further, $h(x)\leq 0$ for all $x\in [0, \xunstab(\epsilon)]$ and $h(x) \geq 0$ for all
 $x \in [\xunstab(\epsilon), \xstab(\epsilon)]$.\label{lem:propertyofh(x)zero}

\item $h'(\xunstable) > 0$ and $h'(\xstable) < 0$; 
$|h'(x)| \leq \dl \dr$ for $x \in [0, 1]$. \label{lem:propertyofh(x)one}

\item
There exists a unique value $0\leq x_*(\epsilon) \leq \xunstable$ so
that $h'(x_*(\epsilon)) = 0$, and there exists a unique value $\xunstable
\leq x^*(\epsilon) \leq \xstable$ so that $h'(x^*(\epsilon))=0$. 
Further, $h(x)$ is decreasing in $[0, x_*(\epsilon)]$. 
\label{lem:propertyofh(x)two}

\item
Let 
$\kappa_*(\epsilon) = \min\{ -h'(0), \frac{-h(x_*(\epsilon))}{x_*(\epsilon)} \}$.
The quantity $\kappa_*(\epsilon)$ is non-negative and depends only
on the channel parameter $\epsilon$ and the degrees $(\dl, \dr)$. \label{lem:propertyofh(x)three}
\item
For $0 \leq \epsilon \leq 1$,
$x_*(\epsilon) > \frac{1}{\dl^2 \dr^2}$.
\label{lem:propertyofh(x)four}
\item
For $0 \leq \epsilon \leq 1$,
$\kappa_*(\epsilon) \geq \frac1{8\dr^2}$.
\label{lem:propertyofh(x)five}

\item
Let $\kappa_*$ and $x_*$ denote the universal lower bounds, given in the previous part, on $\kappa_*(\epsilon)$ and
$x_*(\epsilon)$, respectively. 
If we draw a line from $0$ with slope $-\kappa_*$, then $h(x)$
lies below this line for $x \in [0 , x_*]$. \label{lem:propertyofh(x)six}

\item For $\epsilon \in (\epsilon^{\BPsmall}, 1]$ we have 
\begin{align}\label{equ:boundonxuone}
\xunstab(\epsilon) \geq \xunstab(1) \geq (\dr-1)^{-\frac{\dl-1}{\dl-2}}.
\end{align} \label{lem:propertyofh(x)seven}

\end{enumerate}
\end{lemma}
\begin{remark}
The function $h(x)$ is the DE equation for the $(\dl, \dr)$-regular ensemble when
transmitting over the BEC. The two non-zero solutions, $\xunstab(\epsilon)$ and $\xstab(\epsilon)$
 represent the unstable and the stable FPs of DE \cite{RiU08}. In the following, we will be using extremes of
information combining techniques to relate the Battacharyya parameters via $h(x)$. 
\end{remark}
 
In Figure~\ref{fig:one-sided_fixed_point} we see that within a few sections
the constellation changes from reliable sections (towards the boundary)
to sections which all have more or less the same reliability. In other words,
this transition happens quickly. This is made precise in the following lemma.
\begin{lemma}[Transition Length] \label{lem:transitionlength}
Let $\epsilon^{\BPsmall}$ be the BP threshold for transmission over the BEC using the $(\dl, \dr)$-regular (uncoupled) ensemble.
For $\epsilon \in (\epsilon^{\BPsmall}, 1]$, let $\xunstab(\epsilon)$ be
the smaller of the two strictly positive roots of the equation $h(x)=0$,
where $h(x) = \epsilon (1 - (1-x)^{\dr-1})^{\dl-1} - x$.  For $0 \leq
\epsilon \leq \epsilon^{\BPsmall}$, {\em define} $\xunstab(\epsilon) =
\lim_{\delta \downarrow \epsilon^{\BPsmall}} \xunstab(\delta)$.  

Consider transmission over a BMS channel $\Ldens{c}$. Let $w$ be admissible 
in the sense of property (\ref{equ:admissiblethree}) of Definition~\ref{def:admissible}.
Let $(\Ldens{c}, \Ldens{\x})$ be a proper one-sided
FP on $[-\Lfp, 0]$ with any boundary condition. Let
$\batta_i=\batta(\Ldens{x}_i)$ denote the Battacharyya parameter
of the density associated to the $i$-th section and define
$\epsilon=\batta(\Ldens{c})$.

Then, there exists a positive constant $c(\dl, \dr)$ which depends on
$\dl$ and $\dr$, but not on $\Lfp$ or the channel $\Ldens{c}$,
so that for any $\delta >0$
\begin{align*}
\Big\vert \{ i: \delta < \batta_i < \xunstab(\epsilon) \}\Big\vert \leq w 
\frac{c(\dl, \dr)}{\delta}.
\end{align*}
\end{lemma}
\begin{proof}
Throughout the proof we set $\epsilon=\batta(\Ldens{c})$ and we
write $\batta_i$ for $\batta(\Ldens{x}_i)$.

Note first that we have to prove the statement only for $\epsilon
\in (\epsilon^{\BPsmall}, 1]$. This is true since we have defined
$\xunstab(\epsilon)$ to coincide with $\xunstab(\epsilon^{\BPsmall})$ for
$\epsilon \in [0, \epsilon^{\BPsmall}]$ and since further the function
$h$, which we use to bound the process, is strictly decreasing as a
function of $\epsilon$. Hence, in the sequel our language will reflect
the fact that we have $\epsilon \in (\epsilon^{\BPsmall}, 1]$.

(i) {\em The number of sections such that $\batta_i \in [\delta,
x_*(\epsilon)]$ is at most $w(\frac1{\kappa_*\delta}+1)$}. If $\delta>x_*(\epsilon)$
then the number of sections in this part is 0. Hence wlog assume $\delta < x_*(\epsilon)$.  
Let $i$ be the smallest index so that $\batta_i \geq \delta$. If
$\batta_{i+(w-1)} \geq x_*(\epsilon)$ then the claim is trivially
fulfilled. Assume therefore that $\batta_{i+(w-1)} \leq x_*(\epsilon)$.
From the monotonicity of $g(\cdot)$ and the fact that $\Ldens{\x}$ is increasing,
\begin{align*}
\Ldens{x}_i & = \Ldens{c} \vconv g(\Ldens{x}_{i-(w-1)}, \dots, \Ldens{x}_i, \dots, \Ldens{x}_{i+(w-1)}) \\
& \prec \Ldens{c}\vconv g(\Ldens{x}_{i+(w-1)}, \dots, \Ldens{x}_{i+(w-1)}).
\end{align*}
This implies
\begin{align*}
\batta_i \stackrel{\text{extremes of info. comb.}}{\leq} \epsilon 
g(\batta_{i+(w-1)}, \dots, \batta_{i+(w-1)}).
\end{align*}
As a consequence we get
\begin{align*}
& \batta_{i+(w-1)} - \batta_i  \geq \batta_{i+(w-1)} - \epsilon g(\batta_{i+(w-1)},\dots, \batta_{i+(w-1)}) \\
& = -h(\batta_{i+(w-1)}) 
 \stackrel{\text{Lemma~\ref{lem:propertyofh(x)} (\ref{lem:propertyofh(x)two})}}{\geq} -h(\delta)  
 \stackrel{\text{Lemma~\ref{lem:propertyofh(x)} (\ref{lem:propertyofh(x)six})}}{\geq} 
\kappa_*(\epsilon) \delta.
\end{align*}
This is equivalent to
$\batta_{i+(w-1)} \geq \batta_i + \kappa_*(\epsilon) \delta$.
More generally, using the same line of reasoning,
$\batta_{i+l(w-1)} \geq \batta_i + l \kappa_*(\epsilon) \delta$,
as long as $\batta_{i+l (w-1)} \leq x_*(\epsilon)$.

We summarize, the total distance we have to cover is $x_*-\delta$
and every $(w-1)$ sections we cover a distance of at least $\kappa_*(\epsilon)
\delta$ as long as we have not surpassed $x_*(\epsilon)$.  Therefore,
after $(w-1) \lfloor \frac{x_*(\epsilon)-\delta}{\kappa_*(\epsilon)
\delta} \rfloor$ sections we have either passed $x_*$ or we must be
strictly closer to $x_*$ than $\kappa_*(\epsilon) \delta$.  Hence,
to cover the remaining distance we need at most $(w-2)$ extra sections.
The total number of sections needed is therefore upper bounded by
$w-2+(w-1) \lfloor \frac{x_*(\epsilon)-\delta}{\kappa_*(\epsilon)
\delta} \rfloor$, which, in turn, is upper bounded by $w
(\frac{x_*(\epsilon)}{\kappa_*(\epsilon) \delta}+1)$. The final
claim follows by bounding $x_*(\epsilon)$ with $1$ and $\kappa_*(\epsilon)$
by $\kappa_*$. \\

(ii) {\em The number of sections such that $\batta_i \in [
x_*(\epsilon), \xunstab(\epsilon)]$ is at most $ 2 w (\frac{4}{3 \kappa_* (x_*)^2}+1)$}
Let us define $\bavg_i = \frac1{w^2}\sum_{j,k =0}^{w-1}\batta_{i+j-k}.$
From Lemma~\ref{lem:avgprop}, $\batta_i \leq \epsilon
g(\bavg_i,\bavg_i, \dots, \bavg_i) = \bavg_i + h(\bavg_i)$.  Summing
this inequality over all sections from $-\infty$ to $k \leq 0$ we
get,
\begin{align*}
\sum_{i=-\infty}^{k} \batta_i \leq \sum_{i=-\infty}^{k} \bavg_i + \sum_{i=-\infty}^{k} h(\bavg_i).
\end{align*}
Writing $ \sum_{i=-\infty}^{k} \bavg_i$ in terms of the $\batta_i$s and rearranging terms,
\begin{align*}
- \sum_{i=-\infty}^{k} h(\bavg_i) & \leq
\frac{1}{w^2}\sum_{i=1}^{w-1} {w-i+1 \choose 2} (\batta_{k+i}-\batta_{k-i+1}) \\
& \leq \frac{w}6 (\batta_{k+(w-1)}-\batta_{k-(w-1)}).
\end{align*}
Let us summarize:
\begin{align}\label{equ:momentum}
\batta_{k+(w-1)}-\batta_{k-(w-1)} & \geq - \frac{6}{w} \sum_{i=-\infty}^{k} h(\bavg_i).
\end{align}

Without loss of generality we can assume that there exists a section
$k$ so that $x_*(\epsilon) \leq \batta_{k-(w-1)}$ (we know from
point (i) that we must reach this point unless the constellation is too
short, in which case the statement is trivially fulfilled).  Consider sections
$\batta_{k-(w-1)}, \dots, \batta_{k+(w-1)}$, so that in addition
$\batta_{k+(w-1)} \leq \xunstab(\epsilon)$. If no such $k$ exists
then there are at most $2w-1$ points in the interval $[x_*(\epsilon),
\xunstab(\epsilon)]$, and the statement is correct a fortiori.

Our plan is to use (\ref{equ:momentum}) to lower bound
$\batta_{k+(w-1)}-\batta_{k-(w-1)}$.  This means, we need a lower
bound for $-\frac{6}{w} \sum_{i=-\infty}^{k} h(\bavg_i)$. Since by
assumption $\batta_{k+(w-1)} \leq \xunstab(\epsilon)$, it follows
that $\bavg_{k} \leq \xunstab(\epsilon)$, so that every contribution
in the sum $-\frac{6}{w} \sum_{i=-\infty}^{k} h(\bavg_i)$ is positive
(cf. \text{Lemma~\ref{lem:propertyofh(x)} (\ref{lem:propertyofh(x)zero})}).
Further, by (the Spacing) Lemma~\ref{lem:spacing}, $w(\bavg_i -
\bavg_{i-1}) \leq 1$. Hence,
\begin{align} \nonumber 
& -\frac{6}{w} \sum_{i=-\infty}^{k} h(\bavg_i)  \geq
-6 \sum_{i=-\infty}^{k} h(\bavg_i)(\bavg_{i} - \bavg_{i-1}) \\
& \geq  6\kappa_*(\epsilon) \int_{0}^{x_*(\epsilon)/2} x \,\text{d}x 
 =  \frac{3\kappa_*(\epsilon)(x_*(\epsilon))^2}{4} .  \nonumber 
\end{align}
Let us explain how we obtain the last inequality. First we claim that
there must exist a section $i$ with $\bavg_i$ between $x_*(\epsilon)/2$
and $x_*(\epsilon)$. Indeed, suppose on the contrary that this was not
true. Let $k^* \leq k$ be the smallest section number such that
$\bavg_{k^*} \geq x_*(\epsilon)$.  Clearly, such a $k^*$ exists.
Indeed, since $x_*(\epsilon) \leq \batta_{k-(w-1)}$, it follows
that $\bavg_{k} \geq x_*(\epsilon)$.  Since $\bavg_{-\infty}=0$,
we must have $\bavg_{k^*-1}\leq x_*(\epsilon)/2$.  This implies
that $\bavg_{k^*} - \bavg_{k^*-1} > x_*(\epsilon)/2$.  Using (the
Spacing) Lemma~\ref{lem:spacing} we conclude that $\frac{\dl-1}{w}
\geq x_*(\epsilon)/2$. Hence $w\leq 2\dl/x_*(\epsilon)$. Using the
universal lower bound on $x_*(\epsilon)$, we get $w \leq  2\dl^3\dr^2$,
a contradiction to the hypothesis of the lemma.  Finally, according
to Lemma~\ref{lem:propertyofh(x)} part (\ref{lem:propertyofh(x)three}),
$-h(x) \geq \kappa_*(\epsilon) x$ for $x \in [0, x_*(\epsilon)]$,
which implies the inequality.
Combined with (\ref{equ:momentum}) this implies that
\begin{align*}
\batta_{k+(w-1)}-\batta_{k-(w-1)} \geq \frac{3\kappa_*(\epsilon)(x_*(\epsilon))^2}{4}.
\end{align*}
We summarize, the total distance we have to cover is
$\xunstab(\epsilon)-x_*(\epsilon)$ and every $2(w-1)$ steps we cover
a distance of at least $\frac{3\kappa_*(\epsilon)(x_*(\epsilon))^2}{4}$
as long as we have not surpassed $\xunstab(\epsilon)$. Allowing for
$2(w-1)-1$ extra steps to cover the last part, bounding again $w-1$
by $w$, bounding $\xunstab(\epsilon)-x_*(\epsilon)$ by $1$ and
replacing $\kappa_*(\epsilon)$ and $x_*(\epsilon)$ by their universal
lower bounds, proves the claim.  
\end{proof}


\section{
 Saturation -- Theorem~\ref{thm:existenceintermediateform}
 }\label{app:areathmsforpartialapproxfamily}

Before we proceed to prove the Saturation theorem, we introduce a key technical
element required in the proof, {\em a family of spatial (approximate) FPs}. This is the
content of Definition~\ref{def:partialFPfamily} and Theorem~\ref{lem:interpolationyieldsapproximatefpfamilypartial}. 
Then, Theorem~\ref{the:areatheoremforapproximatefpfamilypartial} shows that the GEXIT
integral of this family depends only on its end-points. Combined with the
Negativity lemma~\ref{lem:asymptoticnegativity} this imposes a strong constraint on the
channel value of the spatial FPs, culminating in the proof of the Saturation theorem. 

\begin{definition}[Interpolation]
\label{def:partialFPfamily} 
Let $(\Ldens{c}^*, \Ldens{\x}^*)$, $\Ldens{c}^* \in \{\Ldens{c}_\ent\}$, denote an increasing
one-sided constellation on $[-\Lfp, 0]$ for the parameters $(\dl, \dr, w)$.
Let $\ent^*=\entropy(\Ldens{c}^*)>0$
and let $0 \leq \Lc \leq \Lfp$.

The family (of constellations) for the $(\dl, \dr, \Lc,
w)$-ensemble, based on $(\Ldens{c}^*, \Ldens{\x}^*)$, is denoted by
$\{\Ldens{c}_{\sigma}, \Ldens{\x}_{\sigma}\}_{\sigma=0}^{\ih^*}$.

Each element $\Ldens{\x}_{\sigma}$ is symmetric
with respect to the spatial index and the components are indexed by $[-\Lc, \Lc]$.  Hence it suffices to
define the constellations in the range $[-\Lc, 0]$ and then we set
$\Ldens{x}_{\sigma, i}=\Ldens{x}_{\sigma, -i}$ for $i \in [0, \Lc]$.
As usual, we set $\Ldens{x}_{\sigma, i}=\Delta_{+\infty}$ for $i\notin
[-\Lc,\Lc]$.  For $i \in [-\Lc, 0]$ and $\sigma \in [0, \ih^*)$
define
\begin{align*}
\Ldens{x}_{\sigma, i}  & = \begin{cases} 
\Ldens{a}_{\sigma, i}, & \sigma \in (\frac{\ih^*}{2}, \ih^*), \\
\frac{2}{\ih^*} \sigma \Ldens{x}^*_{i-\Lfp+\Lc}+(1-\frac{2}{\ih^*}\sigma) \Delta_{+\infty}, & \sigma \in [0, \frac{\ih^*}2], 
\end{cases} 
\end{align*}
where for $\sigma \in (\frac{\ih^*}{2}, \ih^*)$,
\begin{align*}
\Ldens{a}_{\sigma, i}  = &
\alpha(\sigma) \Ldens{x}^*_{i-\lceil (2-\frac2{\ent^*}\sigma)(\Lfp-\Lc)\rceil} + \\
 & (1-\alpha(\sigma)) \Ldens{x}^*_{i-\lceil (2-\frac2{\ent^*}\sigma) (\Lfp-\Lc) \rceil+1}, \\
\alpha(\sigma) =  & \Big( (\Lfp-\Lc) (2-\frac2{\ent^*}\sigma)\Big) \!\!\!\!\mod(1).
\end{align*}
Finally, $\Ldens{c}_{\sigma}=\Ldens{c}_{\ih=\ih^*}=\Ldens{c}^*$.
\qed
\end{definition}

{\em Discussion}:
\begin{enumerate}[(i)]
\item Notice that in the above definition when $\sigma$ approaches $\ih^*$, then 
$\Ldens{x}_{\sigma, i} = \Ldens{x}^*_i$. 

\item 
In the definition above, we keep the channel constant across the sections and over $\sigma$.
In other words, the channel remains constant for all the constellations in the family. 

We denote the two partitions in the interpolation as phases, e.g.,
$(\ih^*/2,\ih^*)$ corresponds to phase I and  $[0, \frac{\ih^*}{2}]$
corresponds to phase II.

\item
The above interpolation might look complicated.  But there is a
straightforward interpretation.  Think of one-sided constellations.
We are interested in a constellation of size $\Lc$.

In phase I, the basic idea is to ``move'' the
constellation $\Ldens{\x}^*$ to the right and at each point in time to
``chop off'' the overhanging parts both on the left and on the right. We
do this until the left most section of $\Ldens{\x}^*$ is at position
$-\Lc$. If $\Ldens{\x}^*$ were a continuous function, i.e., suppose we had a continuum of sections,
then this would be all we need to do. But $\Ldens{\x}^*$ is discrete, so in order to get
a continuous interpolation we interpolate between
two consecutive elements of $\Ldens{\x}^*$. This mimics the ``wave effect'' we
mentioned in the beginning. 

In phase II, the residual constellation is uniformly brought down to
$\Delta_{+\infty}$ in each section.
\end{enumerate}

In the next lemma we show that if we have an interpolated family constructed
via the above definition, then the resulting family is a family of approximate FPs. 
\begin{lemma}[Interpolation Yields Approximate FP Family]
\label{lem:interpolationyieldsapproximatefpfamilypartial} 
Let $(\Ldens{c}^*, \Ldens{\x}^*)$, $\Ldens{c}^* \in \{\Ldens{c}_\ent\}$,
denote an increasing one-sided constellation on $[-\Lfp, 0]$ with
free or fixed boundary condition for the parameters $(\dl, \dr, w)$ and let
$w \leq \Lc < \Lfp$. Assume that $(\Ldens{c}^*, \Ldens{\x}^*)$
fulfills the following conditions, for some $0<\delta\leq \frac1w$.
\begin{enumerate}[(i)]
\item
{\em Constellation is close to $\Delta_{+\infty}$ ``on the left''}:
\begin{align*}
\batta(\Ldens{x}^*_{-\Lfp+\Lc}) \leq \delta.
\end{align*}

\item {\em Constellation is flat ``on the right''}
\begin{align*}
\Ldens{x}^*_{-\Lc} = \Ldens{x}^*_{-\Lc+1} = \dots = \Ldens{x}^*_0 = \Ldens{x}.
\end{align*}
Also, $d(\Ldens{x}^*_{-\Lc-w+1}, \Ldens{x}) \leq \delta.$

\item {\em Constellation is approximate FP}:
For $i \in [-\Lfp, 0]$,
\begin{align*}
d(\Ldens{x}_i^*,\Ldens{c}^* \vconv g(\Ldens{x}_{i-w+1}^*,\dots,\Ldens{x}_{i+w-1}^*)) \leq \delta.
\end{align*}
\end{enumerate}

Let $\{\Ldens{c}_{\sigma}, \Ldens{\x}_{\sigma}\}_{\sigma=0}^{\ih^*}$
denote the family as described in Definition~\ref{def:partialFPfamily}.
Then this family is an approximate FP family. More precisely, for
$\underline{\sigma}=0$ and $\overline{\sigma}=\ih^*$ 
\begin{enumerate}[(i)]
\item $\{\Ldens{c}_\sigma\}_{\underline{\sigma}}^{\overline{\sigma}}$ and 
$\{\Ldens{\x}_{\sigma}\}_{\underline{\sigma}}^{\overline{\sigma}}$ are ordered by degradation, increasing,
and piece-wise linear,
\item  $\Ldens{x}_{\sigma, i}=\Delta_{+\infty}$ for $i \notin [-\Lc, \Lc]$ and for all $\sigma$ and
\item
for any $\sigma\in [\underline{\sigma}, \overline{\sigma})$ and any $i\in [-\Lc+w-1,-w+1] \cup [w-1, L-w+1]$ 
\begin{align}
d(\Ldens{x}_{\sigma,i},& \Ldens{c}_{\sigma}\vconv
g(\Ldens{x}_{\sigma,i-w+1},\dots,\Ldens{x}_{\sigma,i+w-1}))  \nonumber \\ 
& \leq \frac{2(\dl-1)(\dr-1)}{w} + \delta. \label{equ:partialddistance}
\end{align}
\end{enumerate}
\end{lemma}
{\em Discussion:} For the boundary $[-L, -\Lc+w-2] \cup [\Lc-w+2,
L]$ and in the middle $[-w+2, w-2]$ the interpolation does not in
general result in an approximate FP. Fortunately this does not cause
problems.  We will see in
Theorem~\ref{the:areatheoremforapproximatefpfamilypartial} that
each section gives only a small contribution to the GEXIT integral.
If we choose $\Lc$ sufficiently large then we can safely ignore a
fixed number of sections.
\begin{proof}
\begin{enumerate}
\item That $\{\Ldens{c}_\sigma\}_{\underline{\sigma}}^{\overline{\sigma}}$ and 
$\{\Ldens{\x}_{\sigma}\}_{\underline{\sigma}}^{\overline{\sigma}}$ are ordered by degradation, increasing,
and piece-wise linear follows by construction.
\item  In the same way, that $\Ldens{x}_{\sigma, i}=\Delta_{+\infty}$ for $i \notin [-\Lc, \Lc]$ 
and for all $\sigma$ also follows by construction.

\item It remains to check that the family so defined constitutes an
approximate FP family. Since the family, by definition, is symmetric around the
section $0$, we check only for the sections belonging in $[-\Lc+w-1,-w+1]$.
\begin{enumerate}

\item[] {\em Phase I}: Think
of $i$ and $\sigma$ as fixed, $i\in [-\Lc+w-1,-w+1]$. Define $c=c(
\sigma)$ and $j= i-\lceil (2-\frac2{\ent^*}\sigma)(\Lfp-\Lc)\rceil$.
Set $\Ldens{z}^*_j=c \Ldens{x}^*_j+\bar{c} \Ldens{x}^*_{j+1}$. 
With these conventions, we want to bound 
\begin{align*}
d(\Ldens{z}^*_{j}, \Ldens{c}_{\ih^*} \vconv g(\Ldens{z}^*_{j-w+1}, \cdots, \Ldens{z}^*_{j+w-1})).
\end{align*}
Using the convexity property (\ref{lem:blmetricconvexity}) of (the
Wasserstein metric) Lemma~\ref{lem:blmetric}, it is sufficient to bound
\begin{align*}
d(\Ldens{x}^*_{j},\Ldens{c}_{\ih^*} \vconv g(\Ldens{z}^*_{j-w+1}, \cdots, \Ldens{z}^*_{j+w-1})), \;\text{and}\\ 
d(\Ldens{x}^*_{j+1},\Ldens{c}_{\ih^*} \vconv g(\Ldens{z}^*_{j-w+1}, \cdots, \Ldens{z}^*_{j+w-1})
)
\end{align*}
separately. The two bounds are identical and their derivation is
also essentially identical. Let us therefore concentrate on the
first expression. Using first the triangle inequality and 
then the regularity properties (\ref{lem:blmetricregularvconv}) and
(\ref{lem:blmetricregularcconv}) as well as the convexity property
(\ref{lem:blmetricconvexity}), we upper bound the first expression by 
\begin{align*}
& d( \Ldens{c}_{\ih^*}\vconv
g(\Ldens{x}^*_{j-w+1},\dots,\Ldens{x}^*_{j+w-1}), \nonumber \\
& \quad\quad \Ldens{c}_{\ih^*} \vconv g(\Ldens{z}^*_{j-w+1}, \cdots, \Ldens{z}^*_{j+w-1}) ) + \nonumber \\
& + d( \Ldens{x}^*_j,  \Ldens{c}_{\ih^*}\vconv g(\Ldens{x}^*_{j-w+1},\dots,\Ldens{x}^*_{j+w-1})) \\
& \leq  2 d \Bigl(
\Bigl(\frac{1}{w} \sum_{l=0}^{w-1} \bigl(\frac{1}{w} \sum_{k=0}^{w-1} 
\Ldens{x}^*_{j\!+\!l\!-\!k} \bigr)^{\cconv \dr\!-\!1} \Bigr)^{\vconv \dl\!-\!1}, \nonumber \\
& \quad\quad\quad\quad\quad\quad 
\Bigl(\frac{1}{w} \sum_{l=0}^{w-1} \bigl(\frac{1}{w} \sum_{k=0}^{w-1} 
\Ldens{z}^*_{j\!+\!l\!-\!k} \bigr)^{\cconv \dr\!-\!1} \Bigr)^{\vconv \dl\!-\!1}
\Bigr) \!+\!\delta \nonumber \\
& \leq  \frac{2 (\dl\!-\!1)}{w} \sum_{l=0}^{w-1}d \Bigl(
\bigl(\frac1w \sum_{k=0}^{w-1} 
\Ldens{x}^*_{j\!+\!l\!-\!k} \bigr)^{\cconv \dr\!-\!1} , \\
& \quad\quad\quad\quad\quad\quad\quad\quad\quad\quad\quad\quad  \bigl(\frac{1}{w} \sum_{k=0}^{w-1} 
\Ldens{z}^*_{j\!+\!l\!-\!k} \bigr)^{\cconv \dr\!-\!1}
\Bigr) \!+\!\delta \nonumber \\
& \leq  \frac{2 (\dl\!-\!1)(\dr-1)}{w^2} \sum_{l=0}^{w-1} d
\Bigl(
\sum_{k=0}^{w-1} 
\Ldens{x}^*_{j\!+\!l\!-\!k}, 
\sum_{k=0}^{w-1} 
\Ldens{z}^*_{j\!+\!l\!-\!k} 
\Bigr) \!+\!\delta \nonumber \\
& = \frac{2 (\dl\text{-}1)(\dr\text{-}1)}{w^2} \!\!\sum_{l=j}^{j+w-1} 
\!\!\!\!d(\!\sum_{k=0}^{w-1} \!\Ldens{x}^*_{l\!-\!k}, \!\!\sum_{k=0}^{w-1} \!\!c\Ldens{x}^*_{l\!-\!k}\!\!+\!\!\bar{c} \Ldens{x}^*_{l\!-\!k+1}\!) \!+\!\delta  \\
& = \frac{2 (\dl\!-\!1)(\dr\!-\!1)}{w^2} \bar{c} \sum_{l=j}^{j+w-1} d(\Ldens{x}^*_{l-w+1},  \Ldens{x}^*_{l+1})  \!+\!\delta \\
& \leq \frac{2 (\dl\!-\!1)(\dr\!-\!1)}{w} \!+\!\delta,
\end{align*}
where to obtain the first inequality we use the approximate nature of $\Ldens{\x}^*$ and in the last step we have used property (\ref{lem:blmetricboundedness}) of Lemma~\ref{lem:blmetric}.
\item[] {\em Phase II}: 
In this regime we interpolate the ``tail'' of the original constellation
uniformly to $\Delta_{+\infty}$.  From the assumption of the lemma we
have $\batta(\Ldens{x}^*_{-\Lfp+\Lc}) \leq \delta$. Since
$\Ldens{\x}^*$ is increasing we must have $\batta(\Ldens{x}^*_{i-\Lfp+\Lc})
\leq \batta(\Ldens{x}^*_{-\Lfp+\Lc})$ for $i \in [-\Lc, 0]$.
Lemma~\ref{lem:blmetric}, property (\ref{lem:blmetriccontinuity}),
then implies that $d(\Ldens{x}^*_{i-\Lfp+\Lc}, \Delta_{+\infty})
\leq \delta$ for all $i\in [-\Lc,0]$.

Again, think
of $i$ and $\sigma$ as fixed, $i\in [-\Lc+w-1,-w+1]$.
Set $c=2\sigma/\ih^*$ and $j=i-\Lfp+\Lc$. Then 
\begin{align*}
& d(\!c \Ldens{x}^*_{j} \!\!+\! \bar{c} \Delta_{\!+\!\infty},
\Ldens{c}_{\ih^*}\!\!\vconv\! g(\!c \Ldens{x}^*_{j\!-\!w\!+\!1} \!\!+\! \bar{c}\Delta_{\!+\!\infty},\!\dots\!,\!c\! \Ldens{x}^*_{j\!+\!w\!-\!1}\!\!+\! \bar{c}\Delta_{\!+\!\infty}\!) \\
& \leq  d(\!c \Ldens{x}^*_{j} \!\!+\! \bar{c} \Delta_{\!+\!\infty}, \Delta_{+\infty})  \\ 
& + d(\Delta_{+\infty},\Ldens{c}_{\ih^*}\!\!\vconv\! g(\!c \Ldens{x}^*_{j\!-\!w\!+\!1} \!\!+\! \bar{c}\Delta_{\!+\!\infty},\!\dots\!,\!c\! \Ldens{x}^*_{j\!+\!w\!-\!1}\!\!+\! \bar{c}\Delta_{\!+\!\infty}\!) \\
& \leq 2 (\dl-1)(\dr-1) \delta c + c \delta \\
& \stackrel{\delta \leq 1/w}{\leq} \frac{2 (\dl-1)(\dr-1)}{w}  +  \delta,
\end{align*}
where  
 to obtain the penultimate inequality we use Lemma~\ref{lem:sensitivity} to bound
the distance of
$\Ldens{c}_{\ih^*}\!\vconv\! g(c \Ldens{x}^*_{j\!-\!w\!+\!1} +
\bar{c}\Delta_{+\infty},\dots,c \Ldens{x}^*_{j\!+\!w\!-\!1}+
\bar{c}\Delta_{+\infty}\!)$ to $\Delta_{+\infty} (= \Ldens{c}_{\ih^*}\!\vconv\! g(\Delta_{+\infty},\dots,\Delta_{+\infty})$, since $\Delta_{+\infty}$ is always an FP of DE) and the second
expression is the distance of $c \Ldens{x}^*_{j} + \bar{c}
\Delta_{+\infty}$ to $\Delta_{+\infty}$, which is bounded using the previous arguments.
\end{enumerate}
\end{enumerate}
\end{proof}

Next, we show that if we have an approximate family of FPs, then the area under
the GEXIT integral associated to the family depends only on the ``end points'' 
of the interpolated family.  
\begin{theorem}[Area Theorem for Approx. FP Family]\label{the:areatheoremforapproximatefpfamilypartial}
Let $\{\Ldens{c}_\sigma, \Ldens{\x}_\sigma\}_{\underline{\sigma}}^{\overline{\sigma}}$ denote 
an approximate FP family for the $(\dl, \dr, \Lc, w)$ ensemble.
More precisely, 
\begin{enumerate}[(i)]
\item  $\{\Ldens{c}_\sigma\}_{\underline{\sigma}}^{\overline{\sigma}}$ and 
$\{ \Ldens{\x}_\sigma\}_{\underline{\sigma}}^{\overline{\sigma}}$ are ordered 
by degradation, increasing, and piece-wise linear\footnote{In fact, we will
apply this theorem to the family given in Definition~\ref{def:partialFPfamily}. More generally, however, 
given a set of distinct ordered densities $\Ldens{a}_1 \prec
\Ldens{a}_2 \prec \dots \prec \Ldens{a}_n$, we get a piece-wise linear family by
linearly interpolating always between consecutive densities.},  
\item  $\Ldens{x}_{\sigma, i}=\Delta_{+\infty}$ for $i \notin [-\Lc, \Lc]$ and for all $\sigma$, 
\item  $\Ldens{x}_{\underline{\sigma}, i}=\Ldens{x}_{\underline{\sigma}}$ for $i \in [-\Lc, \Lc]$, 
\item  $\Ldens{x}_{\overline{\sigma}, i}=\Ldens{x}_{\overline{\sigma}}$ for $i \in [-\Lc, \Lc]$, and
\item for all 
$i \in [-\Lc+w-1, -w+1] \cup [w-1, \Lc-w+1]$ 
and $\sigma \in
[\underline{\sigma}, \overline{\sigma}]$ 
\begin{align*} d(\Ldens{x}_{\sigma, i}, \Ldens{c}_\sigma \vconv
g(\Ldens{x}_{\sigma, i-w+1}, \dots, \Ldens{x}_{\sigma, i+w-1})) \leq \delta.
\end{align*}
\end{enumerate}
Define
\begin{align*}
A(\{\Ldens{c}_\sigma, \Ldens{\x}_\sigma\}_{\underline{\sigma}}^{\overline{\sigma}})=\sum_{i=-\Lc}^{\Lc} G(\{\Ldens{c}_\sigma, 
\hat{g}(\Ldens{x}_{\sigma, i\!-\!w\!+\!1}, \!\dots\!,  \Ldens{x}_{\sigma, i\!+\!w\!-\!1})
\}_{\underline{\sigma}}^{\overline{\sigma}}),
\end{align*}
where $G(\{\Ldens{c}_\sigma, 
\hat{g}(\Ldens{x}_{\sigma, i\!-\!w\!+\!1}, \!\dots\!,  \Ldens{x}_{\sigma, i\!+\!w\!-\!1})
\}_{\underline{\sigma}}^{\overline{\sigma}})$ is the GEXIT integral introduced in 
Definition~\ref{def:gexitintegralbasic}.  Let
\begin{align*}
A(\Ldens{x}) = \entropy(\Ldens{x}) +  (\dl-1-\frac{\dl}{\dr}) \entropy(\Ldens{x}^{\cconv \dr})  - 
(\dl-1) \entropy(\Ldens{x}^{\cconv \dr-1}).
\end{align*}
Then
$A(\{\Ldens{c}_\sigma, \Ldens{\x}_\sigma\}_{\underline{\sigma}}^{\overline{\sigma}})$
is well defined and  
\begin{align*}
\Big|\!\frac{A(\!\{\Ldens{c}_\sigma, \Ldens{\x}_\sigma\}_{\underline{\sigma}}^{\overline{\sigma}})}{2 \Lc+1}\!&-A(\Ldens{x}_{\overline{\sigma}})+A(\Ldens{x}_{\underline{\sigma}})\!\Big| \leq \cwdldrLdelta,
\end{align*}
where
\begin{align*}
\cwdldrLdelta = & \frac{11 w (1+\dl\dr)}{2 \Lc\!+\!1} \! +\! 4(\sqrt{2}\!+\!\frac2{\ln 2}\dl (\dr\!-\!1)) \sqrt{\delta}.
\end{align*} 
\end{theorem}
{\em Discussion:} In words, the theorem says that for any family of 
spatial FPs which start and end at a constant (over all sections) FP, the 
GEXIT integral is given by the end-points and is close to the difference
of the $A$ expression introduced in Lemma~\ref{lem:areaunderBPGEXIT}. In fact, 
from the Lemma~\ref{lem:areaunderBPGEXIT} we see that, graphically, this is equal
to the area under the BP GEXIT curve of the underlying ensemble between the two 
end-points. 
\begin{proof}
Let us consider the circular ensemble which is associated to $(\dl,
\dr, \Lc, w)$ (see Definition~\ref{def:circularensemble}).  As
defined in the statement of the lemma, for  $i \in [-\Lc, \Lc]$,
the channel ``seen'' at position $i$ is
$\Ldens{c}_{\sigma,i}=\Ldens{c}_{\sigma}$. 
For the remaining sections $i \in [\Lc+1, \Lc+w-1]$ we impose the ``natural'' condition
$\Ldens{c}_{\sigma,i}=\Delta_{+\infty}$.
As a consequence, for these positions
$\Ldens{x}_{\sigma,i}=\Delta_{+\infty}$.

Since $\{\Ldens{c}_\sigma\}$ as well as $\{\Ldens{\x}_\sigma\}$ are
piece-wise linear, all GEXIT integrals are well defined (see the
proof of Lemma~\ref{lem:areaunderBPGEXIT}). Consequently,
$A(\{\Ldens{c}_\sigma,
\Ldens{\x}_\sigma\}_{\underline{\sigma}}^{\overline{\sigma}})$ is
well-defined.

Instead of determining $A(\{\Ldens{c}_\sigma,
\Ldens{\x}_\sigma\}_{\underline{\sigma}}^{\overline{\sigma}})$,
directly, let us determine the equivalent quantity associated to
the circular ensemble, i.e., we include the $w-1$ extra positions
$[\Lc+1, \Lc+w-1]$.  Since for all ``extra'' positions the associated
channel is constant, and so the additional integrals are zero, the
numerical value of these two unnormalized GEXIT integrals is in
fact identical.

We will now derive upper and lower bounds for the GEXIT integrals
for the given approximate FP family. Recall: for $i \in [-\Lc+w-1,
-w+1] \cup [w-1, \Lc-w+1]$ we have a $\delta$-approximate (in the
Wasserstein metric) FP family. For $i \in [-\Lc, -\Lc+w-2] \cup
[-w+2, w-2] \cup [\Lc-w+2, \Lc]$ all we know is that the channel
is a monotone function of $\sigma$.  Finally, for $i \in [\Lc+1,
\Lc+w-1]$ the channel is frozen to ``perfect.''

Let us start by deriving a lower bound.
\begin{enumerate}
\item[] {\em Boundary:} For $i \in [-\Lc, -\Lc+w-2] \cup [-w+2, w-2] \cup [\Lc-w+2, \Lc]$ the GEXIT
integral is non-negative. Thus, in this regime, we get a lower bound by setting each GEXIT integral to 0 (cf. Lemma~\ref{lem:gexitsmoothfamily}). 
\item[] {\em Interior:} Consider the GEXIT integrals for $i \in [-\Lc+w-1, -w+1] \cup [w-1, \Lc-w+1]$.
\begin{itemize}
\item[]{\em Technique:}
Rather than evaluating these integrals directly we use the
technique introduced in \cite{MMU08}, i.e., we consider the
computation tree of height $2$ rooted in node $i$ as shown in
Figure~\ref{fig:computationgraph} for the specific case $(\dl=2, \dr=4)$.
\begin{figure}[hbt] \centering
\setlength{\unitlength}{1.0bp}%
\begin{picture}(160,160)
\put(5,0){\includegraphics[scale=1.0]{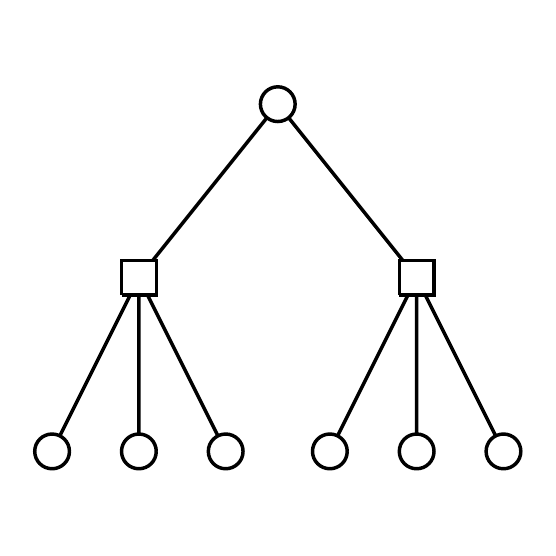}}
\put(85,10){\makebox(0,0){\small{leaves}}}
\put(85,150){\makebox(0,0){\small{root}}} \end{picture}
\caption{Computation tree of height 2 for $(2,4)$-regular LDPC ensemble.}
\label{fig:computationgraph} \end{figure} 
More precisely, there are $\dl$ check nodes connected to this root
variable node and $(\dr-1)$ further variable nodes connected to each
such check node.  So in total there are $\dl$ check nodes in this tree
and $1+\dl(\dr-1)$ variable nodes. We call the starting variable node, 
the {\em root} and all other variable nodes, {\em leaves}.  By symmetry
it suffices to consider one branch of this computation tree in detail.
Let $j$, $j \in [i, i+w-1]$, denote the position of a particular
check node.  We assume that the choice of $j$ is done uniformly
over this interval.  Let $k_{l}$, $l \in [1, \dr-1]$, $k_l \in[j-w+1,
j]$, denote the position of the $l$-th variable node attached to this
check node, and let the index of the root node be $0$. For the leaf nodes we
assume again a uniform choice of $k_l$ over the allowed interval.
Note that, wlog, we have set the position $l=0$ for the root variable
node.  For each computation tree assign to its root node the channel
$\Ldens{c}_{\sigma,i}$, whereas each leaf variable node at position
$k_{l}$ ``sees'' the channel $\Ldens{x}_{\sigma, k_{l}}$.  Note that
for our model of the tree, the distribution (averaged over this choice)
which flows into the root node is exactly $\hat{g}(\Ldens{x}_{\sigma,
i\!-\!w\!+\!1}, \!\dots\!,  \Ldens{x}_{\sigma, i\!+\!w\!-\!1})$,
as required for the computation of $A(\{\Ldens{c}_\sigma,
\Ldens{\x}_\sigma\}_{\underline{\sigma}}^{\overline{\sigma}})$.

Let us describe the basic trick which will help us to accomplish the
computation.  We will first determine the sum of all GEXIT integrals
associated to such a tree. From this we will then subtract the GEXIT
integrals associated to its leaf nodes. This will give us the GEXIT
integral associated to the root node, which is what we are interested in.

More precisely, we use \eqref{eq:totalderivativeequaltosumofpartialderivatives}. 
The lhs of this equation gives us the contribution of the overall tree and the rhs contains
the GEXIT integral of the root node plus the GEXIT integrals of the leaf nodes.
For the current case, we stress that all the operations (integrals of derivatives and partial
derivatives) in \eqref{eq:totalderivativeequaltosumofpartialderivatives} are well-defined since the family we consider is piece-wise linear
\item[] {\em Contributions from overall tree:} Recall that for $i
\in [-\Lc, \Lc]$, $\Ldens{x}_{i,
\underline{\sigma}}=\Ldens{x}_{\underline{\sigma}}$ and $\Ldens{x}_{i,
\overline{\sigma}}=\Ldens{x}_{\overline{\sigma}}$.

Consider first the case $\sigma=\overline{\sigma}$ and $i \in
[-\Lc+w-1, -w+1] \cup [w-1, \Lc-w+1]$.  From Lemma~\ref{lem:entropyoftree}
we know that the conditional entropy $\entropy(X \mid Y)$ of the
tree code is given by
\begin{align*}
& \entropy(\tilde{\Ldens{x}}_{\overline{\sigma}})+ 
\dl(\dr-1) \entropy(\Ldens{x}_{\overline{\sigma}}) - 
\entropy(\tilde{\Ldens{x}}_{\overline{\sigma}}\cconv\Ldens{x}_{\overline{\sigma}}^{\cconv \dr-1}) \\
& - (\dl-1) \entropy(\Ldens{x}^{\cconv \dr-1}_{\overline{\sigma}}),
\end{align*}
where $\tilde{\Ldens{x}}_{\overline{\sigma}}= \Ldens{c}_{\overline{\sigma}}
\vconv (\Ldens{x}_{\overline{\sigma}}^{\cconv \dr-1})^{\vconv
\dl-1}$.  Now recall that $d(\tilde{\Ldens{x}}_{\overline{\sigma}}, \Ldens{x}_{\overline{\sigma}}) \leq \delta$.
Define $T(\Ldens{x})$ as
\begin{align*}
(1\!+\!\dl(\dr\!-\!1)) \entropy(\Ldens{x})\!-\!\entropy(\Ldens{x}^{\cconv \dr})\!-\!(\dl\!-\!1) 
\entropy(\Ldens{x}^{\cconv \dr-1}).
\end{align*}
Then (dropping the subscripts $\overline{\sigma}$ for a moment),
\begin{align*}
 & \vert \entropy(X \mid Y)  - T(\Ldens{x})\vert \\
  &\leq \vert  \entropy(\tilde{\Ldens{x}})\! -\!  
 \entropy(\Ldens{x})\vert \!+\!  \vert
\entropy(\tilde{\Ldens{x}}\cconv\Ldens{x}^{\cconv \dr-1}) \!-\!
\entropy(\Ldens{x}^{\cconv \dr})\vert \\  
& \stackrel{\text{Lem.~\ref{lem:blmetric}.\ref{lem:blmetricwasserboundsbatta}}}{\leq} 
h_2(d(\tilde{\Ldens{x}}, \Ldens{x})/2) \!+\!  h_2(d(\tilde{\Ldens{x}}\cconv\Ldens{x}^{\cconv \dr-1}, \Ldens{x}^{\cconv \dr})/2) \\
& \stackrel{\text{Lem.~\ref{lem:blmetric}.\ref{lem:blmetricregularcconv}}}{\leq}   
2  h_2(d(\tilde{\Ldens{x}}, \Ldens{x})/2) 
\stackrel{\text{(\ref{eq:upperboundbinaryentropyone})}}{\leq}   
4 \sqrt{d(\tilde{\Ldens{x}}, \Ldens{x})/2} \leq 2\sqrt{2 \delta}.
\end{align*}
Exactly the same argument tells us that the entropy of such a tree
for $\sigma=\underline{\sigma}$ is, up to a possible error of size
$2\sqrt{2\delta}$, equal to $T(\Ldens{x}_{\underline{\sigma}})$.  We
conclude: the difference of the total entropy of such a tree is
lower bounded by
$T(\Ldens{x}_{\overline{\sigma}})-T(\Ldens{x}_{\underline{\sigma}}) -
4\sqrt{2\delta}$, call this $B-4\sqrt{2\delta}$.

\item[] {\em Contributions from leaves:} We need to find the
contributions of GEXIT integrals associated to all the leaf nodes
of each such tree rooted at a position $i \in [-\Lc+w-1, -w+1] \cup
[w-1, \Lc-w+1]$.  The exact such sum is difficult to determine. But
we only need an upper bound to derive a lower bound on the overall
GEXIT integral.  Note that GEXIT integrals are non-negative. Hence,
let us compute the sum of GEXIT integrals of leaf nodes of {\em
all} computation trees, whether they are rooted in a position $i
\in [-\Lc+w-1, -w+1] \cup [w-1, \Lc-w+1]$ or not.

By symmetry, this contribution is easy to determine.  More precisely,
consider the following equivalent procedure. Pick a check node at
position $j$, $j \in [-\Lc, \Lc+w-1]$.  Every check node has $\dr$
connected variable nodes, where each variable node is picked with
uniform probability and independently from the range $[j-w+1, j]$
and the choice of the $\dr$ variables is iid (note that the connections
 are taken on the circular ensemble).

\item[] {\em Contributions from checks in the range $[-\Lc, -\Lc+2w-3]
\cup [-w+2, 2w-3] \cup [\Lc-w+2, \Lc+w-1]$:} Check nodes in this
range might see some frozen channels or channels which do not form
approximate FPs. Hence we upper bound all GEXIT integrals associated
to check nodes in this range by $1$ (cf. Lemma~\ref{lem:gexitsmoothfamily}). The number of such integrals
is $(7w-8) \dl (\dr-1)$.

\item[] {\em Contributions from checks in the range $[-\Lc+2w-2,
-w+1] \cup[2w -2, \Lc-w+1]$:} Check nodes in this range only see
channels which are approximate FPs and none of the channels are
frozen.  There are $(2 \Lc-6w+8) \dl (\dr-1)$ such integrals.
Let us determine the contribution for each such integral.
Since we consider an average over all possible
computation trees, the (average) density entering a check node is
equal for all the leaf nodes (there are $\dr-1$ such densities).
Let us call this density $\Ldens{x}_{\sigma}$. If we focus on a
check node at position $j$, this density is equal to
$$
\Ldens{x}_{\sigma} = \frac1w\sum_{k=0}^{w-1} \Ldens{x}_{\sigma, j-k}.
$$
However, the density entering the check node, at position $j$, from the
root node will be different from $\Ldens{x}_{\sigma}$, since we do not have a
family of true FPs. Call this density $\tilde{\Ldens{x}}_{\sigma}$. This density
is equal to 
$$
\tilde{\Ldens{x}}_{\sigma} = \frac1w\sum_{k=0}^{w-1} \Ldens{c}_{\sigma}\vconv
g(\Ldens{x}_{\sigma, j-k-w+1},\dots,\Ldens{x}_{\sigma,j-k+w-1}).
$$
Since we assumed that we have an approximate FP family and due to the
convexity of the Wasserstein metric, we conclude that
$d(\Ldens{x}_{\sigma}, \tilde{\Ldens{x}}_{\sigma}) \leq \delta$.
Let us define $P(\Ldens{x})=\entropy(\Ldens{x}) -
\frac1{\dr}\entropy(\Ldens{x}^{\cconv \dr})$. From Lemma~\ref{lem:entropyofcheck}
 we have that $P(\Ldens{x})$ is the GEXIT integral of a leaf node if we had
a true FP. Since we have an approximate FP, 
each such integral can be upper bounded by
$P(\Ldens{x}_{\overline{\sigma}})-P(\Ldens{x}_{\underline{\sigma}}) +
\frac8{\ln 2} \sqrt{2\delta}$, call it $C+ \frac{8}{\ln 2} \sqrt{2\delta}$. 
We derive this as follows. We want to bound the difference
$$
\Big \vert\int_{\underline{\sigma}}^{\overline{\sigma}}  \entropy(\frac{\dee \Ldens{x}_{\sigma}}{\dee \sigma}\vconv \Ldens{z}_{\sigma})\dee \sigma
- \int_{\underline{\sigma}}^{\overline{\sigma}}\entropy(\frac{\dee \Ldens{x}_{\sigma}}{\dee \sigma}\vconv \tilde{\Ldens{z}}_{\sigma})\dee \sigma
\Big\vert,$$
where $\Ldens{z}_{\sigma}=\Ldens{x}_{\sigma}^{\cconv \dr-1}$ and $\tilde{\Ldens{z}}_{\sigma}=\Ldens{x}_{\sigma}^{\cconv \dr-2}\cconv \tilde{\Ldens{x}}_{\sigma}$.
Since the family, $\{\Ldens{x}_{\sigma}\}$ is piece-wise linear, we use
\eqref{eq:totalderivativeequaltosumofpartialderivatives} (applied in this case to the single parity-check code),
Lemma~\ref{lem:entropyofcheck} and
symmetry to conclude that $\int_{\underline{\sigma}}^{\overline{\sigma}} \dee \sigma \entropy(\frac{\dee \Ldens{x}_{\sigma}}{\dee \sigma}\vconv \Ldens{z}_{\sigma}) = P(\Ldens{x}_{\overline{\sigma}})-P(\Ldens{x}_{\underline{\sigma}}) $.
Since the family, $\{\Ldens{x}_{\sigma}\}$ is piece-wise linear and ordered by degradation, we can reparameterize 
the GEXIT integrals with the Battacharyya parameter which we denote by $b=\batta(\Ldens{x}_{\sigma})$. Thus
\begin{align*}
 \Big \vert\!\!\int_{\underline{b}}^{\overline{b}} \!\!\!\!\entropy(\frac{\dee \Ldens{x}_{b}}{\dee b}\!\vconv\! (\Ldens{z}_{b} \!-\! \tilde{\Ldens{z}}_b)) 
 \dee b
\Big\vert 
\! \leq \!\frac8{\ln 2}\!\sqrt{\!2d(\Ldens{x}^{\cconv \dr\!-\!1}_b\!,\! \Ldens{x}_\sigma^{\cconv
\dr\!-\!2} \!\cconv\! \tilde{\Ldens{x}}_b)}.
\end{align*}
To see the last inequality, using
(\ref{equ:partiallydegradedcase}), Lemma~\ref{lem:magic} we have
$$
\entropy((\Ldens{x}_{b'}\!-\!\Ldens{x}_{b})\vconv(\Ldens{z}_b\!-\!\tilde{\Ldens{z}}_b))
\le
\frac{8}{\ln(2)} \batta(\Ldens{x}_{b'}\!-\!\Ldens{x}_b) \!\sqrt{2d(\tilde{\Ldens{z}}_b,\Ldens{z}_b)},
$$
where $\Ldens{x}_{b} \prec \Ldens{x}_{b'}$. Since
$\batta(\Ldens{x}_{b'})=b'$ and $\batta(\Ldens{x}_b) = b$, we get
$\frac{\entropy((\Ldens{x}_{b'}-\Ldens{x}_{b})\vconv(\Ldens{z}_b-\tilde{\Ldens{z}}_b))}{b'-b}
\le \frac{8}{\ln(2)}\sqrt{2 d(\tilde{\Ldens{z}}_b,\Ldens{z}_b)},$
which gives us the bound.  The last expression can be further upper
bounded (using (\ref{lem:blmetricregularcconv}), Lemma~\ref{lem:blmetric})
by $\frac8{\ln 2} \sqrt{2d(\Ldens{x}_b, \tilde{\Ldens{x}}_b)} \leq
\frac8{\ln 2} \sqrt{2\delta}$.


\item[]{\em Accounting:} Putting everything together, we have
\begin{align*}
& \underbrace{(2 \Lc-4 w+6)}_{\text{nb. of interior nodes\,}}\, 
\underbrace{(B-4\sqrt{2\delta})}_{\text{sum of GEXIT integrals per tree}}  \\
& -\underbrace{(2\Lc-6w+8)\dl (\dr-1) C}_{\text{contributions of approx. FP channels}}  + \\
& -\underbrace{(7 w-8) \dl (\dr-1)}_{\text{frozen and non FP contributions}}+ \\
& -\underbrace{(2\Lc-6w+8) \dl (\dr-1) \frac8{\ln 2} \sqrt{\delta }}_{\text{correction due to approx. FP nature}} \\
\geq & (2 \Lc+1) (A(\Ldens{x}_{\overline{\sigma}})-A(\Ldens{x}_{\underline{\sigma}})) + D,
\end{align*}
where
\begin{align*}
D = & \underbrace{-(4 w-5) B\!-\!(7 w-8) \dl (\dr-1)}_{\text{$\geq - 11  w(1+\dl \dr)$ since $B \leq 1+\dl \dr$}} \\
& - 4\sqrt{\delta} (2 \Lc +1)[\sqrt{2} + \frac2{\ln 2}\dl (\dr-1)].
\end{align*}
\end{itemize}
\end{enumerate}
Let us derive an upper bound in the same manner.  \begin{itemize}
\item[] {\em Boundary:} For $i \in [-\Lc, -\Lc+w-2] \cup [-w+2,
w-2] \cup [\Lc-w+2, \Lc]$ the GEXIT integrals are at most $1$. This
gives a contribution of $4w-5$.  As usual,  for $i \in [\Lc+1,
\Lc+w-1]$ the GEXIT integral is $0$ and does not contribute to the
area.

\item[] {\em Interior:} Consider the GEXIT integrals for $i \in [-\Lc+w-1, -w+1] \cup [w-1, \Lc-w+1]$.
\begin{itemize}
\item[]{\em Technique:} We use the same procedure as beforehand.
But this time we need a lower bound of the GEXIT integrals of the
leaf nodes.

\item[] {\em Contributions from overall tree:} As before, the overall
contribution of each tree is equal to
$T(\Ldens{x}_{\overline{\sigma}})-T(\Ldens{x}_{\underline{\sigma}})$ plus
an error term of absolute value equal to $4\sqrt{2 \delta}$.

\item[] {\em Contributions from leaves:} The idea is same as before and as before, we will
 consider the computation from the point of view of check nodes. As before, we
split the contribution in two regimes, $[-\Lc, -\Lc+2w-3]
\cup [-w+2, 2w-3] \cup [\Lc-w+2, \Lc+w-1]$ and $[-\Lc+2w-2, -w+1] \cup [2w -2,
\Lc - w+1]$. 
\item[] {\em Contributions from checks in the range $[-\Lc, -\Lc+2w-3]
\cup [-w+2, 2w-3] \cup [\Lc-w+2, \Lc+w-1]$:} Check nodes in this range
might see some frozen channels or channels which are not approximate
FPs. Since we are looking for an upper bound, we set the contribution of 
 such check nodes to be 0. 
\item[]
{\em Contributions from checks in the range $[-\Lc+2w-2, -w+1] \cup
[2w -2, \Lc - w+1]$:} As we discussed before, check nodes in this
range only see channels which are approximate FPs and none of the
channels are frozen. Further, all these GEXIT integrals corresponds
to computation trees whose root $i$ is in the range $[-\Lc+w-1,
-w+1]  \cup [w-1, \Lc-w+1]$. We can, therefore, subtract all their
contributions, which are obtained by arguments similar to those
used in the lower bound. There are $(2 \Lc-6w+8) \dl (\dr-1)$ such
integrals and the contribution for each such integral is at least
$C- \frac8{\ln 2} \sqrt{\delta}$. Here, the last term takes into
account the approximate FP nature of the channels and $C$ was defined
in the arguments for obtaining the lower bound.

\item[]{\em Accounting:} We have
\begin{align*}
& \underbrace{(4w-5)}_{\text{boundary}} + \underbrace{(2 \Lc-4w+6)}_{\text{nb. interior nodes}} 
\underbrace{(B+4\sqrt{2 \delta})}_{\text{total contribution per tree}} + \\
& -\underbrace{(2 \Lc-6w+8)\dl(\dr-1) C}_{\text{contr. of interior check nodes}}+ \\
& +\underbrace{(2\Lc-6w+8) \dl (\dr-1) \frac8{\ln 2} \sqrt{\delta}}_{\text{correction due to approx. FP nature}} \\
\leq & (2 \Lc+1) (A(\Ldens{x}_{\overline{\sigma}})-A(\Ldens{x}_{\underline{\sigma}})) + E,
\end{align*}
where
\begin{align*}
E = & \underbrace{(6 w-7) \dl  (\dr-1)C}_{\text{$\leq 6 w \dl \dr$ since $C \leq \frac{\dr+1}{\dr}$}} + \underbrace{(4w-5)}_{\leq 4w\dl\dr} \\
& + 4\sqrt{\delta}(2 \Lc +1) [\sqrt{2}+\frac2{\ln 2}\dl (\dr-1)].
\end{align*}
\end{itemize}
\end{itemize}
\end{proof}

{\em Proof of Theorem~\ref{thm:existenceintermediateform}}: Rather
than deriving the bound $c(\dl, \dr, \delta, w, \Msat, \Lc)$ for all
values of the parameters,  we are only interested in the behavior
of this bound for values of $\delta$ tending to $0$ and values of
$\Msat$ and $\Lc$ tending to $\infty$. Hence, in the sequel, nothing
is lost by assuming at several spots that $\delta$ is ``sufficiently''
small and $\Msat$ and $\Lc$ are ``sufficiently'' large (consequently $N$ is also sufficiently large). 
This will simplify our arguments significantly.

Let $(\Ldens{c}^*, \Ldens{\x}^*)$ denote the proper one-sided FP
on $[-\Lfp, 0]$ with forced boundary condition which fulfills the
stated conditions for some $\delta>0$ and $2(w-1) \leq \Lc$ and $\Lc+w \leq \Msat
\leq \Lfp$.  We prove the claim in several steps, where in each
step we assert further properties that such a FP has to fulfill.

{\em Constellation is almost flat and not too small ``on the right'':}
Recall that by assumption $\batta(\Ldens{x}^*_{-\Msat}) \geq \xunstab(1)$
so that $\batta(\Ldens{x}^*_{i}) \geq \xunstab(1)$ for $i \in [-\Msat,
0]$.  Using the same reasoning as in the discussion at the end of
Lemma~\ref{lem:degradationandwasserstein}, we can conclude that
there exists an $i^* \in [-\Msat, -\Lc-w]$ such that $D(\Ldens{x}^*_j,
\Ldens{x}^*_k) \leq D(\Ldens{x}^*_{i^*},
\Ldens{x}^*_j) +  D(\Ldens{x}^*_j,
\Ldens{x}^*_k) +  D(\Ldens{x}^*_k,
\Ldens{x}^*_{i^*+\Lc+w})  =  D(\Ldens{x}^*_{i^*}, \Ldens{x}^*_{i^*+\Lc+w}) \leq
\frac{2 (\Lc+w)}{\Msat}$ for all $j\leq k$ and $j,k \in [i^*, i^*+\Lc+w]$.  From part
(\ref{lem:blmetricdegradation}) of
Lemma~\ref{lem:degradationandwasserstein} we conclude that
$d(\Ldens{x}^*_j, \Ldens{x}^*_k) \leq \sqrt{8 (\Lc+w)/\Msat}$ for all $i^* \leq j \leq k \leq i^*+\Lc+w$. 
Clearly, the right-hand side can be made
arbitrarily small by picking $\Msat$ sufficiently larger than $\Lc+w$.

{\em Constellation can be made exactly flat and not too small ``on the right'':} 
Create from $(\Ldens{c}^*, \Ldens{\x}^*)$ the increasing constellation
$(\Ldens{c}^*, \Ldens{\z}^*)$ on $[-\Lfp, 0]$ with free boundary condition in the following way,
\begin{align*} 
\Ldens{z}_i^* = \begin{cases} \Ldens{x}_{i}^*,
& i \in [-\Lfp, i^*+w], \\ 
\Ldens{x}_{i^*+w}^*, & i \geq i^*+w.
\end{cases} \end{align*} 
The graphical interpretation is simple.  We replace the ``almost''
flat part on the right plus the extra part on the  right which might
not be flat with an exactly flat part. To simplify our subsequent
notation we set $\Ldens{x}=\Ldens{x}^*_{i^*+w}$ and from above arguments note that
$\batta(\Ldens{x}) \geq \xunstab(1)$. Hence $\batta(\Ldens{z}^*_{i}) \geq \xunstab(1)$ for
all $i \geq i^*+w$.

{\em Constellation is approximate FP}: 
Note that by going from $\Ldens{\x}$ to $\Ldens{\z}$ no component in $[-\Lfp, i^*+\Lc+w]$ is changed by
more than a distance $\kappa=\sqrt{8 (\Lc+w)/\Msat}$.  Therefore, if we run
DE on the modified components it is clear that in this range the  output must still be close to the
original output. More precisely, we have for every $i  \in [-\Lfp, i^*+\Lc+1]$
\begin{align*}
& d(\Ldens{z}_i^*, \Ldens{c}^*\vconv g(\Ldens{z}_{i-w+1}^*,\dots,\Ldens{z}_{i+w-1}^*)) \\
& \leq
 d(\Ldens{z}_i^*, \Ldens{x}_i^*)+
 d(\Ldens{x}_i^*,  \Ldens{c}^* \vconv g(\Ldens{z}_{i-w+1}^*,\dots,\Ldens{z}_{i+w-1}^*)) \\
& \leq \!\kappa\!+\! d(\Ldens{c}^*\!\vconv\! g(\Ldens{x}_{i\!-\!w\!+\!1}^*,\dots,\Ldens{x}_{i\!+\!w\!-\!1}^*), \Ldens{c}^*\!\vconv \!g(\Ldens{z}_{i\!-\!w\!+\!1}^*,\dots,\Ldens{z}_{i\!+\!w\!-\!1}^*)) \\
& \leq \kappa+ 2 (\dl-1) (\dr-1)\kappa,
\end{align*}
where to get the penultimate inequality we first replace $\Ldens{x}_i^*$ by $\Ldens{c}^* \vconv
g(\Ldens{x}_{i-w+1}^*,\dots,\Ldens{x}_{i+w-1}^*)$, since $\Ldens{\x}^*$ is a true FP, and then 
to obtain the  last inequality we apply Lemma~\ref{lem:sensitivity}.
Since $\kappa$ can be made arbitrarily small by choosing $\Msat$
sufficiently large, this verifies the approximate FP nature for $i
\in [-\Lfp, i^*+\Lc+1]$.  Let us now focus on $i  \in [i^*+\Lc+2, 0]$.
Note that since $\Lc \geq 2 (w-1)$, we can use the above argument
in particular for $i=i^*+2w-1$. For this choice of $i$ all involved
densities, $\Ldens{z}_{i-w+1}^*, \dots, \Ldens{z}^*_{i+w-1}$, are equal to $\Ldens{x}$.  Therefore, the previous argument
shows that 
\begin{align}\label{eq:xisapproxFP} 
d(\Ldens{x}, \Ldens{c} \vconv g(\Ldens{x}, \dots,
\Ldens{x})) \leq  \kappa+ 2(\dl-1) (\dr-1)\kappa. 
\end{align}
But for $i \geq
i^*+w$ all components of $\Ldens{\z}^*$ are equal to $\Ldens{x}$
and so the approximate FP nature of $\Ldens{\z}^*$  is also verified
for $i \geq i^*+2w-1$. Since $ i^*+2w-1 \leq i^*+\Lc+1$, we
 conclude that $\Ldens{\z}^*$ is an approximate FP.

{\em From FP to FP family:} From the approximate FP $(\Ldens{c}^*,
\Ldens{\z}^*)$ on $[-\Lfp, 0]$ we create the approximate FP family
$\{\Ldens{c}^*_\sigma,
\Ldens{\z}^*_\sigma\}_{\underline{\sigma}=0}^{\overline{\sigma}=\ih^*}$ on
$[-\Lc, 0]$ as described in Definition~\ref{def:partialFPfamily}.

{\em Computing GEXIT integral --
\protect{Definition~\ref{def:gexitintegralbasic}}:} Using the basic
definition of the GEXIT functional in
Definition~\ref{def:gexitintegralbasic} we conclude that the GEXIT
integral associated to $\{\Ldens{c}^*_\sigma,
\Ldens{\z}^*_\sigma\}_{\underline{\sigma}=0}^{\overline{\sigma}=\ih^*}$, $A(\{\Ldens{c}^*_\sigma, \Ldens{\z}^*_\sigma\}_{\underline{\sigma}}^{\overline{\sigma}})$ is
$0$ since the channel remains constant throughout the interpolation.

{\em Computing GEXIT integral --
\protect{Theorem~\ref{the:areatheoremforapproximatefpfamilypartial}}:} We
now compute the GEXIT integral associated to $\{\Ldens{c}^*_\sigma,
\Ldens{\z}^*_\sigma\}_{\underline{\sigma}=0}^{\overline{\sigma}=\ih^*}$
by first applying
Lemma~\ref{lem:interpolationyieldsapproximatefpfamilypartial} and
then Theorem~\ref{the:areatheoremforapproximatefpfamilypartial}.

More precisely, from the previous arguments we satisfy all the hypotheses of
Lemma~\ref{lem:interpolationyieldsapproximatefpfamilypartial}. This allows
us to conclude that the FP family constructed above is 
$\frac{2(\dl-1)(\dr-1)}{w}+\delta$ approximate FP
(cf.  \eqref{equ:partialddistance}) if $\Msat$ is chosen sufficiently large.
Furthermore, since the starting
($\Ldens{z}_{\sigma=\ih^*, i}^* = \Ldens{x}$ for all sections $i\in [-L,0]$) and ending
constellations ($\Ldens{z}_{\sigma=0,i}^* = \Delta_{+\infty}$ for all
sections $i\in [-L,0]$) are flat, we satisfy all the hypotheses of
Theorem~\ref{the:areatheoremforapproximatefpfamilypartial} from which we 
conclude that the GEXIT integral is upper bounded by $A(\Ldens{x}) +  
\cwdldrLdeltanew$.\footnote{Note that $A(\Delta_{+\infty})=0$.}

{\em Flat region has entropy not much smaller than $\frac{\dl}{\dr}$:} 
From part (\ref{lem:propertyofh(x)seven}) Lemma~\ref{lem:propertyofh(x)} we get 
\begin{align*}
\xunstab^2(1) \geq (\dr\!-\!1)^{-2(\frac{\dl-1}{\dl-2})} \geq
(\dr\!-\!1)^{-3} + \Bigl(\frac34\Bigr)^{\frac{\dl-1}{2}}, \end{align*}
where in the last step we have used condition (\ref{equ:admissiblesix})
in Definition~\ref{def:admissible}.  We conclude that
\begin{align}\label{equ:lowerbound}
\entropy(\Ldens{x}) \stackrel{\text{Lem.~\ref{lem:entropyvsbatta}}}{\geq} \batta^2(\Ldens{x})
\geq \xunstab^2(1) 
\geq (\dr\!-\!1)^{-3} + \Bigl(\frac34\Bigr)^{\frac{\dl-1}{2}}.
\end{align}
We now proceed by contradiction.  
Let us assume
that $\entropy(\Ldens{x}) \leq \frac{\dl}{\dr} - \dl e^{-4
(\dr-1)(\frac{2 \dl}{11 e \dr})^{\frac43}}- \frac{1}{\dr}$.
As we just discussed,
\begin{align*}
d(\Ldens{x}, \Ldens{c}^*\!\vconv \!(\Ldens{x}^{\cconv \dr-1})^{\vconv \dl-1}) \!\leq\!
\frac{2(\dl\!-\!1)(\dr\!-\!1)}{w} \!+\! \delta 
\!\leq\! (\frac{\ln(2) \dl}{16 \sqrt{2} \dr})^2.
\end{align*}
In the last step we assumed without loss of generality that $\delta$
is chosen sufficiently small. The inequality then follows from the
 condition
(\ref{equ:admissiblefour}) in Definition~\ref{def:admissible}.
This, together with (\ref{equ:lowerbound}), guarantees that we
satisfy the hypothesis of (the Negativity)
Lemma~\ref{lem:asymptoticnegativity}.  Hence we conclude
that $A(\Ldens{x}) \leq - \frac{1}{\dr}$. From  condition
(\ref{equ:admissiblefive}) in Definition~\ref{def:admissible}
$4(\sqrt{2}+\frac{2}{\ln 2}\dl (\dr-1))  \sqrt{\frac{2 (\dl-1)(\dr-1)}{w}}
< \frac{1}{\dr}$.
Hence for a sufficiently small $\delta$
and a sufficiently large $\Lc$, this leads to the conclusion that
the GEXIT integral $A(\{\Ldens{c}^*_\sigma, \Ldens{\z}^*_\sigma\}_{\underline{\sigma}}^{\overline{\sigma}}) \leq A(\Ldens{x}) +  
\cwdldrLdeltanew < 0$, a contradiction
 to the previous computation.  As a consequence, we must
have
\begin{align}\label{eq:hstarislarge}
\ent^* = \entropy(\Ldens{c}^*) \geq \entropy(\Ldens{x}) \geq
\frac{\dl}{\dr}\!-\!\dl e^{-4 (\dr-1)(\frac{2 \dl}{11 e \dr})^{\frac43}}
\!-\!\frac{1}{\dr}.
\end{align}

{\em The flat region is close to} $\xBPdens$: We will now show that
$\Ldens{x}$ is close to $\xBPdens(\Ldens{c}^*)$, the BP FP when
transmitting over the channel $\Ldens{c}^*$ using the underlying $(\dl, \dr)$-regular ensemble. In the sequel we
will denote $\xBPdens(\Ldens{c}^*)$ by $\xBPdens$. 
To do this, we will first bound the Wasserstein distance between $\xBPdens$
 and $\tilde{\Ldens{x}}$, where $\tilde{\Ldens{x}}$ is defined to be equal to
$\Ldens{x}^*_{i^*+L+w}$.  Thus to bound the distance between $\Ldens{x}$
and $\xBPdens$ we bound the distances $d(\Ldens{x}, \tilde{\Ldens{x}})$ and
$d(\tilde{\Ldens{x}}, \xBPdens)$. Note from the previous part we have that $d(\Ldens{x},
\tilde{\Ldens{x}}) = d(\Ldens{x}^*_{i^*+w}, \Ldens{x}^*_{i^*+L+w}) \leq \kappa$ and hence the distance between $\Ldens{x}$ and $\tilde{\Ldens{x}}$ can be made arbitrarily small by taking $\Msat$
sufficiently large. Let us now bound $d(\tilde{\Ldens{x}}, \xBPdens)$. 
First, we show that $d(\tilde{\Ldens{x}}, \Ldens{c}^*\vconv
g(\tilde{\Ldens{x}},\dots,\tilde{\Ldens{x}}))$ can be made arbitrarily small. Indeed, 
\begin{align}\label{equ:tildexapproximateFP}
d(\tilde{\Ldens{x}}, \Ldens{c}^*\vconv
g(\tilde{\Ldens{x}},\dots,\tilde{\Ldens{x}}))  \leq &d(\tilde{\Ldens{x}},\Ldens{x}) +  d(\Ldens{x}, \Ldens{c}^*\!\vconv\!
g(\Ldens{x},\dots,\Ldens{x})) \nonumber \\
& + d(\Ldens{c}^*\!\vconv\!
g(\Ldens{x},\dots,\Ldens{x}), \Ldens{c}^*\!\vconv\!
g(\tilde{\Ldens{x}},\dots,\tilde{\Ldens{x}})) \nonumber \\
& \leq \kappa \!+\! \kappa \!+\! 4(\dl\!-\!1)(\dr\!-\!1)\kappa,
\end{align}
where to get the last inequality we have used the approximate FP nature
of $\Ldens{x}$ (cf. \eqref{eq:xisapproxFP}) and the (sensitivity) Lemma~\ref{lem:sensitivity}. Since $\kappa$ can be made arbitrarily
small, we can make the distance $d(\tilde{\Ldens{x}},  \Ldens{c}^*\vconv
g(\tilde{\Ldens{x}},\dots,\tilde{\Ldens{x}}))$ as small as desired. 

Run forward DE, with the channel $\Ldens{c}^*$, starting from
$\tilde{\Ldens{x}}_0 = \tilde{\Ldens{x}}$, $\xBPdens_0=\xBPdens$, and $\Ldens{w}_0
= \Delta_0$, respectively.  Let
$\tilde{\Ldens{x}}_{\ell}=T_{\Ldens{c}^*}(\tilde{\Ldens{x}}_{\ell-1})$, 
$\xBPdens_{\ell}=T_{\Ldens{c}^*}(\xBPdens_{\ell-1})=\xBPdens$,
and $\Ldens{w}_{\ell}=T_{\Ldens{c}^*}(\Ldens{w}_{\ell-1})$,  $\ell
\geq 1$. Recall that $T_{\Ldens{c}^*}(\cdot)$ is the DE operator for the $(\dl,
\dr)$-regular ensemble when transmitting over the channel $\Ldens{c}^*$. We will choose the value of  $\ell$ shortly.  Then
\begin{align*}
& d(\tilde{\Ldens{x}}, \xBPdens) 
\leq  d(\tilde{\Ldens{x}}_0, \tilde{\Ldens{x}}_{\ell}) + d(\tilde{\Ldens{x}}_{\ell}, \Ldens{w}_{\ell})+d(\Ldens{w}_{\ell}, \xBPdens) \\
\leq &\sum_{j=0}^{\ell-1} d(\tilde{\Ldens{x}}_j, \tilde{\Ldens{x}}_{j+1}) \!+ \!
2 \sqrt{\batta(\Ldens{w}_{\ell})\!-\!\batta(\tilde{\Ldens{x}}_{\ell})} \!+\! 2 \sqrt{\batta(\Ldens{w}_{\ell})-\batta(\xBPdens)}).
\end{align*}
In the last step we use that $\Ldens{w}_\ell \succ \tilde{\Ldens{x}}_{\ell}$,
since $\Ldens{w}_0 = \Delta_0 \succ \tilde{\Ldens{x}}_0$ and DE preserves degradation.
Similarly, we use $\Ldens{w}_\ell \succ \xBPdens$. Therefore we can
upper bound the Wasserstein distance in terms of the difference of the
respective Battacharyya constants according to (\ref{lem:blmetricentropy}),
Lemma~\ref{lem:degradationandwasserstein}.

Choose $\ell = \lfloor\frac{L}{w-1}\rfloor$. We then claim that $\Ldens{x}
\prec \tilde{\Ldens{x}}_{j}$ for all $0\leq j \leq \ell$.  Let us
prove this claim immediately. From construction, we have
$\Ldens{x}=\Ldens{x}^*_{i^*+w} \prec \Ldens{x}^*_{i^*+L+w} =
\tilde{\Ldens{x}}_0$. Next, we claim that
$\tilde{\Ldens{x}}_j \succ \Ldens{x}^*_{i^*+L+1-(w-1)(j-1)}$ for $1\leq j\leq
\ell$.  Before we prove this claim, we apply it immediately 
to conclude that $$ \tilde{\Ldens{x}}_j \succ \Ldens{x}^*_{i^*+L+1-(w-1)(j-1)}\stackrel{
j\leq \ell \leq\frac{L}{w-1}}{\succ}
 \Ldens{x}^*_{i^*+ w} = \Ldens{x}.$$ To prove the intermediate claim 
we argue inductively that 
\begin{align*}
\tilde{\Ldens{x}}_j & = \Ldens{c}^*\vconv g(\tilde{\Ldens{x}}_{j-1},\dots,\tilde{\Ldens{x}}_{j-1}) \\
& \succ \Ldens{c}^*\vconv g(\Ldens{x}^*_{i^*+L+1-(w-1)j},\dots,\Ldens{x}^*_{i^*+L+1-(w-1)(j-2)}) \\
& = \Ldens{x}^*_{i^*+L+1-(w-1)(j-1)}.
\end{align*}
The induction is completed by verifying that $\tilde{\Ldens{x}}_1 \succ \Ldens{x}^*_{i^*+L+1}$.
Indeed, from the monotonicity of the spatial FP, $\Ldens{\x}^*$, we get
\begin{align}\label{eq:flatregionclosetoxBPone}
\Ldens{x}^*_{i^*+L+1} & = \Ldens{c}^*\vconv g(\Ldens{x}^*_{i^*+L-w+2},\dots,\Ldens{x}^*_{i^*+L+w}) \nonumber \\
& \stackrel{\Ldens{x}^*_{i^*+L+w}=\tilde{\Ldens{x}}_0}{\prec} \Ldens{c}^*\vconv g(\tilde{\Ldens{x}}_0,\dots,\tilde{\Ldens{x}}_0) = \tilde{\Ldens{x}}_1.
\end{align}

Let us now bound the distance $d(\tilde{\Ldens{x}}_j, \tilde{\Ldens{x}}_{j+1})$ for $1\leq j\leq \ell$. 
Since these elements are derived by DE we can
use our bounds on how the Wasserstein distance behaves under DE (cf. (\ref{lem:blmetricregularDE}), Lemma~\ref{lem:blmetric}) to
conclude that $d(\tilde{\Ldens{x}}_j, \tilde{\Ldens{x}}_{j+1}) \leq \alpha
d(\tilde{\Ldens{x}}_{j-1}, \tilde{\Ldens{x}_{j}})$, where $\alpha= 2(\dl-1)(\dr-1)(1 -
\batta^2(\Ldens{x}))^{\frac{\dr-2}2}$. To obtain $\alpha$ we have used $\tilde{\Ldens{x}}_{j} \succ \Ldens{x}$ for all $0\leq j \leq \ell$  to get $\min\{\batta(\tilde{\Ldens{x}}_{j-1}),
\batta(\tilde{\Ldens{x}}_{j})\} \geq \batta(\Ldens{x})$.
Continuing with above inequality, it is not hard to see that we get
$d(\tilde{\Ldens{x}}_j, \tilde{\Ldens{x}}_{j+1}) \leq \alpha^j d(\tilde{\Ldens{x}}_{0}, \tilde{\Ldens{x}_{1}})$.
This gives a bound of
$$
\sum_{j=0}^{\ell-1} d(\tilde{\Ldens{x}}_j, \tilde{\Ldens{x}}_{j+1}) \leq d(\tilde{\Ldens{x}}_0,
\tilde{\Ldens{x}}_{1}) \frac{\alpha^{\ell}-1}{\alpha-1}\leq d(\tilde{\Ldens{x}}_0,
\tilde{\Ldens{x}}_{1}) \frac{1}{1-\alpha},$$
where in the last inequality we use $\batta(\Ldens{x}) \stackrel{\text{Lemma}~\ref{lem:entropyvsbatta}}{\geq} \entropy(\Ldens{x}) \geq \frac{\dl}{\dr}-\dl e^{-4 (\dr-1)(\frac{2 \dl}{11 e \dr})^{\frac43}}
-\frac{1}{\dr}$ combined with the condition (\ref{equ:admissiblenine}) in Definition~\ref{def:admissible} to get $\alpha < 1$.
From \eqref{equ:tildexapproximateFP} 
we know that we can make $d(\tilde{\Ldens{x}}_0, \tilde{\Ldens{x}}_{1})$ as small
as we want by choosing $\Msat$ sufficiently large.

Let us now bound the two
terms containing Battacharyya parameters. Note that in each iteration
the distance of the respective Battacharyya
constants decreases by a factor of at least
$\beta=\batta(\Ldens{c}^*)(\dl\!-\!1)(\dr\!-\!1)
(1\!-\!\min\{\batta(\Ldens{x}), \batta(\Ldens{x}_{\BPsmall})\}^2)^{\dr-2}$.
Indeed, from Lemma~\ref{lem:dbbound},
\begin{align*}
&\batta(\Ldens{w}_{\ell})\!-\!\batta(\tilde{\Ldens{x}}_{\ell}) \leq  \Big(\!\batta(\Ldens{c}^*)(\dl\!-\!1)(\dr\!-\!1) (1\!-\!\batta(\Ldens{x})^2)^{\dr-2}\Big)^{\ell}  \\ 
&\batta(\Ldens{w}_{\ell})\!-\!\batta(\xBPdens) \!\leq\!
\Big(\!\batta(\Ldens{c}^*)(\dl\!-\!1)(\dr\!-\!1) (1\!-\!\batta(\Ldens{x}_{\BPsmall})^2)^{\dr-2}\Big)^{\ell}.
\end{align*}
For the first inequality we again use $\batta(\tilde{\Ldens{x}}_{j}) \geq \batta(\Ldens{x})$ for all $0\leq j \leq \ell$.
Above we have also used $\batta(\Ldens{w}_0) - \batta(\tilde{\Ldens{x}}_0)
= \batta(\Delta_0) - \batta(\tilde{\Ldens{x}}) \leq 1$ and $\batta(\Ldens{w}_0)
- \batta(\xBPdens) = \batta(\Delta_0) - \batta(\xBPdens) \leq 1$. 
We now have
\begin{align*}
\batta(\Ldens{c}^*)(\dl\!-\!1)(\dr\!-\!1) (1\!-\!\batta(\Ldens{x})^2)^{\dr-2} < 1, \\
\batta(\Ldens{c}^*)(\dl\!-\!1)(\dr\!-\!1) (1\!-\!\batta(\xBPdens)^2)^{\dr-2} < 1.
\end{align*}
For the first inequality we use condition (\ref{equ:admissiblenine}) in Definition~\ref{def:admissible} combined with  
$\batta(\Ldens{x}) \stackrel{\text{Lemma}~\ref{lem:entropyvsbatta}}{\geq} \entropy(\Ldens{x}) \geq \frac{\dl}{\dr}-\dl e^{-4 (\dr-1)(\frac{2 \dl}{11 e \dr})^{\frac43}}
-\frac{1}{\dr}$.
For the second inequality we use condition (\ref{equ:admissibletwo}) in Definition~\ref{def:admissible}
combined with $\ih^*\geq \frac{\dl}{\dr}-\dl e^{-4 (\dr-1)(\frac{2 \dl}{11 e \dr})^{\frac43}}
-\frac{1}{\dr} \geq \entLE$ and Lemma~\ref{lem:continuityforlargeentropy}. 

Therefore we can bound the sum of the two Battacharyya terms by $4 \beta^{\ell/2}$ with 
$\beta=\batta(\Ldens{c}^*)(\dl\!-\!1)(\dr\!-\!1)
(1\!-\!\min\{\batta(\Ldens{x}), \batta(\Ldens{x}_{\BPsmall})\}^2)^{\dr-2} < 1$.

Putting everything together we conclude that by choosing $\Lc, \Msat$ sufficiently large $d(\Ldens{x}, \xBPdens)$ can be
made as small as desired.

$\ih^*$ {\em is close to} $\ih^A$: 
From Theorem~\ref{the:areatheoremforapproximatefpfamilypartial} we have
$$
\Big\vert\frac{A(\!\{\Ldens{c}^*_\sigma, \Ldens{\z}^*_\sigma\}_{\underline{\sigma}}^{\overline{\sigma}})}{2 \Lc+1} - A(\Ldens{x}) \Big\vert \leq \cwdldrLdeltanew. 
$$
From above arguments we have $A(\!\{\Ldens{c}^*_\sigma, \Ldens{\z}^*_\sigma\}_{\underline{\sigma}}^{\overline{\sigma}})=0$ hence
$$
\vert A(\Ldens{x}) \vert \leq \cwdldrLdeltanew. 
$$
Using the formula for $A(\cdot)$ given in Lemma~\ref{lem:areaunderBPGEXIT} and properties (\ref{lem:blmetricregularcconv}) and (\ref{lem:blmetricwasserboundsbatta}) given in Lemma~\ref{lem:blmetric} we have
\begin{align*}
& \vert A(\xBPdens) - A(\Ldens{x}) \vert \leq 2\sqrt{2}\sqrt{d(\Ldens{x}, \xBPdens)} \\
& \times \Big(1 + \sqrt{\dr}(\dl - 1- \frac{\dl}{\dr}) + \sqrt{\dr-1}(\dl-1)\Big). 
\end{align*}
Recall that $\xBPdens=\xBPdens(\Ldens{c}^*)$. Combining, we get
\begin{align*}
& \vert A(\xBPdens) \vert \leq \cwdldrLdeltanew  \\ 
& + 2\sqrt{2}\sqrt{\delta} \Big(1 + \sqrt{\dr}(\dl - 1- \frac{\dl}{\dr}) + \sqrt{\dr-1}(\dl-1)\Big). 
\end{align*}

Further the BP GEXIT value for all channels between $\ih^*$ and $\ih^A$ is lower
bounded by $\frac1{2(\dr-1)^3}$. To show this we first note that
from condition (\ref{equ:admissibletwo}) and (\ref{equ:admissibleseven}) in Definition~\ref{def:admissible} 
we satisfy the hypotheses
of Lemma~\ref{lem:areathresholdapproachesshannon}.
Hence from Lemma~\ref{lem:areathresholdapproachesshannon}  we have $\ih^A\geq \entLE$.
Also, from \eqref{eq:hstarislarge} we have $\ih^*\geq \entLE$.

Then for any $\ih  \geq \min\{\ih^A, \ih^*\}$ we have
 $\batta(\Ldens{x}_{\ih}) \geq \xunstab(1)$ (cf.  Lemma~\ref{lem:continuityforlargeentropy}). Thus we conclude that $\batta(\Ldens{x}_{\ih}) \geq \xunstab(1) \geq
\frac1{(\dr-1)^{3/2}}$ for any $\ih  \geq \min\{\ih^A, \ih^*\}$. 
Denoting $\Ldens{y}_{\ih} = \Ldens{x}^{\cconv \dr-1}_{\ih}$ we have,
\begin{align*}
G(\Ldens{c}_\ih,& \Ldens{y}_{\ih}^{\vconv \dl}) \stackrel{\text{concavity of GEXIT}}{\geq}
2\perr(\Ldens{y}_{\ih}^{\vconv \dl}) \\ 
& \stackrel{\substack{\text{extremes of info.,} \\ \text{mult. prop. of Batta}}}{\geq} 1 - \sqrt{1 - (\batta(\Ldens{y}_{\ih}))^{2\dl}} \\
& \stackrel{\text{(a)}}{\geq} 1 - \sqrt{1 - \frac1{(\dr-1)^3}} \geq \frac1{2(\dr-1)^3}.
\end{align*}
To obtain (a) we use $\batta(\Ldens{x}_{\ih}) =\batta(\Ldens{c}_{\ih}) (\batta(\Ldens{y}_{\ih}))^{\dl-1}$,
since $\Ldens{c}_{\ih}$ and $\Ldens{x}_{\ih}$ form a FP pair. This implies that
$(\batta(\Ldens{y}_{\ih}))^{2\dl} =
(\frac{\batta(\Ldens{x}_{\ih})}{\batta(\Ldens{c}_{\ih})})^{\frac{2\dl}{\dl-1}} \geq
(\batta(\Ldens{x}_{\ih}))^{\frac{2\dl}{\dl-1}} \geq
(\xunstab(1))^{\frac{2\dl}{\dl-1}} \stackrel{\text{Lemma}~\ref{lem:propertyofh(x)}}{\geq} (\dr-1)^{\frac{-2\dl}{\dl-2}} \geq
(\dr-1)^{-3}.$ The last inequality follows since condition
(\ref{equ:admissiblesix}) in Definition~\ref{def:admissible} implies that $\dl \geq 6$.
This implies
\begin{align*}
\Big\vert \int_{\ih^A}^{\ih^*} G(\Ldens{c}_{\ih}, \Ldens{y}_{\ih}^{\vconv \dl}) \dee \ih\Big\vert
\geq \vert \ih^* - \ih^A \vert \frac1{2(\dr-1)^3}.
\end{align*}
Since $\ih^*$ and $\ent^A$ are both greater than $\entLE$, from Lemma~\ref{lem:areaunderBPGEXIT} we have 
$$
\Big\vert \int_{\ih^A}^{\ih^*} G(\Ldens{c}_{\ih}, \Ldens{y}_{\ih}^{\vconv \dl}) \dee \ih\Big\vert
=\vert A(\xBPdens) - A(\xBPdens_{\ent^A})\vert = \vert A(\xBPdens) \vert,
$$
where the last equality follows since $A(\xBPdens_{\ent^A})=0$ (cf. Lemma~\ref{lem:areathresholdapproachesshannon}).

Putting everything together we get
\begin{align*}
& \vert \ih^* - \ih^A \vert \!\leq\! 2(\dr\!-\!1)^3 \Big(\cwdldrLdeltanew  \\ 
& + 2\sqrt{2}\sqrt{\delta} (1 + \sqrt{\dr}(\dl - 1- \frac{\dl}{\dr}) + \sqrt{\dr-1}(\dl-1))\Big). 
\end{align*}

\section{Existence of FP -- Theorem~\ref{thm:existencebasicform}}\label{sec:existencebasicform}
\begin{IEEEproof} Before proceeding to the main part of the proof,
let us show that if we assume that there exists a proper FP on
$[-\Lfp, 0]$, with forced boundary condition on the right and
$\Delta_{+\infty}$ on the left ($i < -\Lfp$) and with Battacharyya
parameter of the constellation (cf. Definition~\ref{def:entropy})
equal to $\xunstab(1)/2$, then the desired properties (i) and (ii)
mentioned in the statement of the theorem follow.

{\em Constellation is close to $\Delta_{+\infty}$ ``on the left''}:
Let $\Lfp_1$ be the largest integer so that for all $i < -\Lfp+\Lfp_1$,
$\batta(\Ldens{x}_i) \leq \delta$.  We have a proper FP and $w >
2\dl^3\dr^2$ (because $w$ is by assumption admissible in the sense of condition 
(\ref{equ:admissiblethree}) in Definition~\ref{def:admissible}).  Hence by applying (the
Transition Length) Lemma~\ref{lem:transitionlength} we conclude
that the number of sections with Battacharyya parameter bounded
between $\delta$ and $\xunstab(1)$ is at most $w c(\dl, \dr)/\delta$,
where $c(\dl, \dr)$ is the constant defined in
Lemma~\ref{lem:transitionlength}.  Since the Battacharyya parameter
of the constellation is $\xunstab(1)/2$, we have
\begin{align*}
(\Lfp+1) \frac{\xunstab(1)}2  \geq 
(\Lfp+1 - \Lfp_1 -
w c(\dl, \dr)/\delta)\xunstab(1). 
\end{align*}
This implies that 
$\Lfp_1 \geq (\Lfp+1)\Big(\frac12 - \frac{wc(\dl, \dr)}{(\Lfp+1)\delta}\Big)$.
Using property (\ref{lem:blmetricbattaboundswasser}) of (the
Wasserstein metric) Lemma~\ref{lem:blmetric}, we conclude that
for all $i< -\Lfp+\Lfp_1$, $d(\Ldens{x}_i, \Delta_{+\infty}) \leq
\delta$.

{\em Constellation is not too small ``on the right''}: 
Let $\Lfp_1$ be as defined previously. Again, since the Battacharyya parameter
of the constellation is equal to $\xunstab(1)/2$ we have
\begin{align*}
(\Lfp+1)\frac{\xunstab(1)}{2} \leq \Lfp_1\delta  + (\Lfp\!+\!1 \!-\! \Lfp_1), 
\end{align*}
where on the rhs above we have replaced the sections with value greater than $\delta$
 by the maximum value of $1$.

This implies that 
$\Lfp_1 \leq (\Lfp+1) \frac{1- \frac{\xunstab(1)}2}{1 - \delta}$.
Thus if we define $\Lfp_2$ as the number of sections with Battacharyya
parameter at least equal to $\xunstab(1)$,
we must have
\begin{align*} \Lfp_2 & \geq (\Lfp+1) - \Lfp_1 - wc(\dl, \dr)/\delta \\ 
& \geq
(\Lfp+1)\Big(\frac{\xunstab(1)}4  - \frac{wc(\dl,
\dr)}{\delta(\Lfp+1)}\Big), 
\end{align*} 
where we used $\delta\leq \frac{\xunstab(1)}{4}$ to obtain the above
expression.

It remains to show the existence of the proper FP itself, with
Battacharyya parameter of the constellation equal to $\xunstab(1)/2$.
We use the Schauder FP theorem in a strong form recently proved by
Cauty \cite{Cau01}: This theorem states that every continuous map
$f$ from a convex compact subset $S$ of a topological vector space
to itself has a FP.

Recall that a {\em topological vector space} ${\mathcal S}$ is a
vector space over a topological field $\field$ (most often the real or
complex numbers with their standard topologies) which is endowed
with a topology such that vector addition ${\mathcal S} \times
{\mathcal S} \rightarrow {\mathcal S}$ and scalar multiplication
$\field \times {\mathcal S} \rightarrow {\mathcal S}$ are continuous
functions.

Let ${\mathcal S} = L_1[0, 1]$ (where $L_1$ denotes the $L_1$ norm).
Note that ${\mathcal S}$ is a real normed vector space and hence a
topological vector space.  Let ${\mathcal P}$ denote the space of
probability measures on $[0, 1]$  endowed with the Wasserstein
metric.  Note that ${\mathcal P} \subset {\mathcal S}$, where we
represent elements of ${\mathcal P}$ by their cumulative distribution functions.
Note that the topology on ${\mathcal P}$ induced by ${\mathcal S}$ coincides with
our choice (cf. second alternative definition in part
(\ref{lem:alternative}) of Lemma~\ref{lem:blmetric}).  Also, on
${\mathcal P}$ the topology induced by the Wasserstein metric is
equivalent to the weak topology.  Since $[0, 1]$ is a complete
separable metric space, so is ${\mathcal P}$, see \cite[Theorem
6.18]{Villani09}.  Since $[0, 1]$ is compact, so is ${\mathcal P}$,
see \cite[Remark 6.19]{Villani09}.

A Cartesian product of a family of topological vector spaces, when
endowed with the product topology, is a topological vector space.
Hence,  ${\mathcal S}^{\Lfp+1}$, endowed with the product topology,
is a topological vector space.

Let $S$ be the subset 
\begin{align*} S  = & \{ \absDdist{\x} \in
{\mathcal S}^{\Lfp+1}: \text{$\absDdist{x}_i$ is a $|D|$-distribution},\;
i \in [-\Lfp,0]; \\ & \batta(\absDdist{\x}) = \xunstab(1)/2;\,
\absDdist{x}_{-\Lfp} \prec \absDdist{x}_{-\Lfp+1} \prec \dots\prec
\absDdist{x}_0\}.  
\end{align*}
{\em Discussion:} As we discussed above, we think of the elements
of ${\mathcal P}$ as cumulative distribution functions. In particular,
these are the cdfs in the so called $|D|$ domain. In the sequel,
rather than only referring to cdfs it will often be more convenient
to write down the $|D|$ distributions $\absDdens{x}$ or $D$
distributions $\Ddens{x}$, directly.

{\em $S$ is non-empty}: Setting all elements of $\absDdens{\x}$
equal to $\xunstab(1)/2 \Delta_0+(1-\xunstab(1)/2) \Delta_1$ gives
an element in this space.

{\em $S$ is convex}: Let  $\Ldens{\x}, \Ldens{\y} \in S$ with $|D|$-distributions given by $\absDdens{\x}$ and $\absDdens{\y}$ respectively. 
Let $\absDdens{\vc}=\beta \absDdens{\x} + (1-\beta)\absDdens{\y}$
for some $\beta \in (0,1)$. 
Since $\batta(\cdot)$ is a linear operator, we see that
$$ 
\batta(\absDdens{\vc}) = \beta \batta(\absDdens{\x}) + (1-\beta)\batta(\absDdens{\y}) = \xunstab(1)/2.
$$ 
Also, using \eqref{equ:degradationcdfs}, we see that
$\absDdens{\vc}_{i-1}\prec \absDdens{\vc}_{i}$ for all $i \in [-\Lfp+1, 0]$.
Hence $\beta \Ldens{\x} + (1-\beta) \Ldens{\y} \in S$. 

{\em $S$ is closed}:
Consider a sequence $\{\absDdens{\x}^{(\ell)} \}_{\ell=1}^{\infty}$
of elements of $S$ and assume that this sequence converges in the
Wasserstein metric to a limit, call it $\absDdens{\x}^{(\infty)}$.  We
need to show that $\absDdens{\x}^{(\infty)} \in S$, i.e., we claim
that $S$ is closed.  In this respect, recall from our discussion
above that $S \subseteq {\mathcal P}^{\Lfp+1}$ and that on ${\mathcal
P}^{\Lfp+1}$ the topology induced by the Wasserstein metric is the
weak topology.

From Lemma $4.25$ in \cite{RiU08} we know that each component of
$\absDdens{\x}^{(\infty)}$ is a symmetric $|D|$ distribution. It
therefore remains to shows that (i)
$\batta(\absDdens{\x}^{(\infty)})=\xunstab(1)/2$, and (ii)
$\absDdens{x}^{(\infty)}_{i-1} \prec \absDdens{x}^{(\infty)}_{i}$
for all $i\in [-\Lfp+1, 0]$. Both claims follow from the fact that
we can encode the above properties in terms of continuous functions
and that continuous functions preserve the properties under limits.

Let us show this in detail. We begin with (i).  Consider the sequence
$\{\absDdens{\x}^{(\ell)}\}$. We have
\begin{align*}
\batta(\absDdens{\x}^{(\ell)})  & \!=\! \frac1{\Lfp+1}\sum_{j=-N}^0 \batta(\absDdens{x}^{(\ell)}_j), \\ 
\batta(\absDdens{x}^{(\ell)}_j) & \! =\! \int_{0}^{1} \absDdens{x}^{(\ell)}_j(y) \sqrt{1-y^2} \,\text{d}y.
\end{align*}
Now note that $\sqrt{1-y^2}$ is a bounded and continuous function
on $[0, 1]$.  Therefore, (weak) convergence of $\{\absDdens{\x}^{(\ell)}\}$
to $\absDdens{\x}^{(\infty)}$ implies (weak) convergence of
$\batta(\absDdens{\x}^{(\ell)})$ to
$\batta(\absDdens{\x}^{(\infty)})=\xunstab(1)/2$.

Let us show (ii). From \eqref{equ:degradationcdfs}, 
$\absDdens{x}^{(\ell)}_{j-1} \prec \absDdens{x}^{(\ell)}_j$ is
equivalent to $\int_{z}^1 \absDdist{x}^{(\ell)}_{j-1}  (x) \,\text{d}x \leq
\int_{z}^1 \absDdist{x}^{(\ell)}_{j} (x) \,\text{d}x $ for all $z\in [0,1]$. We have
\begin{align}\label{equ:inequalitya}
& \int_{z}^1 \absDdist{x}^{(\infty)}_{j-1}  (x) \,\text{d}x  \leq  \int_{z}^1 
\absDdist{x}^{(\infty)}_{j} (x) \,\text{d}x +\nonumber   \\
&  \int_{z}^1 \absDdist{x}^{(\infty)}_{j-1} (x) \,\text{d}x - \int_{z}^1  
\absDdist{x}^{(\ell)}_{j-1}  (x) \,\text{d}x \nonumber \\
& - \int_{z}^1 \absDdist{x}^{(\infty)}_{j} (x) \,\text{d}x + \int_{z}^1  
\absDdist{x}^{(\ell)}_{j} (x) \,\text{d}x. 
\end{align} 
By assumption, the sequence $\{\absDdens{\x}^{(\ell)}\}$ converges in the
sense of the Wasserstein metric.
Therefore from property (\ref{lem:blmetricmetrizable}) 
of Lemma~\ref{lem:blmetric}, for all $j\in [-\Lfp+1, 0]$, $\lim_{\ell\to
\infty}\absDdist{x}^{(\ell)}_{j} (x) = \absDdist{x}^{(\infty)}_{j} (x) $
for all $x\in [0,1]$ such that  $\absDdist{x}^{(\infty)}_{j}$ is
continuous at $x$ (in other words, weak convergence is equal to convergence in distribution). 
This implies that for all $j$
$$ \lim_{\ell \to \infty}\Big\vert \int_{z}^1 \absDdist{x}^{(\ell)}_{j}
 (x) \,\text{d}x - \int_{z}^1 \absDdist{x}^{(\infty)}_{j}
 (x) \,\text{d}x\Big\vert = 0$$
so that from (\ref{equ:inequalitya}) we conclude that
\begin{align*}
& \int_{z}^1 \absDdist{x}^{(\infty)}_{j-1} (x) \,\text{d}x  \leq  \int_{z}^1 
\absDdist{x}^{(\infty)}_{j} (x) \,\text{d}x.
\end{align*}

{\em $S$ is compact}:
Note that $S$ is a closed subset of ${\mathcal P}^{\Lfp+1}$, which is
compact since it is the product of compact spaces.  Hence $S$ is
compact as well.

{\em Definition of map $V(\cdot)$}: 
In order to show (via Schauder's FP theorem) that $S$ contains a
FP of DE we need to exhibit a continuous map which maps $S$ into
itself.  Our first step is to define a map, call it $V(\absDdens{\x})$,
which ``approximates'' the DE equation and is well-suited for
applying the FP theorem.  The final step in our proof is then to
show that the FP of the map $V(\absDdens{\x})$ is in fact a FP of DE
itself.

The map $V(\absDdens{\x})$ is constructed as follows.  For $\absDdens{\x}\in
S$, let $U(\absDdens{\x})$ be the map,
$$
(U(\absDdens{\x}))_i = g(\absDdens{x}_{i-w+1}, \dots, \absDdens{x}_{i+w-1}), \quad i\in [-\Lfp,0],
$$
where $\absDdens{x}_i=\Delta_{+\infty}$ for $i<-\Lfp$, and where
$\absDdens{x}_i=\Delta_0$ for $i>0$.
Define $V: S\to S$ as
\begin{align*}
V(\absDdens{\x}) & = 
\begin{cases}
U(\absDdens{\x}) \vconv \absDdens{c}, &\text{s.t.} \; 
\batta(\absDdens{c})=\frac{\xunstab(1)}{2 \batta(U(\absDdens{\x}))}, \\ 
				&		\xunstab(1)/2 \leq \batta(U(\absDdens{\x})), \\
\underline{\bar{\alpha}}(\absDdens{\x}) U(\absDdens{\x}) + 
\underline{\alpha}(\absDdens{\x}) \underline{\Delta}_0, & \text{otherwise}.
\end{cases}
\end{align*}
In words, if $U(\absDdens{\x})$ is ``too large'', upgrade it by an
appropriate channel $\absDdens{c}$.  If, on the other hand, $U(\absDdens{\x})$
is ``too small'' then we take a convex combination with $\Delta_0$.
 In the preceding expressions, terms like $\underline{\bar{\alpha}}
U(\absDdens{\x})$ denote component-wise products, i.e., the result is
a vector of densities, where the $i$-th component is the result of
multiplying the $i$-th component of $U(\absDdens{\x})$ with the scalar
$\bar{\alpha}(\absDdens{\x})_i$. Further, $\underline{\bar{\alpha}}$
is a shorthand for $(1-\underline{\alpha}(\absDdens{\x}))$.

It remains to specify the components of $\underline{\alpha}(\absDdens{\x})$.
Note that $\underline{\alpha}(\absDdens{\x}) \in [0, 1]^{\Lfp+1}$. Further,
we require that its components are increasing and that they are all
either $0$ or $1$, except possibly one. I.e.,
$\underline{\alpha}(\absDdens{\x})$ has the form $(0, 0, \dots, 0,
\alpha_i, 1, \dots, 1)$, where $i \in [-\Lfp, 0]$, and $\alpha_i \in
[0, 1]$. This defines the vector uniquely. Pictorially we can think
of this map in the following way.  We start at component
$(U(\absDdens{\x}))_0$. We take an increasing convex combination with
$\Delta_0$ until the overall Battacharyya constant is equal
to $\xunstab(1)/2$. If this is not sufficient, then we set
$(V(\absDdens{\x}))_0 = \Delta_0$ and repeat this procedure with
component $(U(\absDdens{\x}))_{-1}$, and so on. To apply Schauder's theorem, we need to 
 show that the map $V(\cdot)$ is well-defined and continuous.  

{\em Map $V(\cdot)$ is well defined}: First consider the case
$\batta(U(\absDdens{\x})) \geq \xunstab(1)/2$.  In this case
$\frac{\xunstab(1)}{2 \batta(U(\absDdens{\x}))} \leq 1$.  Since the
Battacharyya parameter is a strictly increasing and continuous function of the
channel\footnote{That the Battacharyya parameter is continuous follows since the
channel family is smooth. Further, since the Battacharyya kernel is strictly
concave and the channel family is ordered by degradation, the Battacharyya parameter is
strictly increasing.},
there exists a unique $\absDdens{c} \in \{\absDdens{c}_\sigma\}$ such that
$\batta(\absDdens{c})=\frac{\xunstab(1)}{2 \batta(U(\absDdens{\x}))}$.
Note also that $U(\absDdens{\x})$ is monotone (spatially)
since $g(\cdot)$ is monotonic (as a function of its arguments) and $\absDdens{\x}$
is monotone.  Consequently, $U(\absDdens{\x})\vconv \absDdens{c}$
is monotone.  Further, from the multiplicative property of the Battacharyya
parameter at the variable node, we get that $\batta(V(\absDdens{\x})) =
\batta(U(\absDdens{\x}))\batta(\absDdens{c}) = \xunstab(1)/2$. It follows that in this case $V(\absDdens{\x}) \in
S$.

Consider next the case $\batta(U(\absDdens{\x})) < \xunstab(1)/2$.  If
we choose $\underline{\alpha}=\underline{1}$ then we get a Battacharyya
 parameter of $1$. Further, the increase in the Battacharyya
 parameter is continuous. Hence there exists an $\underline{\alpha}$ so that
the resulting constellation has Battacharyya constant equal to
$\xunstab(1)/2$.  Also, by construction the resulting constellation
is monotone.  This shows that also in this case $V(\absDdens{\x}) \in
S$. In both the cases above, the map maintains the symmetric nature of the 
 $D$-distributions. 

We summarize, $V$ maps $S$ into itself.  In the rest of the proof,
we will use the notation $d(\absDdens{\x}, \absDdens{\y}) = \sum_{i=-\Lfp}^{0}
d(\absDdens{x}_i, \absDdens{y}_i)$ to denote the Wasserstein distance between two
constellations $\absDdens{\x}$ and $\absDdens{\y}$.

{\em Continuity of map $V(\cdot)$}: We will show that for every
$\absDdens{\x} \in S$ and for any $\varepsilon>0$, there exists a
$\nu>0$ such, that if $\absDdens{\y} \in S$ and $d(\absDdens{\x}, \absDdens{\y}) \leq \nu$, 
then $d( V(\absDdens{\x}), V(\absDdens{\y}))\leq \varepsilon$.
Note that if $d(\absDdens{\x}, \absDdens{\y})\leq \nu$ then 
\begin{itemize}
\item[(i)] $d( U(\absDdens{\x})_i, U(\absDdens{\y})_i) \leq 2(\dl-1)(\dr-1) \nu$, $i \in [-\Lfp, 0]$;
\item[(ii)] $|\batta(U(\absDdens{\x})_i)- \batta(U(\absDdens{\y})_i)| \leq 
\sqrt{4 (\dl-1)(\dr-1) \nu}$, $ i \in [-\Lfp, 0]$;
\item[(iii)] $d(\absDdens{c}_{\absDdens{\x}}, \absDdens{c}_{\absDdens{\y}}) \leq 
2 \sqrt{\frac{2(\Lfp+1)\sqrt{4 (\dl-1)(\dr-1) \nu}}{\xunstab(1)}}$ if
$\batta(U(\absDdens{\x})) \geq  \xunstab(1)/2$ and 
$\batta(U(\absDdens{\y})) \geq  \xunstab(1)/2$.\footnote{We abuse notation slightly to denote the channel associated to $\absDdens{\x}, \absDdens{\y}$ by $\absDdens{c}_{\absDdens{\x}}, \absDdens{c}_{\absDdens{\y}}$, respectively, rather than denoting them by the standard parameterization $\sigma$.}
\end{itemize}
Assertion (i) is equivalent to Lemma~(\ref{lem:sensitivity}) since
if $d(\absDdens{\x}, \absDdens{\y})\leq \nu$ then a fortiori
$d(\absDdens{x}_i, \absDdens{y}_i)\leq \nu$, $i \in [-\Lfp, 0]$.
Assertion (ii) follows from assertion (i) by applying property
(\ref{lem:blmetricwasserboundsbatta}) of Lemma~\ref{lem:blmetric}.
To see assertion (iii) we write
\begin{align*}
|\batta(\absDdens{c}_{\absDdens{\x}})- \batta(\absDdens{c}_{\absDdens{\y}})| 
& = |\frac{\xunstab(1)}{2\batta(U(\absDdens{\x}))}- \frac{\xunstab(1)}{2\batta(U(\absDdens{\y}))}| \\
& \leq \frac{\xunstab(1)}{2} \frac{|\batta(U(\absDdens{\y}))-\batta(U(\absDdens{\x}))|}{\batta(U(\absDdens{\x})) \batta(U(\absDdens{\y}))} \\
& \leq \frac{2 (\Lfp+1)\sqrt{4 (\dl-1)(\dr-1) \nu}}{\xunstab(1)}.
\end{align*}
The last inequality follows from assertion (ii) and $\batta(U(\absDdens{\x})), \batta(U(\absDdens{\y})) \geq \xunstab(1)/2.$
Recall that the channel family is ordered by degradation. We can
therefore apply property (\ref{lem:blmetricentropy}) of
Lemma~\ref{lem:degradationandwasserstein} to prove our claim.

Choosing $\nu$ as a function of $\absDdens{\x}$ and using assertion (ii) above,
we can therefore assume that either
$\batta(U(\absDdens{\x})) \geq  \xunstab(1)/2$ and 
$\batta(U(\absDdens{\y})) \geq  \xunstab(1)/2$ or
$\batta(U(\absDdens{\x})) \leq  \xunstab(1)/2$ and 
$\batta(U(\absDdens{\y})) \leq  \xunstab(1)/2$.
In the first case,
\begin{align*}
d(V(\absDdens{\x}), V(\absDdens{\y}))  = &
d(\absDdens{c}_{\absDdens{\x}} \vconv U(\absDdens{\x}), \absDdens{c}_{\absDdens{\y}} \vconv U(\absDdens{\y})) \\
\stackrel{\text{(\ref{lem:blmetricregularvconv}), Lem.~\ref{lem:blmetric}}}{\leq} & 2d(U(\absDdens{\x}), U(\absDdens{\y})) +
2d(\absDdens{c}_{\absDdens{\x}}, \absDdens{c}_{\absDdens{\y}}) \\
 \stackrel{\text{(i) \& (iii)}}{\leq} &4 (\dl-1)(\dr-1) \nu (\Lfp+1) + \\
& + 
4 \sqrt{\frac{2(\Lfp+1)\sqrt{4 (\dl-1)(\dr-1) \nu}}{\xunstab(1)}}.
\end{align*}

Let us now focus on the second case. 
Let $i^*$ denote the largest integer in $[-\Lfp, 0]$ such that
$\alpha(\absDdens{\x})_{i^*}$ is non-zero. Clearly if $\batta(U(\absDdens{x})) <
\xunstab(1)$, then $i^*\leq 0$, else we set $i^*=1$. Similarly, let $j^*$ be the
corresponding index in $\underline{\alpha}(\absDdens{\y})$. Let us denote $\alpha(\absDdens{\x})_{i^*}=\alpha$ and
$\alpha(\absDdens{\y})_{j^*}=\beta$. Note that $0\leq \alpha, \beta \leq 1$.
Wlog we can assume that $j^*\leq i^*$.  With this we can upper bound $d(V(\absDdens{\x}), V(\absDdens{\y}))$ by,
\begin{align}\label{eq:continuityofmapV}
& \sum_{i=-\Lfp}^{j^*-1} d(U(\absDdens{\x})_i, U(\absDdens{\y})_i) + d(U(\absDdens{\x})_{j^*},\bar{\beta}U(\absDdens{\y})_{j^*}+\beta\Delta_0) \nonumber \\
& + \sum_{i=j^*+1}^{i^*-1} d(U(\absDdens{\x})_j, \Delta_0) + d(\bar{\alpha}U(\absDdens{\x})_{i^*}+\alpha\Delta_0,\Delta_0).
\end{align}
Above we have used that for $i\geq i^*+1$ we have $V(\absDdens{\y})_i = V(\absDdens{\x})_i = \Delta_0.$
In the case $i^*=j^*$, the terms in the interval $[j^*,i^*]$ collapse to $d(\bar{\alpha}U(\absDdens{\x})_{i^*}+\alpha\Delta_0, \bar{\beta}U(\absDdens{\y})_{i^*}+\beta\Delta_0)$.

Let us first consider the case when $j^* < i^*$. Note that
$\batta(V(\absDdens{\x}))=\batta(V(\absDdens{\y}))=\xunstab(1)/2$. This implies
that if we replace the Wasserstein distance by the Battacharyya parameter in \eqref{eq:continuityofmapV} the expression evaluates to 0.
Then writing the $j^*$ term as $
\bar{\beta}(\batta(U(\absDdens{\x})_{j^*}\!)
-\batta(U(\absDdens{\y})_{j^*}))+\beta(\batta(U(\absDdens{\x})_{j^*}) - \batta(\Delta_0))$  we get
\begin{align}\label{eq:differenceofBattasissmall}
&  \beta(1 - \batta(U(\absDdens{\x})_{j^*}))
 \!+\!\sum_{i=j^*\!+\!1}^{i^*\!-\!1} ( 1- \batta(U(\absDdens{\x})_i)) \nonumber \\
 & + (1 - \batta(\bar{\alpha}U(\absDdens{\y})_{i^*}+\alpha\Delta_0)) \nonumber \\ 
&  \leq  \sum_{i=-\Lfp}^{j^*}\vert \batta(U(\absDdens{\x})_i)
-\batta(U(\absDdens{\y})_i)  \vert,
\end{align}
where above we use $\batta(\Delta_0)=1$.

We now continue with \eqref{eq:continuityofmapV}. We use $d(U(\absDdens{\x})_{j^*},\bar{\beta}U(\absDdens{\y})_{j^*}+\beta\Delta_0) \leq d(U(\absDdens{\x})_{j^*},U(\absDdens{\y})_{j^*})+\beta d(U(\absDdens{\x})_{j^*}, \Delta_0)$, (\ref{lem:blmetricbattaboundswasser}) of Lemma~\ref{lem:blmetric} and \eqref{eq:differenceofBattasissmall} to get the upper bound
\begin{align*}
 \sum_{i=-\Lfp}^{j^*} d(& U(\absDdens{\x})_i, U(\absDdens{\y})_i)   \\
 & + 
 \sqrt{2(N+1)}\sqrt{\sum_{i=-\Lfp}^{j^*}\vert \batta(U(\absDdens{\x})_i)
-\batta(U(\absDdens{\y})_i)  \vert}.
  \end{align*}
  Finally using assertions (i) and (ii) above we get that
\begin{align*} 
d(V(\absDdens{\x}), & V(\absDdens{\y})) \leq 2(N+1)(\dl-1)(\dr-1)\nu \\
& +  2(N+1)((\dl-1)(\dr-1)\nu)^{\frac14}.
\end{align*}
For the case when $j^*=i^*$ we have
\begin{align*}
& d(\bar{\alpha}U(\absDdens{\x})_{i^*}+\alpha\Delta_0,
\bar{\beta}U(\absDdens{\y})_{i^*}+\beta\Delta_0)  \\ 
 &\leq
d(\bar{\alpha}U(\absDdens{\x})_{i^*}+\alpha\Delta_0,
\bar{\alpha}U(\absDdens{\y})_{i^*}+\alpha\Delta_0) \\
& +
d(\bar{\alpha}U(\absDdens{\y})_{i^*}+\alpha\Delta_0,
\bar{\beta}U(\absDdens{\y})_{i^*}+\beta\Delta_0)  \\
& \!\leq\!
d(U(\absDdens{\x})_{i^*},U(\absDdens{\y})_{i^*}\!) \!+\!
d(\bar{\alpha}U(\absDdens{\y})_{i^*}\!+\!\alpha\Delta_0,
\bar{\beta}U(\absDdens{\y})_{i^*}\!+\!\beta\Delta_0).
\end{align*}
Wlog we can assume $\beta \geq \alpha$. This implies $\bar{\alpha}U(\absDdens{\y})_{i^*}\!+\!\alpha\Delta_0 \prec
\bar{\beta}U(\absDdens{\y})_{i^*}\!+\!\beta\Delta_0$. Hence from (\ref{lem:blmetricentropy})
 of Lemma~\ref{lem:degradationandwasserstein} we can bound
the second Wasserstein distance above 
by the difference of the Battacharyya parameters. Further,
\begin{align*}
& \vert \batta(\bar{\alpha}U(\absDdens{\y})_{i^*}\!+\!\alpha\Delta_0) - \batta(\bar{\beta}U(\absDdens{\y})_{i^*}\!+\!\beta\Delta_0)\vert \\
& \leq\vert \batta(\bar{\alpha}U(\absDdens{\y})_{i^*}\!+\!\alpha\Delta_0) - \batta(\bar{\alpha}U(\absDdens{\x})_{i^*}\!+\!\alpha\Delta_0) \vert \\
& + \vert \batta(\bar{\alpha}U(\absDdens{\x})_{i^*}\!+\!\alpha\Delta_0) - \batta(\bar{\beta}U(\absDdens{\y})_{i^*}\!+\!\beta\Delta_0)\vert. 
\end{align*}
The first Battacharyya difference on the rhs can be bounded by 
$\vert \batta(U(\absDdens{\y})_{i^*}) - \batta(U(\absDdens{\x})_{i^*}) \vert$. For the second
difference we use same arguments as \eqref{eq:differenceofBattasissmall} to obtain
\begin{align*}
& \vert \batta(\bar{\alpha}U(\absDdens{\x})_{i^*}\!+\!\alpha\Delta_0) - \batta(\bar{\beta}U(\absDdens{\y})_{i^*}\!+\!\beta\Delta_0)\vert \\ 
& \leq\sum_{i=-\Lfp}^{i^*-1}\vert \batta(U(\absDdens{\x})_i)
-\batta(U(\absDdens{\y})_i)  \vert.
\end{align*}
Combining everything with the assertions (i) and (ii), in this case we get
\begin{align*}
d(V(\absDdens{\x}), & V(\absDdens{\y})) \leq 2(N+1)(\dl-1)(\dr-1)\nu \\
& + 2\sqrt{2}\sqrt{N+1} ((\dl-1)(\dr-1)\nu)^{\frac14}.
\end{align*}

{\em Existence of FP of $V(\cdot)$ via Schauder}: We can invoke
Schauder's FP theorem to conclude that $V(\cdot)$ has a FP in $S$,
call it $\absDdens{\x}^*$.

{\em Existence of FP of DE ($U(\cdot)$)}: Let us show
that, as a consequence, DE itself has a FP $(\absDdens{c}^*,
\absDdens{x}^*)$ with the desired properties.

If $\batta(U(\absDdens{\x}^*)) \geq \xunstab(1)/2$, then
$\absDdens{\x}^*=V(\absDdens{\x}^*)= U(\absDdens{\x}^*) \vconv \absDdens{c}^*$
with $\absDdens{c}^* \in \{\absDdens{c}_\sigma\}$. Hence indeed, $(\absDdens{c}^*,
\absDdens{x}^*)$ is a FP of DE.

Consider hence the case $\batta(U(\absDdens{\x}^*)) < \xunstab(1)/2$.  
We will show that it leads to a contradiction.
Recall that in this case 
\begin{align}\label{equ:deequation}
\absDdens{\x}^* & = (1-\underline{\alpha}(\absDdens{\x}^*)) U(\absDdens{\x}^*) + 
\underline{\alpha}(\absDdens{\x}^*) \underline{\Delta}_0,
\end{align}
and that $\absDdens{x}^*_i=\Delta_0$ for $i \geq 1$.

Given a density $\absDdens{x}$ we say that it has a ``BEC component'' of $u$
if $\absDdens{x}$ contains a delta at $0$ of ``weight'' $u$ (i.e., contains a mass of $u$ at $\Delta_0$).  In the sequel
we will think of $u$ as the erasure probability of a binary erasure channel.

Let $\underline{u}$ be the vector of BEC components corresponding to
$\absDdens{\x}^*$.  Since $\batta(U(\absDdens{\x}^*)) < \xunstab(1)/2$ we
know that $\underline{u}$ has some non-trivial components in $[-\Lfp, 0]$,
and by definition of the right boundary, $u_i=1$ for $i>0$.
We claim that for $i \in [-\Lfp, 0]$,
\begin{align}\label{equ:zdensity}
u_i \geq g(u_{i-w+1}, \dots, u_{i+w-1}).
\end{align}
Let us prove this claim immediately.
Extract the BEC component from both the left-hand as well as the right-hand side
of (\ref{equ:deequation}). This gives
\begin{align}
u_i & = (1-\alpha_i) \text{BEC}(U(\absDdens{\x}^*)_i) + \alpha_i \nonumber \\
    & \geq (1-\alpha_i) g(u_{i-w+1}, \dots, u_{i+w-1}) + \alpha_i, \label{equ:ude}
\end{align}
where we wrote $\alpha_i$ as a shorthand for
$\alpha(\absDdens{\x}^*)_i$ and BEC($\cdot$) denotes weight at $\Delta_0$.  To see the second step,
i.e., to see that $\text{BEC}(U(\absDdens{\x}^*)_i) \geq g(u_{i-w+1},
\dots, u_{i+w-1})$, let $\absDdens{\vc}^*$ denote the density at the
output of the check nodes when the input is $\absDdens{\x}^*$. Let
$\underline{v}$ denote the (BEC) density at the output of the check
nodes when the input is $\underline{u}$.  Some thought shows that
$\underline{v}$ is also the BEC component of $\absDdens{\vc}^*$.  In
words, at check nodes the BEC component evolves according to density
evolution -- we get an erasure at the output of a check node if
and only if at least one of the incoming messages is an erasure.
At variable nodes we only get a bound.  If all inputs to a variable
node are erasures then the output is also an erasure, but this is
only a sufficient condition. Thus \eqref{equ:ude} is proved.
If $\alpha_i=1$, then $u_i=1$
and (\ref{equ:zdensity}) is true.  If $\alpha_i<1$, then $u_i \geq
\frac{u_i-\alpha_i}{1-\alpha_i} \geq g(u_{i-w+1}, \dots, u_{i+w-1})$,
where the second step follows from (\ref{equ:ude}). 

Extend the constellation $\underline{u}$ by $\Lfp_3=\lceil
(\Lfp+1)\frac{w}{\frac{\dr}{\dl}-1} \rceil+1$ sections on the right,
with values equal to $1$, and let $\underline{u}^{(0)}$ denote this
constellation.  We claim that $\underline{u}^{(0)}$ has at least
\begin{align*}
\Lfp_4 \geq (\Lfp+1)\Bigl(\frac12 - \frac{c(\dl, \dr)w}{\delta(\Lfp+1)} \Bigr)
\end{align*}
sections on the left with Battacharyya value between $0$ and $\delta$
where $c(\dl, \dr)$ is the constant of Lemma~\ref{lem:transitionlength} and only depends on the
dd.

To prove this claim, we consider our original $\absDdens{\x}^*$ (before
we extracted the BEC components) which was the FP obtained by Schauder's
theorem. We claim that $\absDdens{\x}^*$ has at least $\Lfp_4$ segments
on the left with Battacharyya constant at most $\delta$, where
\begin{align*}
\Lfp_4 \geq \underbrace{(\Lfp+1)}_{(a)}- \underbrace{\frac{\Lfp+1}{2}}_{(b)}- 
\underbrace{\frac{c(\dl, \dr) w}{\delta}}_{(c)}.
\end{align*}
Let us explain each of the terms on the right. There are $\Lfp+1$
segments to start with, which explains (a).  At most $(\Lfp+1)/2$
sections on the right can have a Battacharyya value of $\xunstab(1)$
or larger (since $\batta(\absDdens{\x}^*)=\xunstab(1)/2$). This accounts for the (b) term. Finally, all sections
$i$, with $i < -(\Lfp+1)/2+1$, must be sections where $\absDdens{\x}^*$
fulfills the actual FP equations, i.e., these cannot be sections
where the map $V(\cdot)$ ``pushes'' the constellation up to $\Delta_0$. More precisely, 
we must have $\alpha(\absDdens{\x}^*)_i = 0$ for $i < -(\Lfp+1)/2+1$. Indeed, from 
 construction, starting from the rightmost section, each section is increased all
 the way up to $\Delta_0$ before we move on to the next section on the left. 
 Since the constellation $\absDdens{\x}^*$ has Battacharyya
 parameter equal to $\xunstab(1)/2\leq 1/2$ we conclude that for $i < -(\Lfp+1)/2+1$ we must have
 $\absDdens{x}^*_i = (U(\absDdens{\x}^*))_i$, which is a true FP of DE for the channel $\Delta_0$. 
 Therefore, for these section we can apply
(the Transition Length) Lemma~\ref{lem:transitionlength}
and conclude that there are at most $c(\dl, \dr) w/\delta$
such section which have Battacharyya value between $\delta$ and
$\xunstab(1)$. This is the term (c).

The claim now follows since the BEC component $u_i$ is upper bounded
by the corresponding Battacharyya parameter, $\batta(\absDdens{x}^*_i)$.

Now consider a further constellation $\vc^{(0)}$ on $[-\Lfp, \Lfp_3]$.  We set
$v_i^{(0)}=0$ for all $i\in [-\Lfp,0]$. For $i \in [1, \Lfp_3]$ we set
$\underline{v}^{(0)}$ to the FP of forward DE according to Lemma 22 in
\cite{KRU10}, where the length of the constellation is taken to be  $\Lfp_3-1$, $\epsilon=1$, and
$\chi=\frac12(1-\dl/\dr)$. More precisely, Lemma 22 in \cite{KRU10} says that
if we run forward DE, with free boundary condition, when transmitting over the
BEC with $\epsilon=1$ and $(\dl, \dr, \Lfp_3-1, w)$ coupled ensemble, then for
large enough length, the one-sided FP of forward DE must be  proper (non-trivial and increasing) and we can
lower bound the Battacharyya parameter of the resulting FP. By our choice of
$\Lfp_3$ this FP (on $[1, \Lfp_3]$) has Battacharyya parameter at least $\frac12(1-\dl/\dr)$.
Now since $w\geq 2\dl^3\dr^2$ we have $\Lfp_3=\lceil
(\Lfp+1)\frac{w}{\frac{\dr}{\dl}-1} \rceil+1 \geq N+1$. This implies that $\frac{N_3}{N+1+N_3}\geq \frac12$.
Thus  $\batta(\underline{v}^{(0)})\geq \frac14(1-\dl/\dr).$
Clearly, $\vc^{(0)} \leq \underline{u}^{(0)}$ (component-wise).

Apply forward DE, when transmitting through BEC with $\epsilon=1$, to both constellation with a fixed boundary condition.
More precisely, we have for all $i\in [-\Lfp, \Lfp_3]$ 
$u_i^{(\ell)} = g(u_{i-w+1}^{(\ell-1)},\dots,u_{i+w-1}^{(\ell-1)})$ and  
$v_i^{(\ell)} = g(v_{i-w+1}^{(\ell-1)},\dots,v_{i+w-1}^{(\ell-1)})$.
We keep
$u_i^{(\ell)}=\Delta_0$, and $v_i^{(\ell)}=\Delta_0$ fixed, for all $i > \Lfp_3$ and
$\ell \in \naturals$ and for $i<-\Lfp$ both the constellations have 
sections fixed to $\Delta_{+\infty}$. Recall that $\underline{u}^{(0)}$ is equal
to $\underline{u}$ on $[-\Lfp,0]$ and equal to $1$ for the sections $[1, N_3]$. Because of (\ref{equ:zdensity}), 
we have $\underline{u}^{(0)} \geq \underline{u}^{(1)}$. From the monotonicity of
the DE operator we conclude that the sequence
$\underline{u}^{(\ell)}$ is decreasing and since it is bounded from below
it must converge. Call this limit $\underline{u}^{(\infty)}$. We claim that 
the sequence $\underline{v}^{(\ell)}$ is increasing in $\ell$ and since it is
 bounded from above it must converge. Call this limit
$\underline{v}^{(\infty)}$.  Let us prove the claim that $\underline{v}^{(\ell)}$ is increasing. 
 Indeed, for $i \in [-\Lfp, 0]$, $v^{(1)}_i \geq v^{(0)}_i=0$, for $i
\in[1, \Lfp_3-w+1]$, $v^{(1)}_i= v^{(0)}_i$ (since $v^{(0)}_i$ is an FP in that
region) and for $i \in[\Lfp_3-w+2, \Lfp_3]$, $v^{(1)}_i \geq v^{(0)}_i$ (since
$v^{(0)}_i$ is an FP with free boundary condition and hence replacing the
boundary with $1$ can only increase the value under DE). Again, from the monotonicity of DE we have that
$\underline{v}^{(\ell)}$ in $\ell$.  Since
 $\underline{v}^{(\ell)}$ is increasing and proper we
conclude that $\underline{v}^{(\infty)}$ exists and is proper. Further,
$\underline{v}^{(\infty)} \leq \underline{u}^{(\infty)}$, since
$\underline{v}^{(0)} \leq \underline{u}^{(0)}$ and the ordering is preserved
under iterations of DE.

Since $\batta(\underline{v}^{(\infty)}) \geq \batta(\underline{v}^{(0)}) \geq
\frac14(1-\dl/\dr)$ we claim that there must exists at least $\Lfp_5 =
\Lfp_3\Big(1 - \frac{1+\frac{\dl}{\dr}}{2(1-\delta)} - \frac{wc(\dl,
\dr)}{\Lfp_3\delta}\Big)$ sections, from the right, with Battacharyya parameter
greater than $\xunstab(1)$.  Indeed, this can be obtained by considering the
sections $[1, \Lfp_3]$ of $\underline{v}^{(0)}$ and then using
$\underline{v}^{(\infty)}\geq \underline{v}^{(0)}$.  More precisely, since the
sections $[1, \Lfp_3]$ of $\underline{v}^{(0)}$ form a proper FP, if we let
$N'_3$ denote the number of sections with Battacharyya parameter less than
$\delta$, then we get $\batta(\sum_{j=1}^{\Lfp_3}v^{(0)}_j)\leq N'_3\delta +
N_3 - N_3'$.  Since $\frac1{\Lfp_3}\batta(\sum_{j=1}^{\Lfp_3}v^{(0)}_j)\geq
\frac12(1 - \frac{\dl}{\dr})$ we get $N_3'\leq
\frac{N_3(1+\frac{\dl}{\dr})}{2(1-\delta)}$ and combining with the transition
length Lemma~\ref{lem:transitionlength}, we get the expression for $N_5$.
Further, from the previous discussion, there are at least $\Lfp_4$ values below
$\delta$ on the left. Thus, it is not hard to see that we can simultaneously choose
$\delta>0, w, \Lc\in \naturals, \Msat\in \naturals, \Lfp \in \naturals$ such that
\begin{align*}
& 2(w-1) \leq \Lc,  \\
& \Lc  \leq \Lfp_4, \\
& \Lc+w\leq \Msat  \leq \Lfp_5 \leq \Lfp_3. 
\end{align*}
We summarize, $\underline{v}^{(\infty)}$ is a proper one-sided FP of DE for $\epsilon=1$ with fixed boundary condition 
and $\batta(v^{(\infty)}_{-\Lfp+\Lc})\leq \delta$ and $\batta(v^{(\infty)}_{\Lfp_3-\Msat}) \geq \xunstab(1)$.
But we know from Theorem~\ref{thm:existenceintermediateform} that such a FP,
$\underline{v}^{(\infty)}$, must have a channel value close to $\epsilon^A(\dl,
\dr)$, the area threshold of $(\dl, \dr)$-regular ensemble when transmitting
over BEC.  More precisely, applying
Theorem~\ref{thm:existenceintermediateform} we conclude that the entropy of the
channel of $\underline{v}^{(\infty)}$ must be less than $\epsilon^A(\dl, \dr) +
c(\dl, \dr, \delta, w, \Msat, \Lc)$. Since $\epsilon^A(\dl, \dr) \leq \frac{\dl}{\dr}
< 1$\footnote{For transmission over the BEC using a $(\dl, \dr)$-regular
ensemble, from Theorem 3.120 in \cite{RiU08} we know that $\epsilon^A(\dl, \dr)=\epsilon^{\MAPsmall}(\dl, \dr)$. 
Further the MAP threshold is upper bounded by the
Shannon threshold, $\frac{\dl}{\dr}$.}, we conclude that by choosing $\delta$
small enough and $\Msat, \Lc, \Lfp$ large enough, $c(\dl, \dr, \delta, w, \Msat, \Lc)$ can be
made arbitrarily small and hence the channel of $\underline{v}^{(\infty)}$ is
strictly less than 1, leading to a contradiction since we started with $\epsilon=1$.  This contradiction tells us
that we cannot have $\batta(U(\absDdens{\x}^*)) < \xunstab(1)/2$ when we
apply the Schauder theorem. Hence the FP must be a true FP of DE.
\end{IEEEproof}

\end{appendices}
\bibliographystyle{IEEEtran}
\bibliography{lth,lthpub}
\newpage
\tableofcontents
\end{document}

%% file: definition.tex
\newtheorem{theorem}{Theorem}
\newcommand{\btheo}{\begin{theorem}}
\newcommand{\etheo}{\end{theorem}}
\newcommand{\bproof}{\begin{proof}}
\newcommand{\eproof}{\end{proof}}
\newtheorem{definition}[theorem]{Definition}
\newcommand{\bdefi}{\begin{definition}}
\newcommand{\edefi}{\end{definition}}
\newtheorem{fact}[theorem]{Fact}
\newcommand{\bprop}{\begin{fact}}
\newcommand{\eprop}{\end{fact}}
\newtheorem{corollary}[theorem]{Corollary}
\newcommand{\bcor}{\begin{corollary}}
\newcommand{\ecor}{\end{corollary}}
\newtheorem{example}[theorem]{Example}
\newcommand{\bex}{\begin{example}}
\newcommand{\eex}{\end{example}}
\newtheorem{lemma}[theorem]{Lemma}
\newcommand{\blemma}{\begin{lemma}}
\newcommand{\elemma}{\end{lemma}}
\newtheorem{remark}[theorem]{Remark}
\newcommand{\bremark}{\begin{remark}}
\newcommand{\eremark}{\end{remark}}
\newtheorem{conj}[theorem]{Conjecture}
\newcommand{\bconj}{\begin{conj}}
\newcommand{\econj}{\end{conj}}

\newcommand{\naturals}{\ensuremath{\mathbb{N}}}
\newcommand{\integers}{\ensuremath{\mathbb{Z}}}
\newcommand{\expectation}{\ensuremath{\mathbb{E}}}

\newcommand{\defas}{\ensuremath{\stackrel{{\vartriangle}}{=}}} 

\def\0{{\tt 0}} 
\def\1{{\tt 1}} 
\def\?{{\tt *}} 
\newcommand{\ddp}{dd~pair~} 
\newcommand{\ledge}{\ensuremath{\lambda}} 
\newcommand{\redge}{\ensuremath{\rho}} 
\newcommand{\ih}{\ensuremath{{\tt{h}}}} 
\newcommand{\Tc}{T}
\newcommand{\dee}{{\text d}}
\newcommand{\BP}{\ensuremath{\text{BP}}} 
\newcommand{\MAPsmall}{\ensuremath{\text{\tiny MAP}}} 
\newcommand{\BPsmall}{\ensuremath{\text{\tiny BP}}} 
\newcommand{\xl}{\ensuremath{{\tt{e}}}}

\newcommand{\qed}{{\hfill \footnotesize $\blacksquare$}}
\renewcommand{\mid}{\,|\,}

\newcommand{\dens}[1]{\mathsf{#1}}
\newcommand{\Ldens}[1]{\dens{#1}}

\newcommand{\Ddens}[1]{\mathfrak{{#1}}}
\newcommand{\BEC}{\ensuremath{\text{BEC}}}

\newcommand{\BSC}{\ensuremath{\text{BSC}}}
\newcommand{\BAWGNC}{\ensuremath{\text{BAWGNC}}}
\newcommand{\BMS}{\ensuremath{\text{BMS}}}

\newcommand{\BSCsmall}{\ensuremath{\text{\tiny BSC}}}

\newcommand{\BECsmall}{\ensuremath{\text{\tiny BEC}}}

\newcommand{\Ddist}[1]{\mathfrak{\MakeUppercase{#1}}}
\newcommand{\absDdist}[1]{\absd{\mathfrak{\MakeUppercase{#1}}}}
\newcommand{\absDdens}[1]{\absd{\Ddens{#1}}}

\newcommand{\absd}[1]{|#1|}

\newcommand{\dr}{d_r}
\newcommand{\dl}{d_l}

\newcommand{\Td}{R}

\DeclareMathOperator{\perr}{\mathfrak{E}}

\newcommand{\ind}{\mathbbm{1}}

\DeclareMathOperator{\batta}{\mathfrak{B}}      
\newcommand{\entropy}{\text{H}}

\newcommand{\Lip}{\text{Lip}}

\newcommand{\Lc}{L} 
\newcommand{\Lfp}{N} 

\newcommand{\vconv}{\circledast}
\newcommand{\cconv}{\boxast}

\DeclareMathOperator{\field}{\ensuremath{\mathbb{F}}}
\newcommand{\moment}{m}

%% file: ps/ebpexit36bec_arxiv.tex
\setlength{\unitlength}{1.0bp}%
\begin{picture}(274,108)(-14,-8)
\put(0,0)
{
\put(0,0){\includegraphics[scale=1.0]{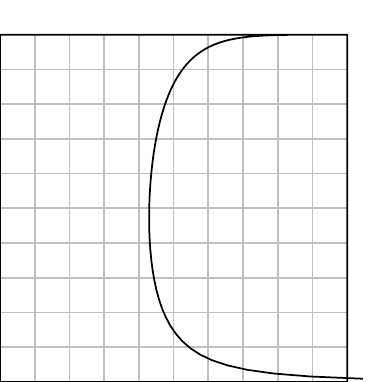}}
\small
\multiputlist(0,-8)(20,0)[cb]{$~$,$0.2$,$0.4$,$0.6$,$0.8$}
\multiputlist(-14,0)(0,20)[l]{$~$,$0.2$,$0.4$,$0.6$}
\put(-14,-8){\makebox(0,0)[lb]{$0.0$}}
\put(100,-8){\makebox(0,0)[rb]{$\epsilon$}}
\put(102,2){\makebox(0,0)[lb]{\rotatebox{90}{$(1, \xunstab(1))$}}}
}
\put(135,0)
{
	\put(0,0){\includegraphics[scale=1.0]{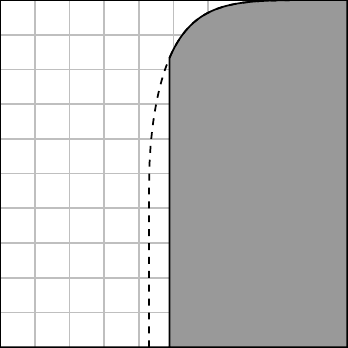}}
	\small
	\multiputlist(0,-8)(20,0)[cb]{$~$,$0.2$,$0.4$,$0.6$,$0.8$}
	\multiputlist(-14,0)(0,20)[l]{$~$,$0.2$,$0.4$,$0.6$}
	\put(-14,-8){\makebox(0,0)[lb]{$0.0$}}
	\put(100,-8){\makebox(0,0)[rb]{$\epsilon$}}
	\put(41,2){\makebox(0,0)[rb]{\rotatebox{90}{$\epsilon^{\BPsmall}$}}}
	\put(51,2){\makebox(0,0)[lb]{\rotatebox{90}{$\epsilon^{\MAPsmall}$}}}
	\put(75, 60){\makebox(0,0)[t]{$\int =\frac12$}}
}
\end{picture}

%% file: ps/ebpexit36.tex
\setlength{\unitlength}{1.0bp}%
\begin{picture}(274,108)(-14,-8)
\put(0,0)
{
\put(0,0){\includegraphics[scale=1.0]{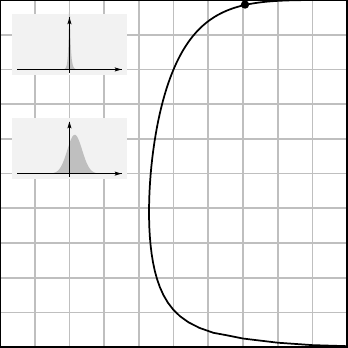}}
\small
\multiputlist(0,-8)(20,0)[cb]{$~$,$0.2$,$0.4$,$0.6$}
\multiputlist(-14,0)(0,20)[l]{$~$,$0.2$,$0.4$,$0.6$}
\put(-14,-8){\makebox(0,0)[lb]{$0.0$}}
\put(100,-10){\makebox(0,0)[rb]{$\entropy(\Ldens{c}_{\sigma})$}}
\put(-14,100){\makebox(0,0)[lt]{\rotatebox{90}{$G(\Ldens{c}_{\sigma}, \cdot)$}}}
\put(30,90){\makebox(0,0)[b]{$\Ldens{x}_\sigma$}}
\put(30,60){\makebox(0,0)[b]{$\Ldens{c}_{\sigma}$}}
}
\put(135,0)
{
\put(0,0){\includegraphics[scale=1.0]{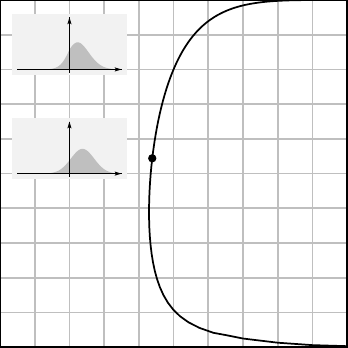}}
\small
\multiputlist(0,-8)(20,0)[cb]{$~$,$0.2$,$0.4$,$0.6$}
\multiputlist(-14,0)(0,20)[l]{$~$,$0.2$,$0.4$,$0.6$}
\put(-14,-8){\makebox(0,0)[lb]{$0.0$}}
\put(100,-10){\makebox(0,0)[rb]{$\entropy(\Ldens{c}_{\sigma})$}}
\put(-14,100){\makebox(0,0)[lt]{\rotatebox{90}{$G(\Ldens{c}_{\sigma}, \cdot)$}}}
\put(30,90){\makebox(0,0)[b]{$\Ldens{x}_\sigma$}}
\put(30,60){\makebox(0,0)[b]{$\Ldens{c}_{\sigma}$}}
}
\end{picture}

%% file: ps/maxwellbsc.tex
\setlength{\unitlength}{1.0bp}%
\begin{picture}(114,108)(-14,-8)
\put(0,0)
{
\put(0,0){\includegraphics[scale=1.0]{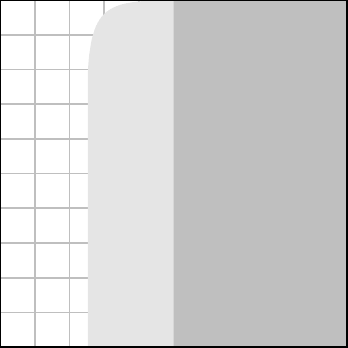}}
\small
\multiputlist(0,-8)(20,0)[cb]{$~$,$0.2$,$0.4$,$0.6$}
\multiputlist(-14,0)(0,20)[l]{$~$,$0.2$,$0.4$,$0.6$}
\put(-14,-8){\makebox(0,0)[lb]{$0.0$}}
\put(100,-10){\makebox(0,0)[rb]{$H(\Ldens{c}_{\sigma})$}}
\put(-14,100){\makebox(0,0)[lt]{\rotatebox{90}{$G(\Ldens{c}_{\sigma}, \Ldens{x})$}}}
\put(24,3){\makebox(0,0)[rb]{\rotatebox{90}{\scriptsize BP threshold}}}
\put(51,3){\makebox(0,0)[lb]{\rotatebox{90}{\scriptsize Area threshold}}}
}
\end{picture}

%% file: ps/maxwell.tex
\setlength{\unitlength}{1.0bp}%
\begin{picture}(114,108)(-14,-8)
\put(0,0)
{
\put(0,0){\includegraphics[scale=1.0]{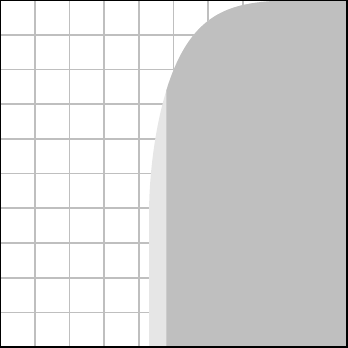}}
\small
\multiputlist(0,-8)(20,0)[cb]{$~$,$0.2$,$0.4$,$0.6$}
\multiputlist(-14,0)(0,20)[l]{$~$,$0.2$,$0.4$,$0.6$}
\put(-14,-8){\makebox(0,0)[lb]{$0.0$}}
\put(100,-10){\makebox(0,0)[rb]{$H(\Ldens{c}_{\sigma})$}}
\put(-14,100){\makebox(0,0)[lt]{\rotatebox{90}{$G(\Ldens{c}_{\sigma}, \Ldens{x})$}}}
\put(42,3){\makebox(0,0)[rb]{\rotatebox{90}{\scriptsize BP threshold}}}
\put(49,3){\makebox(0,0)[lb]{\rotatebox{90}{\scriptsize Area threshold}}}
}
\end{picture}

%% file: ps/lrLexit.tex
\setlength{\unitlength}{0.95bp}%
\begin{picture}(260,120)(0,0)
\put(0,0)
{
\put(0,0){\rotatebox{0}{\includegraphics[scale=0.95]{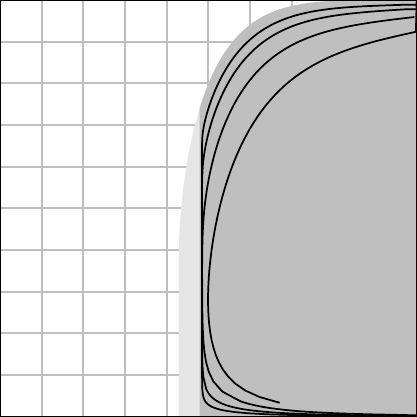}}}
\put(73,79){\makebox(0,0)[l]{\rotatebox{0}{\scriptsize $L\!=\!4$}}}
\put(50,110){\makebox(0,0)[l]{\rotatebox{0}{\scriptsize $L\!=\!32$}}}
\put(120,-10){\makebox(0,0)[rb]{\small $H(\Ldens{c}_{\sigma})$}}
}
\put(140,0)
{
\put(0,0){\rotatebox{0}{\includegraphics[scale=0.95]{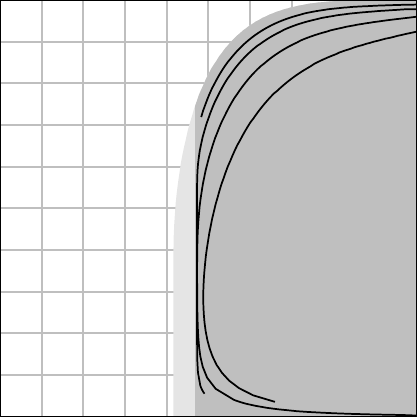}}}
\put(73,79){\makebox(0,0)[l]{\rotatebox{0}{\scriptsize $L\!=\!4$}}}
\put(50,110){\makebox(0,0)[l]{\rotatebox{0}{\scriptsize $L\!=\!32$}}}
\put(120,-10){\makebox(0,0)[rb]{\small $H(\Ldens{c}_{\sigma})$}}
}
\end{picture}

%% file: ps/one-sided_fixed_point_arxiv.tex
\setlength{\unitlength}{1.0bp}%
\begin{picture}(240,65)(0,0)
\put(0,0)
{
\put(0,0){\rotatebox{0}{\includegraphics[scale=1.0]{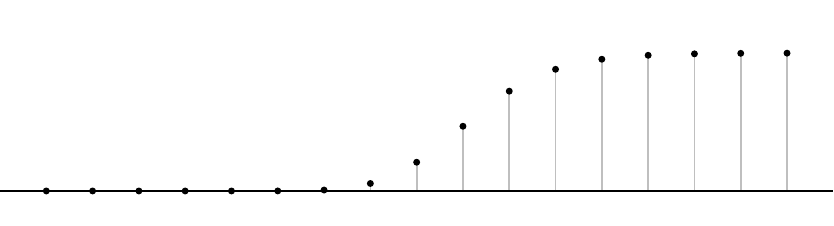}}}
\footnotesize
\multiputlist(8,0)(27.5,0)[b]{$\text{-}16$,$\text{-}14$,$\text{-}12$,$\text{-}10$,$\text{-}8$,$\text{-}6$,$\text{-}4$,$\text{-}2$,0}
}
\end{picture}

%% file: ps/accordeonfp.tex
\setlength{\unitlength}{1.0bp}%
\begin{picture}(240,90)(0,0)
\put(0,0)
{
\put(0,0){\rotatebox{0}{\includegraphics[scale=1.0]{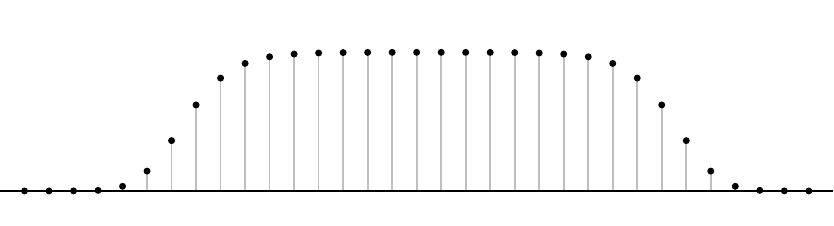}}}
\footnotesize
\multiputlist(6,0)(14.2,0)[b]{$\text{-}16$,$\text{-}14$,$\text{-}12$,$\text{-}10$,$\text{-}8$,$\text{-}6$,$\text{-}4$,$\text{-}2$,0,2,4,6,8,10,12,14,16}
\put(0,55){\rotatebox{0}{\includegraphics[scale=1.0]{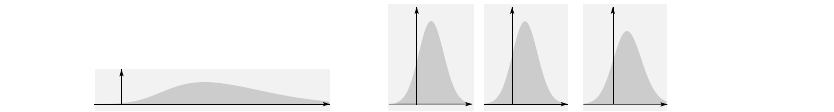}}}
\footnotesize
\multiputlist(6,87)(28.3,0)[b]{$$,$\text{-}12$,,,0,4,$8$}
}
\end{picture}

%% file: ps/crossings_arxiv.tex
\setlength{\unitlength}{1.0bp}%
\begin{picture}(114,108)(-14,-8)
\put(0,0)
{
\put(0,0){\includegraphics[scale=1.0]{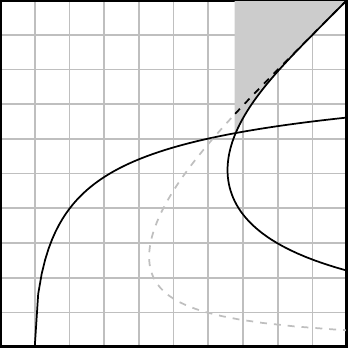}}
\small
\multiputlist(0,-8)(20,0)[cb]{$~$,$0.2$,$0.4$}
\multiputlist(-14,0)(0,20)[l]{$~$,$0.2$,$0.4$}
\put(-14,-8){\makebox(0,0)[lb]{$0.0$}}
\put(100,-10){\makebox(0,0)[rb]{$\batta(\Ldens{c})=g(\beta)$}}
\put(-14,100){\makebox(0,0)[lt]{\rotatebox{90}{$\batta(\Ldens{x})=\beta$}}}
}
\end{picture}